\documentclass[12pt]{scrbook}
\usepackage[utf8]{inputenc}
\usepackage[OT2,OT1]{fontenc}
\usepackage{threeparttable}
\newcommand\cyr{%
\renewcommand\rmdefault{wncyr}%
\renewcommand\sfdefault{wncyss}%
\renewcommand\encodingdefault{OT2}%
\normalfont
\selectfont}
\DeclareTextFontCommand{\textcyr}{\cyr}
\usepackage{a4wide}
\usepackage{url}
\usepackage{algorithmic}
\usepackage[ruled,lined]{algorithm2e}
\usepackage{subfig}
\usepackage{amsmath}
\usepackage{amsthm}
\usepackage{color}
\definecolor{light}{gray}{0.85}
\definecolor{heavy}{gray}{0.30}
\definecolor{mygray}{gray}{.5}
\usepackage{amsfonts}
\usepackage{dsfont}
\usepackage{amssymb}
\usepackage{makeidx}		
 \usepackage[pdftex]{graphicx}
 \newtheorem{theoremS}{Theorem}[section]

\usepackage[table]{xcolor} 
\usepackage{rotating}
\usepackage{multirow}

\usepackage{shortvrb}

\renewcommand{\dotfill}{\leaders\hbox to 5pt{\hss.\hss}\hfill}

\newcounter{definitioncounter}
\newcounter{assertioncounter}
\newcounter{examplecounter}
\newcounter{exercisecounter}

\title{Some aspects of physical prototyping in Pervasive Computing\\\color{white}--\color{black} \\ \Large Distributed adaptive beamforming, Device-free recognition of activities from RF, Secure keys from ambient audio and calculation of mathematical functions on the wireless channel}
\author{\begin{minipage}[t]{10.5cm}
\begin{minipage}[t]{10.5cm}\vspace{1cm}  \centering
Habilitationsschrift zur Erlangung der Lehrbefugnis im Fach Informatik an der Karl-Friedrich-Gauss-Fakultät der Technischen Universität Carolo Wilhelmina zu Braunschweig\vspace{3cm}\end{minipage}\\
\begin{minipage}[t]{5cm}\normalsize \sffamily		
\begin{tabular}{rl}
Vorgelegt von: & Dr. rer. nat. Stephan Sigg\\
 	&sigg$@$ibr.cs.tu-bs.de\\
	&Institut für Betriebssysteme und Rechnerverbund\\
	&Technische Universität Braunschweig\\[.3cm]
	&Braunschweig im Februar 2015\\[.5cm]
	Habilitation: & 17.03.2017
\end{tabular}
 \end{minipage}
\end{minipage}
 }
\date{ }
\makeindex
\begin{document}
\pagestyle{empty}
\maketitle
\pagestyle{plain}

\tableofcontents
\setcounter{definitioncounter}{0}\setcounter{examplecounter}{0}
\setcounter{exercisecounter}{0}

\chapter{Introduction}
This document summarises the results of several research campaigns over the past seven years. 
The main connecting theme is the physical layer of widely deployed sensors in Pervasive Computing domains.
In particular, we have focused on the RF-channel or on ambient audio.
Instead of plugging together existing technologies to solve a particular task, we have been re-prototyping the use and interaction via these interfaces for a particular purpose. 
In particular, the initial problem from which we started this work was that of distributed adaptive transmit beamforming. We have been looking for a simple method to align the phases of jointly transmitting nodes (e.g. sensor or IoT nodes).
The algorithmic solution to this problem was to implement a distributed random optimisation method on the participating nodes in which the transmitters and the receiver follow an iterative question-and-answer scheme. 
In this scheme, phases of transmitters are randomly altered.
The algorithm works on the physical layer, not utilising existing protocols. 
We have been able to derive sharp asymptotic bounds on the expected optimisation time of an evolutionary random optimiser and presented an asymptotically optimal approach (cf. section~\ref{sectionOriginalBF01})\cite{4022}.
The latter approach, however, requires richer feedback from the transmit devices which restricts its application. 
Given the strong unimodality of the underlying search space, we then derived improved sharp bounds on a local random search approach (cf. section~\ref{sectionOriginalBF02})\cite{Beamforming_Sigg_2014}.

One thing that we have learned from the work on these physical layer algorithms was that the signals we work on are fragile and perceptive to physical environmental changes.
These could be obstacles such as furniture, opened or closed windows or doors as well as movement of individuals.

This observation motivated us to view the wireless interface as a sensor for environmental changes in Pervasive Computing environments.
Pioneering this field of device-free recognition of activities and situations, we have demonstrated the feasibility of this sensing paradigm with software radios and also sensor nodes (cf. section~\ref{sectionOriginalRF01})\cite{Pervasive_Sigg_2012}.
The essential novelty of this sensing paradigm enabled by looking at the physical layer directly is that monitored entities do not need to be equipped with any hardware or with any part of the sensing system.
By reflecting and blocking RF-signals, the monitored entities are implicitly integral parts of the sensing system.
Improving the recognition accuracy of these systems, we could show that an accurate recognition of activities is also possible utilising only ambient signals (i.e. not controlling the transmitter) (cf. section~\ref{sectionOriginalRF02})\cite{Pervasive_Shi_2014}.
In particular, ambient FM radio signals have been utilised.
Finally, we could also demonstrate that a (lower accuracy) recognition is also possible on consumer devices, such as smartphones and that software radios are not mandatory for this recognition scheme. 
In this work, gestures and activities have been distinguished by analysing the fluctuation in the received signal strength indicator (RSSI) of received IEEE 802.11 packets (cf. section~\ref{sectionOriginalRF03})\cite{RFSensing_Sigg_2014}.

Another use of physical layer RF-signals is for security applications (e.g.~\cite{Cryptography_Mathur_2011}).
The essential idea is that the signal fluctuations at two distinct physical locations are uncorrelated given that the devices are separated by at least half the wavelength of the signal.
Then, close devices can use their correlated signal for the generation of common secure keys whereas devices which are farther apart are not able to generate identical keys following the same protocol since their input to generate the keys is uncorrelated.
The security of this scheme relies on the difficulty to predict a channel at a particular remote physical location.

However, due to the high frequency of RF-signals, devices can be separated by few centimeters at most in order to generate common secure keys by this approach.

Instead, we exploit ambient audio which shares a number of properties with RF but operates at a lower frequency, so that higher physical separation of devices is acceptable.
In particular, we presented a scheme for the generation of secure cryptographic keys from ambient audio (cf. section~\ref{sectionOriginalSE01})~\cite{Cryptography_Schuerman_2011}. 
The approach has been exploited in various environmental conditions and we have not been able to find BIAS using statistical tests.
Later, we ported this approach towards common smartphones (cf. section~\ref{sectionOriginalSE02})~\cite{Cryptography_Nguyen_2012-2}.
In this, hardware inconsistencies and insufficient synchronisation on common smartphone platforms had to be solved algorithmically. 

This collection of applications demonstrates the potential of physical prototyping of Pervasive Computing applications. 
Existing protocols simplify communication among and with devices and the interaction with common interfaces. 
However, such protocols also introduce overhead and abstract from available information.
While this enables a wide and easy application of the protocols, some information is only available on the physical level and some level of efficiency only possible by re-prototyping the physical layer. 

We are currently working to further push these mentioned directions and novel fields of physical prototyping as detailed, for instance, in \cite{FunctionComputation_Sigg_2012,InNetworkProcessing_Jakimovski_2012,FunctionComputation_Sigg_2013}.
In particular, the calculation of mathematical operations on the wireless channel at the time of transmission appears to contain good potential for gains in efficiency for communication and computation in Pervasive Computing domains.

\section{Contribution}
This thesis presents the work of seven publications at international Journals or Conferences between 2010 and 2014. 
In particular, these are
\begin{enumerate}
\item[\cite{4022}] Stephan Sigg, Rayan Merched El Masri and Michael Beigl: Feedback based closed-loop carrier synchronisation: A sharp asymptotic bound, an asymptotically optimal approach, simulations and experiments, in IEEE Transactions on Mobile Computing (TMC), 2011 \\(DOI: http://dx.doi.org/10.1109\%2FTMC.2011.21)
\item[\cite{Beamforming_Sigg_2014}] Stephan Sigg: A fast binary feedback-based distributed adaptive carrier synchronisation for transmission among clusters of disconnected IoT nodes in smart spaces, Elsevier Journal on Ad Hoc Networks, vol. 16, May 2014, pp. 120-130 \\(DOI: http://10.1016/j.adhoc.2013.12.006)
\item[\cite{Pervasive_Sigg_2012}] Stephan Sigg, Markus Scholz, Shuyu Shi, Yusheng Ji and Michael Beigl: RF-sensing of activities from non-cooperative subjects in device-free recognition systems using ambient and local signals, in IEEE Transactions on Mobile Computing (TMC), Feb. 2013, vol. 13, no. 4 \\(DOI: http://doi.ieeecomputersociety.org/10.1109/TMC.2013.28)
\item[\cite{Pervasive_Shi_2014}] Shuyu Shi, Stephan Sigg, Wei Zhao, and Yusheng Ji: Monitoring of Attention from Ambient FM-radio Signals, IEEE Pervasive Computing, Los Alamitos, CA, USA, IEEE Computer Society, Jan-Mar 2014, vol. 13, no. 1, pp. 30-36, 2014 \\(DOI: http://dx.doi.org/10.1109/MPRV.2014.13)
\item[\cite{RFSensing_Sigg_2014}] Stephan Sigg, Ulf Blanke and Gerhard Troester: The Telepathic Phone: Frictionless Activity Recognition from WiFi-RSSI, IEEE International Conference on Pervasive Computing and Communications (PerCom), Budapest, Hungary, March 24-28, 2014 (DOI: http://dx.doi.org/10.1109/PerCom.2014.6813955)
\item[\cite{Cryptography_Schuerman_2011}] Dominik Schuermann and Stephan Sigg: Secure communication based on ambient audio, in IEEE Transactions on Mobile Computing (TMC), Feb. 2013, vol. 12 no. 2 (DOI: http://doi.ieeecomputersociety.org/10.1109/TMC.2011.271)
\item[\cite{Cryptography_Nguyen_2012-2}] Ngu Nguyen, Stephan Sigg, An Huynh and Yusheng Ji: Pattern-based Alignment of Audio Data for Ad-hoc Secure Device Pairing, in 2012 16th International Symposium on Wearable Computers (ISWC), pp.88-91, 18-22 June 2012 \\(DOI: http://dx.doi.org/10.1109/ISWC.2012.14)
\end{enumerate}

\chapter{Original work}
\section[Feedback based closed-loop carrier synchronisation: A sharp asymptotic bound, an asymptotically optimal approach, simulations and experiments]{Feedback based closed-loop carrier synchronisation: A sharp asymptotic bound, an asymptotically optimal approach, simulations and experiments \footnote{Originally published as 'Stephan Sigg, Rayan Merched El Masri and Michael Beigl: Feedback based closed-loop carrier synchronisation: A sharp asymptotic bound, an asymptotically optimal approach, simulations and experiments, in IEEE Transactions on Mobile Computing (TMC), 2011
(DOI: http://dx.doi.org/10.1109\%2FTMC.2011.21)' 1536-1233/11/\$26.00 \copyright 2011 IEEE  Published by the IEEE CS, CASS, ComSoc, IES, SPS}}\label{sectionOriginalBF01}
We derive an asymptotically sharp bound on the synchronisation speed of a randomised black box optimisation technique for closed-loop feedback based distributed adaptive beamforming in wireless sensor networks.
We also show that the feedback function that guides this synchronisation process is weak multimodal. 
Given this knowledge that no local optimum exists, we consider an approach to locally compute the phase offset of each individual carrier signal.
With this design objective an asymptotically optimal algorithm is derived.
Additionally, we discuss the concept to reduce the optimisation time and energy consumption by hierarchically clustering the network into subsets of nodes that achieve beamforming successively over all clusters.
For the approaches discussed we demonstrate their practical feasibility in simulations and experiments.

\subsection{Introduction}\label{sectionIntroductionBF01}
In recent years, sensor nodes of extreme tiny size have been envisioned \cite{5911,5912,5137}.
In \cite{5916}, for example, applications for square-millimetre sized nodes that seamlessly integrate into an environment are detailed.
At these small form-factors transmission power of wireless nodes is restricted to several microwatts.
Communication between a single node and a remote receiver is then only feasible at short distances.
It is possible, however, to increase the maximum transmission range by cooperatively transmitting information from distinct nodes of a network \cite{5908,5884}.
Cooperation can increase the capacity and robustness of a network of transmitters \cite{5893,5894} and decreases the average energy consumption per node \cite{5885,5888,5939}.

Related research branches are cooperative transmission \cite{5907}, collaborative transmission \cite{4019,4020}, distributed adaptive beamforming \cite{Mudumbai_2009,5919,Barriac_2004,5923}, collaborative beamforming \cite{5930} or cooperative/virtual MIMO for wireless sensor networks \cite{5940,5937,5941,5938}.
One approach is to utilise neighbouring nodes as relays \cite{5898,5899,5900} as proposed by Cover and El Gamal in \cite{5901}.
Cooperative transmission is then achieved by multi-hop \cite{5903,5807,5909} or data flooding \cite{Ochiai_2005,5905,5844,5843} approaches.
The general idea of multi-hop relaying based on the physical channel is to retransmit received messages by a relay node so that the destination will receive not only the message from the source node but also from the relay.
In data flooding approaches, a node will retransmit a received message at its reception. 
It has been shown that the approach outperforms non-cooperative multi-hop schemes significantly.
In particular, the transmission time is reduced compared to traditional transmission protocols \cite{5886}.

In these approaches, nodes are not tightly synchronised and transmission may be asynchronous.
Synchronous transmission, however, is achieved by virtual MIMO techniques.
In these implementations, identical RF carrier signal components from various transmitters that function as a distributed beamformer are superimposed.
When the relative phase offset of these carrier signal components at a remote receiver is small, the signal strength of the received sum signal is improved.
In virtual MIMO for wireless sensor networks, single antenna nodes are cooperating to establish a multiple antenna wireless sensor network \cite{5937,5940,5941}.
Virtual MIMO has capabilities to adjust to different frequencies and is highly energy efficient \cite{5938,5939}.
However, the implementation of MIMO capabilities in WSNs requires accurate time synchronisation, complex transceiver circuits and signal processing that might surpass the power consumption and processing capabilities of simple sensor nodes.

Other solutions proposed are open-loop synchronisation methods such as round-trip synchronisation \cite{5931,5932,5933}.
In this scheme, the destination transmits beacons in opposed directions along a multi-hop circle in which each of the nodes appends its part of the overall message to the beacons.
Beamforming is achieved when the processing time along the multi-hop chain is identical in both directions.
This approach, however, does not scale with the size of a network.

Closed loop feedback based approaches include full-feedback techniques, in which carrier synchronisation is achieved in a master-slave manner.
The phase-offset among the carrier signals of destination nodes is corrected by a receiver node.
Diversity between RF-transmit signal components is achieved over CDMA channels \cite{5934}.
This approach is applicable only to small network sizes and requires sophisticated processing capabilities at the source nodes.

A more simple and less resource demanding implementation is the one-bit feedback based closed-loop synchronisation considered in \cite{5934,5920}.
The authors describe an iterative process in which $n$ source nodes $i\in[1,\dots,n]$ randomly adapt the phases $\gamma_i$ of their carrier signal $\Re\left(m(t)e^{j(2\pi (f_c+f_i)t+\gamma_i)}\right)$.
Here, $m(t)$ is the transmit message and $f_i$ denotes the frequency offset of node $i$ to a common carrier frequency $f_c$. 
Initially, i.i.d. phase offsets $\gamma_i$ of carrier signals are assumed.
When a receiver requests a transmission from the network, carrier phases are synchronised in an iterative process.
\begin{enumerate}
	\item Each source node $i$ adjusts its carrier phase offset $\gamma_i$ and frequency offset $f_i$ randomly.
	\item The source nodes transmit to the destination simultaneously as a distributed beamformer.
	\item The receiver estimates the level of phase synchronisation of the received sum signal (for instance by the SNR).
	\item This value is broadcast as a feedback to the network.
	Nodes interpret this feedback and adapt the phase of their carrier signal accordingly.
\end{enumerate}
These four steps are iterated repeatedly until a stop criterion is met (e.g. maximum iteration count or sufficient synchronisation). 
Fig.~\ref{FigureGeneralProcedure} illustrates this procedure.
\begin{figure}
	\centering
	\includegraphics[scale=.4]{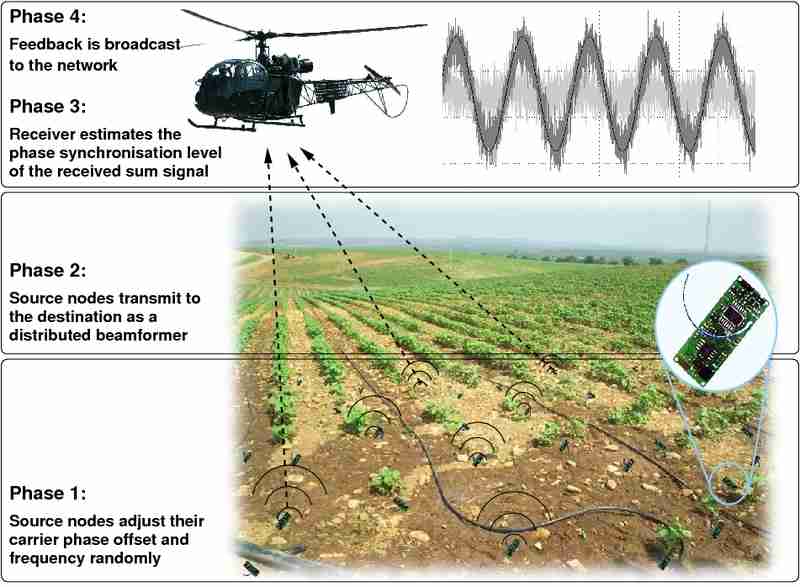}
	\caption{Schematic illustration of feedback based distributed adaptive beamforming in wireless sensor networks {\scriptsize(\copyright 2011 IEEE  Published by the IEEE CS, CASS, ComSoc, IES, SPS)}}
	\label{FigureGeneralProcedure}
\end{figure}
It has been studied by different authors \cite{Mudumbai_2010b,Seo_2008,Bucklew_2008,4019}.
The distinct approaches proposed differ in the implementation of the first and the fourth step specified above.
The authors of \cite{Bucklew_2008} show that it is possible to reduce the count of transmitters in a random process and still achieve sufficient synchronisation among all nodes.

In \cite{Mudumbai_2010b,Seo_2008,Bucklew_2008} a process is described in which each node alters its carrier phase offset $\gamma_i$ according to a normal distribution with small variance in step one.
In \cite{4019} a uniform distribution is utilised instead but the probability for one node to alter the phase offest of its carrier signal is low.
We show in section~\ref{sectionResultsBF01} that both approaches achieve a similar performance.
Only in \cite{Seo_2008} not only the phase but also frequency is adapted.

Significant differences among these approaches also apply to the feedback and the reactions of nodes in step four.
In  \cite{Mudumbai_2009,Seo_2008,Bucklew_2008} a one-bit feedback is utilised.
Nodes sustain their phase modifications when the feedback has improved and otherwise reverse them.
In \cite{Bucklew_2008} it was shown that the optimisation time is improved by a factor of two when a node as response to a negative feedback from the receiver applies a complementary phase offset instead of simply reversing its modification.
In \cite{4019}, authors suppose to utilise more than one bit as feedback so that parameters of the optimisation can be adapted with regard to the optimisation progress.

The strength of feedback based closed-loop distributed adaptive beamforming in wireless sensor networks is its simplicity and low processing requirements that make it feasible for the application in networks of tiny sized, low power and computationally restricted sensor nodes.

We study aspects of this transmission scheme and derive sharp asymptotic lower and upper bounds on the expected optimisation time of a common implementation in section~\ref{sectionAsymptoticBounds}.
Together with these bounds we show that the feedback function is weak multimodal so that no local optimum exists.
By small modifications of the common algorithm, however, further improvements in the synchronisation time can be achieved. 
In section~\ref{sectionClustering} we discuss a hierarchical clustering scheme that exploits that the superimposed signal strength of a set of nodes is increased at a slower pace than the synchronisation time with increasing node count.
When additional information available at a receiver node is utilised, further improvements are possible.
In section~\ref{sectionThreeUnknowns} we show that by providing more than one bit as feedback to the transmitters, knowledge about the feedback function can be derived from measurements of a single node altering the phase offset of its carrier signal.
We present an asymptotically optimal algorithm that utilises this knowledge and significantly improves the synchronisation process.
In section~\ref{sectionResultsBF01}, algorithms are compared for their synchronisation performance in numerical simulations.
In these simulations, the impact of various environmental settings and algorithmic configurations can be approximated.
Finally, in section~\ref{sectionJulian} we demonstrate the feasibility of distributed adaptive beamforming in wireless sensor networks in a near-realistic instrumentation with software radios.
Section~\ref{sectionConclusionBF01} draws our conclusion.

\subsection{Synchronisation time analysis}\label{sectionAsymptoticBounds}
We analyse the process of distributed adaptive beamforming in wireless sensor networks as described in section~\ref{sectionIntroductionBF01} and assume that each one of the $n$ nodes decides with probability $\frac{1}{n}$ to change the phase of its carrier signal uniformly at random in the interval $[0,2\pi]$.
On obtaining the feedback of the receiver, nodes that recently updated the phase of their carrier signal either sustain this decision or reverse it, depending on whether the feedback has improved or not.
A feedback function $\mathcal{F}:\zeta_{\mbox{\footnotesize sum}}^*\rightarrow\mathds{R}$ maps the superimposed received carrier signal 
\begin{equation}
\zeta_{\mbox{\footnotesize sum}}=\Re\left(m(t)e^{j2\pi f_ct}\sum_{i=1}^n\mbox{RSS}_ie^{j(\gamma_i+\phi_i+\psi_i)}\right)\label{equationOneBF01}	
\end{equation}
to a real-valued feedback score.
In equation~(\ref{equationOneBF01}) the $RSS_i$ denotes the received signal strength of the $i$-th signal out of $n$ received signal components.
As local oscillators are not synchronised and nodes are spatially distributed, $\phi_i$ and $\psi_i$ account for the phase offset in the received signal components due to the offset in the local oscillators and due to distinct signal propagation times.
A possible feedback function that is proportional to the distance between an observed superimposed carrier $\zeta_{\mbox{\footnotesize sum}}$ and an optimum sum carrier signal 
\begin{equation}
\zeta_{\mbox{\footnotesize opt}}=\Re\left(m(t)\mbox{RSS}_{\mbox{\footnotesize opt}}e^{j(2\pi f_ct+\gamma_{\mbox{\footnotesize opt}})}\right)	
\end{equation}
is $\mathcal{F}\left(\zeta_{\mbox{\footnotesize sum}}\right)=\int_{t=0}^{2\pi}\left|\zeta_{\mbox{\footnotesize sum}}-\zeta_{\mbox{\footnotesize opt}}\right|$.
Since this function can be mapped onto other feedback measures as, for instance, the signal to noise ratio (SNR) or the received signal strength (RSS), the following discussion remains valid for these feedback measures.
While the multimodality of this feedback function is straightforward, we derive in Appendix~A that it is also weak multimodal so that no local optima exist.

Distributed adaptive beamforming in wireless sensor networks is a search problem.
The search space $\mathcal{S}$ is given by the set of possible combinations of phase and frequency offsets $\gamma_i$ and $f_i$ for all $n$ carrier signals.
A global optimum is a configuration of individual carrier phases that result in identical phase and frequency offset of all received direct signal components.
For the analysis, we assume that the optimisation aim is to achieve for an arbitrary $k$ a maximum relative phase offset of $\frac{4\pi}{k}$ between any two carrier signals.
This means that we can control the quality of the synchronisation achieved by the variable $k$.
An optimum is then reached when the phases of all carrier signal components of a receiver are within an interval of $\frac{4\pi}{k}$ in the phase space.
When $k$ is increased this directly translates to an improved phase synchronisation among signal components. 
Naturally, we can expect that the accuracy of the synchronisation also impacts the synchronisation time. 

For our analysis we logically divide the phase space for a single carrier signal into $k$ intervals of width $\frac{2\pi}{k}$. 
Observe that this is half of the interval that was used to define the optimum synchronisation.
Consequently, when the achieved phase offset of each received signal component is within a maximum distance of $\frac{2\pi}{k}$ to the optimum phase offset, all received carrier phase signals are within an interval of $\frac{4\pi}{k}$ in the phase space and the optimum is reached. 

For a specific superimposed carrier signal $\zeta$ at a receiver we represent the corresponding search point $s_\zeta=(\Gamma_t,F_t)_\zeta\in \mathcal{S}$ at iteration $t$ by a specific combination of phase and frequency offsets with $\Gamma_t=(\gamma_{t,1},\dots,\gamma_{t,n})$ and $F_t=(f_{t,1},\dots,f_{t,n})$.
In order to respect neighbourhood similarities we represent search points as Gray encoded binary strings $s_\zeta\in\mathds{B}^{n\cdot\log(k)}$ so that similar points have a small Hamming distance \cite{2105}.
A search point is then composed from $n$ sections of $\log(k)$ bits each.
Every block of length $\log(k)$ describes one of the $k$ intervals for the phase offset of one carrier signal.
For the analysis, we assume that the frequency offset $f_i$ is zero for all carriers.
Observe, however, that the discussion can be easily adapted to also cover a simultaneous carrier frequency synchronisation. 
In~\cite{Seo_2008} the authors demonstrated, that the same random synchronisation approach can be utilised to synchronise carrier frequencies when in each iteration not only the carrier phase but also the frequency of the transmit signals is altered. 
By this generalisation, the search space of the algorithm is increased. 
For each node, not only $k$ distinct possibilities exist, but $k\cdot\rho$ where $\rho$ denotes the count of distinct frequencies that can possibly be applied for each carrier signal.
The optimisation time then increases by a factor of $\rho$. 
However, the analytical discussion becomes more complicated as the common period of the received sum signal might be increased considerably. 

The optimisation problem is denoted as $\mathcal{P}$ and $T_\mathcal{P}$ describes the count of iterations required to reach one optimum for the problem $\mathcal{P}$.

\subsubsection{An upper bound on the expected synchronisation time}\label{sectionUpperBound}
The value of the feedback function increases with the number of carrier signals $\zeta_i$ that share the same interval for their phase offset $\gamma_i$ at the receiver.
Assume that $\kappa\in[1,k]$ is the interval that contains most of the carrier phase offsets.
As worse feedback values are not accepted, we count the iterations required for all carrier signals to change to interval $\kappa$.
We can roughly divide the values of the feedback function into $n$ partitions $L_1,\dots,L_n$ depending on the number of carrier signals with their phase in the interval $\kappa$.
For each one transmitter, the probability to adapt its phase to one specific interval is $\frac{1}{k}$.
The probability to increase the feedback value so that at least the next partition is reached is then 
\begin{equation}
	\frac{1}{k}\cdot\left( n-L_i\right)\cdot\frac{1}{n}
\end{equation}
since one carrier signal $\zeta_i$ is altered with probability $\frac{1}{n}$ and the probability to reach any particular of the $(n-L_i)$ partitions that would increase the feedback value is $\frac{1}{k}$.
In partition $i$, a total of 
\begin{equation}
	\left(\begin{array}{c}
	      	n-i\\1
	      \end{array}
\right)=n-i
\end{equation}
carrier signals each suffice to improve the feedback value with probability $\frac{1}{n}\cdot\frac{1}{k}$.
We therefore require that at least one of the not synchronised carrier signals is correctly altered in phase while all other $n-1$ signals remain unchanged.
This happens with probability 
\begin{eqnarray}
&	&\left(
	\begin{array}{c}
		n-i\\1
	\end{array}
	\right)\cdot \frac{1}{n}\cdot\frac{1}{k}\cdot\left(1-\frac{1}{n}\right)^{n-1}\nonumber\\
&=&\left(\frac{n-i}{n\cdot k}
\right)\cdot\left(1-\frac{1}{n}\right)^{n-1}.
\end{eqnarray}
Since
\begin{equation}
	\left(1-\frac{1}{n}\right)^n<\frac{1}{e}<\left(1-\frac{1}{n}\right)^{n-1}
\end{equation}
We obtain the probability $P[L_i]$ that $L_i$ is left and a partition $j$ with $j>i$ is reached as
\begin{equation}
	P[L_i]\geq\frac{n-i}{n\cdot e \cdot k}.
\end{equation}
The expected number of iterations to change the layer is bounded from above by $P[L_i]^{-1}$. 
We consequently obtain the overall expected synchronisation time as
\begin{eqnarray}
	E[T_{\mathcal{P}}]&\leq& \sum_{i=0}^{n-1} \frac{e\cdot n \cdot k}{n-i}\nonumber\\
	&=&e\cdot n\cdot k\cdot\sum_{i=1}^n\frac{1}{i}\nonumber\\
	&<&e\cdot n\cdot k\cdot\left(\ln(n)+1\right)\nonumber\\
	&=&\mathcal{O}\left(n \cdot k \cdot \log n\right).
\end{eqnarray}

\subsubsection{A lower bound on the expected synchronisation time}
After the initialisation, the phases of the carrier signals are identically and independently distributed.
Consequently for a superimposed received sum signal $\zeta$, each bit in the binary string $s_\zeta$ that represents the corresponding search point has an equal probability to be $1$ or $0$.
The probability to start from a search point $s_\zeta$ with Hamming distance $h(s_{\mbox{\footnotesize opt}},s_\zeta)$ not larger than $l\in\mathds{N}\; ; \; l\ll n\cdot\log(k)$ to one of the global optima $s_{\mbox{\footnotesize opt}}$ directly after the random initialisation is at most 
\begin{eqnarray}
	P[h(s_{\mbox{\footnotesize opt}},s_\zeta)\leq l]
	&=& \sum_{i=0}^l\left(\begin{array}{c}
	                                	n\cdot\log(k)\\ n\cdot\log(k)-i
	                                \end{array}
\right)\cdot \frac{k}{2^{n\cdot\log(k)-i}}\nonumber\\
&\leq&\frac{\left(n\cdot\log(k)\right)^{l+2}}{2^{n\cdot\log(k)-l}}\nonumber
\end{eqnarray}
In this formula, 
\begin{equation}
\left(\begin{array}{c}
	                                	n\cdot\log(k)\\ n\cdot\log(k)-i
	                                \end{array}
\right)	
\end{equation}
is the count of possible configurations with $i$ bit errors to a given global optimum, $\frac{1}{2^{n\cdot\log(k)-i}}$ represents the probability for all these bits to be correct and $k$ is the count of global optima.
This means that with high probability (w.h.p.) the Hamming distance to the nearest global optimum is at least $l$.
We use the method of the expected progress to calculate a lower bound on the optimisation time.

Let $(s_\zeta,t)$ denote the situation that search point $s_\zeta$ is achieved after $t$ iterations of the algorithm.
We assume a progress measure $\Lambda:\mathds{B}^{n\cdot\log(k)}\rightarrow\mathds{R}^+_0$ such that $\Lambda(s_\zeta,t)<\Delta$ represents the case that a global optimum was not found in the first $t$ iterations.
For every $t\in\mathds{N}$ we have 
\begin{eqnarray}
	E[T_\mathcal{P}]&\geq& t\cdot P[T_\mathcal{P}>t]\nonumber\\ 
	&=& t\cdot P[\Lambda(s_\zeta,t)<\Delta]\nonumber\\
	&=& t\cdot (1-P[\Lambda(s_\zeta,t)\geq\Delta]).
\end{eqnarray}
With the help of the Markov-inequality we obtain 
\begin{equation}
P[\Lambda(s_\zeta,t)\geq\Delta]\leq \frac{E[\Lambda(s_\zeta,t)]}{\Delta}	
\end{equation}
and therefore 
\begin{equation}
	E[T_\mathcal{P}]\geq t\cdot\left(1-\frac{E[\Lambda(s_\zeta,t)]}{\Delta}\right).
\end{equation}
This means that we can obtain a lower bound on the optimisation time by providing the expected progress after $t$ iterations.
The probability for $l$ bits to correctly flip is at most 
\begin{eqnarray}
	& & \left(1-\frac{1}{n\cdot\log(k)}\right)^{n\cdot\log(k)-l}\cdot\left(\frac{1}{n\cdot\log(k)}\right)^l\nonumber\\
	&\leq&\frac{1}{(n\cdot\log(k))^l}.
\end{eqnarray}
In this formula, $\left(1-\frac{1}{n\cdot\log(k)}\right)^{n\cdot\log(k)-l}$ describes the probability that all 'correct' remain unchanged while the remaining $l$ bits flip with probability $\left(\frac{1}{n\cdot\log(k)}\right)^l$.
The expected progress in one iteration is therefore 
\begin{eqnarray}
	E[\Lambda(s_{\zeta},t),\Lambda(s_{\zeta'},t+1)]&\leq& \sum_{i=1}^l\frac{i}{(n\cdot\log(k))^i}\nonumber\\
	&<&\frac{2}{n\cdot\log(k)}
\end{eqnarray}
and the expected progress in $t$ iterations is consequently not greater than $\frac{2t}{n\cdot\log(k)}$.
When we choose $t=\frac{n\cdot\log(k)\cdot\Delta}{4}-1$, the double of the expected progress is still smaller than $\Delta$.
With the Markov inequality we can show that this progress is not achieved with probability $\frac{1}{2}$.
Altogether we conclude that the expected synchronisation time is bounded from below by
\begin{eqnarray}
	E[T_\mathcal{P}]&\geq& t\cdot\left(1-\frac{E[\Lambda(s_\zeta,t)]}{\Delta}\right) \nonumber\\
	&\geq&  \frac{n\cdot\log(k)\cdot\Delta}{4}\cdot \left(1- \frac{\frac{2\cdot n\cdot\log(k)}{4\cdot n\cdot\log(k)}\cdot\Delta}{\Delta}\right)\nonumber\\
	&=&\Omega(n\cdot\log(k)\cdot\Delta)
\end{eqnarray}
With $\Delta=k\cdot\frac{\log(n)}{\log(k)}$ we obtain a lower bound in the same order as the upper bound derived in section~\ref{sectionUpperBound} and consequently an asymptotically sharp bound of
\begin{equation}
E[T_\mathcal{P}]=\Theta\left(n\cdot k\cdot \log(n)\right).
\end{equation}

Note that in \cite{Mudumbai_2010b} an upper bound on the expected asymptotic synchronisation time was derived that scales linearly in the number of nodes $n$ when the probability distribution is optimally altered repeatedly during the synchronisation.
However, simulation results derived for a fixed uniform distribution in this study also indicate a logarithmic factor in the synchronisation time of one-bit feedback based synchronisation.

\subsection{Hierarchical clustering}\label{sectionClustering}
A further improvement of the synchronisation time can be achieved by synchronising smaller clusters of nodes separately.
Since this bound on the synchronisation time grows faster than linearly with the network size $n$ but the received signal strength $\mbox{RSS}_{\mbox{\footnotesize sum}}$ of the received superimposed signal grows linearly with $n$, the overall energy consumption and synchronisation time might be reduced when fewer nodes transmit for a shorter time but with an increased transmission power.
Note that currently most low cost radios are not capable of altering their transmission power and therefore are not able to exploit this property.
More sophisticated radios could, however, achieve carrier phase synchronisation more efficiently when this fact is utilised.
We propose the following hierarchical clustering scheme that synchronises all transmit nodes iteratively in clusters of reduced size.
\begin{enumerate}
	\item Determine clusters (e.g. by a random process initialised by the receiver node)
	\item Synchronise clusters successively as described above with possibly increased transmit power.
	When cluster $\iota$ is sufficiently synchronised, nodes in this cluster sustain their carrier signal and stop transmitting until all clusters are synchronised.
	\item At this stage, carrier signals in all clusters are in phase but carrier phases of distinct clusters might differ.
	Determine representative nodes from all clusters and synchronise these.
	\item Nodes in all clusters alter their carrier phase by the phase offset experienced and broadcast by the corresponding representative node (broadcast).\\
	Let $\zeta_{i}=\Re\left(m(t)\mbox{RSS}_ie^{j2\pi f_ct(\gamma_i+\phi_i+\psi_i)}\right)$ and $\zeta_{i}'=\Re\left(m(t)\mbox{RSS}_ie^{j2\pi f_ct(\gamma_i'+\phi_i+\psi_i)}\right)$ be the carrier signals of representative node $i$ from cluster $\iota$ before and after synchronisation between representative nodes was achieved.
	A node $h$ from cluster $\iota$ alters its carrier signal $\zeta_{h}=\Re\left(m(t)\mbox{RSS}_he^{j2\pi f_ct(\gamma_h+\phi_h+\psi_h)}\right)$ to $\zeta_{h}'=\Re\left(m(t)\mbox{RSS}_he^{j2\pi f_ct(\gamma_h+\phi_h+\psi_h+\gamma_i-\gamma_i')}\right)$.
	Under ideal conditions, all nodes are now in phase.
	\item To account for synchronisation errors a final synchronisation phase in which all nodes participate concludes the overall synchronisation process.
\end{enumerate}
Fig.~\ref{figureClustering} illustrates this procedure. 
\begin{figure}
\centering
	\includegraphics[width=\textwidth]{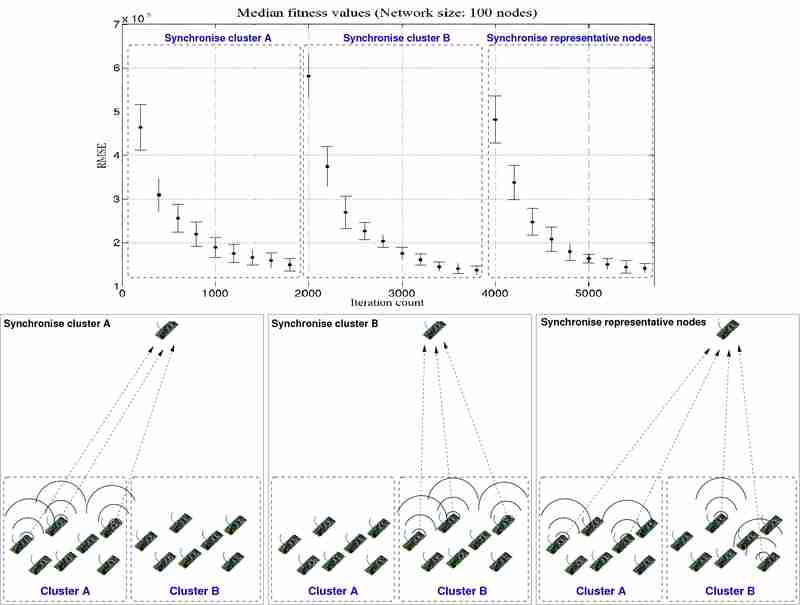}
	\caption{Illustration of the approach to cluster the network of nodes in order to improve the synchronisation time of feedback based closed-loop distributed adaptive beamforming. {\scriptsize(\copyright 2011 IEEE  Published by the IEEE CS, CASS, ComSoc, IES, SPS)}}
	\label{figureClustering}
\end{figure}
The crucial idea of this approach is applied in step~4. 
Since nodes inside a cluster have already been synchronised, they are still in phase after all apply an identical phase offset.
Because this offset is the phase alteration the representative nodes experienced to during their synchronisation, all nodes should be synchronised after this step. 

A potential problem for this approach is phase noise.
Since only one cluster is synchronised at a time, phases of nodes in the inactive clusters experience phase noise and start drifting out of phase due to practical properties of oscillators.
However, we show in section~\ref{sectionResultsBF01} that sufficient synchronisation is possible in the order of milliseconds.
Therefore, we do not consider phase noise an important issue.
Observe that all coordination is initiated by the receiver node so that no inter-node communication is required for coordination.

Depending on the network size, more than one hierarchy stage might be optimal for the synchronisation time and the energy consumption.
To estimate the optimal hierarchy depth and the optimum cluster size, the count of nodes participating in the synchronisation must be computed.
We assume that the nodes themselves do not know the network size. 
This means that the remote receiver derives the network size, calculates optimal cluster sizes and hierarchy depths and broadcasts this information.
In \cite{5811} it was demonstrated that the superimposed sum signal from arbitrarily synchronised nodes is sufficient to estimate the number of transmitters.
We derive the optimum hierarchy depth and cluster size by integer programming in time $\mathcal{O}(n^2)$ (cf. Appendix~B).
The expected synchronisation time is dependent on the cluster count, cluster size and hierarchy depth. 
Since for each cluster a small instance of the original problem is solved, the synchronisation time can be composed from the synchronisation times of individual clusters.

\subsection{An asymptotically optimal algorithm}\label{sectionThreeUnknowns}
Since no local optimum exists in the search space due to its weak multimodality (cf. Appendix~A), the performance of $\Theta(n\cdot k\cdot\log(n))$ of the random search seems weak.
Previously, only one feedback bit was utilised. 
Due to this reduced information, some of the information available at the receiver is not available by the nodes that actually react on the feedback.
When more information is included in the feedback of the receive node we are able to design an asymptotically optimal synchronisation algorithm.

In every iteration the receiver provides additional information over a feedback value so that a node $i$ can learn the optimum phase offset of its own carrier 
\begin{equation}
\zeta_{i}=\Re\left(m(t)\mbox{RSS}_ie^{j2\pi f_ct(\gamma_i+\phi_i+\psi_i)}\right)	\nonumber
\end{equation}
relative to the superimposed sum signal 
\begin{equation}
\zeta_{\mbox{\footnotesize sum}\smallsetminus i}=\Re\left(m(t)e^{j2\pi f_ct}\sum_{o\in[1,n];o\not= i}\mbox{RSS}_oe^{j(\gamma_o+\phi_o+\psi_o)}\right)	\nonumber
\end{equation}
of all other nodes, provided that the latter does not change significantly.
$\zeta_{\mbox{\footnotesize sum}\smallsetminus i}$ is a sinusoidal signal.
The feedback is maximal when $\zeta_i$ and $\zeta_{\mbox{\footnotesize sum}\smallsetminus i}$ have identical phase offset at a receiver.
With increasing phase offset 
\begin{equation}
\left|\left(\gamma_i+\phi_i+\psi_i\right)-\left(\gamma_{\mbox{\footnotesize sum}\smallsetminus i}+\phi_{\mbox{\footnotesize sum}\smallsetminus i}+\psi_{\mbox{\footnotesize sum}\smallsetminus i}\right)\right|\nonumber
\end{equation}
the feedback value decreases symmetrically.
Consequently, the feedback function has the form $\mathcal{F}(\gamma_i)=A\sin\left(\gamma_i+\varPhi\right)+c$.
This is an equation with the three unknowns $A$ (amplitude), $\varPhi$ (phase offset of $\mathcal{F}$) and the additive term $c$ so that a node $i$ can calculate it with three distinct measurements.
Fig.~\ref{figureOneSender}  illustrates the accuracy of this procedure for 100 transmitters.
\begin{figure}
\centering
	\includegraphics[scale=.25]{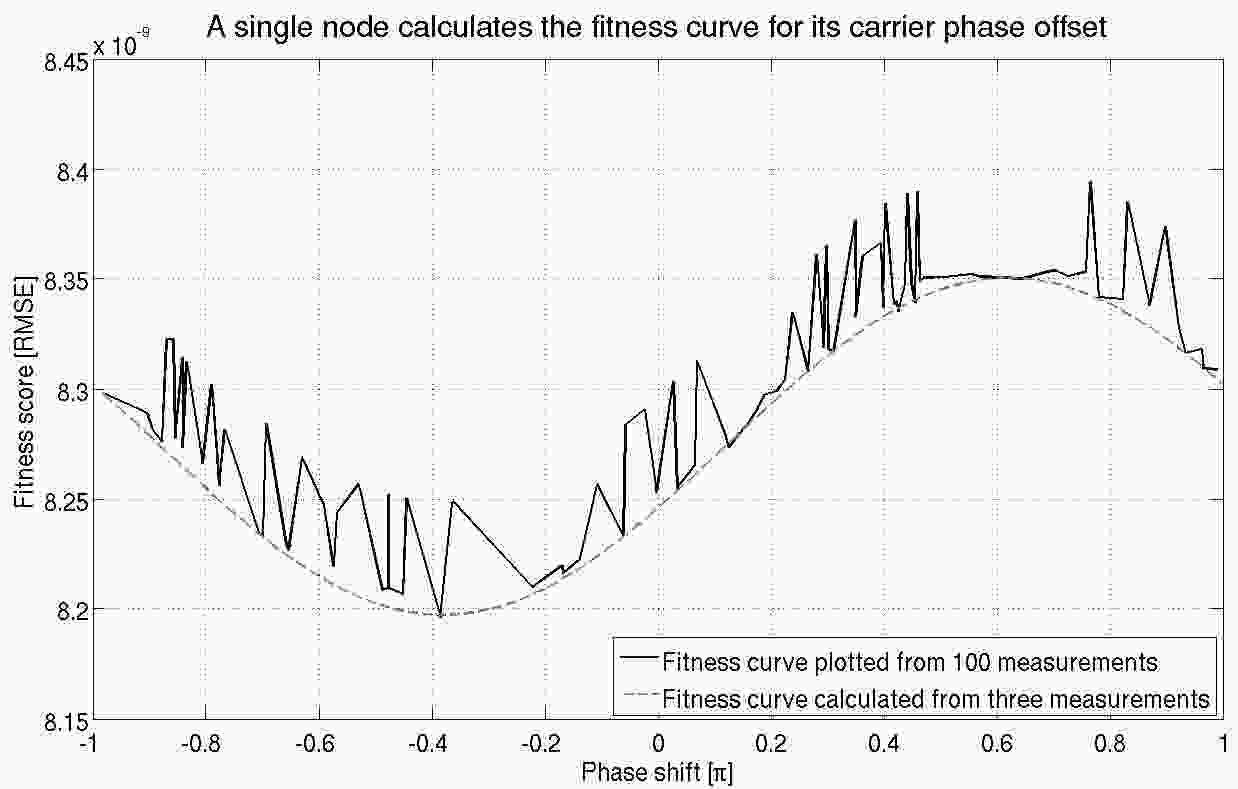}
	\caption{Deviation of the feedback curve calculated from three measurements to the feedback curve plotted from 100 measurements. {\scriptsize(\copyright 2011 IEEE  Published by the IEEE CS, CASS, ComSoc, IES, SPS)}}
	\label{figureOneSender}
\end{figure}
The root of the mean square error (RMSE) is calculated as  
\begin{equation}
RMSE=\sqrt{\sum_{t=0}^{\tau}
		\frac{\left(
				\zeta_{\mbox{\footnotesize sum}}+\zeta_{\mbox{\footnotesize noise}}-\zeta_{\mbox{\footnotesize opt}}
			\right)^2}{n}
		}.\label{equationRMSEBF01}
\end{equation}
Here, $\tau$ is chosen to cover several signal periods.

For the optimisation process, a node will during each of four subsequent iterations either alter the phase offset of its carrier signal or sustain it for all four iterations.
The probability to alter the phase offset should be low as, for instance, $\frac{1}{n}$.
A node that decides to alter its phase offset, will do this three times to measure feedback values for distinct phase offsets, derive with these measurements the feedback function, alter its phase offset accordingly and finally transmit a fourth time to obtain the amount by which the achieved feedback value deviates from the expected value.
If the deviation is small, the node will not alter its phase further since the current phase offset is considered optimal.
All other nodes then adapt the probability to alter their carrier phase so that one node alters the phase of its carrier signal on average per iteration (for instance from $\frac{1}{n}$ to $\frac{1}{n-1}$).

As nodes are chosen according to a random process, it is possible that more than one node simultaneously alters its phase offset. 
In this case, the node's conclusions on the impact of their phase-alteration on the feedback value are biased.
Therefore, in the fourth measurement, when the measured value deviates significantly from the expected feedback a node concludes that it was not the only one to alter its phase and reverses its decision.

In our measurements, the deviation of the calculated feedback curve did not exceed $0.6\%$ when only one node adapts its phase offset.
With two nodes simultaneously adapting their phase offset we already experienced a deviation of approximately $1.5\%$.
As this procedure is guided purely by the feedback broadcast by the receiver, inter-node communication is not required.

Asymptotically, the synchronisation time of this algorithm is $\Theta(n)$ since on average the count of carrier signals that are in phase increases by $1$ in each iteration.
Further performance improvements can be achieved when nodes utilise only three subsequent iterations and acquire the first measurement from the last transmission of the preceding three subsequent iterations.

The asymptotic synchronisation time derived for this approach is optimal when we assume that individual nodes have to compute their optimal carrier phase offset independently since $n$ carrier signals have to be adapted.
When, however, a synchronisation scheme is utilised in which information about the optimum relative carrier phase offsets of all nodes is provided, as e.g. in typical open-loop carrier synchronisation schemes (cf.~\cite{5923}), the asymptotic synchronisation time can be further reduced. 

This improved carrier synchronisation scheme can be applied in any scenario in which a rich feedback as, for instance, the SNR can be provided.
It is, however, not applicable when only binary feedback is provided by the receiver. 
When, for example, high noise and interference would force an impractically complex error correction scheme, it might be beneficial to utilise the one bit feeback based carrier synchronisation instead.

\subsection{Simulation studies}\label{sectionResultsBF01}
We have implemented the scenario of distributed adaptive beamforming in Matlab to obtain a better understanding of the impact of environmental parameters and algorithmic configurations.
In particular, the effect of distinct probability distributions as well as the count of transmitters and the transmission distance are considered.
In these simulations, 100 transmit nodes are placed uniformly at random on a $30m\times30m$ square area.
The receiver is located $30m$ ($100m, 200m, 300m$) above the centre of this area.
Receiver and transmit nodes are stationary.
Simulation parameters are summarised in Table~\ref{tableSimulationConfig}.
\begin{table}
\centering
\begin{footnotesize}
\caption{Configuration of the simulations. $P_{rx}$ is the the received signal power, $d$ is the distance between transmitter and receiver and $\lambda$ is the wavelength of the signal {\scriptsize(\copyright 2011 IEEE  Published by the IEEE CS, CASS, ComSoc, IES, SPS)}}
\begin{tabular}[b]{l|c}\hline
	Property & Value\\\hline
	Node distribution area & $30m\times30m$\\
	Location of the receiver & $(15m,15m,30m)$\\
	Mobility & stationary nodes\\
	Base band frequency & $f_{base}=2.4$ GHz\\
	Transmission power of nodes  & $P_{tx}=1$ mW \\
	Gain of the transmit antenna  & $G_{tx}=0$ dB\\
	Gain of the receive antenna & $G_{rx}=0$ dB \\
	Iterations per simulations& 6000 \\
	Identical simulation runs & 10\\
	Random noise power \cite{062} & $-103$ dBm \\
 	Pathloss calculation ($P_{rx}$)& \begin{minipage}{2.64cm}$P_{tx}\left(\frac{\lambda}{4\pi d}\right)^2 G_{tx} G_{rx}$ \end{minipage}\\\hline
\end{tabular}
\label{tableSimulationConfig}
\end{footnotesize}	
\end{table}

Frequency and phase stability are considered perfect.
We derived the median and standard deviation from 10 simulation runs.
One iteration consists of the nodes transmitting, feedback computation, feedback transmission and feedback interpretation by transmitters.
It is possible to perform these steps within few signal periods so that the time consumed for 6000 iterations is in the order of milliseconds for a base band signal frequency of $2.4$ GHz.
Signal quality is measured by the RMSE of the received signal to an expected optimum signal as detailed in equation~(\ref{equationRMSEBF01}).
The optimum signal is calculated as a perfectly aligned and properly phase shifted received sum signal from all transmit sources.
For the optimum signal, noise is disregarded.

Fig.~\ref{figureUniformDistribution-ReceivedSignal} depicts the optimum carrier signal, the initial received sum signal and the synchronised carrier after $6000$ iterations when carrier phases are altered with probability $\frac{1}{n}$ in each iteration according to a uniform distribution.
\begin{figure*}\centering
	\subfloat[Received sum signal from 100 transmit nodes without synchronisation and after 6000 iterations]{
		\includegraphics[scale=.185]{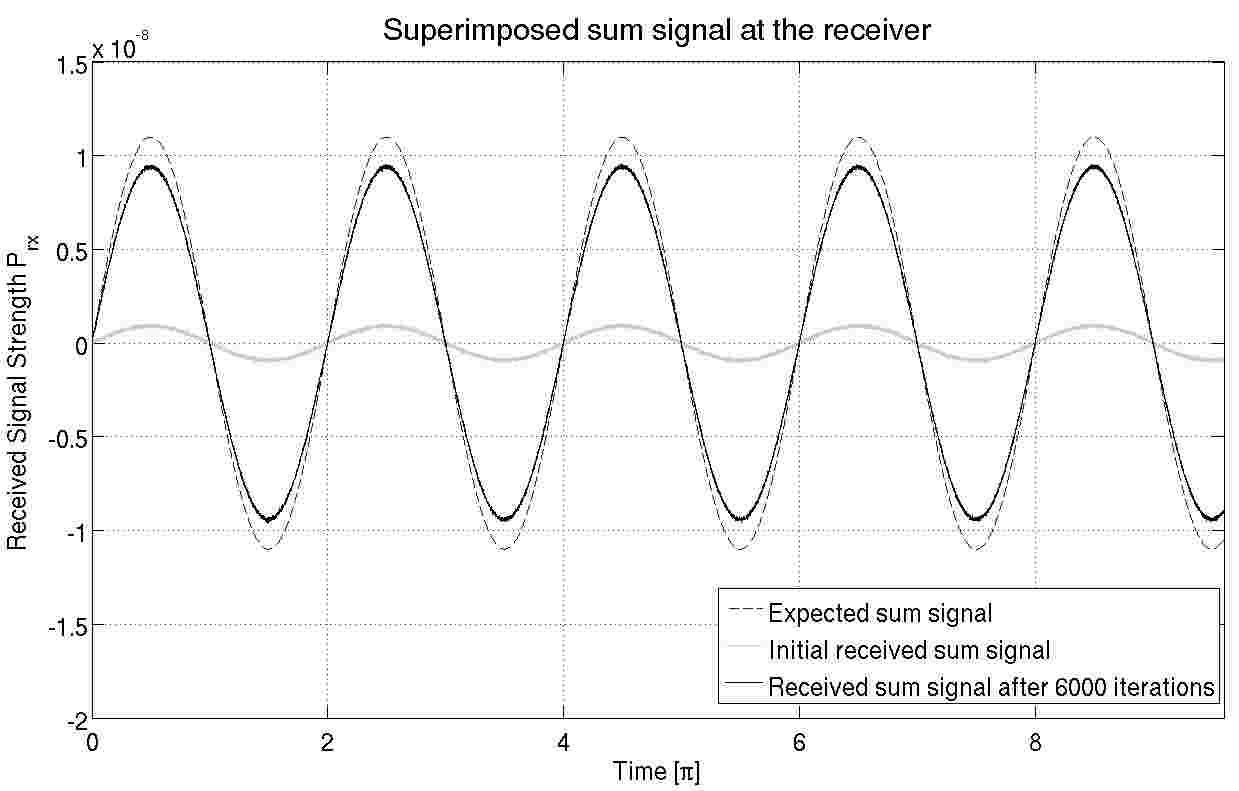}
	 	\label{figureUniformDistribution-ReceivedSignal}}
	\subfloat[Evolution of the phase adaptation process]{
		\includegraphics[scale=.185]{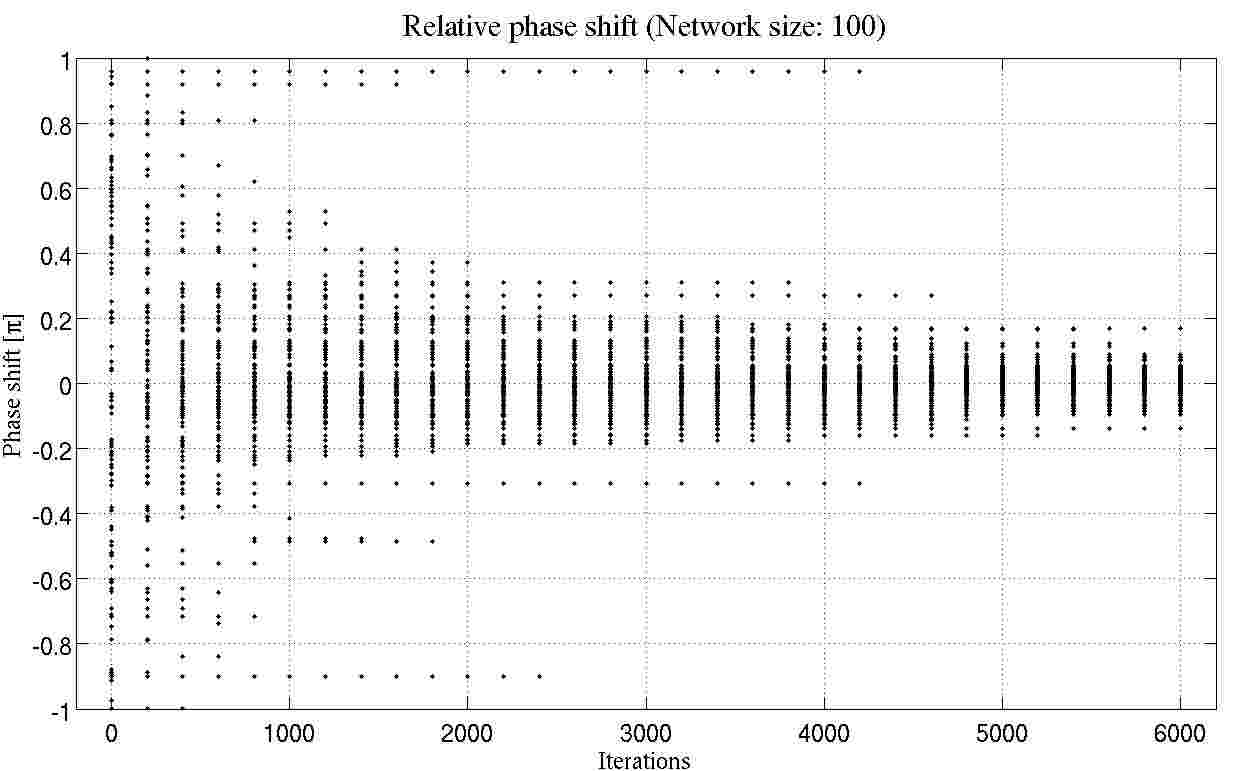}
		\label{figureUniformDistribution_PhaseAdaptation_1}
	}	
	\caption{Simulation results for a simulation with 100 transmit over 6000 iterations of the random optimisation approach to distributed adaptive beamforming in wireless sensor networks. {\scriptsize(\copyright 2011 IEEE  Published by the IEEE CS, CASS, ComSoc, IES, SPS)}}
\end{figure*}
In Fig.~\ref{figureUniformDistribution_PhaseAdaptation_1}, the phase offset of received signal components for an exemplary simulation run with the same parameters are illustrated.
We observe that after $6000$ iterations about $98\%$ of all carrier signals converge to a relative phase offset of about +/- $0.1\pi$.
The median of all variances of the phase offsets for simulations with this configuration is $0.2301$ after $6000$ iterations. 
The actual synchronisation time is dependent on the time to complete a single iteration. 
In each iteration, a synchronisation signal is transmitted, the received sum signal is analysed, feedback is calculated, broadcast to the network and interpreted by transmit nodes.
While the processing speed  might be improved with improved hardware, the round trip time of the signal poses a definite lower bound for the time a single iteration lasts. 
At a distance of 30 meters, for instance, we can not hope to complete a single iteration in less than $0.2\mu s$. 

\subsubsection{Uniform vs. normal distribution}
Distributed adaptive beamforming in wireless sensor networks has been studied in the literature according to various random phase alteration processes.
The authors in \cite{Mudumbai_2009,5919,Barriac_2004} report good results when the probability $p_\gamma$ to alter the phase of a single carrier signal in one iteration is $1$ for all nodes and the phase offset is chosen according to a normal distribution.
The variance $\sigma_\gamma^2$ applied is not reported.
In \cite{4019,4020} $p_\gamma$ was set to $\frac{1}{n}$ for each one of the $n$ nodes while the phase is altered according to a uniform distribution.

For both, uniform and normal distributed processes, we consider several values for $p_\gamma$ and $\sigma_\gamma^2$.
Generally, we achieved good performance when modifications in one iteration were small.
For the uniform distribution this translates to $p_\gamma=\frac{1}{n}$.
For the normal distribution, good results are achieved when $\sigma_\gamma^2$ and $p_\gamma$ are balanced so that the modification to the overall sum signal is small.
With increasing $p_\gamma$ good results are achieved with decreasing $\sigma_\gamma^2$.
Fig.~\ref{figureNormalDistribution_var001} depicts the results for $p_\gamma=\frac{1}{n}$ and $\sigma_\gamma^2=0.5\pi$
\begin{figure}\centering
	\includegraphics[width=.7\textwidth]{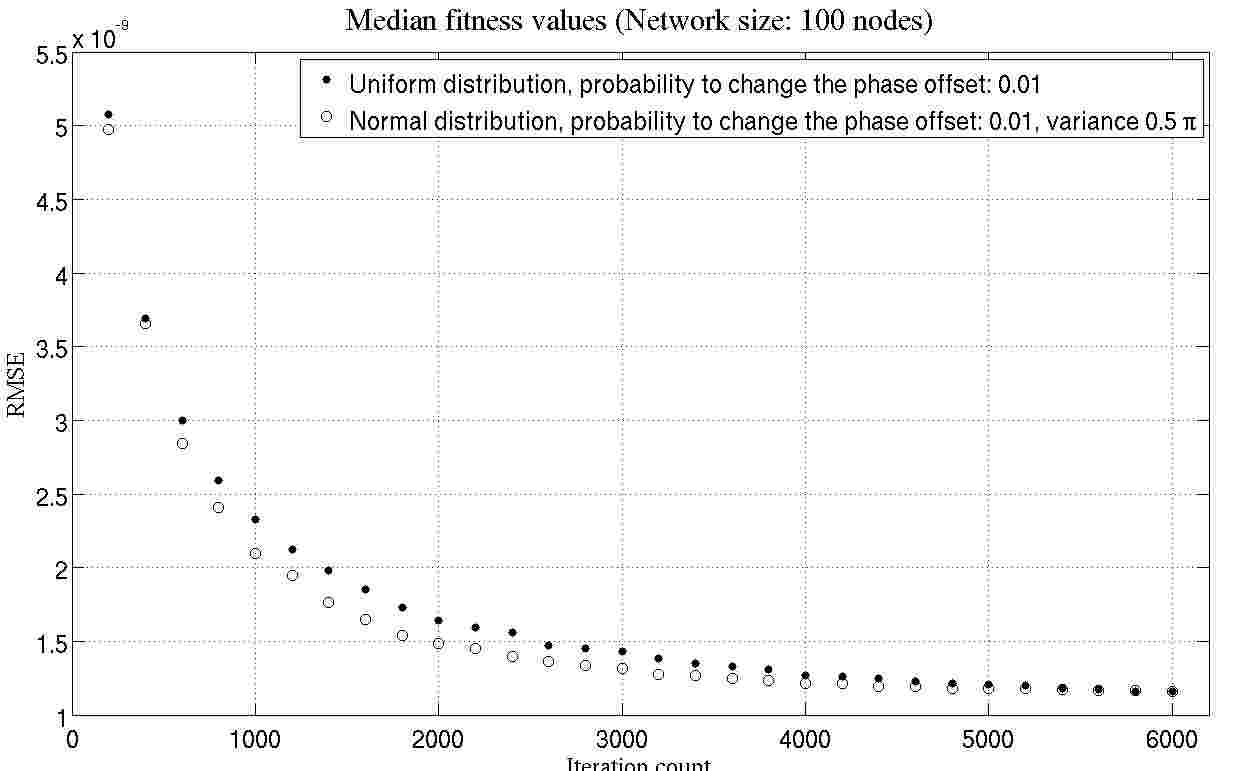}
	\caption{Performance of normal and uniform distributions for a network size of 100 nodes and $p_\gamma=0.01,\sigma_\gamma^2=0.5\pi$. {\scriptsize(\copyright 2011 IEEE  Published by the IEEE CS, CASS, ComSoc, IES, SPS)}}
	\label{figureNormalDistribution_var001}
\end{figure}

The figure shows the median RMSE value achieved in $10$ simulations by normal and uniform distributed processes over the course of 6000 iterations.
For ease of presentation, error bars are omitted in this figure.
However, the standard deviation is low for both processes (the standard deviation of this normal distributed process is depicted in Fig.~\ref{figureRayan2-b}).

The normal distributed process has a slightly improved synchronisation performance.
The optimum feedback value reached is, however, identical.

\subsubsection{Impact of the network size}
When the count of nodes that participate in the synchronisation is altered, this also impacts the performance of this process (cf. section~\ref{sectionAsymptoticBounds}).
We conducted several simulations with network sizes ranging from 20 to 100 nodes.
Fig.~\ref{figureUniformDistribution-NetDecrease} depicts the performance of several synchronisation processes with varying network sizes.
\begin{figure}
\centering
	\includegraphics[width=.7\textwidth]{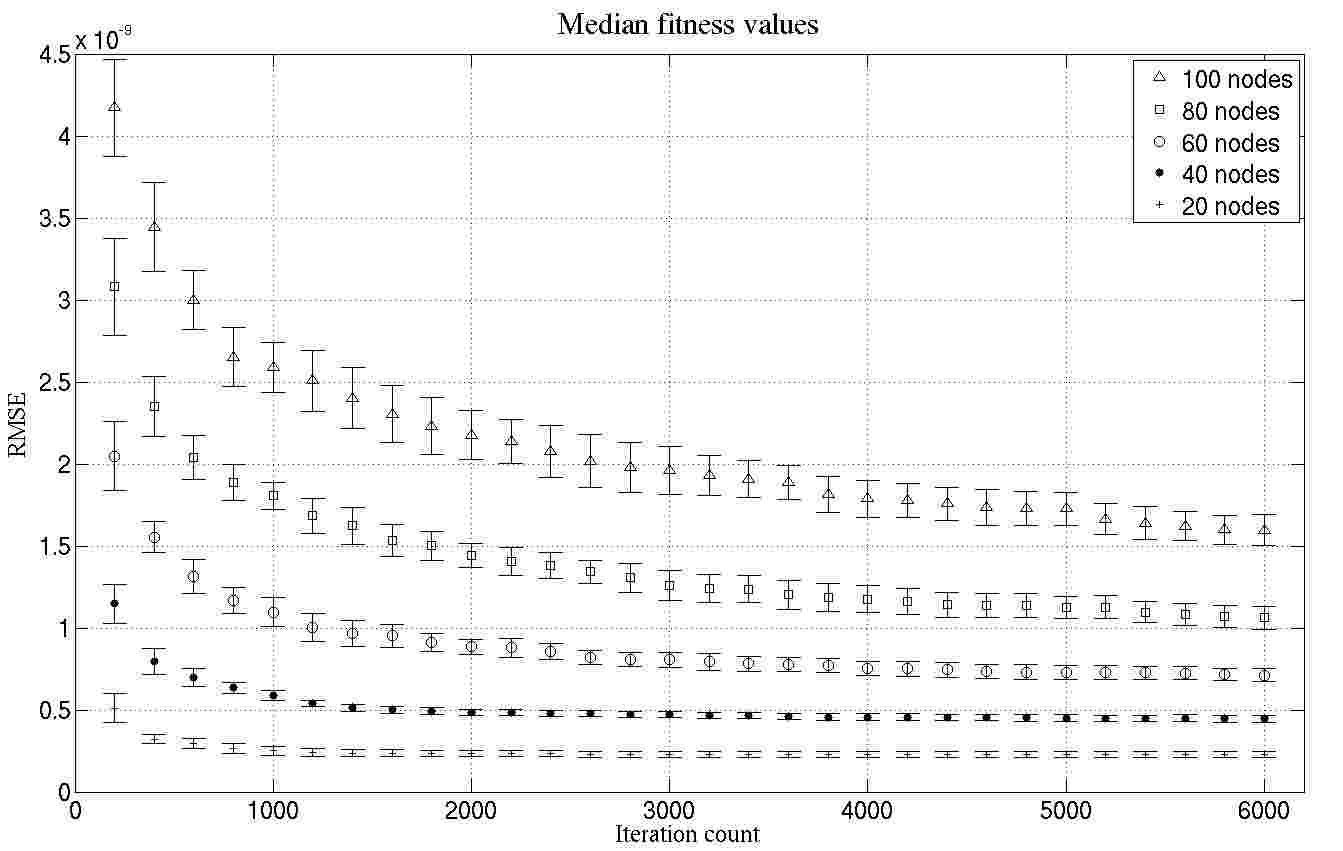}
	\caption{The synchronisation performance for various network sizes in a uniformly distributed process with $p_\gamma=0.05$. {\scriptsize(\copyright 2011 IEEE  Published by the IEEE CS, CASS, ComSoc, IES, SPS)}}
	\label{figureUniformDistribution-NetDecrease}
\end{figure}

In these simulations, we set $p_\gamma=0.05$ and utilised a uniformly distributed phase alteration process.
We see that the maximum feedback value achieved is lower for smaller network sizes.
This is due to the RMSE measure that compares the achieved sum signal to an expected optimum superimposed signal.
As the count of participating nodes diminishes, also the amplitude of the optimum signal decreases.
As expected, the optimum value is reached earlier for smaller network sizes.

\subsubsection{Impact of the transmission distance}
We are also interested in the performance of distributed adaptive beamforming when the distance between the network and a receiver is increased.
For a uniformly distributed phase alteration process with $p_\gamma=\frac{1}{n}$ we increase the transmission distance successively.
Fig.~\ref{figureIncreasedDistance} depicts the phase coherency achieved and the received sum signal for various transmission distances.
\begin{figure}
\centering
	\subfloat[Receiver distance: 100 meters -- Received RF signal]{\includegraphics[width=0.48\textwidth]{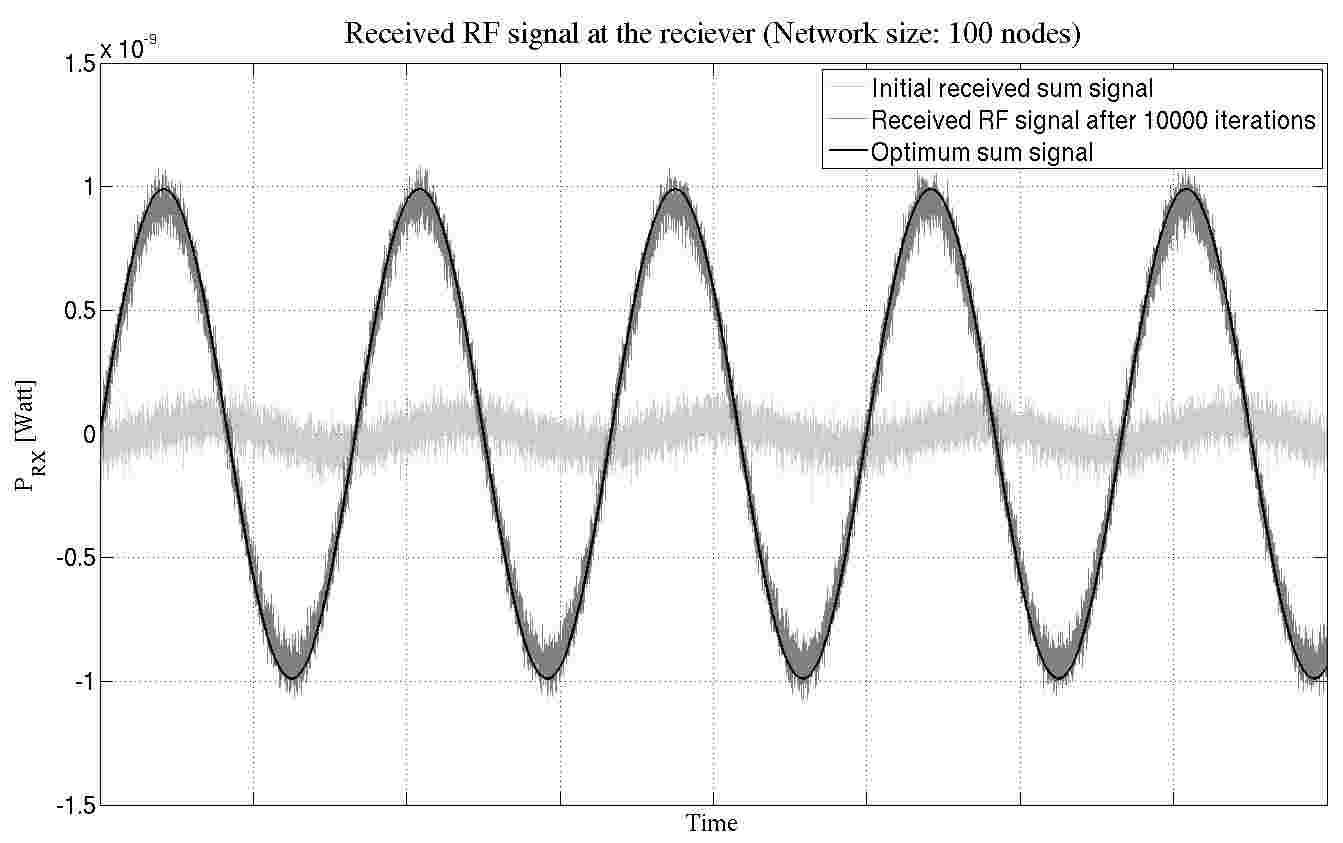}
	\label{label5}}
	\subfloat[Receiver distance: 100 meters -- Relative phase shift of signal components]{\includegraphics[width=0.48\textwidth]{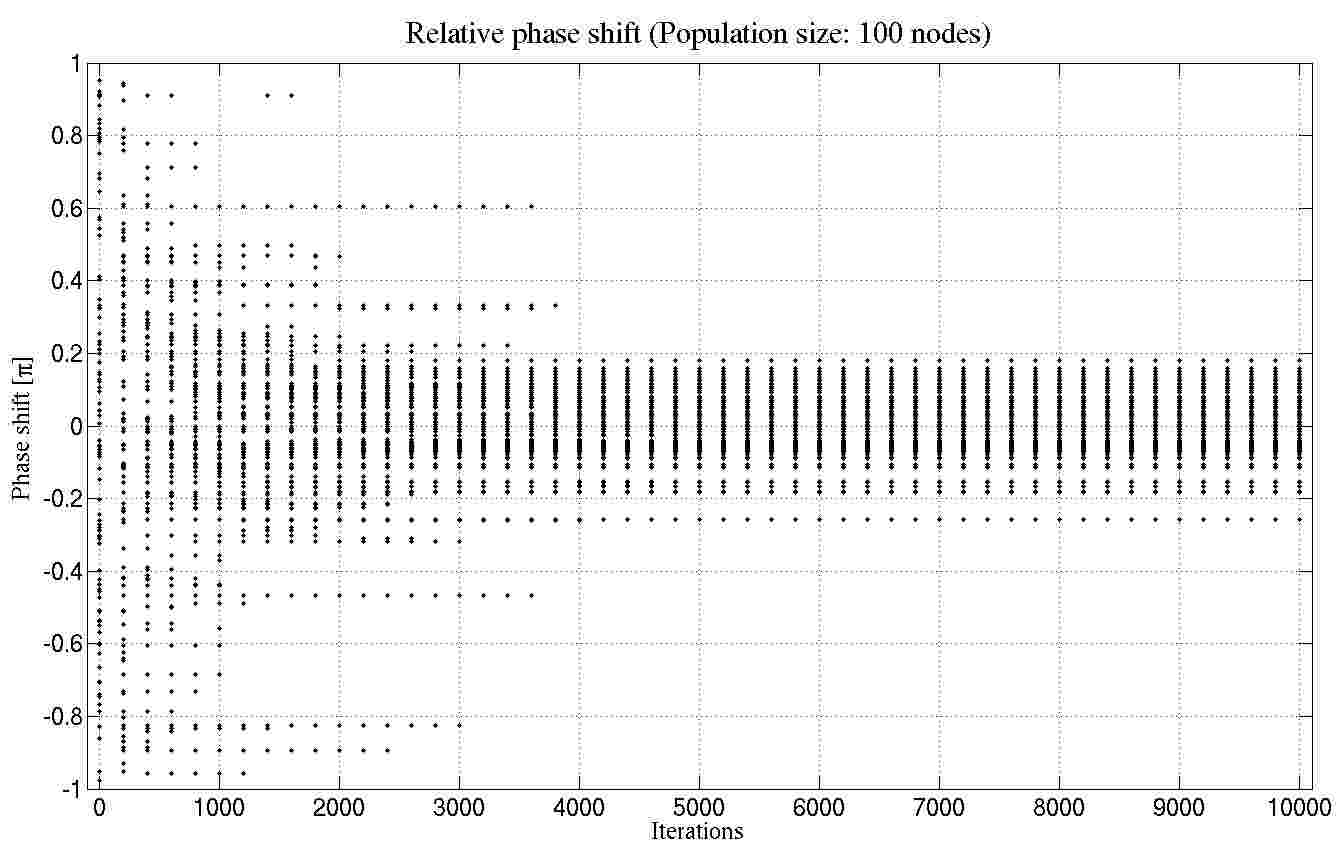}
	\label{label6}}
	
	\subfloat[Receiver distance: 200 meters -- Received RF signal ]{\includegraphics[width=0.48\textwidth]{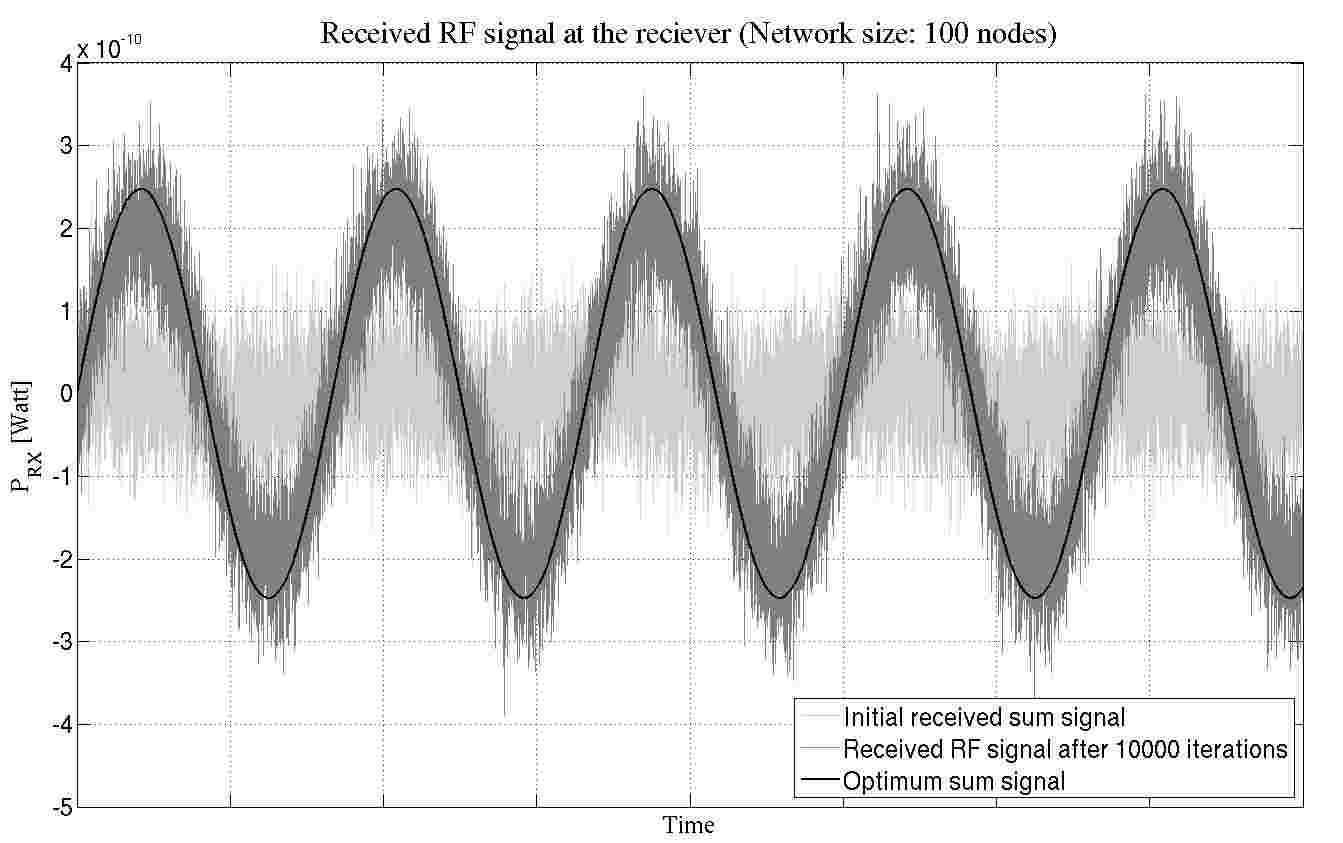}
	\label{label1}}
	\subfloat[Receiver distance: 200 meters -- Relative phase shift of signal components ]{\includegraphics[width=0.48\textwidth]{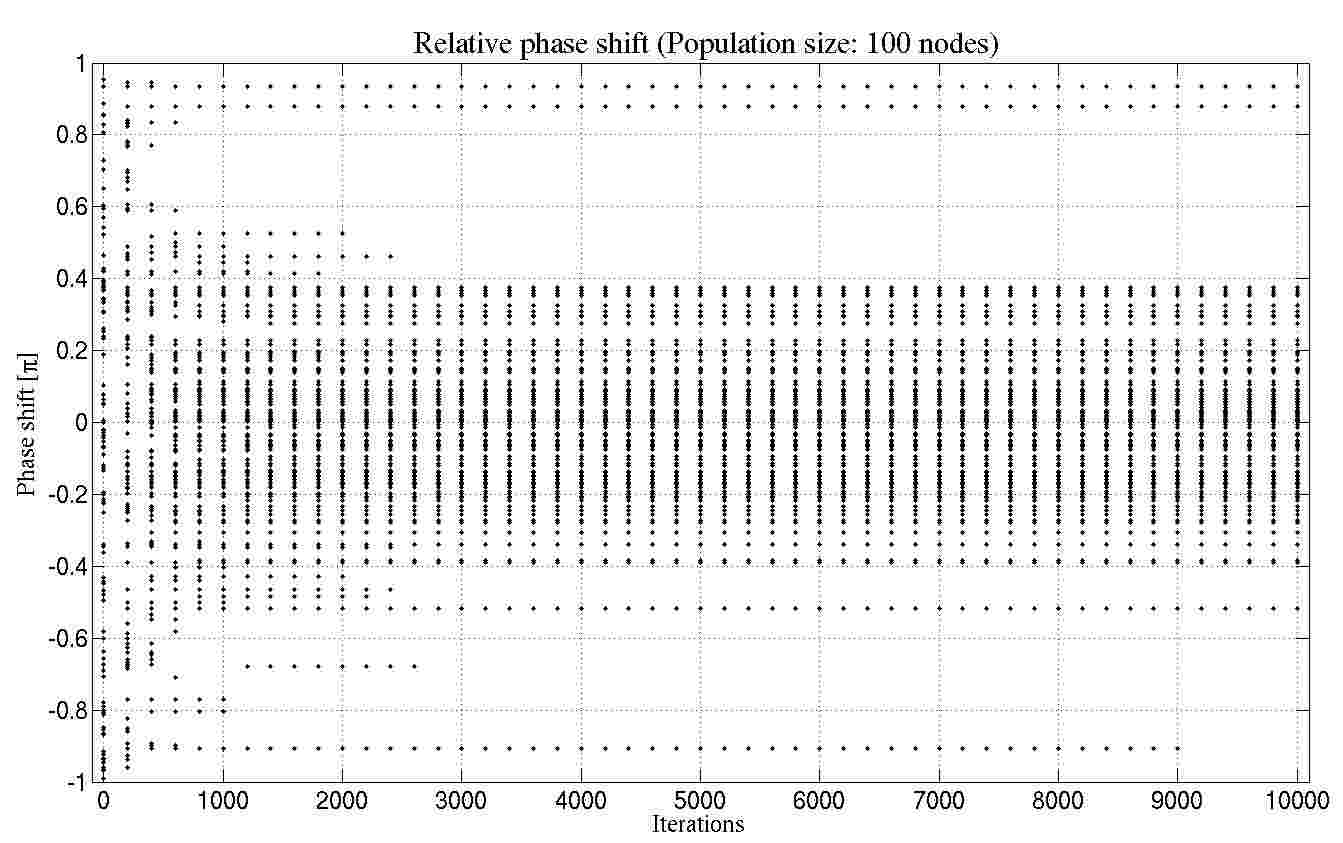}
	\label{label2}}
	
	\subfloat[Receiver distance: 300 meters -- Received RF signal ]{\includegraphics[width=0.48\textwidth]{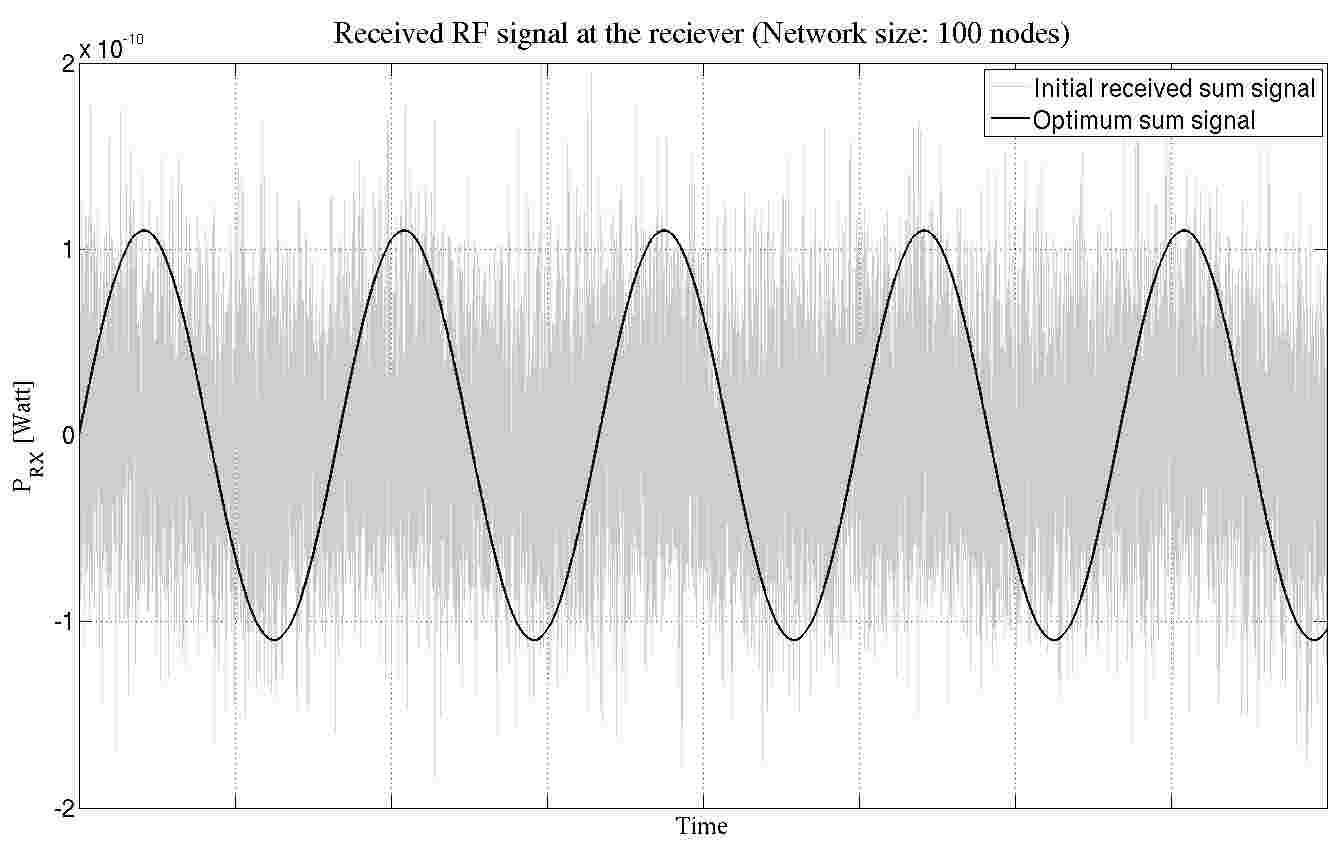}
	\label{label7}}
	\subfloat[Receiver distance: 300 meters -- Relative phase shift of signal components ]{\includegraphics[width=0.48\textwidth]{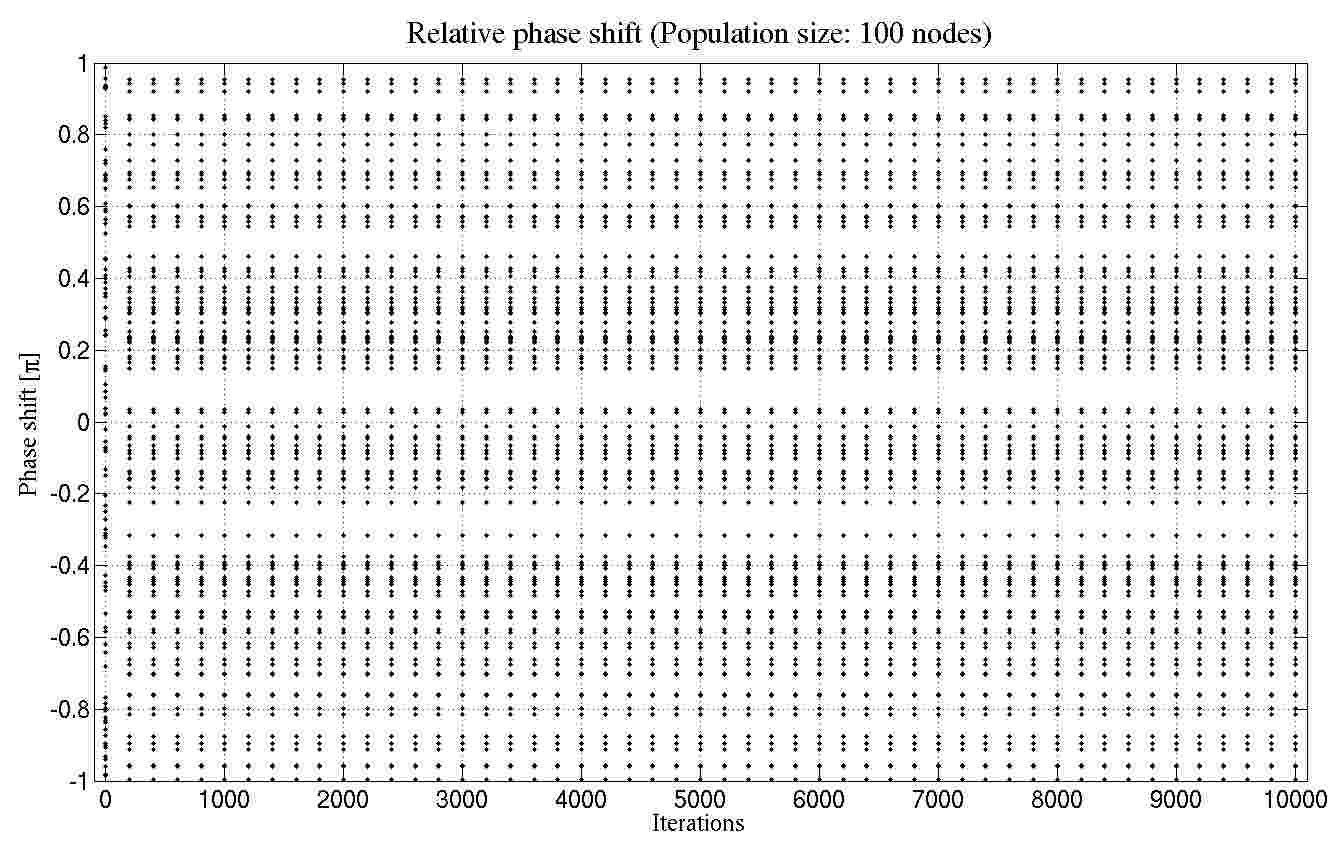}
	\label{label8}}
	\caption{RF signal strength and relative phase shift of received signal components for a network size of 100 nodes after 10000 iterations. Nodes are distributed uniformly at random on a $30m\times30m$ square area and transmit at $P_{TX}=1 mW$ with $p_\gamma=\frac{1}{n}$. {\scriptsize(\copyright 2011 IEEE  Published by the IEEE CS, CASS, ComSoc, IES, SPS)}}
	\label{figureIncreasedDistance}
\end{figure}

Although the noise power relative to the sum signal increases, synchronisation is possible at about 200 meters distance. 
Observe that in our model with $P_{tx}=1 mW$ we expect a signal strength at the receiver of $0.1 \mu W$ or $-40 dBm$ for each single carrier at this distance.

When the distance is further increased to 300 meters, however, synchronisation is not possible with this configuration, due to the high impact of the noise fluctuation on the received signal.
This  has a higher impact on the signal than the alteration of single carrier signals.

However, when more carrier signals are altered simultaneously, a weak synchronisation is still possible.
Fig.~\ref{figureIncreasedDistance-2} depicts the received carrier signal after 100 iterations for the uniformly distributed process with $p_\gamma=0.2$ and $p_\gamma=0.6$.
We see that the synchronisation quality is improved with increasing $p_\gamma$. 
While the superimposed signal is indistinguishable for $p_\gamma=0.2$, the synchronisation quality increases with $p_\gamma=0.6$.
Although the signal is heavily distorted, the carrier can be extracted.
\begin{figure}
\centering
	\subfloat[Receiver distance: 300 meters, $p_\gamma=0.2$]{\includegraphics[width=0.48\textwidth]{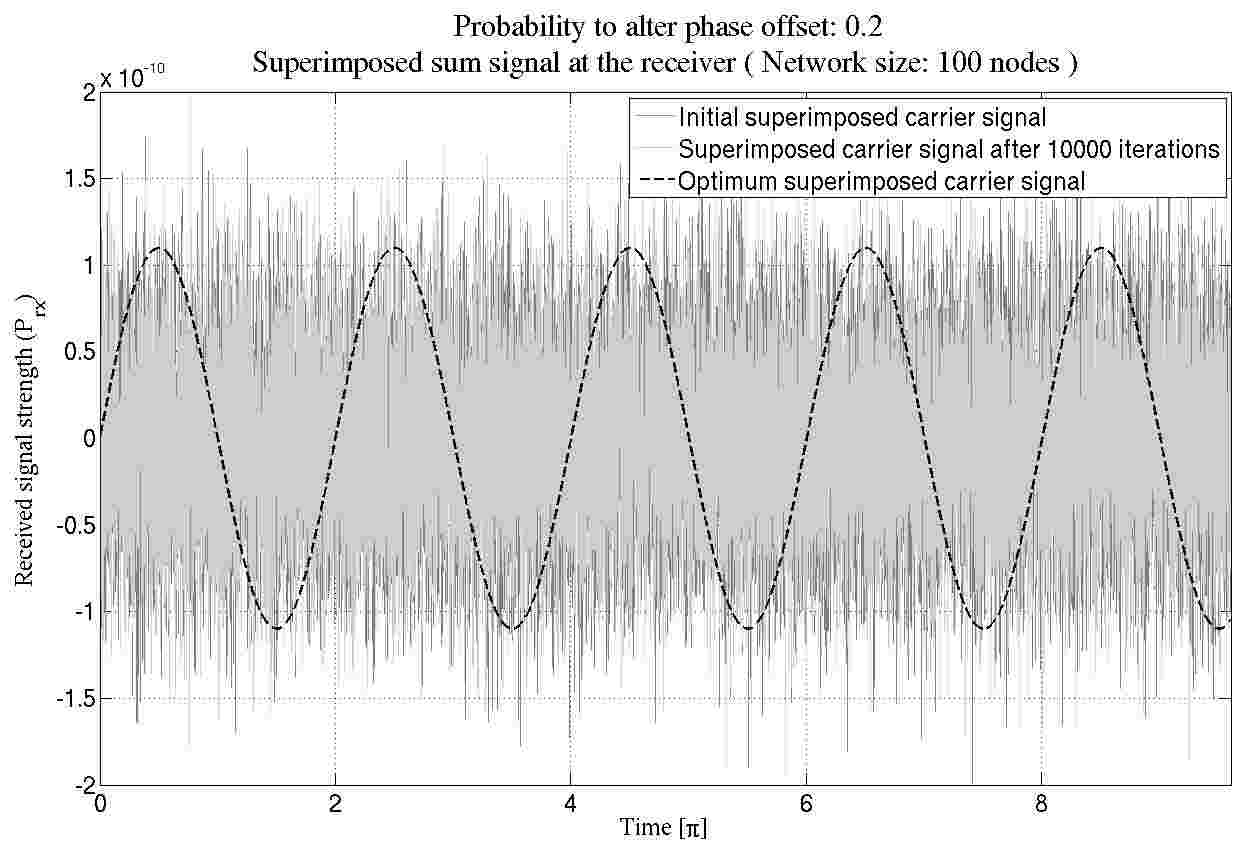}
	\label{label11}}
	\subfloat[Receiver distance: 300 meters, $p_\gamma=0.6$]{\includegraphics[width=0.48\textwidth]{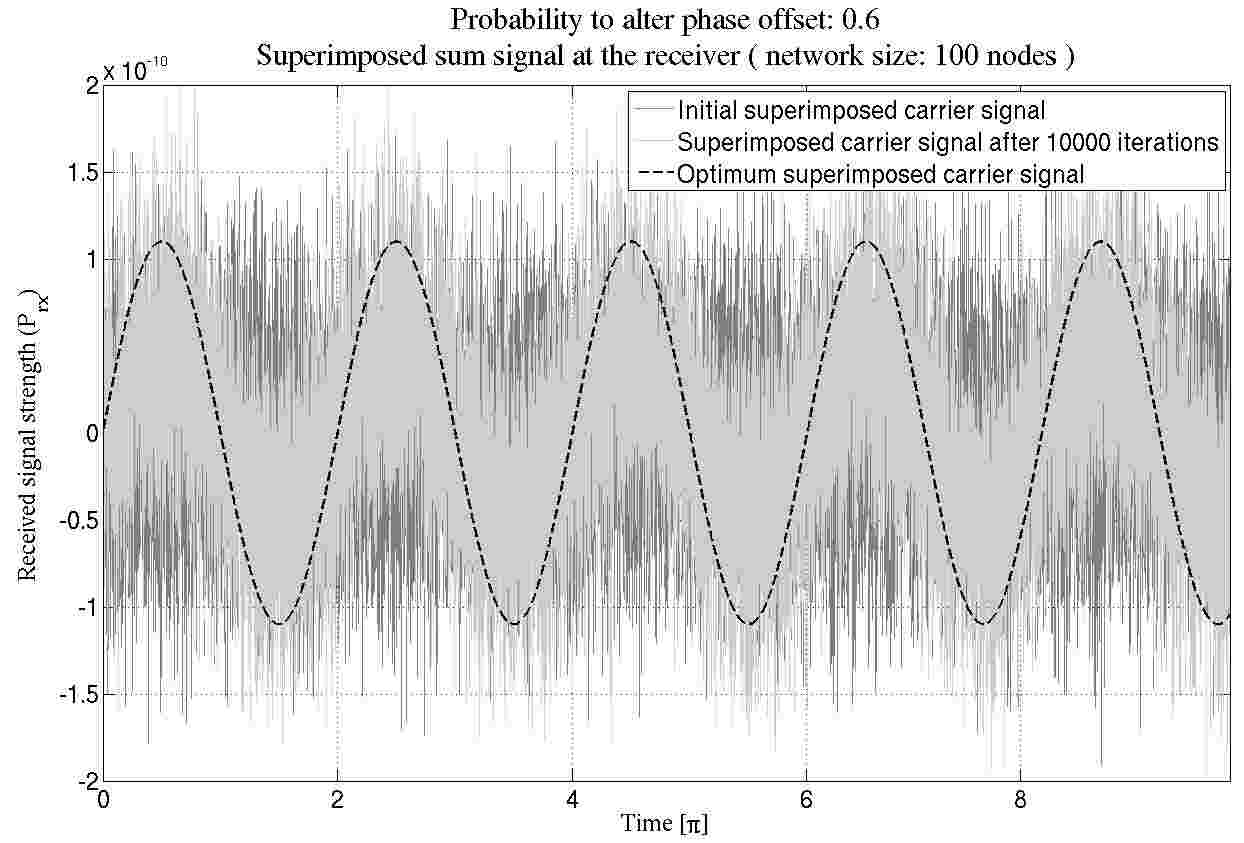}
	\label{label3}}
	\caption{RF signal strength and relative phase shift of received signal components for a network size of 100 nodes after 10000 iterations. Nodes are distributed uniformly at random on a $30m\times30m$ square area and transmit at $P_{TX}=1 mW$. {\scriptsize(\copyright 2011 IEEE  Published by the IEEE CS, CASS, ComSoc, IES, SPS)}}
	\label{figureIncreasedDistance-2}
\end{figure}

\subsubsection{Utilisation of additional feedback information}
We also conducted simulations in which our implementation of the asymptotically optimal algorithm described in section~\ref{sectionThreeUnknowns} is compared to the classical process with normal distributed phase alterations.
When optimum phase offsets are calculated by solving multivariable equations at the transmit nodes, the synchronisation performance can be greatly improved as detailed in section~\ref{sectionThreeUnknowns}.
Fig.~\ref{figureRayan2} depicts the performance improvement achieved by solving multivariable equations to determine the feedback function compared to a global random search approach.
We observe that the global random search heuristic is outperformed already after about 1000 iterations and the feedback value reached is greatly improved.
\begin{figure}\centering
	\subfloat[Phase offset achieved by the proposed optimisation algorithm for distributed adaptive beamforming in WSNs]{
		\includegraphics[width=0.48\textwidth]{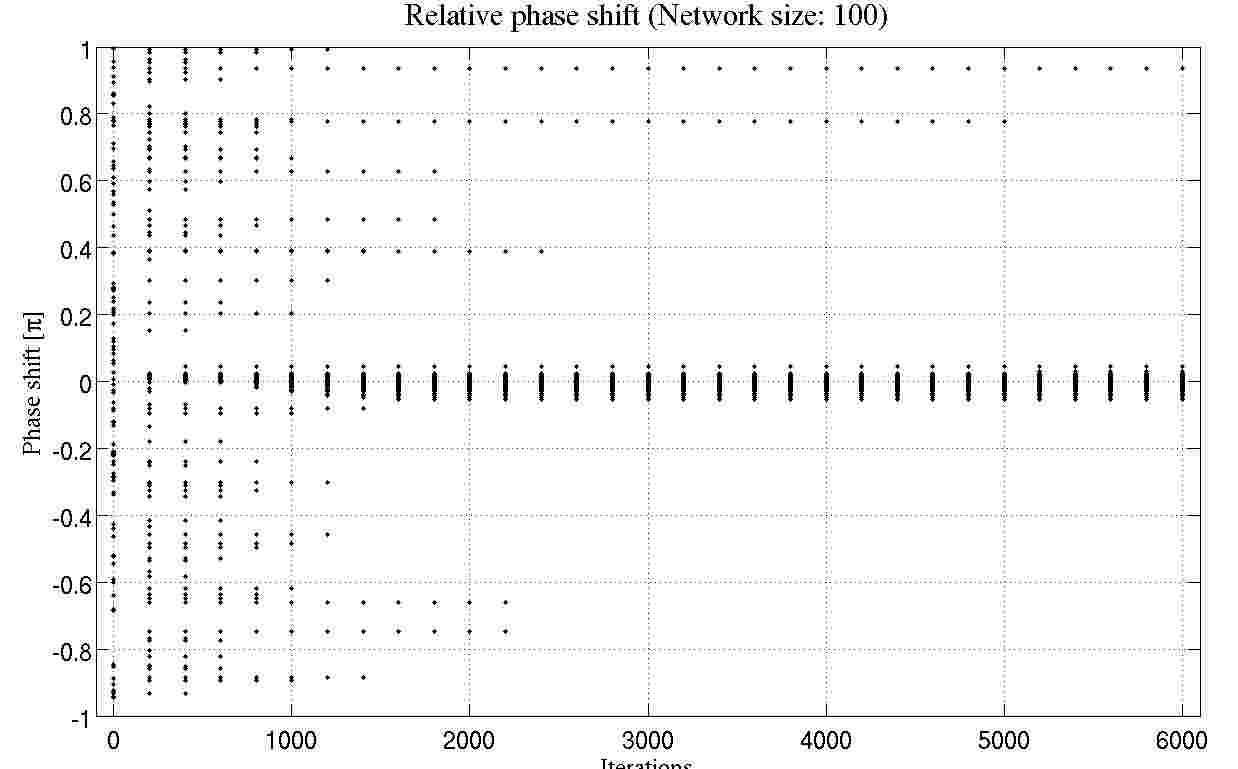}
		}
	\subfloat[Performance of the proposed optimisation algorithm for distributed adaptive beamforming in WSNs]{
		\includegraphics[width=0.48\textwidth]{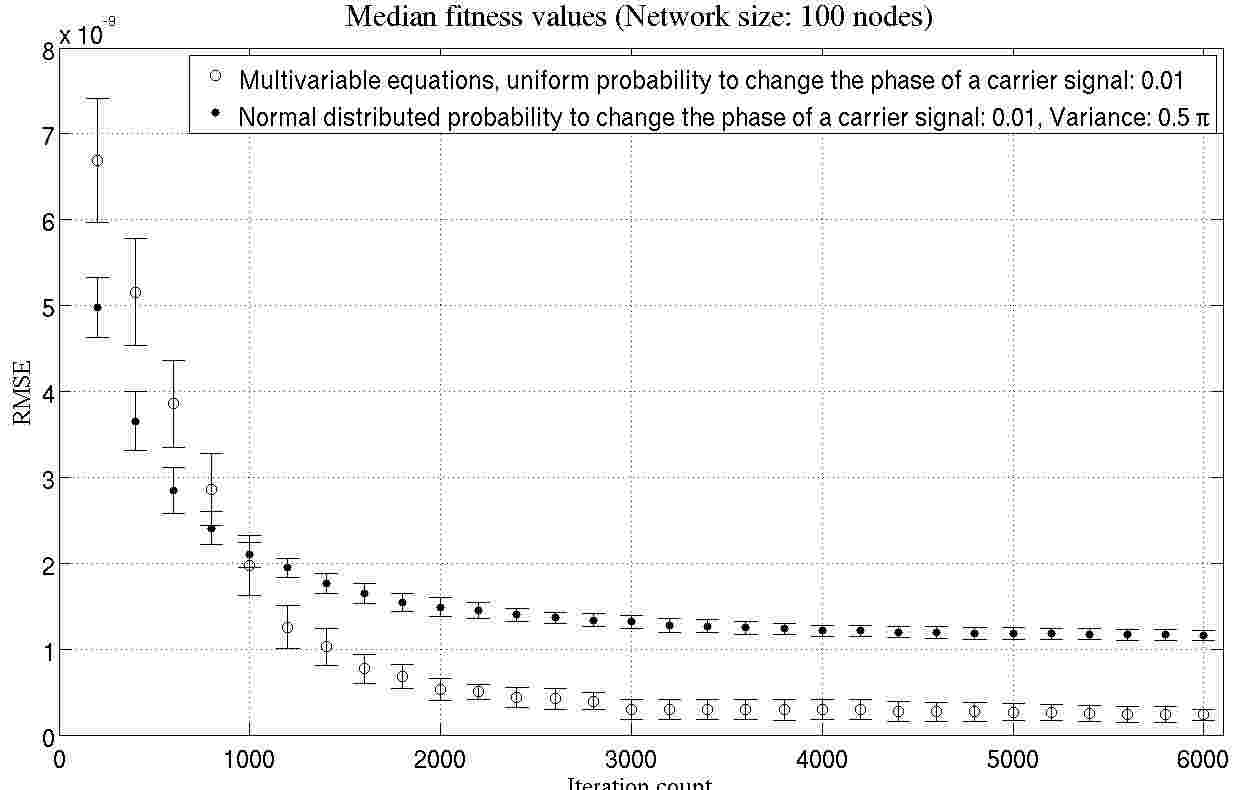}
		\label{figureRayan2-b}
		}
	\caption{Distributed adaptive beamforming with a network size of 100 nodes where phase alterations are drawn uniformly at random. Each node adapts its carrier phase offset with probability $0.01$ in one iteration. In this case, multivariable equation are solved to determine the optimum phase offset of the carrier signal. {\scriptsize(\copyright 2011 IEEE  Published by the IEEE CS, CASS, ComSoc, IES, SPS)}}
	\label{figureRayan2}
\end{figure}

The phase offset of distinct nodes is within $+/- 0.05 \pi$ for up to $99\%$ of all nodes.

\subsection{Near realistic instrumentation}\label{sectionJulian}
We have utilised USRP software radios (http://www.ettus.com) to model a sensor network capable of distributed adaptive transmit beamforming.
The software radios are controlled via the GNU radio framework (http://gnuradio.org).
The transmitter and receiver modules implement the feedback based distributed adaptive beamforming\footnote{The software for our feedback based closed loop implementation is constantly further improved and extended in student projects. It is currently not recommended for productive environments. If you are interested to receive a copy of the code in order to participate in the development and testing of the implementation please contact sigg@ibr.cs.tu-bs.de. }.
For the superimposed transmit channel and the feedback channel we utilised widely separated frequencies so that the feedback could not impact the synchronisation performance.
We conducted experiments with several transmit frequencies of nodes.
In these experiments we repeatedly synchronised the carrier phases of the three transmit devices with the help of the 1-bit feedback based algorithm described in~\cite{5919,5923} with uniform or normal probability distribution on the phase modulation.
Table~\ref{tableJulian} summarises the configuration and results of two experiments with low and high transmit frequencies of $27$MHz and $2.4$GHz, respectively.
\begin{table}
\centering
	\caption{Experimental results of software radio instrumentations {\scriptsize(\copyright 2011 IEEE  Published by the IEEE CS, CASS, ComSoc, IES, SPS)}}
	\label{tableJulian}
	\begin{tabular}{l|l|l}\hline
	Experimental setting	 & Experiment 1 & Experiment 2 \\\hline
	Separation of TX antennas [m] & $\approx 0.21$ & $\approx 0.3$ \\
	Distance to receiver [m] & $\approx 0.75$ & $\approx 4$ \\
	Transmit RF Frequency [MHz] & $f_{TX}=2400$ & $f_{TX}=27$ \\
	Receive RF Frequency [MHz] & $f_{RX}=902$ & $f_{RX}=902$ \\
	Iterations per experiment & $500$ & $200$ \\
	Mobility & stationary & stationary\\
	Identical experiments & $14$ & $10$ \\
	Transmit devices	 & $4$ & $3$ \\
	Receive devices 	& $1$ & $1$\\
	& &\\
	Algorithmic configuration & \\\hline
	Random distribution & uniform & uniform\\
	Phase alteration probability & 0.25 & 0.33\\
	& &\\
	Hardware & &\\\hline
	Transmit board & RFX2400 & LFTX\\
 	Receive board & RFX900 & RFX900\\
	Gain of receive antenna [dBi] & $G_{RX}=3$ & $G_{RX}=3$ \\
	Gain of transmit antenna [dBi] & $G_{TX}=3$ & $G_{TX}=1.5$ \\
	Median gain ($P_{RX}$) [dB] & $2.19$ & $3.72$\\\hline
	\end{tabular}
\end{table}
After 10 experiments at an RF transmit frequency of $27$MHz we achieved a median gain in the received signal strength of $3.72$dB for three independent transmit nodes after 200 iterations.
In $14$ experiments with $4$ independent nodes that transmit at $2.4$GHz the achieved median gain of the received RF sum signal was $2.19$dB after 500 iterations.

For the transmitters we utilised the clock of the first device for all transmit nodes. 
The receive node utilises its own clock and is therefore not synchronised to any of the transmit nodes. 
Apart from this clock synchronisation no other communication or synchronisation between transmitters was applied.
In future implementations it is possible to utilise GPS for the clock synchronisation.

In a third experiment we altered the transmission distance and the phase alteration variance for a normal distributed random process. 
Fig.~\ref{figureExperimentalSetting} depicts our experimental setting.
\begin{figure}
	\centering
	\includegraphics[width=9cm,height=7cm]{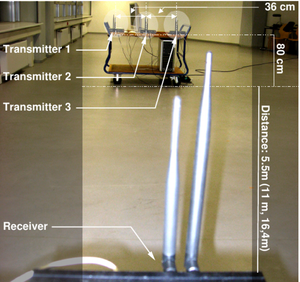}
\caption{Experimental instrumentation of distributed adaptive beamforming among three transmit USRP devices and one receive USRP device. {\scriptsize(\copyright 2011 IEEE  Published by the IEEE CS, CASS, ComSoc, IES, SPS)}}
\label{figureExperimentalSetting}
\end{figure}
Table~\ref{tableInstrumentation} summarises our experimental configuration.
\begin{table}
\centering
	\caption{Configuration of the experiment {\scriptsize(\copyright 2011 IEEE  Published by the IEEE CS, CASS, ComSoc, IES, SPS)}}
	\label{tableInstrumentation}
	\begin{tabular}{l|l}\hline
	Experimental setting & Experiment 3\\\hline
	Separation of transmit antennas [m] & 0.36 \\
	Distance to receive antenna [m] & 5.5 / 11 / 16.4\\
	Transmit frequency [MHz] & $f_{TX}=2400$ \\
	Receive frequency [MHz] & $f_{RX}=902$ \\
	Iterations per experiment & $400$ \\
	Mobility & stationary \\
	Identical experiments & $12$ \\
	Transmit devices & 3\\
	Receive devices & 1\\\\
	
	Algorithmic configuration & \\\hline
	Random distribution of the phase alteration & normal\\
	Phase alteration probability & 0.33 / 0.66 / 1.00\\
	Variance of normal distribution [$\pi$] & $0.25$ / 1\\\\
	
	Hardware & \\\hline
	Transmit board &  RFX2400\\
	Receive board & RFX900\\
	Transmit antenna & VERT2450\\
	Receive antenna& VERT900\\
	Gain of receive antenna [dBi] & $G_{RX}=3$ \\
	Gain of transmit antenna [dBi] & $G_{TX}=3$ \\\hline
	\end{tabular}
\end{table}

Carrier phases have been adapted for each transmit device independently following a normal distributed random process.
We modified the probability to alter the phase offset of one device and the variance for its normal distributed random process as well as the distance between transmit and receive devices.

Some results derived are depicted in Fig.~\ref{figureExperimentalSettingResults}.
In the figure, the mean gain of the received sum signal over all 12 experiments to the initially received sum signal is depicted.
As expected, we observe that the synchronisation process differs for different parameter settings.
Again, best results are achieved when small changes are applied in each iteration. 
Therefore, the experiments in which the phase alteration probability and the variance are small achieve superior results.
\begin{figure}
	\centering
	\subfloat[Mean gain in the signal strength at a transmission distance of 16.4 meters and a variance of the random process of $0.25\pi$]{\includegraphics[width=8cm,height=5cm]{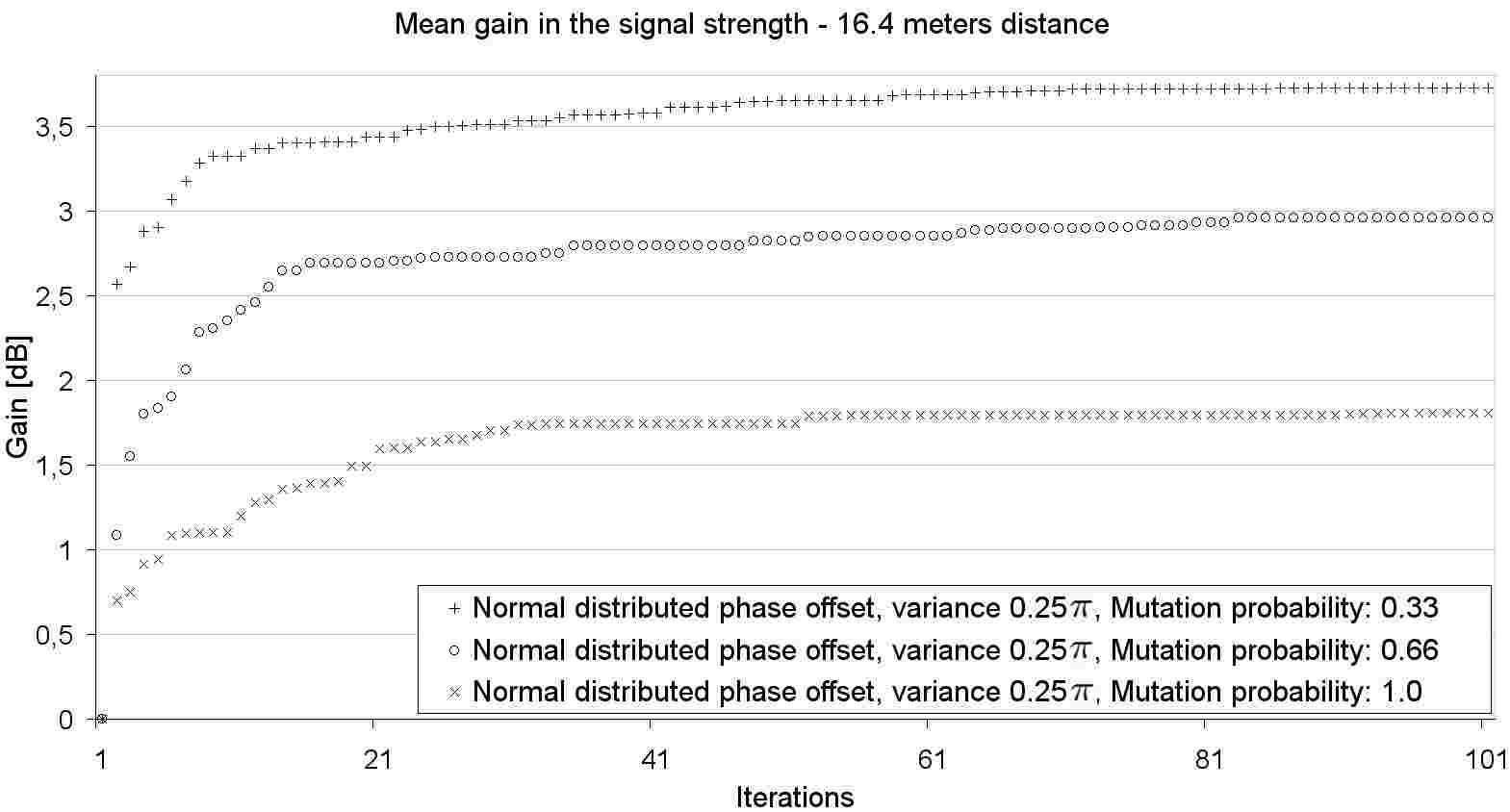}\label{figureE}}
	\subfloat[Mean gain in the signal strength at a transmission distance of 5.5 meters and a variance of the random process of $0.25\pi$ and $\pi$]{\includegraphics[width=8cm,height=5cm]{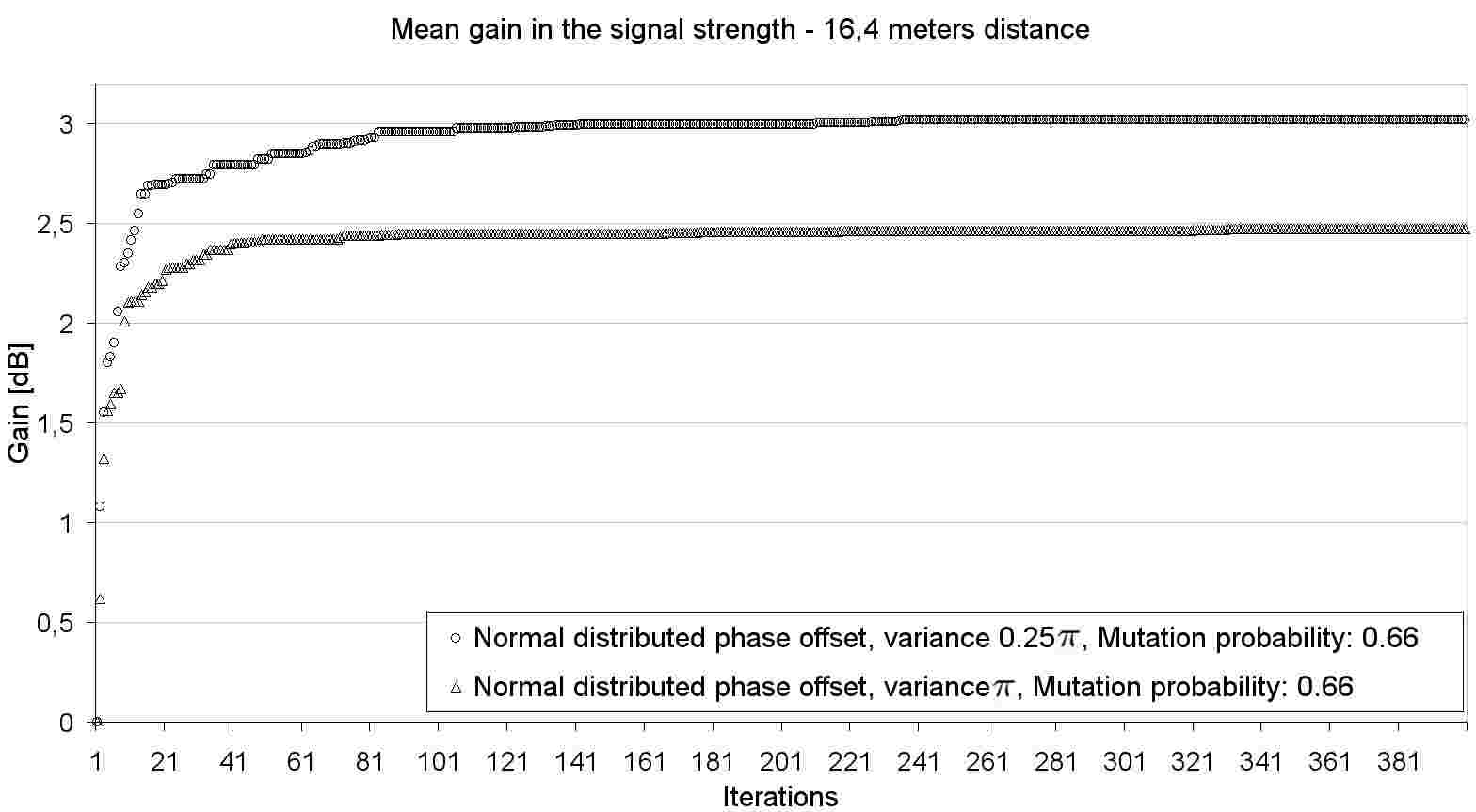}}
	\caption{Mean gain in the signal strength of three collaboratively transmitting devices {\scriptsize(\copyright 2011 IEEE  Published by the IEEE CS, CASS, ComSoc, IES, SPS)}}
	\label{figureExperimentalSettingResults}
\end{figure}

\subsection{Conclusion}\label{sectionConclusionBF01}
We have considered randomised search approaches to solve the problem of distributed adaptive transmit beamforming.
In an analytic consideration an asymptotically tight bound on the expected optimisation time of $\Theta(n\cdot k\cdot\log(n))$ was derived.

Additionally, a protocol to further reduce the optimisation time and energy consumption of distributed adaptive beamforming was introduced.
In this protocol, the problem was divided into sub-problems that were solved iteratively.
Since the decrease in the synchronisation time is greater than the increase in transmission power in smaller clusters, this approach can improve the optimisation time and reduce energy consumption.

Furthermore, an asymptotically optimal algorithm was derived.
For this approach we considered the possibility to estimate the unknown feedback function by an individual node so that an optimisation approach is possible that scales linearly with the network size $n$.
This approach is asymptotically optimal since each carrier signal has to be considered at least once individually in order to find its optimum phase offset.

In mathematical simulations we demonstrated the effect of several configurations for distributed adaptive transmit beamforming with uniform and normal distributed phase alteration methods.
Generally, a low mutation probability translates to a better performance in the phase synchronisation process.
An adaptive probability over the course of the optimisation might further improve the optimisation speed. 
While a moderate mutation probability is beneficial at the beginning of the simulation, a smaller mutation probability shows an improved optimisation speed later in the process.
Also, our implementation of the asymptotically optimal method greatly outperforms the global random search approach in the synchronisation achieved and the optimisation speed.

Finally, in an instrumentation with USRP software radios we demonstrated the feasibility of distributed adaptive transmit beamforming in a concrete implementation with up to four transmitters.

\subsection*{APPENDIX A: On the multimodality of the feedback function}\label{AppendixA}
We easily see that the feedback function is multimodal.
The reason is that, given the search point corresponding to an optimum sum signal $\zeta_{\mbox{\footnotesize opt}}$ we can state another optimum by adding the same phase offset $\gamma'$ to all carrier signals.
In particular, the feedback function is weak multimodal so that no local optimum exists.

\subsubsection{Identical transmit frequencies}
When carrier frequencies among nodes are identical a local optimum exists if we can identify at least one search point $s_{\overline{\zeta}}$ for which all small phase modifications decrease the feedback value, while some larger modifications increase it.
The smallest possible modification is realised when the transmit phase is altered for exactly one carrier signal $\zeta_i$.
Fig.~\ref{figureLocalMinima} illustrates that the feedback of a signal is given by the distance between the rotation angles $\varphi_{\mbox{\footnotesize opt}}$ and $\varphi_i$ of an optimal configuration $s_{\zeta_{\mbox{\footnotesize opt}}}$ and $s_{\zeta_i}$ as $\left|\cos(\varphi_{\mbox{\footnotesize opt}})-\cos(\varphi_i)\right|$.
\begin{figure}
	\centering
	\includegraphics[width=.9\textwidth]{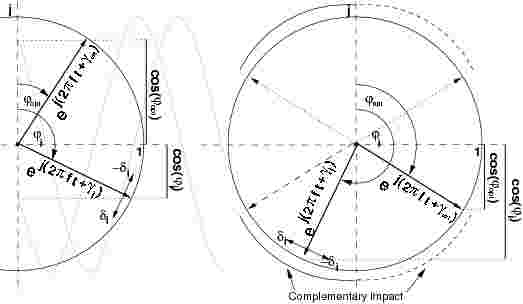}
	\caption{Fitness calculation of signal components. The feedback of the superimposed sum signal is impacted by the relative phase offset of an optimally aligned signal and a carrier signal $i$. {\scriptsize(\copyright 2011 IEEE  Published by the IEEE CS, CASS, ComSoc, IES, SPS)}}
	\label{figureLocalMinima}
 \end{figure}

Compared to $s_{\mbox{\footnotesize opt}}$ no configuration short of the optimum configuration $s_i=s_{\mbox{\footnotesize opt}}$ exists for which the phase offset between signal components is increased for phase offset $\delta_i$ regardless of the sign of $\delta_i$ (cf. Fig.~\ref{figureLocalMinima}).

\subsubsection{Distinct transmit frequencies}
When signal frequencies differ, the feedback function is not affected by phase modifications only.
The reason for this is that we can for every positive contribution to the feedback function also find a negative contribution of the same amount in the common period of $\zeta_{\mbox{\footnotesize opt}}$ and $\zeta_i$.

\subsection*{APPENDIX B: Calculation of optimal hierarchy depth and cluster size}\label{AppendixD}
We estimate the expected optimisation time for a network of size $n$ by $E[T_{\mathcal{P}n}]=c\cdot k\cdot n\cdot\log(n)$ for a suitable constant $c$ and the expected energy consumption by $E[\mathcal{E_P}_n]=c\cdot k\cdot n\cdot\log(n)\cdot \overline{\mathcal{E_P}_n}$ where $\overline{\mathcal{E_P}_n}$ is the energy consumption of all $n$ nodes in one iteration \cite{4020}.
A hierarchy and cluster structure that minimises these formulae when summed over all hierarchy stages is optimal in our sense.
We derive the optimum cluster sizes and hierarchy depths by integer programming.
For a cluster size of $m$ the above formulae have the property 
\begin{eqnarray}
	E[T_{\mathcal{P}n}]=E[T_{\mathcal{P}\frac{n}{m}}]\cdot\frac{n}{m}\cdot E[T_{\mathcal{P}m}] \nonumber\\
	E[\mathcal{E_P}_n]=E[\mathcal{E_P}_\frac{n}{m}]\cdot\frac{n}{m}\cdot E[\mathcal{E_P}_m].\nonumber
\end{eqnarray}
We define the recursion by 
\begin{eqnarray}
	E_{\mbox{\footnotesize opt}}[T_{\mathcal{P}n}]&=&min_m\left[E_{\mbox{\footnotesize opt}}[T_{\mathcal{P}\frac{n}{m}}]\cdot\frac{n}{m}\cdot E_{\mbox{\footnotesize opt}}[T_{\mathcal{P}m}] \right]\nonumber\\
	E_{\mbox{\footnotesize opt}}[\mathcal{E_P}_n]&=&min_m\left[E_{\mbox{\footnotesize opt}}[\mathcal{E_P}_\frac{n}{m}]\cdot\frac{n}{m}\cdot E_{\mbox{\footnotesize opt}}[\mathcal{E_P}_m]\right]\nonumber
\end{eqnarray}
and the start of the recursion by $E_{\mbox{\footnotesize opt}}[T_{\mathcal{P}\eta}]$ and $E_{\mbox{\footnotesize opt}}[\mathcal{E_P}_\eta]$ with $\eta$ being the minimum feasible cluster size when the maximum transmission power and distance are given.
Since $\eta$ is dependent on the distance to the receiver, it can be calculated over the round trip time between the sensor network and the receiver.

The time required for the calculation of the optimum hierarchy depth and cluster sizes is quadratic.
With a network of $n$ nodes, at most $n^2$ distinct terms $E_{\mbox{\footnotesize opt}}[T_{\mathcal{P}i}]$ and $E_{\mbox{\footnotesize opt}}[\mathcal{E_P}_i]$ with $i\in\{1,\dots,n\}$ are of relevance.
We can start by calculating $E_{\mbox{\footnotesize opt}}[\mathcal{E_P}_\eta]$ and $E_{\mbox{\footnotesize opt}}[T_{\mathcal{P}\eta}]$ and obtain all other values by table look-up according to $E_{\mbox{\footnotesize opt}}[T_{\mathcal{P}n}]$ and $E_{\mbox{\footnotesize opt}}[\mathcal{E_P}_n]$ in time $\mathcal{O}(n^2)$ since every one of the (at most) $n$ entries has not more than $n$ possible predecessors.
\vfill
\pagebreak

\section[A fast binary feedback-based distributed adaptive carrier synchronisation for transmission among clusters of disconnected IoT nodes in smart spaces]{A fast binary feedback-based distributed adaptive carrier synchronisation for transmission among clusters of disconnected IoT nodes in smart spaces \footnote{Reprinted from 'Elsevier Journal on Ad Hoc Networks, vol. 16, Stephan Sigg, A fast binary feedback-based distributed adaptive carrier synchronisation for transmission among clusters of disconnected IoT nodes in smart spaces, pp. 120-130, May 2014, with permission from Elsevier.'}}\label{sectionOriginalBF02}
We propose a transmission scheme among groups of disconnected IoT devices in a smart space. 
In particular, we propose the use of a local random search implementation to speed up the synchronisation of carriers for distributed adaptive transmit beamforming. 
We achieve a sharp bound on the asymptotic carrier synchronisation time which is significantly lower than for previously proposed carrier synchronisation processes.
Also, we consider the impact of environmental conditions in smart spaces on this synchronisation process in simulations and a case study.

\subsection{Introduction}
The advancing miniaturisation of electronics and its integration into everyday objects fosters smart spaces as an antecedent to an Internet of Things (IoT).
In such environments, arbitrarily distributed, sharply resource restricted devices share data acquired by their sensors and cooperate in their data processing in order to establish an intelligent and responsive smart space.
Instead of being equally distributed among an environment, the processing and communicating devices are likely clustered in distant physical spaces.
Consequently, for the sharply resource restricted devices it might be difficult to establish a connection among the spread clusters of devices.
Natural clusters are given, for instance, by the set of devices worn or carried by a person or also by a working place constituted of a high density of electronically enhanced tools.
In a smart space, clusters should be connected to share information and provide additional value to an individual in this space.

Since IoT devices (possibly featuring RFID or Organic Electronics~\cite{InNetworkProcessing_Jakimovski_2011}) will have a sharply restricted transmission range for their low energy budget available, these clusters, however, might be frequently disconnected.
From a communication perspective, the signal strength of such resource restricted devices in one cluster might be too weak to reach a remote cluster at sufficient Signal-to-Noise Ratio (SNR).
Therefore, although nodes in a remote cluster might sense some activity on the channel, the signal strength is too weak for them to decode information.

One solution to increase the transmission range of nodes in a cluster and thereby to establish a connection is to combine their transmit signals during simultaneous, phase-aligned transmission on the wireless channel. 

By superimposing signals on the wireless channel in-phase, they are accumulated and therefore strengthened so that the transmission range can be extended.

In the literature, several approaches for such beamforming among distributed nodes are proposed~\cite{Seo_2008,Ochiai_2005,Mudumbai_2010b,4023}.
The most common ones require either code divisioning techniques or for the receiving node to conduct significant computation~\cite{Barriac_2004,Hagmann_1981}.
Such process, however, is exhaustive for sharply resource restricted IoT nodes in a smart space.
A simpler, less resource consuming, method was proposed by Mudumbai and others in~\cite{Mudumbai_2010b,Mudumbai_2009,Bucklew_2008}.
The authors employ an iterative random search mechanism in which nodes in each iteration may randomly change the phase of their carrier signal conditioned on a binary feedback from the receiver.
This approach is better suited for IoT nodes for its low computational complexity to randomly draw alternative signal phases at nodes.
Since the binary feedback can be encoded as an energy efficient on/off (burst transmission $=1$, no transmission $=0$) scheme, the required SNR can be low.
This carrier synchronisation scheme is applicable also with inexpensive crystal oscillators with high frequency derivations~\cite{DistributedBeamforming_Mudumbai_2011,DistributedBeamforming_Quitin_2013} such as we can expect for IoT devices.

For this scheme we derived a sharp asymptotic bound on the expected optimisation time for $n$ transmit devices and $k$ possible transmit carrier phases each, in the order of $\Theta(n\cdot k \cdot \log n)$ iterations in~\cite{4022}. 
This performance is the main drawback of the beamforming scheme for smart spaces.
Significant count of iterations and therefore a high number of transmissions are required which slows down the synchronisation~\cite{CarrierSynchronisation_2012_Mudumbai}.

In~\cite{4032} an alternative asymptotically optimal, iterative optimisation approach was presented.
Although its optimisation performance was as low as $\mathcal{O}(n)$, this improved performance was achieved at the cost of a more descriptive receiver feedback so that it can not be implemented as a simple on/off scheme and is therefore less well suited for resource restricted IoT nodes.

Another possibility to improve the synchronisation performance is to modify the random search for synchronised transmit phases at nodes.
The original approach employed an evolutionary random search~\cite{5919}.
However, as indicated in~\cite{4023}, the search space of the problem is rather simple and does not contain any local optima.

Therefore, we propose in this paper to utilise a fast local random search to establish carrier synchronisation among nodes.
In particular, we derive an asymptotic upper and lower bound for the expected optimisation time and compare the approach to the one presented in~\cite{4022,5919} in simulations and a case study.

The contributions of this paper are
\begin{enumerate}
\item an improved and more exact upper bound for iterative feedback-based closed-loop carrier synchronisation with local random search, 
\item a lower bound in the same asymptotic order,
\item a discussion of environmental impacts on the performance of iterative feedback-based carrier synchronisation,
\item simulations and
\item a case study with software-defined radio devices.
\end{enumerate}
This consideration of a local random search method for feedback-based iterative carrier synchronisation improves and extends the discussion on a simple bound in~\cite{4023}.
Our analysis provides an improved and more exact upper bound and in addition derives a lower bound in the same asymptotic order.

After introducing the related work and discussing iterative random carrier synchronisation in section~\ref{sectionRelatedWorkBF02} we propose a local random search mechanism and study its expected synchronisation time in section~\ref{sectionLocalRandomSearch}.
In section~\ref{sectionImpacts} we show that the synchronisation quality of iterative feedback-based carrier synchronisation among IoT devices in a smart space is impacted by environmental stimuli.
In section~\ref{sectionSimulation}, the impact of environmental stimuli on the optimum choice of optimisation parameters is demonstrated in mathematical simulations and experimental case studies.
Section~\ref{sectionConclusionBF02} draws our conclusion.

\subsection{Distributed adaptive carrier synchronisation}\label{sectionRelatedWorkBF02}
For distributed IoT devices in a smart space to establish a transmission beam to a remote receiver, carrier phases of transmit signals have to be synchronised with respect to the receiver location and the phase and frequency offset of the distributed local oscillators.
After synchronisation, a message $m(t)$ is transmitted simultaneously by all transmit devices $i\in[1..n]$ as 
\begin{equation}
	\zeta_i(t)=\Re\left(m(t)e^{j(2\pi (f_c+f_i)t+\gamma_i)}\right)\label{equationZero}
\end{equation}
so that the receiver observes the superimposed signal 
\begin{eqnarray}
\zeta_{\mbox{\footnotesize sum}}(t)+\zeta_{\mbox{\footnotesize noise}}(t)=\hspace{4.5cm}\nonumber\\
\Re\left(m(t)\sum_{i=1}^n\mbox{RSS}_ie^{j2\pi (f_c+f_i)t+(\gamma_i+\phi_i+\psi_i)}\right)+\zeta_{\mbox{\footnotesize noise}}(t)\label{equationOneBF02}	
\end{eqnarray}
with minimum phase offset between carrier signals: 
\begin{eqnarray}
\min\left(\left|(\gamma_i+\phi_i+\psi_i)-(\gamma_j+\phi_j+\psi_j)\right|\right)\\
\forall i,j\in[1..n],i\not=j.     \nonumber
\end{eqnarray}
In equation~(\ref{equationZero}) and equation~(\ref{equationOneBF02}), $f_i$ denotes the frequency offset of device $i$ to a common carrier frequency $f_c$.  
The values $\gamma_i$, $\phi_i$ and $\psi_i$ represent the carrier phase offset of node $i$ as well as the phase offset in the received signal component due to the offset in the local oscillators of nodes ($\phi_i$) and due to distinct signal propagation times ($\psi_i$).
$\zeta_{\mbox{\footnotesize noise}}(t)$ denotes the noise and interference in the received sum signal.
We assume additive white Gaussian noise (AWGN) here.
With $\mbox{RSS}_i$ we describe the received signal strength of IoT device $i$.

Algorithms for distributed adaptive carrier synchronisation are distinguished by closed-loop carrier synchronisation and open-loop carrier synchronisation techniques~\cite{Mudumbai_2009}.
Closed-loop synchronisation can be achieved by a master-slave approach as detailed in~\cite{5934}.
The central idea is that transmit IoT devices send a synchronisation sequence simultaneously on code-divisioned channels to a destination device in the smart space.
The destination calculates the relative phase offset of the received signals and broadcasts this information to all transmitters which then adapt their carrier signals accordingly.

Due to the high computational complexity burden for the destination node to derive the relative phase offset of all received signals and for all nodes due to the utilisation of code divisioning techniques, this implementation is not suggestive for the application in smart spaces with a high count of strictly resource limited devices.

Alternatively, in a master-slave open-loop synchronisation~\cite{5923}, the relative phase offset among nodes is determined by the transmit nodes with a method similar to~\cite{5934} but among transmit IoT devices only.
The receiver then broadcasts a carrier signal once so that the transmit nodes are able to correct their frequency offsets.
In this method, however, the high complexity for the nodes is shifted from the receiver node to one of the transmit nodes.
Therefore, this approach also suffers from its high computational complexity.

A simpler and less resource demanding distributed carrier synchronisation scheme was proposed in~\cite{5919}. 
This closed-loop approach is computationally cheap at the cost of increasing the time required for carrier synchronisation. 
It utilises a binary feedback on the achieved synchronisation quality that is transmitted in each iteration from a remote receiver~\cite{Mudumbai_2009,5920}.
In particular, such binary feedback can be implemented by a simple on/off burst scheme also for sharply resource restricted IoT devices.

The central optimisation procedure consists of $n$ devices $i\in[1,\dots,n]$ randomly altering the phases $\gamma_i$ of their carrier signal $\zeta_i(t)$
in each iteration.
Implicitly, with this process a global random search is applied.
The search space $\mathcal{S}$ is spanned by all possible combinations of carrier frequencies and carrier phase offsets for all transmit nodes (cf. figure~\ref{figure01BF02}).
\begin{figure}\centering
     \includegraphics[width=10cm]{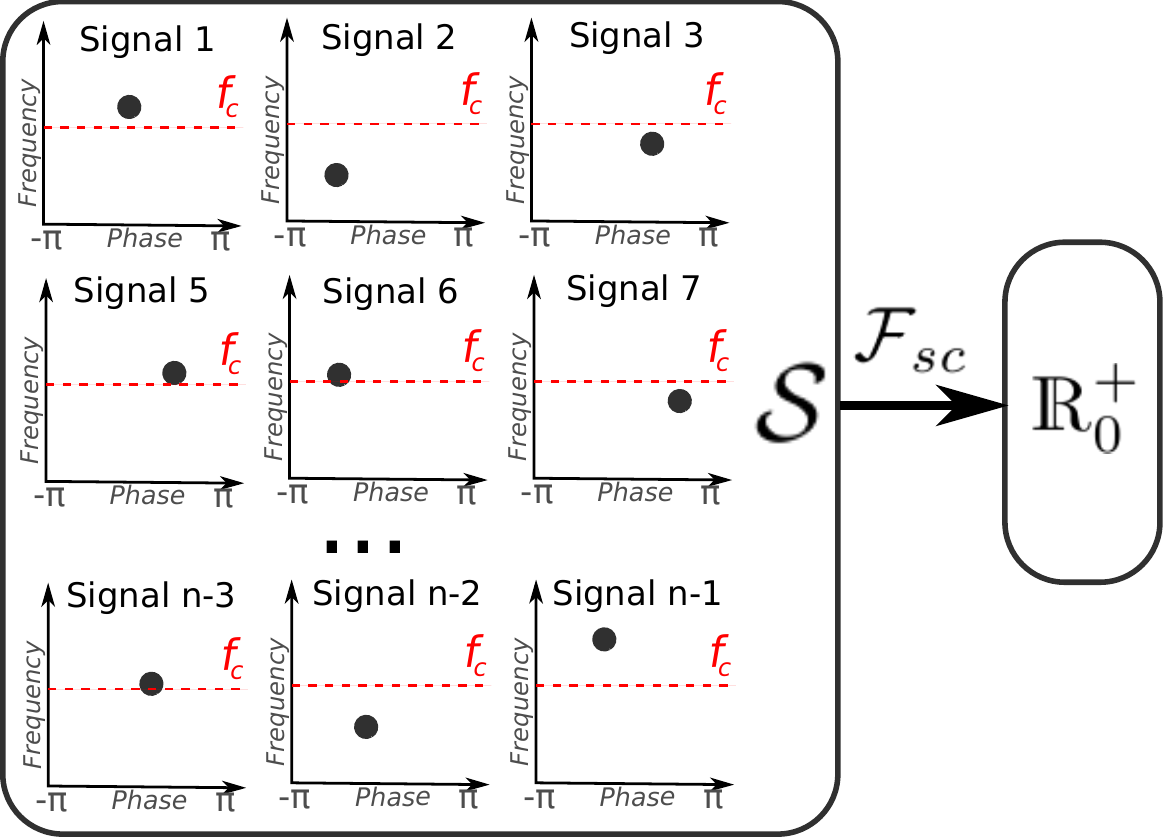}
\caption{Search space and score function for binary feedback-based distributed adaptive carrier synchronisation}
\label{figure01BF02}
\end{figure}

The figure illustrates the global search space constituted from all possible combinations of phase and frequency configurations of transmit signals at local IoT nodes.

Each specific phase-frequency combination $s\in\mathcal{S}$ is associated with a score $\mathcal{F}_{sc}:\mathcal{S}\rightarrow\mathds{R}^+_0$ that denotes its synchronisation quality.
Without loss of generality we assume that the optimisation aim is to maximise $\mathcal{F}_{sc}$.
A natural choice to compute such a score value is, for instance, the Signal-to-Noise-Ratio (SNR) of the received sum signal detailed in equation~(\ref{equationOneBF02}).

Feedback-based distributed carrier synchronisation approaches are characterised by the parameters
\begin{description}
	  \item[$P_{\mbox{\footnotesize mut},\gamma_i}$] Probability to alter the phase-offset of device $i$ ($P_{\mbox{\footnotesize mut},\gamma_i}\in[0,1]$)
	  \item[$P_{\mbox{\footnotesize mut},f_i}$] Probability to alter the frequency-offset of device $i$ ($P_{\mbox{\footnotesize mut},f_i}\in[0,1]$)
	  \item[$P_{\mbox{\footnotesize dist},\gamma_i}$] Probability distribution (phase) for the random process at device~$i$ \linebreak ($P_{\mbox{\footnotesize dist},\gamma_i} \in\{\mbox{normal, uniform, \dots}\}$)
	  \item[$P_{\mbox{\footnotesize dist},f_i}$] Probability distribution (frequency) for the random process at device~$i$  ($P_{\mbox{\footnotesize dist},f_i}\in\{\mbox{normal, uniform, \dots}\}$)
	  \item[$V_{\gamma_i}$] Variance for the random phase alteration process at device $i$ ($V_{\gamma_i} \in[0,\pi]$)
	  \item[$V_{f_i}$] Variance for the random frequency alteration process at device $i$ ($V_{f_i} \in[-f_\Delta,f_\Delta]$ for frequency range $f_\Delta$)
\end{description}
The carrier synchronisation process is described by algorithm~\ref{algorithmOne}.
\begin{algorithm}[t]
     \begin{algorithmic}[1]
          \REPEAT
	  \STATE With probability $P_{\mbox{\footnotesize mut},\gamma_i}$ and $P_{\mbox{\footnotesize mut},f_i}$, each transmit node $i$ adjusts its carrier phase offset $\gamma_i$ and frequency offset $f_i$ following a probability distribution $P_{\mbox{\footnotesize dist},\gamma_i}$ ($P_{\mbox{\footnotesize dist},f_i}$) with variance $V_{\gamma_i}$ ($V_{f_i}$).\label{stepOne}
	  \STATE Nodes transmit to the destination simultaneously as a distributed beamformer.\label{stepTwo}
	  \STATE Receiver estimates the level of phase synchronisation of the received sum signal $\zeta_{\mbox{\footnotesize sum}}(t)+\zeta_{\mbox{\footnotesize noise}}(t)$ (for instance by the SNR).
	  \STATE A binary feedback (e.g. burst/no burst) indicating whether this value has improved is broadcast.\label{stepFour}
	  \STATE When the feedback is worse than in the previous iteration, each transmit node that altered its phase offset in this iteration reverses this decision, by re-setting its carrier phase offset $\gamma_i$ to the previous value.\label{stepFive}
	  \UNTIL{Sufficient synchronisation achieved}
     \end{algorithmic}
\caption{Feedback-based distributed adaptive carrier synchronisation \label{algorithmOne}}
\end{algorithm}
Intuitively, each IoT node may in one iteration alter its transmit carrier phase offset (step~\ref{stepOne}), superimpose a synchronisation signal simultaneously with all other smart devices (step~\ref{stepTwo}) and receive a binary feedback on the quality of the synchronisation (better/worse; step~\ref{stepFour}).
These iterations are repeated until a random distribution of carrier phases is achieved that scores a sufficient synchronisation quality~\cite{Seo_2008,Mudumbai_2010b,4024}.
Initially, independent and identically distributed (i.i.d.) phase offsets $\gamma_i$ of carrier signals are assumed.
Since a decreasing signal quality is not accepted (cf.~step~\ref{stepFive} in algorithm~\ref{algorithmOne}), and since a global random search is implemented by this approach (every possible combination of carrier phase offsets of nodes has a positive probability in each iteration) the method eventually converges to the optimum with probability~1~\cite{Mudumbai_2010b}. 
For this result an idealised environment without noise and interference was considered.
In a realistic environment, the impact of the noise figure determines the accuracy that can be achieved.

In~\cite{CarrierSynchronisation_2012_Rahman}, an implementation of this carrier synchronisation approach was presented for software defined radio (SDR) devices which does not rely on any wired connections between devices (for instance, for clock synchronisation of the SDR nodes). 

The authors of~\cite{Seo_2008} then demonstrated in a case study that the method is feasible to synchronise frequency as well as phase of carrier signal components. 
Without loss of generality we will in our discussion only consider phase synchronisation and assume the frequency synchronisation as perfect.   
Our discussion can be easily extended to cover frequency synchronisation also by adding additional dimensions for the frequency of carrier signals to the search space $\mathcal{S}$~\cite{Seo_2008}.
As an alternative, sufficiently accurate separate frequency synchronisation schemes have been discussed for this approach in~\cite{DistributedBeamforming_Quitin_2012}.

The distinct implementations in the literature differ in the~\ref{stepOne}nd and the~\ref{stepFive}th step of algorithm~\ref{algorithmOne}.
For instance, in~\cite{Seo_2008,Bucklew_2008}, devices alter their carrier phase $\gamma_i$ according to a normal distribution with small variance.
In~\cite{4019}, a uniform distribution with a small probability to alter the phase offset of one individual device is utilised instead.
For a fixed uniform distribution over the whole optimisation process, a sharp asymptotic bound of $\Theta(n\cdot k \cdot \log n)$ on the expected optimisation time was derived~\cite{4022}.
Here, $k$ denotes the maximum number of distinct phase offsets a physical transmitter can generate.

In all previous studies, a global random search is considered, in which nodes choose their next carrier phase and frequency offset uniformly at random from all possible values.
Since the search space does not contain local optima~\cite{4023} we restrict the search neighbourhood to reduce the number of possible next  configurations in one iteration that would worsen the synchronisation quality.
We propose to modify step~\ref{stepOne} of algorithm~\ref{algorithmOne} to follow a local random search instead of the previously applied global random search mechanism.
In particular, an IoT node~$i$ will, when it changes its phase and frequency offset, draw the new values from a restricted neighbourhood of size~$\mathcal{N}$ that is centred around the current values of $\gamma_i$ (and $f_i$).     
This addresses a recent critique expressed in~\cite{CarrierSynchronisation_2012_Mudumbai} regarding the convergence speed for this binary feedback-based iterative adaptive carrier synchronisation.
In section~\ref{sectionLocalRandomSearch} we derive upper and lower bounds on the expected synchronisation performance of a local random search mechanism for feedback-based distributed adaptive carrier synchronisation. 
These bounds improve the existing bounds known for the global random search method. 
Section~\ref{sectionImpacts} shows that the optimum values for the parameters $P_{\mbox{\footnotesize mut},\gamma_i}$, $P_{\mbox{\footnotesize mut},f_i}$, $P_{\mbox{\footnotesize dist},\gamma_i}$, $P_{\mbox{\footnotesize dist},f_i}$, $V_{\gamma_i}$ and $V_{f_i}$ are conditioned on environmental situations in a smart space. 

\subsection{Local random search}\label{sectionLocalRandomSearch}
Recent approaches to 1-bit feedback-based distributed carrier synchronisation utilise a global random search that reaches any search point $s\in\mathcal{S}$ with a positive probability in each iteration~\cite{Seo_2008,Mudumbai_2010b,4023,Bucklew_2008}.
The probability to achieve by these random phase and frequency perturbations of $s$ a search point $s'$ with $\mathcal{F}_{sc}(s')\geq\mathcal{F}_{sc}(s)$ decreases with increasing synchronisation quality ($\mathcal{F}_{sc}(s)$ score).
A restricted search neighbourhood can, however, ensure a constant steady progress since the search space does not contain local optima as derived in~\cite{4022}.
Any local search heuristic that manages to follow a path with increasing $\mathcal{F}_{sc}$ score will find a global optimum with probability~$1$.
We assume that each transmit node is able to apply $k$ distinct phase-offsets and define a global optimum as superimposition of transmit signals in which all phases are within $\frac{2\pi}{k}$ of a superimposition with perfect phase coherency.

\begin{theoremS}\label{theoremZero}
Let $s\in\mathcal{S}$ be a current search point of a local random search algorithm $\mathcal{A}$ for feedback-based distributed adaptive carrier synchronisation with neighbourhood size $\mathcal{N}\leq|\mathcal{S}|$ and $\forall i:P_{\mbox{\footnotesize dist},\gamma_i}=\mbox{uniform}$.
For each phase and frequency perturbation (step~\ref{stepOne} in algorithm~\ref{algorithmOne}) of one transmit carrier signal, the probability to arrive at a search point (superimposition of transmit signals) $s'$ with $\mathcal{F}_{sc}(s')\geq\mathcal{F}_{sc}(s)$ is at least~$\frac{1}{2}$ for each transmit signal that has the optimum search point $s^*$ not within its neighbourhood $\mathcal{N}$.
\end{theoremS}
(Refer to the Appendix for a proof of the results)

Consequently, we divide the following analysis into two phases. 
In the first, the optimum is not within the neighbourhood of at least one node so that at least one node can improve the fitness with probability~$\frac{1}{2}$ or more (cf. figure~\ref{figure02BF}).
\begin{figure}
     \centering
     \includegraphics[width=11cm]{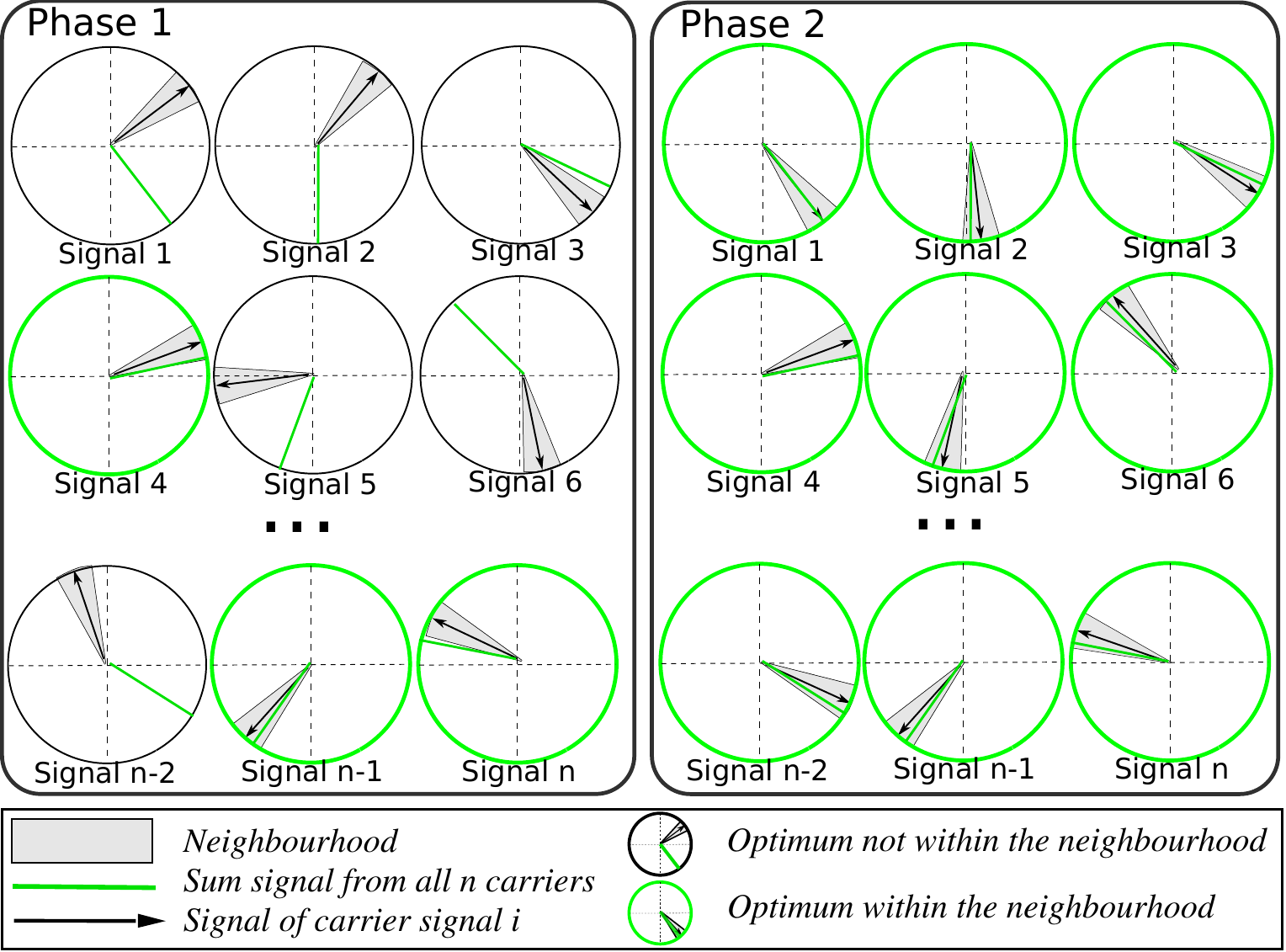}
     \caption{Phase~1 and phase~2 of the synchronisation process}
     \label{figure02BF}
\end{figure}

In the second phase, all nodes have the optimum within their neighbourhood. 
The probability to decrease the distance to the optimum might then be worse than $\frac{1}{2}$.
An optimisation process with restricted neighbourhood-size therefore has a probability of $\frac{1}{2}$ to increase the fitness-value for a long time until the optimum point is within the neighbourhood of each single carrier signal. 
The price for this high probability to improve the fitness-value in each iteration is that the chance to achieve great progress in one step (as possible with an unrestricted neighbourhood-size) is lost.
Since this event is significantly less probable, we are prepared to pay this price.
An individual node $i$ then alters the phase-offset $\gamma_i$ of its carrier-signal uniformly at random within a range of $\left[\gamma_i-\frac{\mathcal{N}}{2},\gamma_i+\frac{\mathcal{N}}{2}\right]$ for suitable $\mathcal{N}$.
To simplify the analysis we represent search points (superimpositions of transmit signals) in a binary encoding.
Figure~\ref{figure03BF} sketches this encoding.
\begin{figure}
\centering
     \includegraphics[width=14cm]{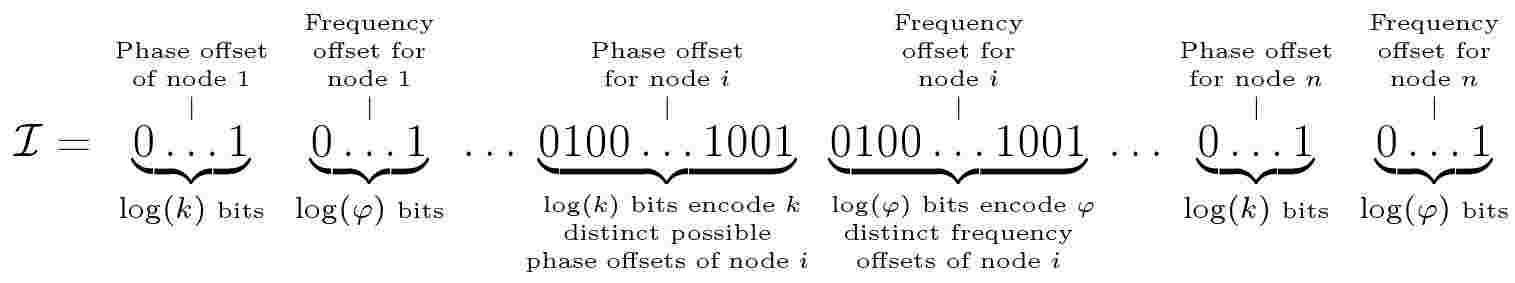}     
\caption{Binary representation of search points as a concatenation of grey encoded phase and frequency offsets}
\label{figure03BF}
\end{figure}

We assume that each transmit node is able to apply $k$ distinct phase-offsets.
We encode a search point $s\in\mathcal{S}$ represented by $n\cdot k$ distinct phase-offsets as binary string of length $n\cdot\log(k)$
\footnote{When distinct frequency offsets are also considered, a search point $s\in\mathcal{S}$ would be represented by $n\cdot k \cdot f$ distinct phase- and frequency-offsets as binary string of length $n\cdot\log(k)\cdot\log{f}$}.
We assume that configurations are encoded so that their Hamming-distance increases with increasing difference in phase-offsets~\cite{Algorithm_Savage_1995,Algorithms_Knuth_2011}.
We analyse the count of bit mutations of this bit-string until an encoding of a global optimum is found.
We choose the probability to alter a bit in the binary sequence as $\frac{1}{n\cdot \mathcal{N}}$ for $n$ nodes with neighbourhood size $\mathcal{N}$.
For the binary representation this is analogue to having a probability of $\frac{1}{n}$ for each signal to alter its phase uniformly at random within the $\mathcal{N}$ possible values in $\left[\gamma_i-\frac{\mathcal{N}}{2},\gamma_i+\frac{\mathcal{N}}{2}\right]$.                                                                                                                                                                                                                                                         Then, one node on average alters its phase offset within the neighbourhood boundaries in each iteration.
With Chernoff bounds we can show that with high probability the Hamming-distance to an optimum configuration of these offsets for all carrier signals is not much smaller than $\frac{n\cdot\log(k)}{2}$.
\begin{theoremS}\label{theoremOne}
For a network of $n$ transmit and one receive node, let $\mathcal{N}$ be the neighbourhood size of a local random search method $\mathcal{A}$ for feedback-based distributed adaptive carrier synchronisation with $\forall i:P_{\mbox{\footnotesize dist},\gamma_i}=\mbox{uniform},P_{\mbox{\footnotesize mut},\gamma_i}=\frac{1}{n}$. 
Further assume that each node is capable of transmitting signals at up to $k$ distinct carrier phases and that new carrier phases are drawn uniformly at random from the neighbourhood. 
The expected number of iterations for distributed adaptive carrier synchronisation is bounded by 
\begin{equation}
\mathcal{O}\left(n\cdot \mathcal{N}\cdot\log(n)+\frac{\log(k)}{\mathcal{N}}\right).     
\end{equation}      
\end{theoremS}
(Refer to the Appendix for a proof of the results)

\begin{theoremS}\label{theoremTwo}
For a network of $n$ transmit and one receive node, let $\mathcal{N}$ be the neighbourhood size of a local random search method $\mathcal{A}$ for feedback-based distributed adaptive carrier synchronisation with $\forall i:P_{\mbox{\footnotesize dist},\gamma_i}=\mbox{uniform},P_{\mbox{\footnotesize mut},\gamma_i}=\frac{1}{n}$.
Further assume that each node is capable of transmitting signals at up to $k$ distinct carrier phases and that new carrier phases are drawn uniformly at random from the neighbourhood of size~$\mathcal{N}$. 
For a suitable $\Delta$ the expected number of iterations for distributed adaptive carrier synchronisation is bounded by 
\begin{equation}
\Omega(n\cdot \mathcal{N}\cdot\Delta)\nonumber
\end{equation}      
\end{theoremS}
(Refer to the Appendix for a proof of the results)

\begin{theoremS}\label{theoremThree}
For a network of $n$ transmit and one receive node, let $\mathcal{N}$ be the neighbourhood size of a local random search method $\mathcal{A}$ for feedback-based distributed adaptive carrier synchronisation with $\forall i:P_{\mbox{\footnotesize dist},\gamma_i}=\mbox{uniform},P_{\mbox{\footnotesize mut},\gamma_i}=\frac{1}{n}$. 
Further assume that each node is capable of transmitting signals at up to $k$ distinct carrier phases and that new carrier phases are drawn uniformly at random from the neighbourhood of size~$\mathcal{N}$. 
The expected number of iterations for distributed adaptive carrier synchronisation is bounded by 
\begin{equation}
E[T_\mathcal{P}]=\Theta\left(n\cdot \mathcal{N}\cdot\log(n)+\frac{\log(k)}{\mathcal{N}}\right).\label{equationResult}
\end{equation}      
\end{theoremS}
(Refer to the Appendix for a proof of the results)

Observe that equation~(\ref{equationResult}) evolves to 
\begin{equation}
\Theta\left(n\cdot\log(n)+\log(k)\right)     
\end{equation}
for $\mathcal{N}\rightarrow 1$ and to 
\begin{equation}
\Theta\left(n\cdot k\cdot\log(n)\right)     \label{equationTwentyOne}
\end{equation}
for $\mathcal{N}\rightarrow k$.
Equation~(\ref{equationTwentyOne}) is identical to the bound derived on the expected optimisation time of the global random search method~\cite{4023}, where in fact the neighbourhood size is $\mathcal{N}=k$.
Observe that it is more beneficial to have a smaller local search neighbourhood than to utilise a global random search method with $\mathcal{N}=k$. 

The minimum value for $E[T_\mathcal{P}]$ is achieved for $\mathcal{N}\rightarrow1$.
Since this is an asymptotic consideration, the optimum absolute value for $\mathcal{N}$ depends on the choice of $n$ and $k$.

\subsection{Environmental impacts}\label{sectionImpacts}
The performance of the local random search guided carrier synchronisation is impacted by $P_{\mbox{\footnotesize mut}, \gamma_i}$, $P_{\mbox{\footnotesize mut},f_i}$, $P_{\mbox{\footnotesize dist},\gamma_i}$, $P_{\mbox{\footnotesize dist},f_i}$, $V_{\gamma_i}$ and $V_{f_i}$.
In~\cite{4022} we observed that a good synchronisation quality is achieved when either $P_{\mbox{\footnotesize mut},\gamma_i}$, $P_{\mbox{\footnotesize mut},f_i}$, $V_{f_i}$ or $V_{\gamma_i}$ are small, so that the search space is propagated in rather small steps, eventually approaching the optimum.
This observation also agrees with our discussion in the last section.
However, the environment may impact the optimum value for these parameters.
We discuss three possible impacts, namely the number of participating nodes, the noise figure and movement of devices.

\subsubsection{Impact of noise and interference}\label{sectionImpactsNoise}
The signal observed by a receiver is composed of the signal $\zeta_{\mbox{\footnotesize sum}}(t)$ and noise $\zeta_{\mbox{\footnotesize noise}}(t)$ (cf.~equation~(\ref{equationOneBF02})).
Noise and interference might differ due to opened windows or doors, people moving or other nearby electronic devices~\cite{4036,ContextAwareness_Sigg_2011,Pervasive_Sigg_2012}.
The impact of the phase alteration of a single link $i\in[1..n]$ on the SNR of $\zeta_{\mbox{\footnotesize sum}}(t)+\zeta_{\mbox{\footnotesize noise}}(t)$ is not greater than $2\cdot\mbox{RSS}_i$.
This is the case when the phase of signal $\zeta_i(t)$ and the sum signal $\zeta_{\mbox{\footnotesize sum-i}}(t)$ without the signal of $i$ have been separated in phase by $\pi$ before $\gamma_i$ is then shifted by $\pi$. 
With the cosine rule we can calculate the change in the received signal strength of the received superimposed signal at the event of a change of the carrier phase from $\gamma_i$ to $\gamma'_i$ as
\begin{eqnarray}
\overline{\mbox{RSS}}(\gamma_i,\gamma'_i)=\hspace{5.5cm}\nonumber\\
\sqrt{\mbox{RSS}_{\mbox{\footnotesize \mbox{\footnotesize sum}-i}}^2+\mbox{RSS}_i^2-2\mbox{RSS}_{\mbox{\footnotesize \mbox{\footnotesize sum}-i}}\mbox{RSS}_i\cos(\gamma_i+\gamma'_i)}\nonumber\\
	-\sqrt{\mbox{RSS}_{\mbox{\footnotesize \mbox{\footnotesize sum}-i}}^2+\mbox{RSS}_i^2-2\mbox{RSS}_{\mbox{\footnotesize \mbox{\footnotesize sum}-i}}\mbox{RSS}_i\cos(\gamma_i)}\label{equationThree}
\end{eqnarray}
as illustrated in figure~\ref{figure04BF}.
\begin{figure}
	\centering
	\includegraphics[width=13cm]{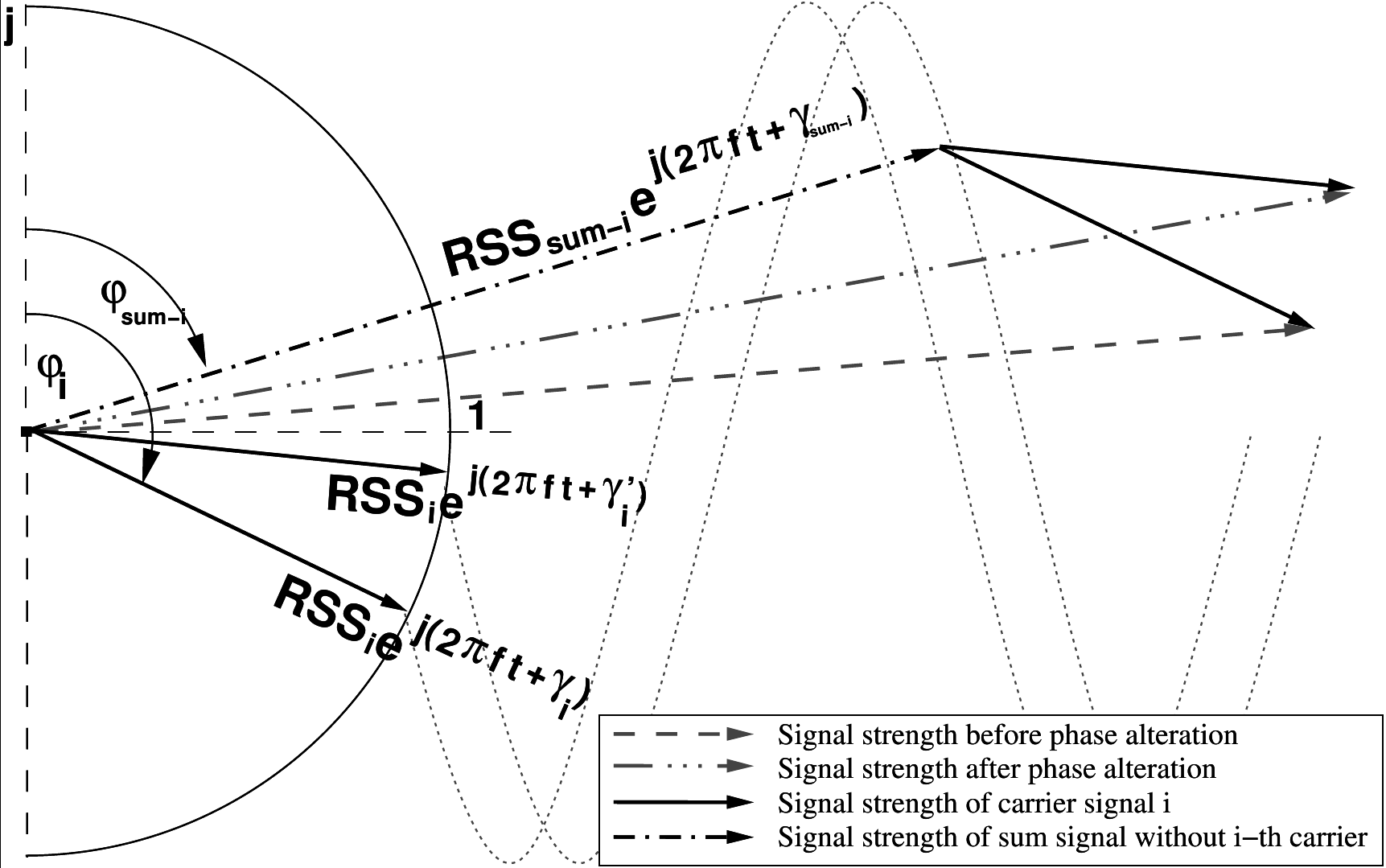}
	\caption{Illustration of the impact of carrier phase alteration on the overall received signal strength}
	\label{figure04BF}
\end{figure}
In this equation we denote the received signal strength achieved by the superimposition of all signals short of $i$ by 
\begin{eqnarray}
\mbox{RSS}_{\mbox{\footnotesize \mbox{\footnotesize sum}-i}}=\sum_{\overline{i}\not=i}\mbox{RSS}_{\overline{i}}e^{j(2\pi (f_c+f_{\overline{i}})t+\gamma_{\overline{i}}+\phi_{\overline{i}}+\psi_{\overline{i}})}; 
\overline{i}\in[1..n].	
\end{eqnarray}

Since the phase alteration is a random process, the actual gain of a single phase modification is typically smaller than the maximum possible value.
When we assume that a single node can establish up to $k$ equally probable carrier phases, the average gain by the alteration of one carrier signal is then
\begin{equation}
	\frac{\sum_{i=1}^k\overline{\mbox{RSS}}\left(\gamma_i,\gamma_i+\frac{2\pi}{k}\cdot i\right)}{k}.
\end{equation}
Consequently, when the noise figure is in the same order or greater, alterations of individual carriers have little effect.
In such a situation it is beneficial to increase the average distance of consecutive search points in the search space in a single iteration.
This can be achieved by increasing the variance $V_{\gamma_i}$, the neighbourhood size $\mathcal{N}$ or the probability $P_{\mbox{\footnotesize mut},\gamma_i}$ (cf. section~\ref{sectionSimulation} and section~\ref{sectionCaseStudyBF02}).

\subsubsection{Impact of the network size}\label{sectionImpactsSize}
The number of nodes that participate also impacts the performance.
Since the synchronisation is achieved by a random process over all possible combinations of phase and frequency offsets, the synchronisation time is increased with the count of nodes~\cite{5923}.
The optimum performance is achieved with small $P_{\mbox{\footnotesize mut},\gamma_i}$, $P_{\mbox{\footnotesize mut},f_i}$, $V_{f_i}$ and $V_{\gamma_i}$~\cite{4023}.
On the other hand, the relative impact of an individual node on $\zeta_{\mbox{\footnotesize sum}}(t)$ decreases with increasing node count.
We can see this again from equation~(\ref{equationThree}).
The value $\overline{\mbox{RSS}}(\gamma_i,\gamma'_i)$ decreases with increasing $\mbox{RSS}_{\mbox{\footnotesize \mbox{\footnotesize sum}-i}}$.
With increasing node count $n$, it is therefore beneficial to chose $P_{\mbox{\footnotesize mut},\gamma_i}$, $P_{\mbox{\footnotesize mut},f_i}$, $V_{\gamma_i}$ and $V_{f_i}$ slightly higher than $\frac{1}{n}$ in order to increase the impact of modifications in one iteration (cf. section~\ref{sectionSimulation}).

\subsubsection{Impact of node mobility}\label{sectionImpactsMobility}
Movement impacts the synchronisation of nodes since phases drift apart when the receiver or transmit nodes move~\cite{4032}.
Synchronisation has to be significantly faster than the velocity experienced.
An increased value for $P_{\mbox{\footnotesize mut},\gamma_i}$, $P_{\mbox{\footnotesize mut},f_i}$, $V_{\gamma_i}$ or $V_{f_i}$ might therefore be beneficial in the presence of node mobility (cf. section~\ref{sectionSimulation}).

\subsection{Simulation and case studies}\label{sectionSimulation}
In a Matlab-based simulation, up to $100$ IoT devices are distributed uniformly at random on a $4$~m$\times6$~m square area (e.g. spread across a wall in a factory building) with a remote IoT receiver located up to $11$~m in orthogonal direction from the centre of this area.
Frequency and phase stability are considered perfect.
We calculate the phase offset of the received dominant signal component from each transmitter according to the transmission distance in a direct line of sight.
Path loss was calculated by the Friis free space equation $P_{tx}\left(\frac{\lambda}{2\pi d}\right)^2 G_{tx} G_{rx}$ with antenna gain for transmitter and receiver as $G_{rx}=G_{tx}=0$~dB.
Signals are transmitted at $2.4$~GHz with transmit power $P_{tx}=1$~mW.
All received signal components calculated in this manner are then summed up in order to achieve the superimposed sum signal 
\begin{equation}
	\zeta_{\mbox{\footnotesize sum}}(t)=\sum_{i=1}^n \left(\Re\left(m(t)\mbox{RSS}_ie^{j(2\pi (f_c+f_i)t+\gamma_i+\phi_i+\psi_i)}\right)\right).
\end{equation}
Finally, a noise signal $\zeta_{\mbox{\footnotesize noise}}(t)$ is added to $\zeta_{\mbox{\footnotesize sum}}(t)$ in order to estimate the signal at the receiver.
We utilise AWGN at $-103$~dBm as proposed in~\cite{062}.
For a given configuration we repeated each simulation 10 times with identical parameters.

Each simulation lasts for 6000 iterations.
Signal quality of a signal during the synchronisation phase is measured by the Root of the Mean Square Error (RMSE) of the received signal $\zeta_{\mbox{\footnotesize sum}}(t)$ to an expected optimum signal $\zeta_{\mbox{\footnotesize opt}}(t)$:
\begin{equation}
RMSE=\sqrt{\sum_{t=0}^{\varrho}
		\frac{\left(
				\zeta_{\mbox{\footnotesize sum}}(t)+\zeta_{\mbox{\footnotesize noise}}(t)-\zeta_{\mbox{\footnotesize opt}}(t)
			\right)^2}{n}
		}\label{equationRMSEBF02}
\end{equation}
In equation~(\ref{equationRMSEBF02}), $\varrho$ is chosen to cover several signal periods.

The optimum signal $\zeta_{\mbox{\footnotesize opt}}(t)$ is calculated as perfectly aligned and properly phase shifted received sum signal from all transmit sources.
For the optimum signal, noise is disregarded.

\subsubsection{Performance of a local random search}
We implement a local random search with neighbourhood radius $\frac{\mathcal{N}}{2} \in [0,\pi]$ where each node~$i\in[1,n]$ alters the phase offset $\gamma_i$ of its carrier signal $\zeta_i(t)$ with probability $\frac{1}{n}$ to
$\gamma'_i\in[\gamma_i-\frac{\mathcal{N}}{2},\gamma_i+\frac{\mathcal{N}}{2}]$.
Figure~\ref{figure05BF} depicts the performance of the algorithm with a neighbourhood size of $\mathcal{N}=0.6\pi$ compared to a global random search approach ($\mathcal{N}=2\pi$).
\begin{figure}
\centering
\includegraphics[width=13.5cm]{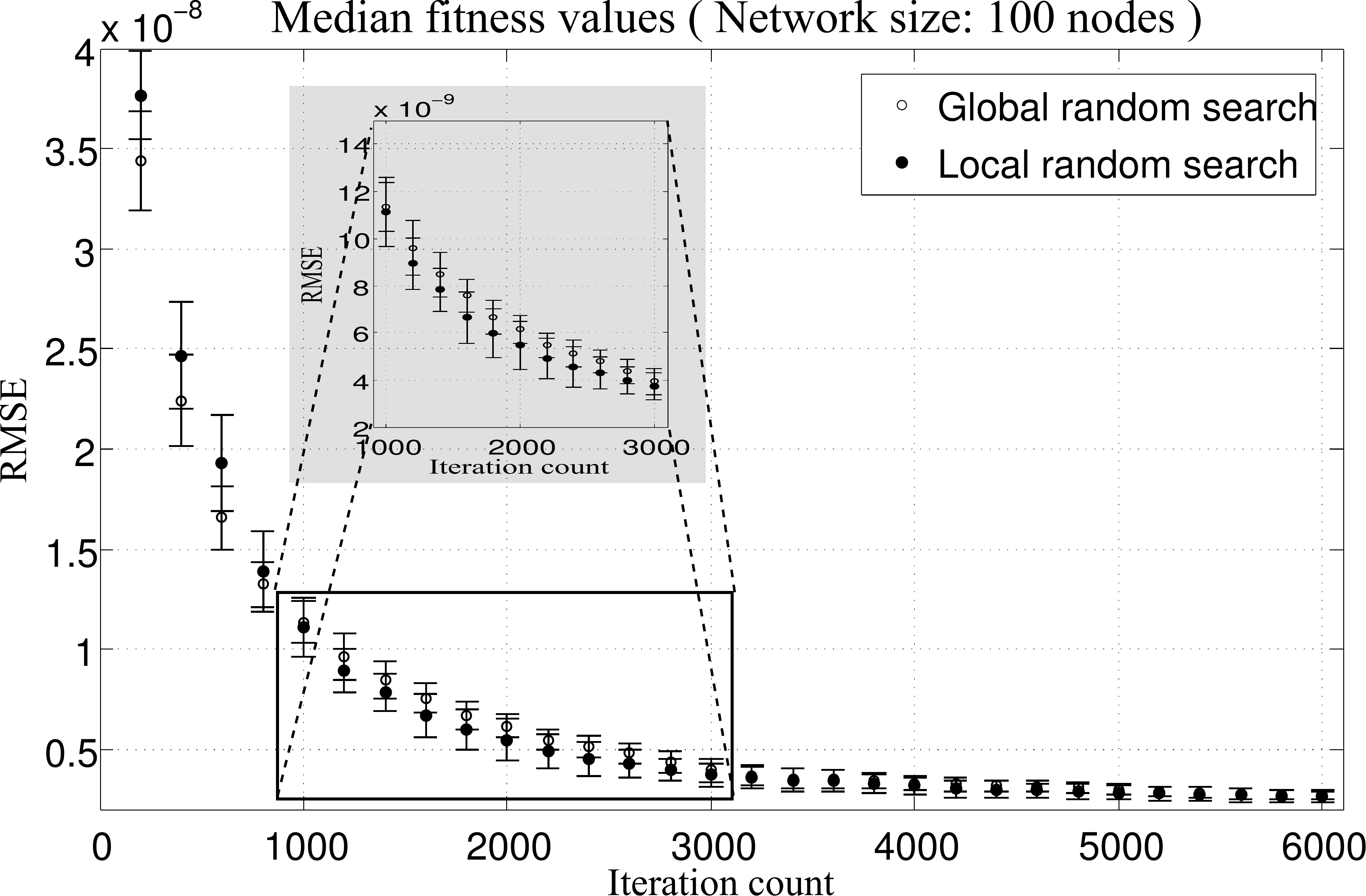}
	\caption{Comparison of the performance achieved by global and local random search}
	\label{figure05BF}
\end{figure}
We observe that, although the local random search method naturally has a slower start than the global random search method, it then reaches lower RMSE values faster. 

In particular, in the critical part ot the optimisation, the RMSE values reached by the local random search approach are reached only about 400-500 iterations later by the global random search method.
Due to noise and therefore a general saturation of the optimisation process, the synchronisation quality is not much improved afterwards so that the global random search eventually catches up.
\subsubsection{Case study with software defined radio devices}\label{sectionCaseStudyBF02}
We approximated realistic conditions in an experimental setting with Universal Software Radio Peripheral (USRP) software radios\footnote{http://www.ettus.com} to represent distributed devices.

Three USRP devices have been utilised as transmitters with one device as receiver.
In order to achieve identical transmit frequencies among devices, the clock of the first transmit device was utilised as reference for the other two transmitters.
The clock of the receiver node was, however, independent.
Alternatively, clocks might be synchronised via GPS or by the iterative frequency synchronisation approach described in~\cite{Seo_2008}.
Table~\ref{tableInstrumentation2} summarises our experimental configuration.
\begin{table}
\centering
	\caption{Configuration of the experimental case study}
	\label{tableInstrumentation2}
	\begin{tabular}{l|l}
	\textbf{Experimental setting} & \\\hline
	Separation of antennas [m] & 0.44 \\
	Distance to receive antenna [m] & 5.5 / 11 / 16.4\\
	Transmit frequency [MHz] & $f_{TX}=2400$ \\
	Receive frequency [MHz] & $f_{RX}=902$ \\
	Iterations per experiment & $400$ \\
	Mobility & stationary \\
	Identical experiments & $12$ \\
	Transmit devices & 3\\
	Receive devices & 1\\[.1cm]
	
	\textbf{Algorithmic configuration} & \\\hline
	\begin{minipage}{5cm}Random distribution of the phase alteration\end{minipage} & normal\\[.1cm]
	Phase alteration probability $P_{\mbox{\footnotesize mut},\gamma_i}$ & 0.33/0.66/1.0\\[.1cm]
	\begin{minipage}{5cm}Variance $V_{\gamma_i}$ [$\pi$] \end{minipage}& $0.25$ / 1
	\end{tabular}
\end{table}

The experimental setting is sketched in figure~\ref{figure06BF}.
\begin{figure}
     \centering	
     \includegraphics[width=11cm]{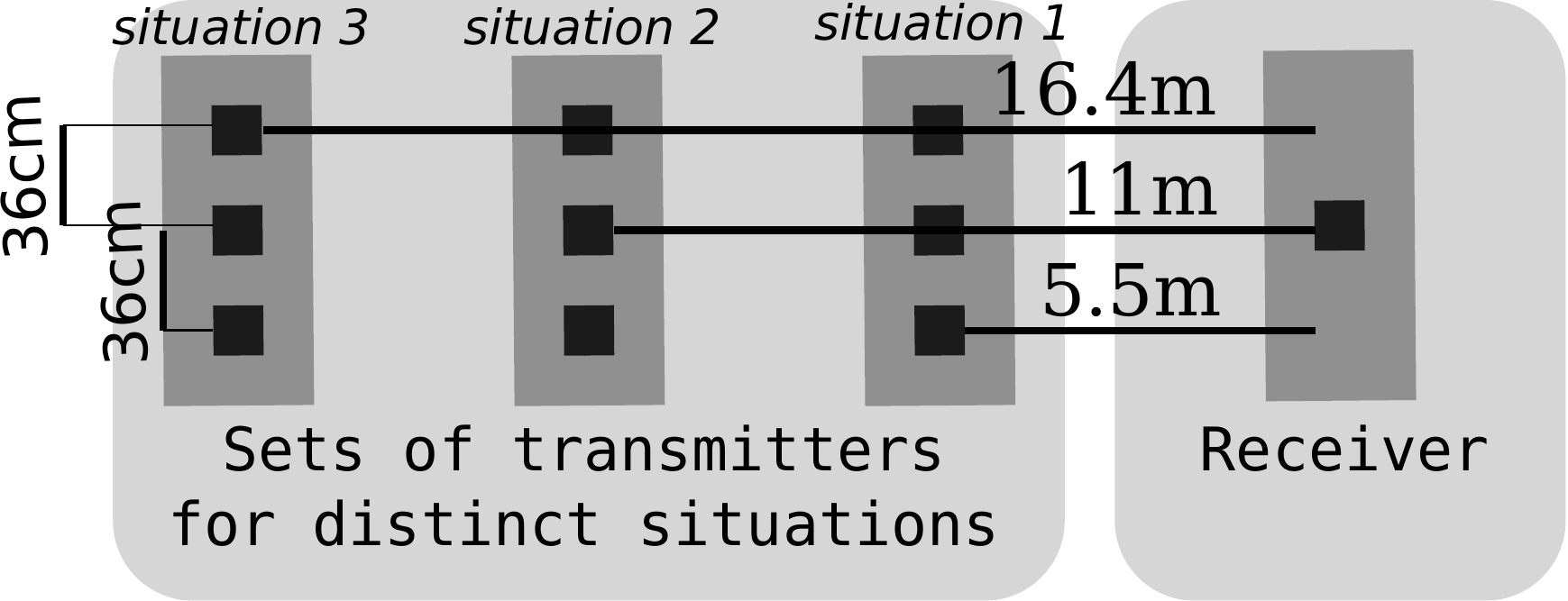}
     \caption{Illustration of the experimental setting utilised. Devices are placed on tables with a height of 72~cm.}
     \label{figure06BF}
\end{figure}
For the three different situations, the transmit nodes were moved to various distances accordingly.
We modified the probability to alter the phase offset of one device and the variance for its normal distributed random phase perturbation process as well as the distance between transmit and receive devices to account for distinct environmental situations.
The three transmit devices were controlled by a single computer running three independent and non-communicating processes. 
The receiver device was controlled by a second computer.
During the experiments the room was vacated so that no movement or presence of individuals could impact the synchronisation process.

Results derived in these experiments are depicted in figure~\ref{figure08BF} and figure~\ref{figure09BF}.
\begin{figure}
  	\centering
	  \includegraphics[width=13.5cm]{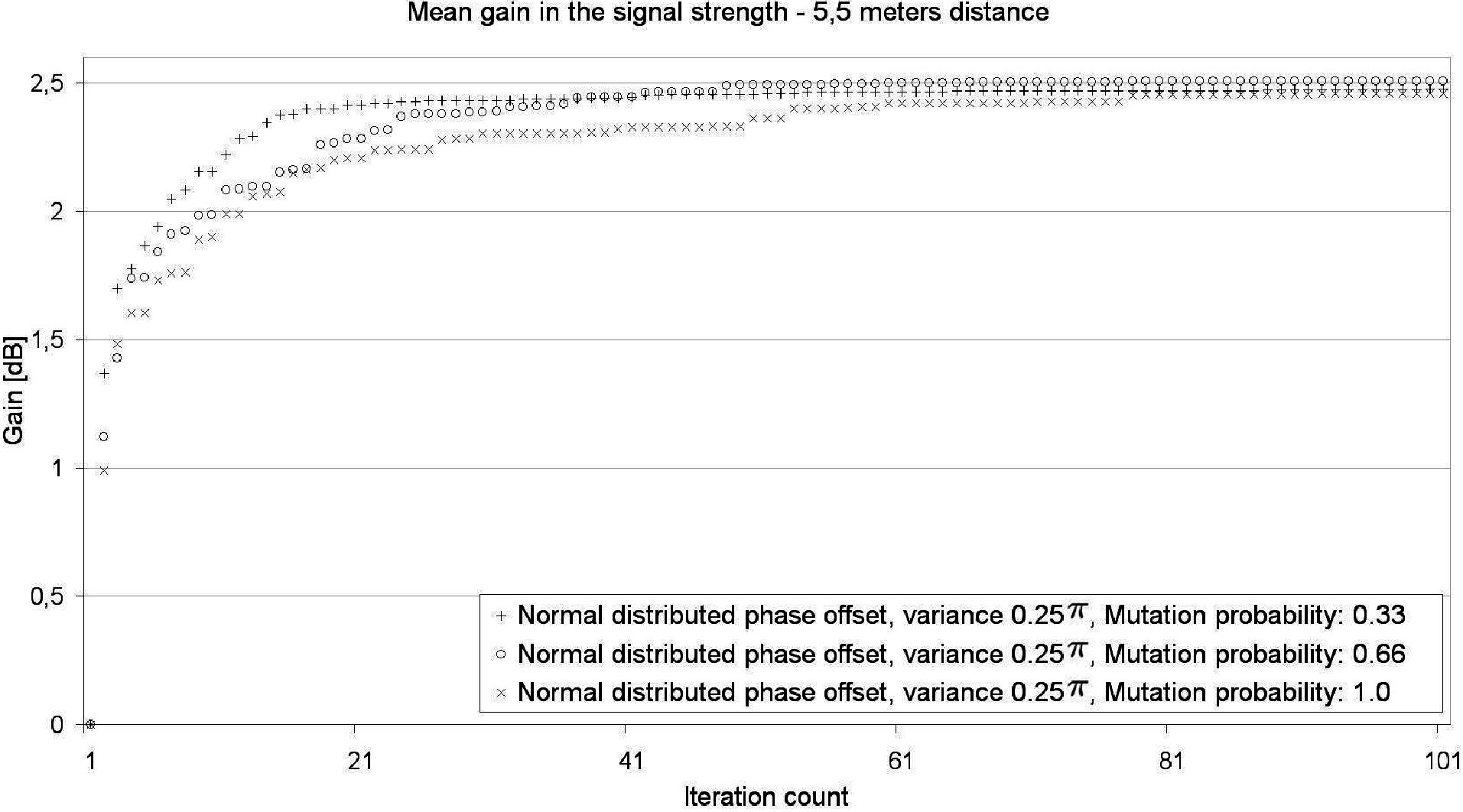}
  	\caption{Mean gain in the signal strength at a transmission distance of 5.5 meters and a variance of the random process of $0.25\pi$}
	  \label{figure08BF}
\end{figure}
\begin{figure}
\centering
	\includegraphics[width=13.5cm]{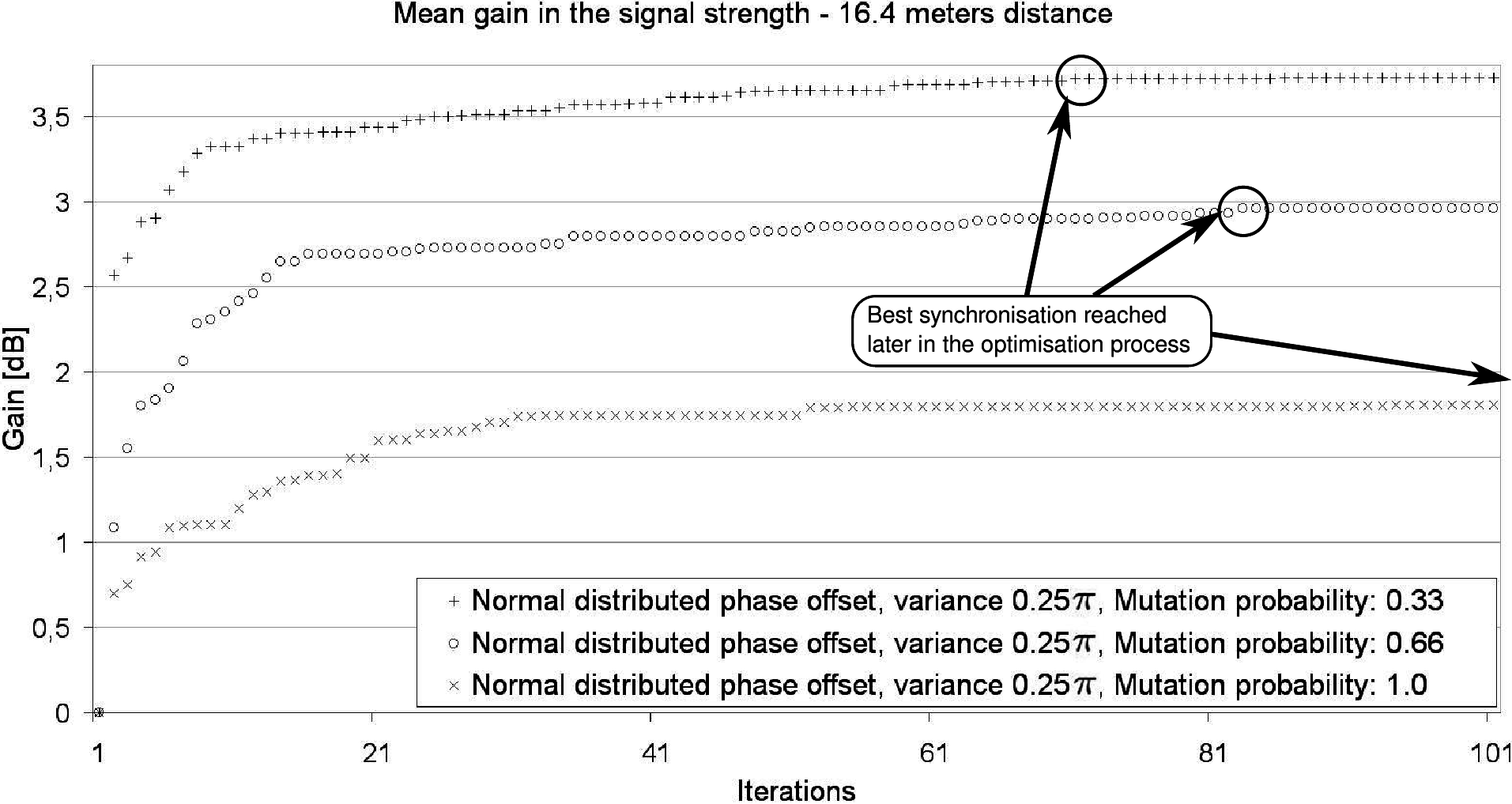}
	  \caption{Mean gain in the signal strength at a transmission distance of 16.4 meters and a variance of the random process of $0.25\pi$}
	  \label{figure09BF}
\end{figure}

The figure displays the mean gain in signal strength compared to an unsynchronised transmission at the beginning of the synchronisation.
The synchronisation performance differs for different environmental situations.
When the transmission distance increases, the relative noise figure also increases.
The best synchronisation is then generally reached later in the synchronisation process.
For instance, in figure~\ref{figure08BF}, at a transmission distance of 5.5~meters, the best value is reached after about 40 to 50 iterations. 
In figure~\ref{figure09BF} (16.4~meters) we observe that the optimum synchronisation is reached after about 60 to 80 iterations.
Also, the choice of the optimum configuration differs dependent on the scenario.
In figure~\ref{figure09BF} at a distance of 16.4~meters, a variance of $0.25\pi$ and a probability to alter the phase offset of $P_{\mbox{\footnotesize mut},\gamma_i}=0.33$ achieves the best results.
At shorter distances, the configuration with $P_{\mbox{\footnotesize mut},\gamma_i}=0.66$ results in a slightly better synchronisation performance.
For the variance, a similar effect was not observed.

\subsection{Conclusion}\label{sectionConclusionBF02}
We analysed  and evaluated a local random search-based approach to distributed adaptive carrier synchronisation for IoT nodes in a smart space with an iterative feedback-based carrier synchronisation method.
We derived a sharp asymptotic bound of 
\begin{equation}
E[T_\mathcal{P}]=\Theta\left(n\cdot \mathcal{N}\cdot\log(n)+\frac{\log(k)}{\mathcal{N}}\right)\nonumber     
\end{equation}
on the expected synchronisation performance.
This bound is significantly lower than the expected synchronisation performance derived recently for a global random search heuristic for this carrier synchronisation method.
Intuitively, although the global random search approach has, unlike the local random search, a positive (but very small) probability to reach a global optimum in each iteration, its probability to generally reach any point that would improve the synchronisation quality decreases with increasing synchronisation quality.
For the local random search, however, we could show that there is at least one node that would improve the synchronisation with probability not smaller than $\frac{1}{2n}$ for a long time during the synchronisation process.

Also, we discussed the impact of environmental effects on the synchronisation performance.
In particular, the relative noise figure, the count of participating devices and the mobility of nodes have been identified to impact the synchronisation performance.
However, by changing the probabilities  $P_{\mbox{\footnotesize mut},\gamma_i}$, $P_{\mbox{\footnotesize mut},f_i}$ to alter the phase offset or the variance $V_{\gamma_i}$, $V_{f_i}$ for each node~$i$, the synchronisation approach can be adapted to these environmental impacts.

We presented simulations and case studies with software defined radios on the iterative feedback-based carrier phase synchronisation by a local random search approach that also showed an improved performance compared to the global random search approach and an effect of the distance between nodes on the synchronisation performance.

\section*{Appendix -- Proofs}
\begin{proof}[Proof of Theorem~\ref{theoremZero}]
We see this from figure~\ref{figure07BF}.
\begin{figure}[p]
	\centering
	\includegraphics[width=13.5cm]{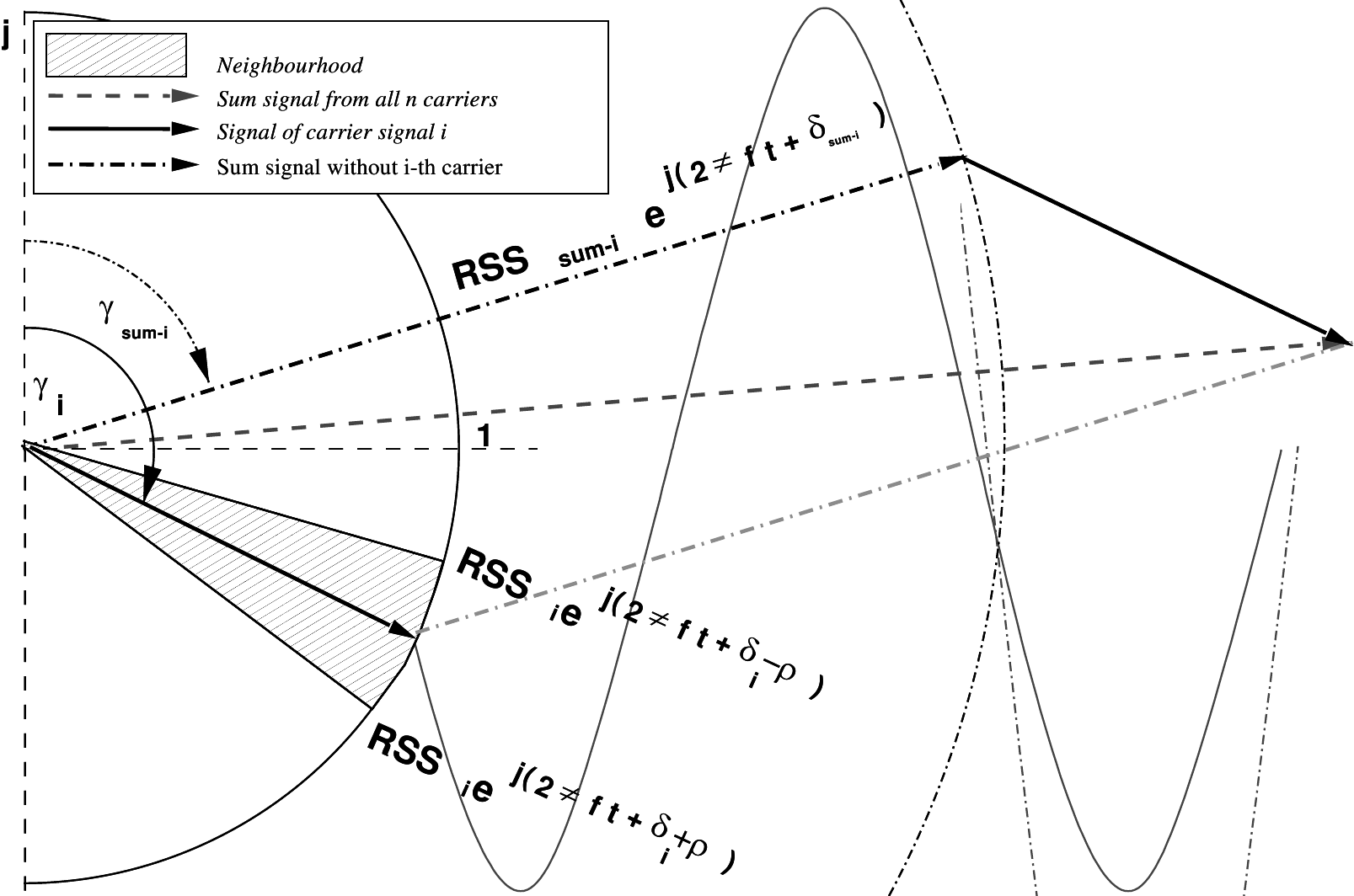}
	\caption{Illustration of the neighbourhood of the local random search approach}
     \label{figure07BF}
\end{figure}
In the figure, a vector of all carrier signals short of signal~$i$ is denoted by 
\begin{equation}
     \zeta_{\mbox{\footnotesize sum}-i}(t)=\mbox{RSS}_{\mbox{\footnotesize sum}-i}e^{j(2\pi ft+\gamma_{\mbox{\tiny sum}-i})}.     
\end{equation}
The vector associated with carrier $i$ is identified by 
\begin{equation}
\zeta_i(t)=\mbox{RSS}_{i}e^{j(2\pi ft+\gamma_{i})}.     
\end{equation}

When a single carrier signal $\zeta_i(t)$ is modified within the neighbourhood of $\mathcal{N}$, this means that $\zeta_i(t)$ is rotated by $\nu$ or $-\nu$ with $\nu\in\left[0,\frac{\mathcal{N}}{2}\right]$.
As long as $\zeta_{\mbox{\footnotesize sum}-i}(t)$ is not within the neighbourhood, a rotation of $\nu$ will increase (decrease) the amplitude of $\zeta_{\mbox{\footnotesize sum}-i}(t)+\zeta_i(t)$ while a rotation by $-\nu$ will decrease (increase) it.
The probability to improve the fitness value is $\frac{1}{2}$ when carrier phase offsets are chosen uniformly at random.
\end{proof}

\begin{proof}[Proof of Theorem~\ref{theoremOne}]   
We divide the analysis into two phases (cf.~theorem~\ref{theoremZero}).
In the first phase, at least one node does not have the optimum phase offset within its neighbourhood.
Then, there is always at least one node that will by altering its carrier phase improve the synchronisation with probability at least $\frac{1}{2}$.
The probability that in one iteration one such node alters its phase offset while all other $n-1$ nodes do not change it is at least
\begin{equation}
     \frac{1}{n}\cdot\left(1-\frac{1}{n}\right)^{n-1} \geq \frac{1}{e\cdot n}.
\end{equation}
We define the expected progress as the expected count of bits in the binary representation of search points that are altered in one iteration.
Since a new search point is drawn uniformly at random from all possible values in the neighbourhood of size $\mathcal{N}$, the expected progress when a node that has not the optimum within its neighbourhood randomly alters its carrier phase and improves the overall synchronisation is therefore at most
\begin{equation}
     \frac{1}{2}\cdot\frac{\mathcal{N}}{2}\cdot e\cdot n=\frac{e\cdot n \cdot \mathcal{N}}{4}.
\end{equation}
The expected upper bound on the iterations required to reach a global optimum is then determined by the maximum distance to an optimum.
The Hamming-distance to a binary representation that describes a global optimum is $n\cdot\log(k)$ at most.
Consequently, the expected number of these iterations until a binary representation is found for which all nodes have the optimum within their neighbourhood is at most 
\begin{equation}
\frac{n\cdot\log(k)\cdot4}{\mathcal{N}\cdot e\cdot n}=\mathcal{O}\left(\frac{\log(k)}{\mathcal{N}}\right).     
\end{equation}
In the second phase, each node has the optimum carrier phase offset within the neighbourhood around its current carrier phase.
Assume that a set of $i$ nodes has already reached an optimum synchronisation of their carrier phases.
In this case, the probability that one of the $n-i$ nodes which have not yet found the optimum phase offset applies a correct mutation which would alter the carrier phase to the optimum value with respect to all other carrier phases is 
\begin{eqnarray}
\left(\begin{array}{c}n-i\\1\end{array}\right)\cdot\frac{1}{n}\cdot\frac{1}{\mathcal{N}}\cdot\left(1-\frac{1}{n}\right)^{n-1}
\geq\frac{n-i}{n\cdot \mathcal{N}\cdot e}.     \label{equationEight}
\end{eqnarray}
In equation~(\ref{equationEight}), the term $\left(\begin{array}{c}n-i\\1\end{array}\right)\cdot\frac{1}{n}$ describes the number of possible cases that one node out of $n-i$ nodes which are not yet perfectly synchronised alters its phase offset with probability $\frac{1}{n}$.
Since all phases are with equal probability drawn from the Neighbourhood of size~$\mathcal{N}$, this alteration leads to the one optimum phase offset within the neighbourhood with probability $\frac{1}{\mathcal{N}}$.
The term $\left(1-\frac{1}{n}\right)^{n-1}$ describes the probability that all other $n-1$ nodes do not alter their phase offset in this iteration.
When this event happens $n-1$ times for each possible number of already synchronised nodes ($n-i$ with $i\in[1..n]$), the carrier phase offsets of all nodes are finally synchronised.
Therefore, an upper bound on the synchronisation time in the second phase is given by
\begin{eqnarray}
&&\sum_{i=0}^{n-1}\frac{n\cdot \mathcal{N}\cdot e}{n-i}\nonumber\\
&=&\sum_{i=1}^{n}\frac{n\cdot \mathcal{N}\cdot e}{i}\nonumber\\
&=&\mathcal{O}\left(n\cdot \mathcal{N}\cdot\log(n)\right).     
\end{eqnarray}
Overall, the expected asymptotic synchronisation time is then 
\begin{equation}
\mathcal{O}\left(n\cdot \mathcal{N}\cdot\log(n)+\frac{\log(k)}{\mathcal{N}}\right).     
\end{equation}
\end{proof}

\begin{proof}[Proof or Theorem~\ref{theoremTwo}]
After initialisation, the phases of the carrier-signals are identically and independently distributed.
Consequently for a superimposed received sum-signal $\zeta_{\mbox{\footnotesize sum}}(t)$, each of the $n\cdot\log(k)$ bits in the binary string $s_{\zeta_{\mbox{\tiny sum}}}$ that represents the corresponding search-point has an equal probability to be $1$ or $0$.
The probability to start from a search-point $s_{\zeta_{\mbox{\tiny sum}}}$ with Hamming-distance $h(s_{\mbox{\footnotesize opt}},s_{\zeta_{\mbox{\tiny sum}}})$ not larger than $l\in\mathds{N}\; ; \; l\ll n\cdot\log(k)$ to one of the global optima $s_{\mbox{\footnotesize opt}}$ is at most 
\begin{eqnarray}
	P[h(s_{\mbox{\footnotesize opt}},s_{\zeta_{\mbox{\tiny sum}}})\leq l]
	&=& \sum_{i=0}^l\left(\begin{array}{c}
	                                	n\cdot\log(k)\\ n\cdot\log(k)-i
	                                \end{array}
\right)
\cdot \frac{k}{2^{n\cdot\log(k)-i}}\nonumber\\
&\leq&\frac{k\cdot\left(n\cdot\log(k)\right)^{l+2}}{2^{n\cdot\log(k)-l}}
\end{eqnarray}
In this formula, 
\begin{equation}
\left(\begin{array}{c}
	                                	n\cdot\log(k)\\ n\cdot\log(k)-i
	                                \end{array}
\right)	
\end{equation}
is the count of possible configurations with $i$ bit-errors to a given global optimum, $\frac{1}{2^{n\cdot\log(k)-i}}$ represents the probability for all these bits to be identical to the respective bits in one of the $k$ global optima.
Observe that we have a global optimum for each possible~$k$ phase offsets since from one global optimum we reach an arbitrary other global optimum by shifting the carrier phases of all nodes by an equal amount.

This means that with high probability the Hamming-distance to the nearest global optimum is at least~$l$.
We will use the method of the expected progress to calculate a lower bound on the optimisation-time required to flip these $l$ bits.
The general idea is the following.

Let $(s_{\zeta_{\mbox{\tiny sum}}},\tau)$ denote the situation that search-point $s_{\zeta_{\mbox{\tiny sum}}}$ was achieved after $\tau$ iterations of the algorithm.
We define a progress measure $\Lambda:\left(\mathds{B}^{n\cdot\log(k)},t\right)\rightarrow\mathds{R}^+_0, t\in\mathds{N}$ such that $\Lambda(s_{\zeta_{\mbox{\tiny sum}}},\tau)<\Delta$ represents the case that a global optimum was not found in the first $\tau$ iterations.
For every $\tau\in\mathds{N}$ we have 
\begin{eqnarray}
	E[T_\mathcal{P}]&\geq& \tau\cdot P[T_\mathcal{P}>\tau]\nonumber\\ 
	&=& \tau\cdot P[\Lambda(s_{\zeta_{\mbox{\tiny sum}}},\tau)<\Delta]\nonumber\\
	&=& \tau\cdot (1-P[\Lambda(s_{\zeta_{\mbox{\tiny sum}}},\tau)\geq\Delta]).
\end{eqnarray}
With the help of the Markov-inequality we obtain 
\begin{equation}
P[\Lambda(s_{\zeta_{\mbox{\tiny sum}}},\tau)\geq\Delta]\leq \frac{E[\Lambda(s_{\zeta_{\mbox{\tiny sum}}},\tau)]}{\Delta}	
\end{equation}
and therefore 
\begin{equation}
	E[T_\mathcal{P}]\geq \tau\cdot\left(1-\frac{E[\Lambda(s_{\zeta_{\mbox{\tiny sum}}},\tau)]}{\Delta}\right).
\end{equation}
This means that we can obtain a lower bound on the optimisation-time by providing the expected progress after $\tau$ iterations.
The probability for $l$ bits to correctly flip is at most 
\begin{eqnarray}
	& & \left(1-\frac{1}{n\cdot \mathcal{N}}\right)^{n\cdot \mathcal{N}-l}\cdot\left(\frac{1}{n\cdot \mathcal{N}}\right)^l
     \leq
\frac{1}{(n\cdot \mathcal{N})^l}.
\end{eqnarray}
In this formula, $\left(1-\frac{1}{n\cdot \mathcal{N}}\right)^{n\cdot \mathcal{N}-l}$ describes the probability that all 'correct' bits do not flip while the remaining $l$ bits mutate with probability $\left(\frac{1}{n\cdot \mathcal{N}}\right)^l$.

This means that for all $n$ nodes in one iteration on average $1$ bit flips inside their neighbourhood of size $\mathcal{N}$.
The expected progress in one iteration is therefore 
\begin{eqnarray}
	E[\Lambda(s_{\zeta_{\mbox{\tiny sum}}},\tau),\Lambda(s_{\zeta_{\mbox{\tiny sum}}'},\tau+1)]&\leq& \sum_{i=1}^l\frac{i}{(n\cdot \mathcal{N})^i}
<\frac{2}{n\cdot \mathcal{N}}
\end{eqnarray}
and the expected progress in $\tau$ iterations is consequently not greater than $\frac{2\tau}{n\cdot \mathcal{N}}$.

When we choose $\tau=\frac{n\cdot N\cdot\Delta}{4}-1$, the double of the expected progress is still smaller than $\Delta$.
With the Markov inequality we can show that this progress is not achieved with probability~$\frac{1}{2}$.
Altogether we conclude that the expected optimisation-time is bound from below by
\begin{eqnarray}
	E[T_\mathcal{P}]&\geq& \tau\cdot\left(1-\frac{E[\Lambda(s_{\zeta_{\mbox{\tiny sum}}},\tau)]}{\Delta}\right) \nonumber\\
	&\geq&  \frac{n\cdot \mathcal{N}\cdot\Delta}{4}\cdot \left(1- \frac{\frac{2\cdot n\cdot \mathcal{N}}{4\cdot n\cdot \mathcal{N}}\cdot\Delta}{\Delta}\right)\nonumber\\
	&=&\Omega(n\cdot \mathcal{N}\cdot\Delta)
\end{eqnarray}
\end{proof}

\begin{proof}[Proof of Theorem~\ref{theoremThree}]
The asymptotically sharp bound is a result of theorem~\ref{theoremOne} and theorem~\ref{theoremTwo} with 
\begin{equation}
\Delta=\log(n)+\frac{\log(k)}{n\cdot \mathcal{N}^2}.     
\end{equation}
\end{proof}

\pagebreak

\section[RF-sensing of activities from non-cooperative subjects in device-free recognition systems using ambient and local signals]{RF-sensing of activities from non-cooperative subjects in device-free recognition systems using ambient and local signals \footnote{Originally published as ' Stephan Sigg, Markus Scholz, Shuyu Shi, Yusheng Ji and Michael Beigl: RF-sensing of activities from non-cooperative subjects in device-free recognition systems using ambient and local signals, in IEEE Transactions on Mobile Computing (TMC), Feb. 2013, vol. 13, no. 4 (DOI: http://doi.ieeecomputersociety.org/10.1109/TMC.2013.28)' (1536--1233 \copyright 2013 IEEE) }}\label{sectionOriginalRF01}
We consider the detection of activities from non-cooperating individuals with features obtained on the Radio Frequency channel.
Since environmental changes impact the transmission channel between devices, the detection of this alteration can be used to classify environmental situations.
We identify relevant features to detect activities of non-actively transmitting subjects. 
In particular, we distinguish with high accuracy an empty environment or a walking, lying, crawling or standing person, in case-studies of an active, device-free activity recognition system with software defined radios.
We distinguish between two cases in which the transmitter is either under the control of the system or ambient. 
For activity detection the application of one-stage and two-stage classifiers is considered.
Apart from the discrimination of the above activities, we can show that a detected activity can also be localised simultaneously within an area of less than $1$~meter radius.

\subsection{Introduction}\label{sectionIntroductionRF01}
In the approaching Internet of Things (IoT), virtually all entities in our environment will be enhanced by sensing, communication and computational capabilities~\cite{IoT_Li_2012,IoT_Haller_2010}.
These entities will provide information on environmental situations, interact in the computation and processing of data~\cite{FunctionComputation_Sigg_2012} and store information.
In order to sense environmental situations, common sensors in current applications are light, movement, pressure, audio or temperature~\cite{5840}.
Clearly, for reasons of cost and sensor size it is desired to minimise the count of distinct sensors in IoT entities.
The one sensor class that defines the minimum set naturally available in virtually all IoT devices is the Radio Frequency (RF)-transceiver to communicate with other wireless entities~\cite{Pervasive_Scholz_2011}. 
It is also shipped with nearly every contemporary electronic device like mobile phones, notebooks, media players, printers as well as keyboards, mouses, watches, shoes and rumour has spread about even media cups. 
Therefore, the RF transceiver is a ubiquitously available sensor class.
It is capable of sensing changes or fluctuation in a received RF-signal.
Radio waves are blocked, reflected or scattered at objects.
At a receiver, the signal components from distinct signal paths add up to form a superimposition.
When objects that block or reflect the signal path of some of these signal components are moved, this is reflected in the superimposition of signal waves at the receiver.
We assert that specific activities in the proximity of a receiver generate characteristic patterns in the received superimposed RF-signal.
By identifying and interpreting these patterns, it is possible to detect activities of non-cooperating subjects in an RF-receiver's proximity.

Although the wireless channel is occasionally utilised for location detection of other RF devices~\cite{RFSensing_Woyach_2006,RFSensing_Muthukrishnan_2007} or passive entities~\cite{Pervasive_Youssef_2007,Pervasive_Seifeldin_2013}, it is seldom used to detect other contexts like activities from entities which are not equipped with a RF-transceiver.

We consider the detection of activities of device-free entities from the analysis of RF-channel fluctuations induced by these very activities.
In analogy to the definition of device-free radio-based localisation systems (DFL)~\cite{Pervasive_Youssef_2007} we define device-free radio-based activity recognition systems (DFAR) as \textit{systems which recognise the activity of a person using analysis of radio signals while the person itself is not required to carry a wireless device} (cf.~\cite{Pervasive_Scholz_2011b}). 
In addition to the sensor type employed, we further categorise radio-based activity recognition systems by the parameters enlisted in table~\ref{tableClassification}.
In particular, we distinguish between passive and active systems depending on whether a transmitter is part of and under control of the radio-based recognition system. Also, an ad-hoc system can be installed in a new environment without re-training the classifier, while a non-ad-hoc system requires initial training or configuration.

In this work, we focus on the detection of static and dynamic activities of single individuals by active and passive, non-ad-hoc DFAR systems.
The active system employs a dedicated transmitter as part of the recognition hardware while the passive system utilises solely ambient FM radio from a transmitter not under the control of the system.
Compared to preliminary work on RF-based activity recognition~\cite{Pervasive_Shi_2012,Pervasive_Shi_2012b,4036,ContextAwareness_Sigg_2011,Pervasive_Scholz_2011}, the novel contributions are
\begin{enumerate}
     \item A comprehensive discussion of research campaigns utilising RF-channel based features for the detection of location or activities (section~\ref{sectionRelatedWorkRF01})
     \item A concise investigation on possible features for RF-based activity recognition (section~\ref{sectionFeaturesRF01})
     \item A case study on activity classification of a single individual from RF-channel based features for
     \begin{itemize}
          \item[3a)] an active DFAR system utilising 900 MHz software defined radio nodes (section~\ref{sectionActiveDetection}), and
          \item[3b)] a passive DFAR system utilising ambient FM radio signals at 82.5 MHz (section~\ref{sectionPassiveDetection}) considering in both cases
     \end{itemize}
     \item the classification accuracy with respect to activity and location.
\end{enumerate}

The majority of the features we consider are amplitude-based.
Since with the Received Signal Strength Indicator (RSSI), a related value is commonly provided by contemporary transceiver hardware, the features utilised in this study can be implemented similarly for most current mobile devices.

Our discussion is structured as follows.
In section~\ref{sectionRelatedWorkRF01} we review the related work on activity and location recognition with a particular focus on radio frequency based or related environmental features.
Section~\ref{sectionUseCases} then discusses use-cases and application scenarios for RF-based activity recognition.  
The features utilised in our case-studies are introduced, analysed and discussed in section~\ref{sectionFeaturesRF01}.
Based on some of these features, we report from the experiments in section~\ref{sectionActivityDetection}.
In particular, we demonstrate the detection of five activities with active and passive DFAR systems.
We can also show that a localisation of these activities is feasible simultaneously from the same set of features.
Section~\ref{sectionConclusionRF01} summarises the results and closes our discussion.
\begin{table}
\caption{Classification-parameters for radio-based context recognition systems {\scriptsize (1536--1233 \copyright 2013 IEEE)}}
\centering
\begin{tabular}{p{2.4cm}|p{5.0cm}}
         \textbf{Parameter} & \textbf{Values} \\\hline 
                Sensor type & Device-bound; Device-free\\\hline
                Sensing modality & Passive; Active \\\hline
            Setup & Ad-hoc; Non-ad-hoc (requires training)\\\hline
     \end{tabular}
\label{tableClassification}
\end{table}

\subsection{Related work}\label{sectionRelatedWorkRF01}
Activity recognition comprises the challenge to recognise human activities from the input of sensor data.
A broad range of sensors can be applied for this task.
Traditionally, accelerometer devices have evolved as the standard equipment for activity recognition both for their high diffusion and convincing recognition rates~\cite{Pervasive_Ravi_2005,Pervasive_Cao_2012}.
General research challenges for activity recognition regard the accurate classification of noisy data captured under real world conditions~\cite{Pervasive_Bao_2004} or the automation of recognition systems~\cite{Pervasive_Ploetz_2012}.
Another problem that is addressed in depth only recently is the creation of classification systems that scale to a large user base. 
With increasing penetration of sensor enriched environments and devices, the diversity in user population poses new challenges to activity recognition.
Abdullah et al. for instance address this challenge by maintaining several groups of similar users during training to identify inter-user differences without the need for individual classifiers~\cite{Pervasive_Abdullah_2012}.

Even more fundamental and aligned to this scaling problem is the required cost for accurately equipping subjects, training them to the system, equipping the environment or the users and most importantly, having them to actually wear the sensing hardware.
The classification accuracy is highly dependent on the accurate sensor location.
The integration of sensors in clothing as well as the recent remarkable progress in the robustness to rotation or displacement have improved this situation greatly~\cite{Pevasive_Chavarriaga_2011}.
However, a subject is still required to cooperate and at least wear the sensors~\cite{Pervasive_Cohn_2012}.
This requirement can not be assured generally in real-world applications. 
In particular, even devices as private as mobile phones, which are frequently assumed to be constantly in the same context as its owner~\cite{Pervasive_Patridge_2008, Pervasive_Varshavsky_2008, Pervasive_Lane_2010}, can not serve as a sensor platform suitable to accurately capture the context of an individual.   
Dey et al. investigated in~\cite{Pervasive_Dey_2011} that users have their mobile phone within arms reach only 54\% of the time.
This confirms a similar investigation of Patel et al. in 2006~\cite{Pervasive_Patel_2006} which reported a share of 58\% for the same measure.

These general challenges of activity recognition can be overcome be using an environmental sensing modality.
Naturally, vision-based approaches, such as video~\cite{Pervasive_Hongeng_2004} and recently also the Kinect and wii concepts have been employed by scientists to classify gestures and activities~\cite{Pervasive_Bannach_2008}.
However, the burden of installation and cost make such approaches hard to deploy at scale~\cite{Pervasive_Cohn_2012}.
Recently, researchers therefore explore alternative sensing modalities that are pre-installed and readily available in environments and therefore minimise installation cost.

Patel et. al. coined the term infrastructure-mediated sensing and demonstrated in 2007 that alterations in resistance and inductive electrical load in a residential power supply system due to human interaction can be automatically identified~\cite{Pervasive_Patel_2007}.
They leveraged transients generated by mechanically switched motor loads to detect and classify such human interaction from electrical events.
In a related work from 2010, Gupta et al. analysed electromagnetic interference (EMI) from switch mode power supplies~\cite{RFSensing_Gupta_2010}.
In~\cite{Pervasive_Gupta_2011} they showed that it is even possible to detect simple gestures near compact fluorescent light by analysing the EMI-structures, effectively turning common light bulbs in a house into sensors.
Environmental sensing with atypical sensing devices is also considered by Campbell et al. and Thomaz et al. who present an activity detection method utilising residential water pipes~\cite{Pervasive_Campbell_2010,Pervasive_Thomaz_2012}.
In~\cite{Pervasive_Cohn_2010}, Cohn et al. form a residual Power-line system into a large distributed antenna to sense low power signals from parasitic or distant devices.
These approaches all require explicit interaction between a cooperating individual and a specific sensing entity.
These approaches are bound to specific environments or installations and typically are only feasible indoors.
An infrastructure mediated sensing medium with greater range is the RF-channel.
Signal strength, amplitude fluctuation or noise level provide information that can be utilised to classify environmental situations.

Several authors considered the localisation of individuals based on measurements from the RF-sensor.
Results are typically achieved by analysing the RF-signal amplitude, namely the RSSI of a received signal.
Classical approaches are device-bound and utilise the RF-sensor for location estimation of an active entity equipped with a RF transceiver.
In these approaches, the impact of multi-path fading and shadowing on the transmission channel and therefore the strength of an RF signal is exploited.
These approaches were driven by the attempt to provide capabilities of indoor localisation.
The first promising work was the RADAR system presented by Bahl et al.~\cite{Pervasive_Bahl_2000}.
The authors took advantage of existing communications infrastructure, WiFi access points, and employed RSSI fingerprints to identify locations off-line. 
With location, this approach was then applied also with GSM networks~\cite{Pervasive_Otsason_2005,Pervasive_Varshavsky_2007}, FM radio signals~\cite{Pervasive_Krumm_2003,Pervasive_Youssef_2005} and domestic powerline~\cite{Pervasive_Patel_2006,Pervasive_Stuntebeck_2008}.
Recently, automations have been proposed for such fingerprinting approaches~\cite{Pervasive_Jiang_2012,Pervasive_Pulkkinen_2012}.
  
While these systems rely on a two staged approach in which first a map of fingerprints is created off-line, recent work achieves on-line real-time localisation of entities equipped with a wireless transceiver based on WiFi or FM radio~\cite{Pervasive_Schougaard_2012,Pervasive_Wang_2012,Pervasive_Chen_2012}. 
These studies were initiated in~2006 by Woyach et al. who detail various environmental changes and their effect on a transmit signal~\cite{RFSensing_Woyach_2006}.
The authors utilise MICAz nodes to show that motion detection based on RSSI measurements can be more accurate than accelerometer data when changes are below the sensitivity of the accelerometer.
They experimentally employ several indoor-settings in which a receive node analyses a signal obtained from a transmitter. 
In this study they focused on an increased fluctuation in the RSSI signal level.
Additionally, the authors showed that velocity of an entity can be estimated by analysing the RSSI pattern of continuously transmitted packets of a moving node.
This work was advanced by Muthukrishnan et al. who study in 2007 the feasibility of motion sensing in a WiFi network~\cite{RFSensing_Muthukrishnan_2007}.
They analyse fluctuation in the 1~byte 802.11 RSSI indicator to sense whether a device is moving.
The authors consider only the two cases of motion and no motion and achieve a classification accuracy of up to $0.94$.
A more fine grained distinction was made by Anderson et al. and Sohn et al. based on fluctuations in GSM signal strength~\cite{RFSensing_Anderson_2006,RFSensing_Sohn_2006}.
The authors of~\cite{RFSensing_Anderson_2006} implement a neural network to detect the travel mode of a mobile phone.
They monitor the signal strength fluctuation from cells in the active set to distinguish between walking, driving and stationary with an accuracy between $0.8$ and $0.9$.

Sohn et. al describe a system that extracts seven features from GSM signal strength measurements to distinguish six velocity levels with an accuracy of $0.85$~\cite{RFSensing_Sohn_2006}.
The features mainly build on distinct measures of variation in signal strength and the frequency of cell-tower changes in the active set.

While all previously mentioned results considered special installations of the wireless transmitters, Sen et al presented a system that allows the localisation of a wireless device with an accuracy of about 1 meter from WiFi physical layer information even when the receiver is carried by a person that might induce additional noise to the captured features~\cite{RFSensing_Sen_2012}.

Summarising, these studies are examples of device-bound and active velocity and location estimation approaches since they require that the located entity is equipped with a RF transceiver. 

Recently, some authors also consider RF-sensing to detect the presence or location of passive entities.
Since these systems require at least one active transmitter, they can be classified as active, device-free systems.

Youssef defines this approach as Device-Free Localisation (DFL) in~\cite{Pervasive_Youssef_2007} to localise or track a person using RF-Signals while the entity monitored is not required to carry an active transmitter or receiver.
They localised individuals by exchanging packets between 802.11b nodes in corners of a room and analysed the moving average and its variance of the RSSI~\cite{Pervasive_Youssef_2007}.
Classification accuracy reached up to $1.0$ for some configurations.
Additionally, they presented a fingerprint-based localisation system with an accuracy of $0.9$.
Later, they improved their approach using less nodes~\cite{Pervasive_Seifeldin_2013}.
A passive radio map was constructed offline before a Bayesian-based inference algorithm estimated the most probable location.
These experiments have been conducted under Line-of-Sight (LoS) conditions.
Also, Wilson and Patwari showed in conformance with the findings of Kosba et al.~\cite{Pervasive_Kosba_2012b} that the variance of the RSSI can be used as an indicator of motion of non-actively transmitting individuals regardless of the average path loss that occurs due to dense walls and stationary objects~\cite{RFSensing_Wilson_2009}.
The area in which environmental changes impact signal characteristics was then considered by Zhang et al.
They used 870~MHz nodes arranged in a grid to show that for each link an elliptical area of about $0.5$ to 1 meters diameter exists for which RSSI fluctuation caused by an object traversing this area exceeds measurements in a static environment~\cite{RFSensing_Zhang_2009}.
They identified a valid region for detecting the impact (i.e. the RSSI fluctuations exceeding the measured threshold in a static environment) for transceiver distances from 2~m to 5~m for the considered 870~MHz frequency range~\cite{RFSensing_Zhang_2011}.
By dividing a room into hexagonal cell-clusters with measurements following a TDMA scheduling, an object position could be derived with an accuracy of around 1~meter.
This accuracy was further improved by Wilson and Patwari in 2011~\cite{RFSensing_Wilson_2009}.
They utilised a dense node array to locate individuals within a room with an average error of about $0.5$~meters.
This was possible by instrumenting a tomographic image over the 2-way RSSI fluctuations of nodes~\cite{RFSensing_Wilson_2010}.
All these studies consider a single experimental setting.

In a related work, Lee et al. sense the presence of an individual in five distinct environments~\cite{RFSensing_Lee_2010}.
They showed that the RSSI peak is concentrated in a restricted frequency band in a vacant environment while it is spread and reduced in intensity in the presence of an individual.
In 2011, Kosba et al. presented a new system for the detection of human movement in a monitored area~\cite{Pervasive_Kosba_2011}.
Using anomaly detection methods they achieved 6\% miss detection and a 9\% false alarm rate when utilising the mean and standard deviation of the RSSI in two environments.
They further implemented techniques to counteract effects of dispersion.
This was accomplished by continuously adding newly measured data which did not trigger the detection.
The previous results all considered the localisation of a single individual.

The simultaneous localisation of multiple individuals at the same time was first mentioned and studied by Patwari and Wilson in~\cite{Pervasive_Patwari_2011b}.
The authors derive a statistical model to approximate the position of a person based on RSSI variance which can be extended to multiple persons.
This aspect together with the previously untackled problem that environmental changes over time might necessitate frequent calibration of the location system was approached by Zhang and others in~\cite{Pervasive_Zhang_2012}.
The authors isolate the LoS path by extracting phase information from the differences in the RSS on various frequency spectrums at distributed nodes.
Their experimental system is with this approach able to simultaneously and continuously localise up to 5 persons in a changing environment with an accuracy of 1~meter.

We summarise that most work conducted in the area of RF-based classification with passive participants is related to the localisation of individuals.
The feasibility of this approach was verified in various environmental settings and at various frequencies.
The features  utilised are mostly the RSSI, its moving average, mean or RSSI fingerprint.
Also, 2-way RSSI variance was employed.
With these features a localisation accuracy of about $0.5$~meters was possible or the simultaneous localisation of up to 5 persons in a changing environment with an accuracy of 1 meter.

While the localisation of individuals based on features from the radio channel can therefore generally be considered as solved, recently, some authors considered active DFAR approaches to also detect activities.

Patwari et al. monitor breathing based on RSS analysis~\cite{RFSensing_Patwari_2011}.
The monitored area was surrounded by twenty 2.4~GHz nodes and the two-way RSSI was measured.
Using a maximum likelihood estimator they approximated the breathing rate within $0.1$ to $0.4$~beats accuracy.

Recently, we also conducted preliminary studies regarding the use of features from a RF-transceiver to classify static environmental changes such as opened or closed doors, presence, location and count of persons with an accuracy of $0.6$ to $0.7$~\cite{4036,ContextAwareness_Sigg_2011,Pervasive_Scholz_2011,OrganicComputing_Sigg_2011}.
We utilised USRP Software defined radio devices (SDR)\footnote{http://www.ettus.com} from which one constantly transmits a signal that is read and analysed by other nodes.
Devices were equipped with 900~MHz transceiver boards.
With the software radios a higher sampling frequency than in previous studies is possible and we can also sample the actual channel instead of only tracking the RSSI.
In these studies we concentrated on features related to the signal amplitude and derivation of the instantaneous amplitude from its mean.
Furthermore, we conducted preliminary studies on passive device free situation awareness by utilising ambient signals from a FM radio station not under the control of the recognition system. 
In these studies, static environmental changes such as opened doors have been detected with an accuracy of about $0.9$~\cite{Pervasive_Shi_2012} and a first study on suitable features to detect human activities could achieve an accuracy of about $0.8$ with a two stage recognition approach~\cite{Pervasive_Shi_2012b}. 

DFAR is still a mostly unexplored research field.
Open research questions regard the optimum frequencies and the impact of the frequency on the classification accuracy, the optimum sampling rate of the signal, the detection range and the impact of this distance on the classification accuracy as well as the minimum Signal-to-Noise Ratio (SNR).
Furthermore, a set of activities that can be recognised by RF-based classification is yet to be identified as well as a suitable design of the detection system.
In particular, the impact of the count and height of transmitting and receiving nodes has not yet been considered comprehensively as well as even the actual necessity of a transmit node as part of the recognition system since potentially the system might utilise ambient radio.
Also, it is not clear whether and how activities of multiple persons can be identified simultaneously and if features exist that enable ad-hoc DFAR systems.
A more detailed discussion of most of these aspects is given by Scholz et al. in~\cite{Pervasive_Scholz_2011b}.

In the present study, we identify and evaluate features for the classification of activities from RF-signals in two frequency bands (900~MHz and 82.5~MHz) with systems utilising ambient radio as well as a system-generated signal.
Four activities, two dynamic and two static, together with the empty environment are considered.

\subsection{Application scenarios for DFAR}\label{sectionUseCases}
We believe that DFAR research can provide a foundation for the realisation of an IoT and for Ubicomp in general.
The RF-Sensor has a high penetration in common equipment and will be available in virtually all IoT devices.
To reduce cost and complexity, hardware designers and application developers might then rather investigate and utilise the common RF-transceiver to sense environmental stimuli than integrating additional sensing hardware.
Currently, the information provided by the RF-channel is, although available virtually for free, mostly disregarded and discarded unused.

Apart from modulated data, the signal strength, amplitude fluctuation or noise level provide additional information about environmental situations.
In the following sections we exemplify two applications for DFAR in emergency situations and elderly monitoring.

\subsubsection{Monitoring in disaster stricken areas}
Despite tremendous efforts, careful preparation and training for a 'worst case', increased security precautions and costly installations of early warning systems, disaster situations either caused by nature or human intervention frequently strike also highly developed countries.
Recent cautionary tales are the flooding in Thailand or also the Tohoku earthquake near Sendai, Japan that let to a devastating tsunami and was the cause of the atomic crisis around the Fukushima-Daichi power plant.

In the time since this event, research efforts have been taken in the search of systems that can assist auxiliary forces in areas where most of the infrastructure is destroyed.
One important and urgent issue in such situations is the search for survivors and injured persons that might reside, for instance, in partly destroyed buildings~\cite{Pervasive_Song_2010,Pervasive_Palmer_2012}.

When the existing infrastructure is destroyed, RF-sensing might provide a cheap and wide-ranging alternative to assist rescue forces.
With a single RF-transmitter such as an RF-radio tower or a base station, a large area can not only be supplied with voice and data communication but the fluctuation in RF-channel characteristics might be employed to detect individuals and identify their status from activities such as lying, crawling, standing or walking.
Auxiliary forces might bring out a network of RF-transceiver devices in order to monitor an area via RF-channel fluctuation as part of their professional routine while at the same time establishing communication means via this RF-transceiver infrastructure~\cite{Pervasive_Ramirez_2011}.
The range, optimum installation height and features for ad-hoc operation are still open research questions for DFAR but the results presented in this work show that assistance in such scenarios can be provided by RF-sensing (although due to the lack of prior training the set of activities recognised might be reduced, for instance, to 'some movement' and 'some static alteration'). 
These additional sensing capabilities come virtually for free on top of the installation of wireless communication.

\subsubsection{Supporting well-being in domestic areas}
Most accidents happen at home. 
The primary reason for these accidents are falls which make up about 40\% of the total number of accidents~\cite{HassLass2002}. 
Most of these accidents leave the affected person in an unusual posture such as lying at an unusual location. 
While the automatic detection of fall and fall prevention has gained large interest in the research community and various approaches have been proposed, these alarm system either need body-attached sensors, require the installation of a complex infrastructure or have strong privacy related implications as, for instance, video based systems~\cite{Pervasive_Nouri_2007, Pervasive_Cucchiara_2007,Pervasive_Lin_2007}.
By utilising the RF-sensor for this kind of detection we would reduce privacy issues, avoid the need of having to carry sensors and ideally reduce installation requirements to a minimum. 

The sensor could further become a crucial component of (Health) Smart Home systems~\cite{Noury_Book_2011,Pervasive_Vacher_2011} relieving users from the necessity to wear a device. 
In fact, for Smart Home systems, the sensor needs to provide a rough localisation capability as well as the recognition of at least a basic set of activities of daily living. Among such activities are walking, standing and sleeping~\cite{Pervasive_Lukovicz_2012}.

Considering the demographic change in developing and developed countries, the application of the RF-sensor for alarm systems or Smart Homes could further play an important role towards the extension of self-sustained living of the elderly.
The present study illustrates the potential of the ubiquitously available RF-sensor for the detection of relevant activities in Smart Home environments.

\subsection{Features for DFAR}\label{sectionFeaturesRF01}
In the following we discuss the RF-based features we considered and their achieved classification accuracy.
We identify a set of three most relevant features for active and passive DFAR systems.

For our active DFAR system, we deploy a USRP SDR transmit node constantly broadcasting a signal $m(t)$ at a frequency of $f_c=900$~MHz.
In the passive DFAR system a FM radio signal $m(t)$ from a local radio station at $f_c=82.5$~MHz is utilised.
In both cases, the received signal
\begin{equation}
\zeta_{\mbox{\footnotesize rec}}(t)=\Re\left(m(t)e^{j2\pi f_ct}\mbox{RSS}e^{j(\psi+\phi)}\right)\label{equationOneRF01}  
\end{equation}
is read by one USRP SDR node and is analysed for signal distortion and its fluctuation due to channel characteristics.
In equation~(\ref{equationOneRF01}) the RSS denotes the Received Signal Strength.
The value $\phi$ accounts for the phase offset in the received signal due to the signal propagation time.
This continuous received signal is sampled from the USRP devices $64\cdot10^9$ times per second at distinct time intervals $t=1,2,\dots$ in a resolution of 12 bits.

We considered the following features for activity classification.
For all features we employed a window $\mathcal{W}$ of $|\mathcal{W}|$ samples to calculate their value.
The blocking or damping of signal components by subjects or other entities impacts the amplitude of the received signal.
A feature to measure this property is the maximum peak of the signal amplitude. 
We calculate it by the difference between the maximum and minimum amplitude within one sample window
\begin{equation}
     \mathcal{P}_{\mbox{\footnotesize eak}}=\max_{t\in \mathcal{W}}\left(\zeta_{\mbox{\footnotesize rec}}(t)\right)-\min_{t\in\mathcal{W}}\left(\zeta_{\mbox{\footnotesize rec}}(t)\right)
\end{equation}

We utilise the Mean amplitude $\mu$ of the received signal frequently as a reference value to compare the current amplitude of a signal $\zeta_{\mbox{\footnotesize rec}}(t)$ to the average amplitude in a training situation:
     \begin{equation}
          \mu=\frac{\sum_{t=1}^{\mathcal{|W|}} \zeta_{\mbox{\footnotesize rec}}(t)}{\mathcal{|W|}}
     \end{equation}
The Root of the Mean Square (RMS) deviation of the signal amplitude $\left|\zeta_{\mbox{\footnotesize rec}}(t)\right|$ to the mean $\mu$ is also utilised.
With lower RMS we expect fewer alterations in an environment.
\begin{equation}
     \mbox{RMS}=\sqrt{\frac{\sum_{t=1}^{|\mathcal{W}|}{\left(\zeta_{\mbox{\footnotesize rec}}(t)-\mu\right)}}{|\mathcal{W}|}^2}\label{equationStd}
\end{equation}

Furthermore, we investigate the second and third central moment that express the shape of a cloud of measured points.
The second central moment describes the variance $\sigma^2$ of a set of points. 
It can be used to measure how far a set of points deviates from its mean.
\begin{equation}
     \sigma^2=\frac{\sum_{t=1}^{|\mathcal{W}|}{\left(\zeta_{\mbox{\footnotesize rec}}(t)-\mu\right)}}{|\mathcal{W}|}^2\label{equationVar}
\end{equation}

Additionally, we consider the third central moment.
\begin{equation}
     \gamma=\mbox{E}[\zeta_{\mbox{\footnotesize rec}}(t)-\mu)^3]\label{equationThirdMoment}
\end{equation}
In equation~(\ref{equationThirdMoment}), $\mbox{E}[x]$ defines the expectation of a value $x$.

All above features are taken from the time domain of the received signal. 
In the frequency domain, we consider the DC component $a_0$, the spectral energy~$\mathcal{E}$ and the entropy~H of the signal. 

The feature $a_0$ represents the average of all samples a Fast Fourier Transform (FFT) was applied to. 
It describes the vertical offset of an observed signal. 	

We calculate its $i^{th}$ frequency component as 
\begin{equation}
 \mbox{FFT}(i)=\sum_{t=1}^{|\mathcal{W}|}{{\zeta_{\mbox{\footnotesize rec}}}(t){e^{-j{\frac{2\pi}{N}}it}}}. \label{equationFFT}
\end{equation}
In equation~(\ref{equationFFT}) we choose the window size $|\mathcal{W}|$ as the quantity of the samples in the FFT.

The DC component is defined by the first Fourier coefficient FFT$(i)$ and is separately calculated as 
\begin{equation}
     a_0=\frac{\int_{-\frac{|\mathcal{W}|}{2}}^{\frac{|\mathcal{W}|}{2}}\left(\zeta_{\mbox{\footnotesize rec}}(t)\right)dt}{|\mathcal{W}|}
\end{equation}

The signal energy~$\mathcal{E}$ can be computed as the squared sum of its probability density of spectrum in each frame. 
The probability of each spectral FFT$(i)$ band is 
\begin{equation}
 \mbox{P}(i)= \frac{{\mbox{FFT}(i)}^2}{\sum_{j=1}^{{|\mathcal{W}|}/{2}}{\mbox{FFT}(j)}^2}. \label{equationPro}
\end{equation}
Consequently, we calculate the spectral energy as
\begin{equation}
 \mathcal{E}=\sum_{i=1}^{{|\mathcal{W}|}/{2}}{\mbox{P}(i)}^2. \label{equationEnergy}
\end{equation}

We compute the entropy of a set of points as 
\begin{equation}
     \mbox{H}=\sum_{i=1}^{{|\mathcal{W}|}/{2}}{\mbox{P}(i)}\cdot\ln\left(\mbox{P}(i)\right).
\end{equation}

For all possible combinations of up to three of these features we exploited their classification accuracy of the five activities considered in section~\ref{sectionActiveDetection} and in section~\ref{sectionPassiveDetection}.
Table~\ref{tableFeatureAccuracy} details the accuracy for the best five feature combinations\footnote{The complete table with all results is available at http://klab.nii.ac.jp/\textasciitilde sigg/TMC-2012-01-0047\_PassiveDFARAcc.pdf} of the passive DFAR system with~$|\mathcal{W}|=32$.
\begin{table}
     \caption{Best feature combinations for the passive DFAR system {\scriptsize (1536--1233 \copyright 2013 IEEE)}}
\subfloat[Distinction between dynamic and static activities]{\begin{footnotesize}\begin{tabular}{r|cccccccc}
acc&$\mathcal{P}_{\mbox{eak}}$&$\mu$&$a_0$&$\mathcal{E}$&H&$\sigma^2$&$\gamma$&RMS\\\hline
.866&x&&x&&&x&&\\
.863&x&&&&&x&x&\\
.861&x&&x&&&&x&\\
.861&x&&x&&&&&\\
.861&x&&&&&&x&\\

\end{tabular}\end{footnotesize}}
\hfill
\subfloat[Distinction between standing, lying and empty]{\begin{footnotesize}\begin{tabular}{r|cccccccccc}
acc&$\mathcal{P}_{\mbox{eak}}$&$\mu$&$a_0$&$\mathcal{E}$&H&$\sigma^2$&$\gamma$&RMS\\\hline
.902&x&&x&&&&&\\
.898&x&x&x&&&&&\\
.898&x&&x&&&&x&\\
.896&x&x&&&&&&x\\
.894&x&x&&&&&x&\\
\end{tabular}\end{footnotesize}
}

\subfloat[Distinction between walking and crawling]{\begin{footnotesize}\begin{tabular}{r|cccccccccc}
acc&$\mathcal{P}_{\mbox{eak}}$&$\mu$&$a_0$&$\mathcal{E}$&H&$\sigma^2$&$\gamma$&RMS\\\hline
.817&&&x&&x&&x&\\
.706&&x&x&&&&&x\\
.701&x&&x&&&x&&\\
.701&x&&&x&x&&&\\
.701&x&&&&&x&x&
\end{tabular}\end{footnotesize}
\label{tableFeatureAccuracyC}
}
\hfill
\subfloat[Distinction between all five activities: standing, walking, crawling, lying and empty]{\begin{footnotesize}\begin{tabular}{r|cccccccccc}
acc&$\mathcal{P}_{\mbox{eak}}$&$\mu$&$a_0$&$\mathcal{E}$&H&$\sigma^2$&$\gamma$&RMS\\\hline
.694&x&&&&x&&x&\\
.686&x&&x&&&&x&\\
.683&x&&x&&&&&\\
.679&x&x&x&&&&&\\
.679&x&&x&&&&&x\\
\end{tabular}\end{footnotesize}
}
\label{tableFeatureAccuracy}
\end{table}

The table distinguishes between one-stage and two-stage classification.
For the one-stage classification, all five activities are distinguished in one single classification step.
For two-stage classification, first the classifier distinguishes between dynamic (walking or crawling) and static (lying, standing or empty) activities.
Then, the final classification is done in one of these classes.
We observe that in particular $\mathcal{P}_{\mbox{eak}}$ and $a_0$ are well suited to achieve a high classification accuracy.

A high $\mathcal{P}_{\mbox{eak}}$ value indicates a dynamic activity.
The feature is therefore well suited to distinguish between dynamic and static activities.
Consequently, for the distinction between the two dynamic activities in table~\ref{tableFeatureAccuracyC}, this feature is less prominent.
The DC-component $a_0$ mostly represents the vertical offset of the signal.
In can therefore serve as an indicator to distinguish whether a person is standing or walking, lying or standing or whether the room is empty.

For the active DFAR system, the most significant features are the variance~$\sigma^2$, the third central moment $\gamma$ when applied twice and the minimum over a window of maximum values.
We achieved good results for a window size of $|\mathcal{W}|=20$ which translates to $|\mathcal{W}|=400$ for features applied on preprocessed data.
By the adding further combinations of features, the overall classification accuracy can be further improved slightly.

Generally, for the activities considered, the dynamic activities have a greater number of significant features as they also have characteristic alterations over time.
Static activities are therefore in principle harder to distinguish from each other and it will likely not be possible to re-use a trained classifier for static activities without re-training in another scenario.

\subsection{RF-based DFAR}\label{sectionActivityDetection}
In order to explore the activity recognition capabilities and limits of the RF-sensor, we conducted case studies for active and passive DFAR implementations.
In particular, three subjects have conducted the four activities lying, standing, crawling and walking in a corridor of our institute.
Additionally, the empty corridor was considered as a baseline activity. 
All experiments have been conducted in after-hours to ensure a controlled environment in which all important external parameters are kept stable.
In particular, no additional subjects have been present in the corridor or in adjacent rooms that could have interfered with the experimental conditions.
Figure~\ref{figureCaseStudy} depicts the setting employed for the case study.
\begin{figure*}
     \centering
\includegraphics[width=\textwidth]{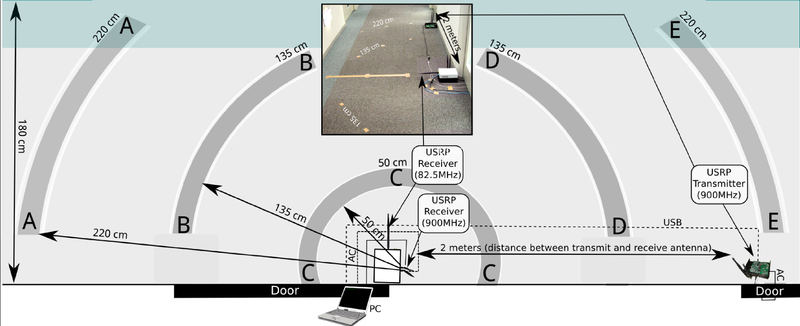}
      \caption{Schematic illustration of the corridor in which the case-study was performed. Locations at which activities were conducted are marked (A,B,C,D,E). Both receive nodes are located in the center of the recognition area on top of each other. {\scriptsize (1536--1233 \copyright 2013 IEEE)}}
     \label{figureCaseStudy}
\end{figure*}

The experimental space was divided into five areas with respect to their distance to the receiver.
For active and passive systems, the receiver was placed at the same location in the center of the detection area.
For the active DFAR system, the transmitter was positioned in two meters distance from the receiver.

The activities were conducted at the five locations which are labelled A, B, C, D, and E.
Locations A and E are in a distance of $2.20$~meters from the receiver, locations B and D are separated by $1.35$~meters and location C is $0.5$~meters apart.
All locations are arranged in a circle around the receiver in their center.

Each of the three subjects repeated all activities at every location for about 60 seconds.
We took arbitrary patterns from these sample sequences for classification.

For the active DFAR system the transmitter constantly modulated a signal to a 900~MHz carrier which was then sampled at the receiver at 70~Hz.
USRP~1 devices\footnote{https://www.ettus.com/product/details/USRP-PKG} were utilised as transmitter and receiver with RFX900 daughterboards\footnote{https://www.ettus.com/product/details/RFX900} and VERT900 Antennas\footnote{https://www.ettus.com/product/details/VERT900} with 3dB antenna gain.

The receiver of the passive DFAR system sampled a signal from an ambient FM radio station at 82.5 MHz with a sample rate of 255 kHz.
We employed a USRP N200\footnote{https://www.ettus.com/product/details/UN200-KIT} device with a WBX daughterboard\footnote{https://www.ettus.com/product/details/WBX} together with a VERT900 Antenna\footnote{https://www.ettus.com/product/details/VERT900} with 3dB antenna gain~\cite{Pervasive_Shi_2012b}.

\subsubsection{Active device-free activity recognition}\label{sectionActiveDetection}
For the detection of the described activities with our active DFAR system we utilise a one-stage classification approach.
In particular we use as features the mean~$\mu$, the variance~$\sigma^2$, the third central moment~$\gamma$, the RMS, the count of amplitude peaks within 90\% of the maximum, the distance of zero crossings, the Energy~$\mathcal{E}$ and the entropy~H, over a window of 400 samples.

For classification we utilise a k-nearest neighbour (k-NN) classifier with $k=10$ and a decision tree (DT).
Figure~\ref{figureRawFeatures} depicts values for the variance~$\sigma^2$ and the third central moment $\gamma$ applied twice for part of the sample data.
\begin{figure}
     \includegraphics[width=\textwidth]{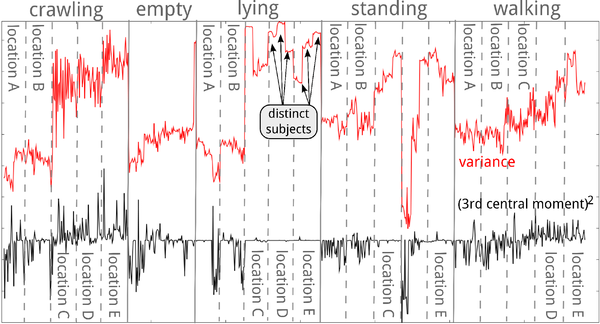}
     \caption{Exemplary feature samples (variance and twice applied 3rd central moment; over 400 samples each) from all activities, locations and subjects for active DFAR {\scriptsize (1536--1233 \copyright 2013 IEEE)}}
     \label{figureRawFeatures}
\end{figure}
Distinct activities are clearly distinguishable in this plot already. 

From this data we observe that activities conducted at locations~A and~B are seemingly harder to distinguish from the empty case.
The reason is that activities at these locations are conducted relative to the transmitter behind the receiver and therefore have less impact on the received signals. 

Classification results after 10-fold cross validation are depicted in table~\ref{tableConfusionOverall}.
\begin{table}\begin{footnotesize}
        \caption{Classification of activities conducted by three subjects at Locations A to E by a k-nearest neighbour and a decision tree classifier in an active DFAR system {\scriptsize (1536--1233 \copyright 2013 IEEE)}}
         \subfloat[Confusion matrix for the k-NN classifier over samples from all locations and subjects]{\begin{tabular}{l | c c c c c}
      &\multicolumn{5}{c}{Classification}\\
      &cr & em & ly & st & wa\\
      \hline
crawling&\textbf{.713}&.024&.06& &.204\\
empty&.022&\textbf{.593}&.187&.121&.077\\
lying&.042&.048&\textbf{.743}&.144&.024\\
standing&.011&.067&.078&\textbf{.777}&.067\\
walking&.166&.029&.034&.051&\textbf{.72}
    \end{tabular}}
  \hfill
         \subfloat[Confusion matrix for the classification tree classifier over samples from all locations and subjects]{\begin{tabular}{l | c c c c c}
      &\multicolumn{5}{c}{Classification}\\
      &cr & em & ly & st& wa \\
      \hline
crawling&\textbf{.659}&.006&.054&.024&.257\\
empty&.055&\textbf{.582}&.154&.121&.088\\
lying&.054&.042&\textbf{.784}&.102&.018\\
standing&.022&.056&.095&\textbf{.771}&.056\\
walking&.189&.023&.011&.057&\textbf{.72}
    \end{tabular}}
        \label{tableConfusionOverall}
\end{footnotesize}\end{table}
Table fields with very low values (i.e. $0.0$) are left blank.
The table depicts the classification accuracy when classifiers for the five activities empty, walking, standing, lying, crawling have been trained on features obtained for all five locations and subjects.
Due to the challenging feature value fluctuations for locations~A and~B we have not been able to achieve a higher accuracy in this case.
In particular, we notice that the distinction of the empty class is hard for the classifiers since other activities conducted at locations~A and~B have a similar feature value footprint.
The overall classification accuracies are $0.714$ and $0.722$ for the classification tree and the k-NN classifier as depicted in table~\ref{tableStatisticsOverall}.
\begin{table}\centering
        \caption{Accuracy, Information score and Brier score for the classification algorithms {\scriptsize (1536--1233 \copyright 2013 IEEE)}}
\begin{tabular}{l | c c c }
 &Accuracy&Information score&Brier score\\
      \hline
 Classification tree&0.716&1.529&0.567\\
 k-NN classifier&0.722&1.518&0.4
   \end{tabular}
        \label{tableStatisticsOverall}
\end{table}
The table also shows the Brier score and the Information score as defined by Kononenko and Bratko~\cite{MachineLearning_Kononenko_1991}.
These basic accuracies can be improved when classifiers are trained at specific locations and when the classification of activities is segmented for distinct locations as derived in the next sections.

\paragraph{Spatial impact on accuracy}
To improve accuracy we spatially restricted the classification area.
In particular, we utilised feature values only from activities conducted at one distinct location (A,B,C,D or E).
Table~\ref{tableDistinctLocationC} shows the classification results for location~C.
\begin{table}\begin{footnotesize}
        \caption{Classification of activities conducted by three subjects at Location C by a k-nearest neighbour and a decision tree classifier in an active DFAR system {\scriptsize (1536--1233 \copyright 2013 IEEE)}}
         \subfloat[Confusion matrix for the k-NN classifier over samples from all subjects at location C]{\begin{tabular}{l | c c c c c}
      &\multicolumn{5}{c}{Classification }\\
      &cr &em&ly&st&wa\\
      \hline
crawling&\textbf{.848}&&&.03&.121\\
empty&&\textbf{.978}&&&.022\\
lying&&&\textbf{.839}&.161&\\
standing&&&.108&\textbf{.892}&\\
walking&.143&&&&\textbf{.857}
    \end{tabular}}
  \hfill
         \subfloat[Confusion matrix for the classification tree classifier over samples from all subjects at location C]{\begin{tabular}{l | c c c c c}
      &\multicolumn{5}{c}{Classification }\\
 &cr&em&ly&st&wa\\      \hline
crawling&\textbf{.727}&.03&.03&.03&.182\\
empty&&\textbf{.978}&.022&&\\
lying&&.032&\textbf{.677}&.29&\\
standing&.054&.027&.135&\textbf{.784}&\\
walking&.2&&&&\textbf{.8}
    \end{tabular}}
        \label{tableDistinctLocationC}
\end{footnotesize}\end{table}
The classification accuracy is increased in this case compared to the previous general setting.
This is also due to the subjects conducting activities in only about 50~cm distance from the receive antenna.
The impact on the signal is therefore significant. 

With increasing distance to the receiver, the classification accuracy slowly deteriorates as visible in table~\ref{tableAllDistinctLocations}.
\begin{table}\begin{footnotesize}
        \caption{Classification of activities conducted by three subjects at Locations A, B, D or E by a k-nearest neighbour classifier in an active DFAR system {\scriptsize (1536--1233 \copyright 2013 IEEE)}}
         \subfloat[Confusion matrix for the k-NN classifier over samples from all subjects at location A]{\begin{tabular}{l | c c c c c}
      &\multicolumn{5}{c}{Classification }\\
      &cr & em & ly & st & wa \\
      \hline
crawling&\textbf{.556}&.167&.194&&.083\\
empty&.044&\textbf{.769}&.055&.066&.066\\
lying&.125&.25&\textbf{.625}&&\\
standing&&.286&&\textbf{.686}&.029\\
walking&.03&.394&&&\textbf{.576}
    \end{tabular}}
  \hfill
         \subfloat[Confusion matrix for the k-NN classifier over samples from all subjects location B]{\begin{tabular}{l |  c c c c c}
      &\multicolumn{5}{c}{Classification }\\
      &cr & em & ly & st & wa \\
      \hline
crawling&\textbf{.8}&.133&&&.067\\
empty&&\textbf{.78}&.132&.077&.011\\
lying&&.242&\textbf{.727}&&.03\\
standing&&.278&&\textbf{.722}&0\\
walking&.03&.242&.091&.061&\textbf{.576}
    \end{tabular}}

         \subfloat[Confusion matrix for the k-NN classifier over samples from all subjects at location D]{\begin{tabular}{l | c c c c c}
      &\multicolumn{5}{c}{Classification }\\
      &cr & em & ly & st & wa \\
      \hline
crawling&\textbf{.765}&&&.029&.206\\
empty&&\textbf{.967}&.033&&\\
lying&&&\textbf{.971}&.029&\\
standing&&.194&.056&\textbf{.722}&.028\\
walking&.111&&&&\textbf{.889}
    \end{tabular}}
  \hfill
         \subfloat[Confusion matrix for the k-NN classifier over samples from all subjects location E]{\begin{tabular}{l | c c c c c}
      &\multicolumn{5}{c}{Classification }\\
      &cr & em & ly & st & wa \\
      \hline
crawling&\textbf{.676}&&.029&&.294\\
empty&&\textbf{.967}&.033&&\\
lying&&&\textbf{.865}&.135&\\
standing&.029&&.114&\textbf{.829}&.029\\
walking&.263&&&&\textbf{.737}
    \end{tabular}}
        \label{tableAllDistinctLocations}
\end{footnotesize}\end{table}
The table depicts the classification accuracy of the k-NN classifier. 
Classification accuracies for the decision tree are comparable as shown in table~\ref{tableStatisticsAllLocations}.
\begin{table}\centering
        \caption{Accuracy, Information score and Brier score for the classification algorithms in an active DFAR system {\scriptsize (1536--1233 \copyright 2013 IEEE)}}
\begin{footnotesize}\begin{tabular}{l | c c c }
 &Accuracy&Information score&Brier score\\
      \hline
\textit{\hfill Location A} & & &\\
 Classification tree&0.674&1.284&0.652\\
 k-NN classifier&0.674&1.309&0.449\\\hline
\textit{\hfill Location B} & & &\\
 Classification tree&0.825&1.696&0.35\\
 k-NN classifier&0.735&1.385&0.383\\\hline
\textit{\hfill Location C} & & &\\
 Classification tree&0.841&1.702&0.317\\
 k-NN classifier&0.907&1.892&0.129\\\hline
\textit{\hfill Location D} & & &\\
 Classification tree&0.832&1.709&0.337\\
 k-NN classifier&0.887&1.846&0.168\\\hline
\textit{\hfill Location E} & & &\\
 Classification tree&0.779&1.549&0.442\\
 k-NN classifier&0.851&1.767&0.214
   \end{tabular}\end{footnotesize}
        \label{tableStatisticsAllLocations}
\end{table}
We observe that indeed locations E, D and C achieve best classification results.
The impact of reflected signals from an action conducted behind the receiver quickly diminishes, so that the classification accuracy of activities conducted at these locations (A~and~B) quickly worsens with distance.

\paragraph{Localising an action}
In the previous case, the classifiers were trained for actions of a specific location without considering actions taking place at other locations.
We now train the classifiers on all five activities at all five locations, respectively.
The action 'empty' is identical regardless of the location.
Overall, we then distinguish between 21 classes.
Table~\ref{tableCasestudyOverallSummary} depicts our results.
\begin{sidewaystable}
\centering
\setlength{\tabcolsep}{3pt}
\begin{footnotesize}
        \caption{Accuracy when training is accomplished including activities at all locations in an active DFAR system {\scriptsize (1536--1233 \copyright 2013 IEEE)}}
          \begin{tabular}{l |ccccc|c|ccccc|ccccc|ccccc}
      &\multicolumn{21}{c}{Classification}\\
     & \multicolumn{5}{c}{crawling}&empty&\multicolumn{5}{c}{lying}&\multicolumn{5}{c}{standing}&\multicolumn{5}{c}{walking}\\
 &A&B&C&D&E& &A&B&C&D&E&A&B&C&D&E&A&B&C&D&E\\\hline
crawling at A&\textbf{.528}&.028&&&&.111&.167&.028&&&&&&&&&&.111&.028&&\\
crawling at B&.033&\textbf{.867}&&&&.067&&&&&&&&&&&&.033&&&\\
crawling at C&.03&&\textbf{.455}&.152&.091&&&&&&&&&&&&&.03&.061&&.182\\
crawling at D&&&.176&\textbf{.176}&.147&&&&.029&&&&&&&&&&.029&.118&.324\\
crawling at E&&&.088&.176&\textbf{.529}&&&&&&&&&&&&&&&&.206\\\hline
empty&.022&&&&&\textbf{.681}&.011&.088&&.011&.022&.044&.077&&&&.044&&&&\\\hline
lying at A&.094&&&&&.188&\textbf{.594}&.125&&&&&&&&&&&&&\\
lying at B&.091&&&&&.152&.152&\textbf{.545}&&&&&&&&&.03&.03&&&\\
lying at C&&&&&&&&&\textbf{.677}&&.032&&&.097&.065&.129&&&&&\\
lying at D&&&&&&&&&.059&\textbf{.765}&.088&&&&&.088&&&&&\\
lying at E&&&&&&&&&&.162&\textbf{.595}&&&.135&&.108&&&&&\\\hline
standing at A&&&&&&.2&&&&&&\textbf{.686}&.057&&&&.057&&&&\\
standing at B&&&&&&.194&&&&&&.306&\textbf{.444}&&&&.056&&&&\\
standing at C&&&&&&&&&.054&.027&.054&&&\textbf{.676}&&.189&&&&&\\
standing at D&&.028&&&&.028&&&.111&&&.056&&.028&\textbf{.472}&.111&.028&.111&.028&&\\
standing at E&&&&&.029&&&&.057&.086&.114&&&.114&.057&\textbf{.514}&&&&.029&\\\hline
walking at A&&&&&&.152&&.061&&&&.061&.061&&&&\textbf{.364}&.242&.061&&\\
walking at B&.061&.03&&&&.03&.03&.03&&&&.03&&&.061&&.303&\textbf{.364}&.061&&\\
walking at C&&&.086&&&&&&&&&&&&.029&&.029&.086&\textbf{.657}&.114&\\
walking at D&&&.056&.028&&&&&&&&&&&&&&&.139&\textbf{.75}&.028\\
walking at E&&&.158&.184&.158&&&&.026&&&&&&&.026&&&.026&&\textbf{.421}
    \end{tabular}
        \label{tableCasestudyOverallSummary}
\end{footnotesize}\end{sidewaystable}

We observe that the right action is classified for the right location most often.
Moreover, we see a locality in the classifications.
Misclassifications are seldom in different activities but most often regarding the correct activity in a neighbouring location.
With increasing distance to the place where the action was trained, the misclassification error increases.
The distance between locations was $85$~cm. 
We therefore conclude that a localisation of activities is possible alongside classification with an error of less than 1~meter.
We further observe that the static activities standing and lying as well as the dynamic activities walking and crawling are harder to distinguish for the classifier since their features are not so well separated.

\paragraph{Summary on active DFAR studies}\label{sectionCaseStudySummary}
All five activities are classified with varying accuracy depending on the setting considered. 
Higher accuracy can be achieved when the activities are conducted near the receiver node or between the transmitter and receiver.
With increasing distance to the receiver the classification accuracy deteriorates.
When the activities are trained at various locations, a localisation of the classified activity within less than 1~meter radius is possible.

\subsubsection{Passive device-free activity recognition}\label{sectionPassiveDetection}
In the previous section we considered an active transmitter as one part of the classification system.
The disadvantage in such a system is that in a practical situation, a separate transmitter has to be brought out and positioned in the proximity of the receiver that is constantly transmitting.
However, since the freely available frequency spectrum is sparse, we can assume to be exposed to some kind of radio signals continuously.
The highest coverage is probably reached by FM radio.
We attempt to utilise ambient FM radio signals from a nearby FM radio station in order to detect the five activities described above in the same setting.
In our case study, the FM-receiver was placed on top of the 900~MHz USRP receiver to sample ambient signals.
Samples have been taken simultaneously to the active DFAR studies described above.
We utilise 10 fold cross validation with k-NN and decision tree classifiers.
A two-stage classification approach with the feature sets that reached best classification accuracy was shown in table~\ref{tableFeatureAccuracy}).

Table~\ref{tablePassiveAllC} details the classification accuracy when activities are conducted at location~C only and classifiers are trained only on these feature values.
\begin{table*}
\setlength{\tabcolsep}{1pt}
\centering
\caption{Classification of activities of all subjects conducted at Location~C by a k-NN and a classification tree classifier in a passive DFAR system {\scriptsize (1536--1233 \copyright 2013 IEEE)}}
\begin{footnotesize}
\setlength{\tabcolsep}{3pt}
 \subfloat[Confusion matrix for the k-NN classifier]{
     \begin{tabular}{l |ccccc}
&\multicolumn{5}{c}{Classification }\\
 & empty & lying & standing & walking &crawling\\
        \hline
 empty &\textbf{.942} (.152)&&.058 (.008)&& (.001)\\
 lying&&\textbf{.773} (.268)&.027 (.007)&.093 (.021)&\\
 standing &.093 (.024)&.027 (.005)&\textbf{.853} (.236)&.027 (.005)&\\
 walking & .0 (.001)&.026 (.008)&.077 (.014)&\textbf{.795} (.261)&.103 (.032)\\
 crawling & .0 (.001)&.121 (.038)&.014 (.004)&.176 (.071)&\textbf{.689} (.214)\\
      \end{tabular}
    }
\hfill
\subfloat[Confusion matrix for the classification tree classifier]{
     \begin{tabular}{l |ccccc}
&\multicolumn{5}{c}{Classification }\\
 & empty & lying & standing  & walking &crawling\\
        \hline
empty &\textbf{.986} (.008)&&.014 (.003)&&\\
 lying&&\textbf{.787} (.252)&.013 (.004)&.133 (.006)&.067\\
  standing&.027 (.006)&.013 (.004)&\textbf{.933} (.009)&&.027 (.005)\\
 walking&&.026 (.005)&.077 (.014)&\textbf{.769} (.210)&.128 (.033)\\
 crawling&&.108 (.031)&.027 (.006)&.135 (.045)&\textbf{.730} (.245)\\
      \end{tabular}
    }
\end{footnotesize}
\label{tablePassiveAllC}
\end{table*}
The table depicts the median classification accuracy and the variance over 10 separate classifications.
For ease of presentation, table entries with very low values ($0.0$ $(0.0)$) are left empty.

With increasing distance to the receiver, the classification accuracy deteriorates.
Naturally, since we utilise ambient signals, the direction in which activities are moved away from the receiver is of minor importance (cf. table~\ref{tablePassiveAllABDE}).
\begin{table*}
\centering
\caption{Classification of activities of all subjects at Locations A,B,D and E by a k-NN classifier {\scriptsize (1536--1233 \copyright 2013 IEEE)}}
\begin{footnotesize}
\setlength{\tabcolsep}{3pt}
 \subfloat[Confusion matrix for the classification at location A]{
     \begin{tabular}{l |ccccc}
&\multicolumn{5}{c}{Classification }\\
 & empty & lying & standing  & walking &crawling\\
        \hline
 empty &\textbf{.812} (.231)&.145 (.074)&&.014 (.006)&.029 (.010)\\
 lying&.155 (.091)&\textbf{.239} (.177) &.155 (.066)&.169 (.073)&.282 (.139)\\
 standing &&.27 (.098)&\textbf{.893} (.158)&&.080 (.004)\\
 walking &.054 (.003)&.162 (.055)&&\textbf{.649} (.294)&.135 (.042)\\
 crawling &.053 (.009)&.027 (.004)&.093 (.007)&.120 (.007)&\textbf{.707} (.302)\\
      \end{tabular}
    }
\hfill
 \subfloat[Confusion matrix for the classification at location B]{
     \begin{tabular}{l |ccccc}
&\multicolumn{5}{c}{Classification }\\
 & empty & lying & standing  & walking &crawling\\
        \hline
 empty &\textbf{.783} (.243)&.072 (.009)&.116 (.009)& .0 (.002)&.029 (.006)\\
 lying& .0 (.001)&\textbf{.507} (.419)&.440 (.338)&.053 (.006)&\\
 standing &&.214 (.110)&\textbf{.786} (.208)&&\\
 walking &.035 (.005)&.047 (.006)&.012 (.003)&\textbf{.659} (.361)&.247 (.157)\\
 crawling &.013 (.003)&.041 (.009)&&.230 (.108)&\textbf{.716} (.293)\\
      \end{tabular}
    }

 \subfloat[Confusion matrix for the classification at location D]{
     \begin{tabular}{l |ccccc}
&\multicolumn{5}{c}{Classification }\\
 & empty & lying & standing  & walking &crawling\\
        \hline
 empty &\textbf{.875} (.186)&&.125 (.045)&&\\
 lying&&\textbf{.758} (.308)&.242 (.112)&&\\
 standing &.222 (.097)&.25 (.113)&\textbf{.500} (.301)&&.027 (.012)\\
 walking & .0 (.001)&.120 (.034)&.173 (.044)&\textbf{.680} (.143)&.027 (.005)\\
 crawling &&.214 (.085)&.071 (.026)&.262 (.063)&\textbf{.452} (.296)\\
      \end{tabular}
    }
\hfill
 \subfloat[Confusion matrix for the classification at location E]{
     \begin{tabular}{l |ccccc}
&\multicolumn{5}{c}{Classification }\\
 & empty & lying & standing  & walking &crawling\\
        \hline
 empty &\textbf{.725} (.262)&.087 (.021)&.159 (.034)&.029 (.004)&\\
 lying&.289 (.103)&\textbf{.461} (.186)&.145 (.032)&.105 (.028)& .0 (.001)\\
 standing &.130 (.041)&.273 (.125)&\textbf{.558} (.371)&.039 (.004)&\\
 walking &.025 (.007)&.120 (.058)&.025 (.004)&\textbf{.667} (.286)&.160 (.071)\\
 crawling &.026 (.007)&.013 (.004)&&.171 (.068)&\textbf{.789} (.228)\\
      \end{tabular}
    }
\end{footnotesize}
\label{tablePassiveAllABDE}
\end{table*}
These tables show classification results of the k-NN classifier. 
Classification accuracies of the classification tree have been comparable.
However, we observe that especially the detection of static activities, in particular lying and standing suffer from the increased distance to the receiver.
We explain this with the missing LoS signal component between transmitter and receiver.
Without a dominant signal component which could have high impact on the received signal when, for instance, blocked, all incoming signal components have equal or similar impact on the signal at the receiver.

For the utilisation of an ambient signal, the distance to a receiver is more critical than in the active DFAR case.
In particular, short of the empty corridor, all classification accuracies drop significantly when activities are conducted at locations remote to the location at which the classifier was trained.
Table~\ref{tableAllLocationsClassifierC} shows this property for a case in which the classifier is trained from feature values of activities conducted at location~C but applied to activities conducted at all five locations.
\begin{table}
\centering
\caption{Accuracy for the k-NN classifier when training is accomplished with activities from all subjects conducted at Location C only in a passive DFAR system. {\scriptsize (1536--1233 \copyright 2013 IEEE)}}
\begin{footnotesize}\begin{tabular}{l | ccccc}
&\multicolumn{5}{c}{Classification }\\
 & empty & lying & standing& walking &crawling\\
        \hline
 empty at location A&\textbf{1.0}&&&&\\
 lying at location A&.355&\textbf{.156}&.298&.085&.106\\
 standing at location A &&.691&\textbf{.215}&.054&.04\\
 walking at location A &.068&.102&.245&\textbf{.251}&.333\\
 crawling at location A &.094&.168&.309&.342&\textbf{.087}\\
   \hline
empty at location B &\textbf{1.0}&&&&\\
 lying at location B&.207&&.427&.327&.04\\
 standing at location B &.446&&\textbf{.331}&.223&\\
 walking at location B &.053&.059&.112&\textbf{.647}&.129\\
 crawling at location B &.034&.047&135&.608&\textbf{.176}\\
 \hline
 empty at location C &\textbf{.986}&&.014&&\\
 lying at location C&&\textbf{.787}&.013&.133&.067\\
 standing at location C &.027&.013&\textbf{.933}&&.027\\
 walking at location C &&.026&.077&\textbf{.769}&.128\\
 crawling at location C &&.108&.027&.135&\textbf{.730}\\
 \hline
empty at location D &\textbf{1.0}&&&&\\
 lying at location D&.113&\textbf{.035}&.423&.373&.056\\
 standing at location D &.165&.152&\textbf{.468}&.171&.044\\
 walking at location D &.04&.228&.201&\textbf{.262}&.268\\
 crawling at location D &.054&.118&.344&.269&\textbf{.215}\\
 \hline
empty at location E &\textbf{1.0}&&&&\\
 lying at location E&.408&&.355&.224&.013\\
 standing at location E &.442&.013&\textbf{.325}&.201&.019\\
 walking at location E &.204&.043&.228&\textbf{.42}&.105\\
 crawling at location E &.099&.132&.152&.503&\textbf{.113}\\
      \end{tabular}\end{footnotesize}
\label{tableAllLocationsClassifierC}
\end{table}

When training the classifier with feature values from activities conducted at various locations the accuracy decreases (cf.~table~\ref{tableClassificationAllC}).
\begin{table}
\caption{Classification of activities at all locations by a k-NN and a classification tree classifier trained on activities conducted by all subjects on location C only in a passive DFAR system {\scriptsize (1536--1233 \copyright 2013 IEEE)}}
\centering
 \subfloat[Confusion matrix for the k-NN classifier]{
     \begin{footnotesize}\begin{tabular}{l | ccccc}
&\multicolumn{5}{c}{Classification }\\
 & em & ly & st  & wa &cr\\
        \hline
 empty &\textbf{.949}&&.043&.007&\\
 lying&.154&\textbf{.473}&.241&.065&.067\\
 standing &.192&.156&\textbf{.577}&.049&.025\\
 walking &.054&.128&.126&\textbf{.476}&.216\\
 crawling &.032&.153&.126&.199&\textbf{.490}\\
      \end{tabular}\end{footnotesize}
    }
\hfill
\subfloat[Confusion matrix for the classification tree]{
     \begin{footnotesize}\begin{tabular}{l | ccccc}
&\multicolumn{5}{c}{Classification }\\
 & em & ly & st  & wa &cr\\
        \hline
 empty &\textbf{.978}&&.022&&\\
 lying&&\textbf{.500}&.198&.301&.001\\
 standing &&&\textbf{.787}&.212&.001\\
 walking &&.001&&\textbf{.824}&.130\\
 crawling &&&.051&.433&\textbf{.516}\\
      \end{tabular}\end{footnotesize}
    }
\label{tableClassificationAllC}
\end{table}

Furthermore, a localisation of the conducted activities as it was feasible for the active DFAR case is hardly possible with the passive DFAR system. 
The classifier can at most give a hint on the possible location as it can be observed from table~\ref{tablePassiveAllAll}.
\begin{sidewaystable}
\caption{Passive DFAR classification accuracy of the k-NN classifier and localisation of activities when training is accomplished including activities at all locations {\scriptsize (1536--1233 \copyright 2013 IEEE)}}
\centering
\setlength{\tabcolsep}{3pt}
     \begin{tabular}{l | c c c c c| c|c c c c c|ccccc|ccccc}
      &\multicolumn{21}{c}{Classification}\\
     & \multicolumn{5}{c}{crawling}&empty&\multicolumn{5}{c}{lying}&\multicolumn{5}{c}{standing}&\multicolumn{5}{c}{walking}\\
&A&B&C&D&E& &A&B&C&D&E&A&B&C&D&E&A&B&C&D&E\\\hline
crawling at A &.134&.027&.027&.067&.013&.014&.027&.04&.027&.027&.02&.06&.047&.013&.067&.04&.08&.034&.06&.081&.094\\
 crawling at B&.007&.304&.041&.014&.155&.014&.027&.007&&.02&&&&.027&.007&.007&.007&.135&.101&.054&.074\\
 crawling at C&.02&.048&.476&.02&.014&&.007&.027&.109&&&&&.007&.007&&.041&.061&.068&.082&.014\\
 crawling at D&.086&.022&.043&.097&.054&.011&.022&.075&&.065&.011&.108&.022&.022&.043&.043&.032&.065&.075&.065&.043\\
 crawling at E&.007&.106&.013&.046&.285&.02&&.007&.04&.106&&&&&.007&&&.172&.079&.02&.093\\
\hline
empty&&&&&&1.0&&&&&&&&&&&&&&&\\
\hline
lying at A&.035&.014&.007&.021&.007&.141&.206&.092&.014&.043&.043&.014&.028&.071&.035&.05&.071&.021&.014&.050&.021\\
 lying at B&.08&.013&.02&.047&.007&.027&.1&.147&&.033&.087&&.073&.047&.06&.133&.007&&.02&.027&.073\\
 lying at C&.04&&.107&&.047&&.013&&.638&&&.054&&.007&.02&&.013&.013&.007&.04&\\
 lying at D&.014&.035&.007&.056&.134&.007&.014&.035&&.282&.035&.028&.028&.042&.049&.028&.035&.014&.063&.063&.028\\
 lying at E&.026&&&.007&&.125&.013&.072&&.033&.191&&.118&.086&.066&.145&.02&&.02&.013&.066\\
\hline
 standing at A&.06&&&.067&&&.013&&.054&.013&&.544&&.007&.121&.013&.04&&&.067&\\
 standing at B&.036&&&.014&&.115&.014&.101&&.036&.144&&.223&.029&.072&.173&.007&&.007&&.029\\
 standing at C&.007&.027&.007&.013&&0.18&.08&.053&.033&.027&.093&.007&.033&.293&.067&.007&.04&.007&.007&.007&.013\\
 standing at D&.063&.006&&.025&.013&.013&.044&.063&.013&.025&.063&.114&.063&.07&.228&.07&.013&&.025&.082&.006\\
 standing at E&.039&&&.032&.006&.102&.071&.123&&.013&.117&&.156&.026&.065&.13&.013&.013&.013&.026&.052\\
\hline
 walking at A&.061&.007&.048&.027&.007&.021&.048&.007&.027&.034&.034&.034&.014&.054&.02&.02&.395&.027&.007&.088&.02\\
 walking at B&.006&.153&.071&.018&.147&.006&.029&.012&.012&.029&&&&.012&&.012&.024&.224&.147&.024&.076\\
 walking at C&.032&.097&.045&.039&.084&&.013&.013&.013&.045&.013&&&.006&.006&.013&.006&.155&.258&.058&.103\\
 walking at D&.06&.04&.067&.04&.027&&.074&.04&.04&.054&.007&.087&&.027&.067&.02&.067&.020&.054&.154&.054\\
 walking at E&.093&.074&.012&.024&.148&.037&.006&.056&&.031&.062&&.031&.024&.006&.043&.019&.056&.08&.062&.136\\
      \end{tabular}
\label{tablePassiveAllAll}
\end{sidewaystable}
The table details the classification accuracy for all 21 classes considering the activities and their respective locations.

\paragraph{Summary on passive DFAR studies}
Summarising, we conclude that activity classification is also feasible with a passive DFAR system utilising ambient FM radio signals. 
We achieved best classification accuracies when the activity was conducted within $0.5$ to $1$~meters from the receiver.
At higher distances, however, the classification accuracy quickly deteriorated and is hardly usable.
A passive DFAR system must therefore employ a higher count of receive devices but can omit a dedicated transmitter.
In short distance, classification accuracy is comparable to active DFAR systems.

\subsection{Conclusion} \label{sectionConclusionRF01}
We have proposed a classification scheme for device-free radio-based activity recognition systems.
Following this scheme, we considered non-ad-hoc, active and passive, device-free activity recognition systems. 

Classification was achieved by k-NN and decision tree classifiers with similar classification accuracy.
For one-stage and two-stage active and passive DFAR systems we derived a set of most significant features with respect to their classification accuracy in our case studies.
The presented work is the first to detect the considered activities from RF-channel measurements and also the first to do this with active and passive DFAR systems.
Despite some recent advances on device-free radio-based localisation systems, this is also the first study to combine an activity recognition and localisation in one classification algorithm on a common set of features.
For the activities lying, crawling, standing and walking we were able to localise them within less than $1$~meter in USRP-SDR-based case-studies with the active DFAR system.

The results of this study effectively enable the use of arbitrary wireless devices as sensing equipment.

Still, open challenges remain and present future research questions for radio-based activity recognition systems. 
Among them are the development of algorithms and features which reduce the amount of training effort or the amount of additional required knowledge in order to use the RF-sensor in a different setting. 
We further need to investigate the required coverage, height and relative location of the sensor in order to deduce how the number of sensor entities affect the resolution of the system. 
Other questions include the activity detection of multiple persons or the inclusion of mobile nodes within a DFAR system. 

Nevertheless, with the presented investigation results it could be shown that the RF-sensor can support applications such as monitoring of emergency situations or the creation of Smart Home systems. 
In both applications the sensor could not only provide a classification accuracy comparable to the currently used technologies but also provide novel services, such as detecting non-cooperating persons, and increases the level of convenience, for instance, by not having to wear an actual monitoring device. 
Based on these findings and the truly pervasive character of the underlying physical entity we believe that RF-based sensing can be essential in the pervasive systems of the upcoming Internet of Things.
\vfill
\pagebreak

\section[Monitoring of Attention Using Ambient FM-radio Signals]{Monitoring of Attention Using Ambient FM-radio Signals \footnote{Originally published as 'Shuyu Shi, Stephan Sigg, Wei Zhao, and Yusheng Ji: Monitoring of Attention from
Ambient FM-radio Signals, IEEE Pervasive Computing, Los Alamitos, CA, USA, IEEE Computer Society, Jan-Mar 2014, vol. 13, no. 1, pp. 30-36, 2014 (DOI: http://dx.doi.org/10.1109/MPRV.2014.13)' (Published by the IEEE CS n 1536-1268/14/\$31.00 \copyright 2014 IEEE}}\label{sectionOriginalRF02}
We investigate the classification of FM-radio signal fluctuation for the monitoring of attention by individuals in motion towards a static object.
In particular, we distinguish in a corridor, whether it is empty or populated by moving or standing individuals as well as the attention of these subjects towards poster frames in that corridor.
We consider the distinction in front of which poster these subjects are walking or standing as well as their walking speeds or changes therein.
This information can provide some hint whether a person is paying attention to a specific poster in this corridor as well as the location of the particular poster. 

\subsection{Introduction}
Attention determines for a system the potential to impact the actions and decisions taken by an 
individual~\cite{AttentionMonitoring_Xu_2012}.
The management of attention covers the activation of attention as well as its detection and timely exploitation.
The same action of the same system might be considered either as annoyance or be 
appreciated as helpful depending on whether the individual was focusing part or all of her attention towards the system or not.
In the literature, we find various definitions that classify attention as well as its determining 
characteristics~\cite{AttentionMonitoring_Wu_2007}. 
A straightforward measure of attention might be the tracking of gaze~\cite{AttentionMonitoring_Yonezawa_2007}. 
In general, aspects such as Saliency, Effort, Expectancy and Value are important indicators of attention~\cite{AttentionMonitoring_Wickens_2008}. 
Alois Ferscha and others extended this model and put a greater stress on the effort a person takes towards an object~\cite{AttentionMonitoring_Ferscha_2012}.

We consider the following scenario.
In a corridor, a series of electronic poster frames are installed while people are walking by these frames.
From the perspective of a specific poster, a significant part of its message shall be recognised by passers-by.
Therefore, the poster should draw the attention of people passing by and, when this is achieved, it might possibly 
transport additional information.
Consequently, the poster frame should be aware of people passing by, know where people are in 
order to attract attention at the right moment and detect whether attention is attracted.
In this work we assume that the monitored individuals are not cooperating with the system and hence are not equipped with any part of the sensing hardware.

Such detection and management of attention may require elaborate installations and very 
specific sensors in order to accurately sense quantities such as Saliency, 
Effort, Expectancy and Value~\cite{AttentionMonitoring_Xu_2012}. 
However, we believe that for many commercial installations, cost and ease of installation and not primarily the highest achievable accuracy are most important.
Also more general, environmental sensors can provide sufficient information to estimate the attention state of individuals.

We propose to utilise ambient FM-radio signals for the detection of attention since it has a nearly perfect coverage in populated areas and features cheap receiver hardware~\cite{Percom_Popleteev_2012}.
\begin{center}
\textsl{\dots but which aspects of attention can actually be captured by an FM-receiver?}
\end{center}

Ferscha and others~\cite{AttentionMonitoring_Ferscha_2012} discuss various aspects of attention and identify as most distinguishing factors changes in walking speed, direction or orientation.
From FM-radio signals it is hard to detect the orientation of a person.
However, it is feasible to classify walking speeds, walking direction or location of individuals.   
We show that with a straightforward installation, we can distinguish 
\begin{itemize}
     \item an empty corridor, a person walking by and a person standing in front of a poster frame
     \item the specific poster the person is observing
     \item the location where a person is walking
     \item the walking speed of a person
\end{itemize}
We therefore argue that attention levels can be inferred upon interpretation of the changes in walking speed or direction as derived from our system.
This information can provide some indication on the attention of persons towards a poster frame.
Also, it can enable a frame to take action in order to catch the attention of a person just in the right moment.

\subsection{Sensing passive entities from the RF-channel}\label{sectionRelatedWorkRF02}
The RF-channel has been recently utilised by various authors for the detection of location or activities of passive entities, not equipped with a transmitter or receiver (See~\cite{Pervasive_Sigg_2013} and references therein). 
These studies exploit fluctuation of received signal strength or signal amplitude conditioned on changes (for instance, movement, altered location of objects) in the physical proximity of a receiver~\cite{RFSensing_Woyach_2006}. 
This work is related to the passive radar literature.
In an early publication, Kaipin Tan et. al presented the potential of a GSM-based passive radar prototype for detecting and tracking different types of ground-moving targets~\cite{DeviceFreeRecognition_Tan_2005} in an outdoor environment. 
For the simple binary detection of presence in an indoor environment, Masahiro~Nishi et. al proposed an indoor human detection system leveraging the multi-path radio propagations of VHF-FM and UHF-TV broadcasting signals~\cite{PIMRC_Nishi_2006}. 
In their work, they leveraged incoming signal waves.
Moustafa~Youssef and others then demonstrated the localisation of individuals by analysing the Received-Signal-Strength indicator (RSSI) in received packets at 802.11b nodes~\cite{Pervasive_Youssef_2007}.
Neal~Patwari and Joey~Wilson introduced a statistical, empirically verified model to approximate the position of a person based on the variance in the Received Signal Strength  Indicator~\cite{RFSensing_Patwari_2011}. 
The aspect of localising up to five individuals at a time, together with the previously untackled problem that environmental changes over time might necessitate frequent calibration of the location system was approached by Dian~Zhang and others using a grid of wireless sensor nodes in~\cite{Pervasive_Zhang_2012}.
They isolated the Line-of-Sight path by extracting phase information from the differences in the received signal strength on various frequency spectrums at distributed nodes.

These studies on localisation of individuals assume a transmitter as part of and under the control of the recognition system. 
However, we recently demonstrated that a recognition is also possible when an ambient signal source is utilised. 
In particular, we analysed fluctuation in ambient signals from an FM-radio station not under the control of the recognition system. 
Static environmental changes such as opened doors have been detected with an accuracy of about $0.9$ 
and a first study on suitable features to detect human activities could achieve an accuracy of about $0.8$ with a two-staged recognition approach~\cite{Pervasive_Shi_2012b}. 
In this article we extend this work towards the monitoring of attention of a single subject.

\subsection{Monitoring attention from FM-radio} \label{feasibilitystudy}
For the monitoring of attention of non-cooperating individuals, an environmental signal source is required.
We propose the utilisation of RF-signals for their ease of deployment and nearly perfect coverage in indoor and outdoor locations through existing, pre-installed systems in populated areas~\cite{Percom_Popleteev_2012}, argue why we favour FM-radio above other RF-technologies and detail the features we utilise for the detection of 
attention.

\subsubsection{Radio-based monitoring of attention}\label{sectionRFAttention}
Radio waves are electromagnetic waves, defined by their amplitude, phase and frequency. During signal 
propagation from the transmitter to the receiver, the radio waves are impacted by physical phenomena, for instance, damping, reflection and scattering.
Assume a signal observed at a receiver, at some frequency $f_c$~[Hz].
Naturally, as we are considering an indoor environment in which no direct-line-of-sight exists, incoming signals arrive over mutliple paths at roughtly equal strenth from all directions.
In the event that a signal wave encounters any structure such as an object or individual, the main signal component will be damped (continue its path with reduced energy) or even completely blocked.
Additionally, the signal is typically reflected or scattered at this occasion. 
Reflection describes the event that the signal wave bounces away from an object in a modified direction.
Typically, the signal will also experience scattering, which is the splitting of a signal wave due to the not perfectly even structure of an encountered object and the propagation of these signal components into diverse directions. 

Therefore, the composition of incoming signal components at a receiver is conditioned on the movement and position of objects in its proximity~\cite{Pervasive_Sigg_2013}.
A change in position of objects or static activities like standing will generally affect the mean amplitude, while movement, such as walking, induces a characteristic pattern on the signal over time~\cite{Pervasive_Shi_2012b}.
Figure~\ref{rawdata} illustrates the received signal's strength and some extracted feature sequences from an ambient FM-radio station over 1~minute for 3 different situations of a single subject in a corridor~(the setting is illustrated in~Figure~\ref{corridor}). 
\begin{figure}
	\centering
	\includegraphics[width=\columnwidth]{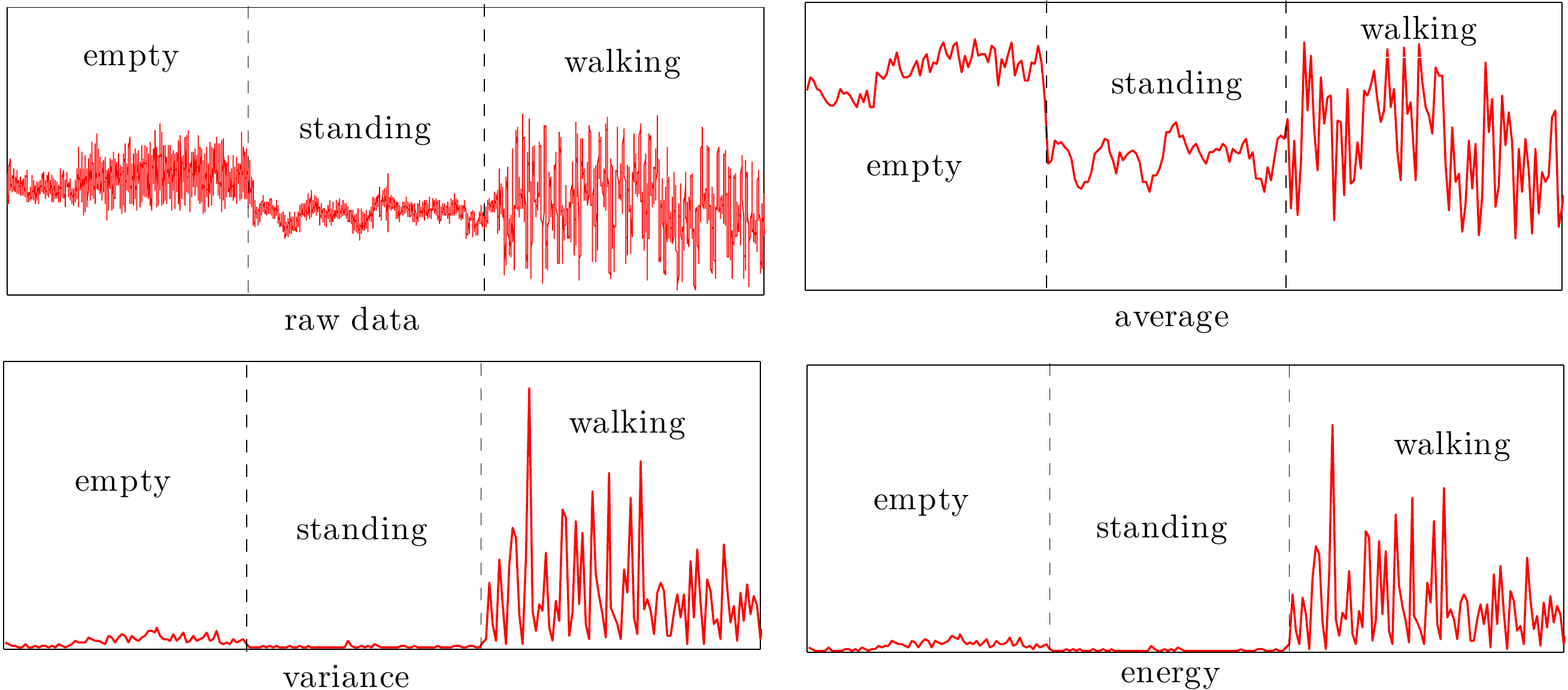}
	\caption{
Evolution of signal strength for empty corridor, standing and walking at the area B, performed by a single subject {\scriptsize (Published by the IEEE CS n 1536-1268/14/\$31.00 \copyright 2014 IEEE}}
	\label{rawdata}
\end{figure}
Figure~\ref{rawdata} indicates a correlation between the characteristics of an RF signal and the activities conducted. 

\subsubsection{Consideration of various RF-technologies}\label{sectionArgueFM}
In the literature, RF-sensing is applied on various signal frequencies and technologies such as WiFi, GSM or FM-radio~\cite{RFSensing_Zhang_2009,Pervasive_Youssef_2007}.
We believe that FM-radio is best suited for attention monitoring for the following reasons.

Since FM-radio features a low operating frequency, a simple modulation mechanism and a wide area of coverage, it is possible to design more robust and discriminative signatures for RF-fingerprinting than for WiFi and GSM~\cite{Percom_Popleteev_2012}. 

FM-radio signals experience, when compared with WiFi, 3G or 4G signals, lower variation in signal strength over time~\cite{Percom_Popleteev_2012}. 
Consequently, for attention monitoring, FM-radio signals induce a lower process noise than signals from WiFi, 3G or 4G systems.
Also, FM-radio is, compared to the other named systems which operate at higher frequencies, less susceptible to weather conditions, such as rain and fog~\cite{Percom_Popleteev_2012}.

Additionally, in order to increase spectrum efficiency, spread spectrum techniques such as frequency hopping or code divisioning are employed in WiFi, 3G and 4G access points.
CDMA interleaves the transmissions to multiple devices including additional potential noise and for hopping schemes, a passive recognition system would need to follow this RF-signal activity variations in order to extract the signal fluctuations on-top which are caused by activities in proximity. 
This would be a more difficult task. 
Therefore, the use of FM-radio signals which are not conditioned on the use of such data transfer schemes appears to be a more appropriate choice.

Furthermore, FM-radio stations are widely implemented and continuously broadcast signals with higher coverage than WiFi, 3G or 4G systems. 
Finally, an FM-radio receiver is less costly than receivers for the other mentioned systems.

For these reasons, we believe that FM-radio is best suited for the utilisation in a passive attention monitoring system.

\subsubsection{Features for FM-based attention monitoring}~\label{fmbasedactivitiesrecognition}
We extract features for attention monitoring from FM-signals continuously broadcast by an FM-radio station (cf. figure~\ref{corridor}).
\begin{figure*}
    \center
	  \includegraphics[width=\textwidth]{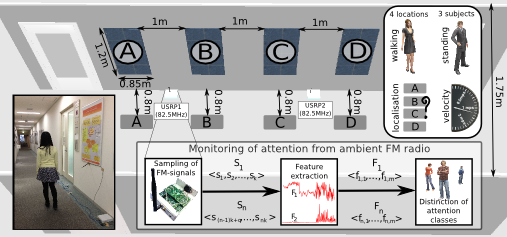}
     \caption{Sketch of the evaluation setting. The attention-monitoring system extracted features combining the data acquired by both USRP devices. {\scriptsize (Published by the IEEE CS n 1536-1268/14/\$31.00 \copyright 2014 IEEE}}
     \label{corridor}
\end{figure*}
The features are obtained from a series of continuous measurements $s_1,s_2,\dots,s_t$ which are samples of the amplitude of ambient FM signals and grouped in windows $S_1,\dots, S_n$ of $k$ consecutive samples each $\langle s_{(i-1)k+1}$,$s_{(i-1)k+2},\dots,s_{ik}\rangle$. 
From these $\mbox{S}_i$, sets of features $\mbox{F}_i=\langle f_{i,1}$,$f_{i,2},\dots,f_{i,m} \rangle$ are extracted and used for the monitoring of attention. 
The features we utilized are the mean ($\mbox{Avg}_i$), the variance ($\mbox{Var}_i$) and the energy ($\mbox{E}_i$) of $\mbox{S}_i$ as detailed in figure~\ref{figureFeaturesAttentionMonitoring}. 
\begin{figure}
      \centering
           \includegraphics[width=\textwidth]{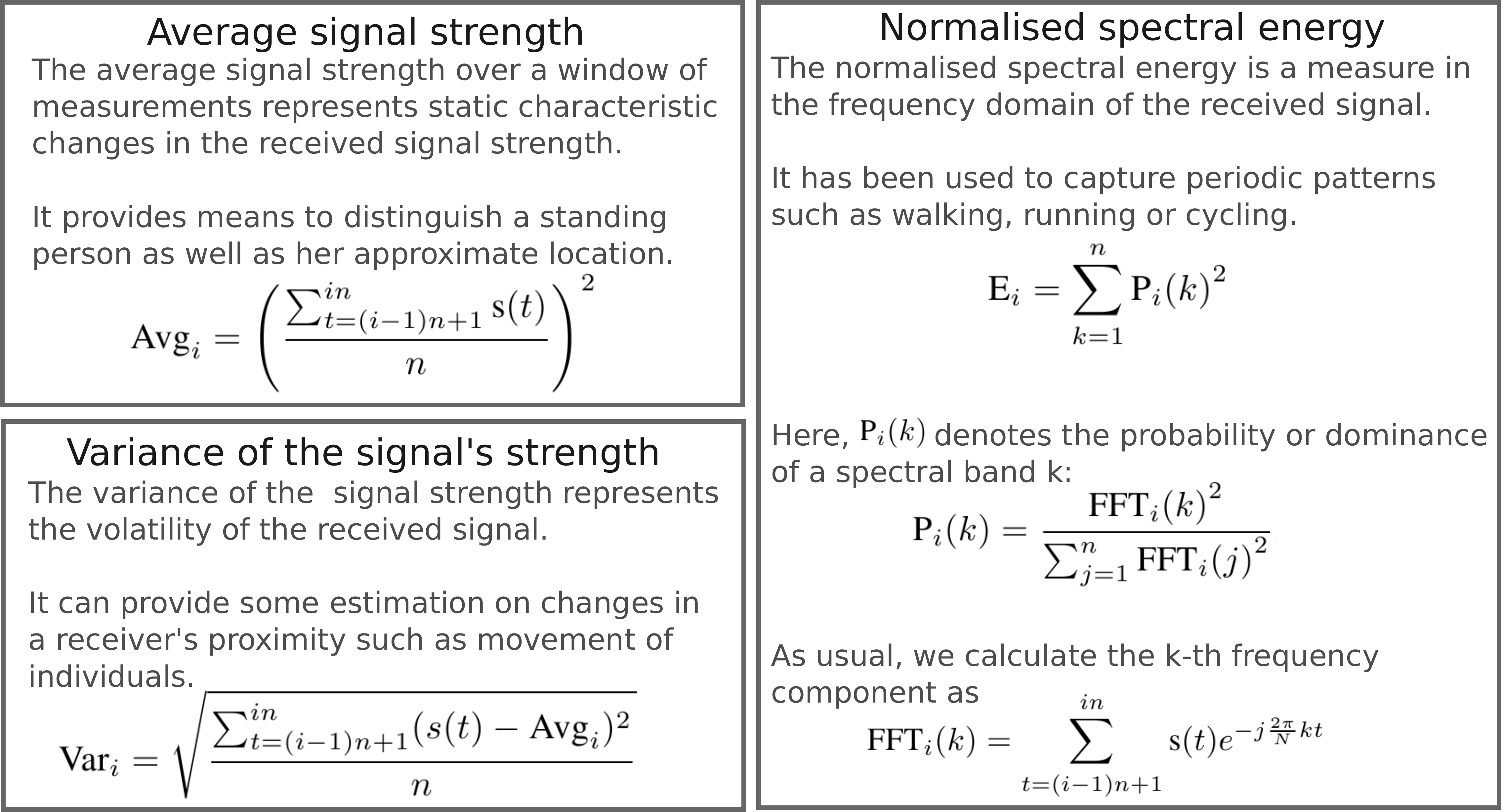}
  \caption{Features utilised for the monitoring of attention via received FM-signals {\scriptsize (Published by the IEEE CS n 1536-1268/14/\$31.00 \copyright 2014 IEEE}}
       \label{figureFeaturesAttentionMonitoring}
  \end{figure}
These features have been derived among a greater set of features as well suited to achieve high accuracy for activity recognition conditioned on passive RF-based recognition systems in~\cite{Pervasive_Sigg_2013}.
After the extraction of features, we randomly divided the collection~F of all feature sets~$\mbox{F}_i$ into a training set~Tr and a classification set~Cl that met the conditions $\mbox{Tr}\cup \mbox{Cl}=\mbox{F}$ and $\mbox{Tr}\cap \mbox{Cl}=\emptyset$. 
The set~Tr is used to train the classifiers. 
After the training, classifiers will process~Cl.

For a set of $k$ activities $\mathcal{A}=\{a_1,\dots,a_k\}$ let $\mathcal{I}(a_i)$, with $\mathcal{I}(a_1)\cup \mathcal{I}(a_2)\dots\cup\mathcal{I}(a_k)=\mbox{Cl}$, be the total number of instances for activity $a_i$ and $\mathcal{I}_{\mbox{\footnotesize cor}}(a_i)$ be the number of correctly classified instances for this activity in which the classification matches the ground truth.
We define the accuracy by which an activity $a_i$ can be detected as 
\begin{equation}\nonumber
     \mathcal{ACC}(a_i)=\frac{\mathcal{I}_{\mbox{\footnotesize cor}}(a_i)}{\mathcal{I}(a_i)}\nonumber
\end{equation}

\subsection{Evaluation}\label{sectionEvaluation}
In this section, we discuss case studies to demonstrate the viability of monitoring attention of people passing by several poster frames towards these frames. 
In all cases, we consider a corridor with posters attached along one side (cf. figure~\ref{corridor}).

Four posters of $0.85\mbox{ m}\times1.2\mbox{ m}$ which are separated by $1\mbox{ m}$ are attached alongside one wall of the corridor.
We place the USRP devices between the two leftmost and rightmost poster frames on the floor.
These N210\footnote{https://www.ettus.com/content/files/2987\_Ettus\_N200-210\_DS\_FINAL\_1.27.12.pdf} USRP devices are equipped with WBX\footnote{https://www.ettus.com/product/details/WBX} daughter boards and VERT900\footnote{https://www.ettus.com/product/details/VERT900} antennas with 3~dBi antenna gain.
Both devices continuously recorded the signal strength with a sample rate of $64\mbox{ Hz}$, emitted by an ambient FM-radio station at 82.5~MHz while the attention of subjects towards the poster frames is monitored. 
We distinguish between four locations, $0.8\mbox{ m}$ in front of the posters (labelled A, B, C, and D) and the rest of the corridor.
During the case studies, the subjects were walking along the corridor and through the marked areas or standing in front of one of the posters at the marked locations.
As a baseline, the received signal from the empty corridor was recorded.
For each action and all three subjects about two minutes of sample data each have been collected.

For all features detailed in figure~\ref{figureFeaturesAttentionMonitoring}, we utilise a window of 128~signal measurements, spanning a total of 2~seconds. 
Features are extracted from the data sets collected by USRP1 and USRP2 (cf. figure~\ref{corridor}) and are merged for the distinction of attention classes. 
For this, we utilise a decision tree (DT) and a k-nearest-neighbour (k-NN) classifier from the Orange data mining Toolkit\footnote{http://orange.biolab.si/}. 
The k-NN classifier utilised 5~neighbours and weights their distance by the Euclidean distance. 
The decision tree utilises at minimum 10 instances in its leaves for pre-pruning and a recursive merge of leaves of the same major class with an m-estimate of~2 for post pruning. 
We apply a 10-fold cross validation. 

\subsubsection{State of the corridor}\label{sectionEvaluationA}
In our scenario, a corridor is equipped with electronic poster frames which shall detect the attention of passers-by and act accordingly.
The most basic case to distinguish for the frames is the state of the corridor.
In particular, we consider whether the corridor is empty or occupied by a person and, when it is occupied, whether this person is walking or standing.
In the case of electronic poster frames, the devices might change into an energy saving mode when the corridor is empty or also display more or less complex information conditioned on whether the person in the corridor is walking or standing.

Table~\ref{3activities} depicts the classification accuracy for these classes.
\begin{table}
\setlength{\tabcolsep}{2pt}
\caption{Mean accuracy for the distinction of the corridor states 'empty', (person)'standing' and (person)'walking' {\scriptsize (Published by the IEEE CS n 1536-1268/14/\$31.00 \copyright 2014 IEEE}}
  \subfloat[
Classification accuracy achieved by a k-NN classifier]{
     \begin{tabular}{r l | ccc}
 &&\multicolumn{3}{c}{Classification}\\
&&empty& standing & walking\\\hline
 \multirow{3}{6pt}{\begin{sideways}Truth\end{sideways}}
  &empty&\textbf{.906}&.034&.06\\
  &standing&.136&\textbf{.765}&.099\\
  &walking&.021&.100&\textbf{.879}\\
 \end{tabular}
}\hfill
 \subfloat[Classification accuracy achieved by a DT classifier]{
       \begin{tabular}{r l | ccc}
 &&\multicolumn{3}{c}{Classification}\\
&&empty& standing & walking\\\hline
 \multirow{3}{6pt}{\begin{sideways}Truth\end{sideways}}
  &empty&\textbf{.877}&.064&.059\\
  &standing&.041&\textbf{.852}&.107\\
  &walking&&.071&\textbf{.929}\\
 \end{tabular}
   }
  \label{3activities}
\end{table} 
For all classes, the mean classification accuracy over the sample windows of 2 seconds is near or above $0.8$. 
In a second stage, we can now obtain information related to the attention of passers-by.

\subsubsection{Focused attention towards specific frames}\label{sectionEvaluationB}
While walking by poster frames, brief snippets of the content can be grasped by individuals. 
However, an intense engagement with the more complex content of a poster requires a person to slow down her walking speed~\cite{AttentionMonitoring_Ferscha_2012} and possibly come to a stand in front of the poster.
We demonstrate the distinction in front of which poster a person is standing in the scenario depicted above. 
All parameters of the recognition system remain identical to section~\ref{sectionEvaluationA}.

All subjects have been standing and observing a poster at one of the locations labelled A, B, C or D in figure~\ref{corridor}.

The most characteristic feature to distinguish these cases is the mean of the signal strength.
The average classification accuracy after 10-fold cross validation is depicted in table~\ref{localisationofstand}
\begin{table}
\setlength{\tabcolsep}{3pt}
\caption{Mean accuracy for the distinction in front of which poster a person is standing {\scriptsize (Published by the IEEE CS n 1536-1268/14/\$31.00 \copyright 2014 IEEE}}
  \subfloat[
Classification accuracy achieved by a k-NN classifier]{
     \begin{tabular}{r l | cccc}
 &&\multicolumn{4}{c}{Classification (standing at)}\\
&& Loc.A&Loc.B&Loc.C&Loc.D \\\hline
 \multirow{4}{6pt}{\begin{sideways}Truth\end{sideways}}
  &Loc.A&\textbf{.876}&.011&.022&.090\\
  &Loc.B&&\textbf{.788}&.203&.008\\
  &Loc.C&&.063&\textbf{.929}&.009\\
  &Loc.D&.009&.027&&\textbf{.964}\\  
 \end{tabular}
}\hfill
 \subfloat[Classification accuracy achieved by a DT classifier]{
  \begin{tabular}{r l | cccc}
 &&\multicolumn{4}{c}{Classification (standing at)}\\
&&Loc.A&Loc.B&Loc.C&Loc.D \\\hline
 \multirow{4}{6pt}{\begin{sideways}Truth\end{sideways}}
  &Loc.A&\textbf{.933}&&.011&.056\\
  &Loc.B&&\textbf{.735}&.257&.009\\
  &Loc.C&&.080&\textbf{.911}&.009\\
  &Loc.D&.054&&.027&\textbf{.919}\\ 
 \end{tabular}
}
  \label{localisationofstand}
\end{table}    
We observe that the classification accuracy is in most cases above $0.9$, in all cases it is near or above $0.8$.

\subsubsection{Tracking individuals in motion}\label{sectionEvaluationC}
While people are passing by poster frames in a corridor, a specific poster frame might have the intention to actively attract the attention of a passer-by.
This attempt is most successful when the person is in the proximity of the poster, facing towards it.
In order to optimally schedule such action, the location of the walking person has to be available at the system.
We show that the location of a single person walking along a corridor can be traced by analysing fluctuation of an incoming FM-radio signal.
Similar to the case study detailed in section~\ref{sectionEvaluationB} we detect in front of which poster a person walking in the corridor is located.

Table~\ref{localisationofwalk} depicts our results.
\begin{table}
\setlength{\tabcolsep}{3pt}
\caption{Mean accuracy for the distinction of walking at location~A, B, C or D in the environment depicted in figure~\ref{corridor} {\scriptsize (Published by the IEEE CS n 1536-1268/14/\$31.00 \copyright 2014 IEEE}}
  \subfloat[
Classification accuracy achieved by a k-NN classifier]{
     \begin{tabular}{r l | cccc}
 &&\multicolumn{4}{c}{Classification (walking at)}\\
&& Loc.A &Loc.B & Loc.C & Loc.D \\\hline
 \multirow{4}{6pt}{\begin{sideways}Truth\end{sideways}}
  &Loc.A&\textbf{.779}&.118&.015&.088\\
  &Loc.B&.107&\textbf{.804}&.071&.018\\
  &Loc.C&&.017&\textbf{.933}&.05\\
  &Loc.D&.139&&&\textbf{.785}\\ 
 \end{tabular}
}\hfill
 \subfloat[Classification accuracy achieved by a DT classifier]{
  \begin{tabular}{rl| cccc}
 &&\multicolumn{4}{c}{Classification (walking at)}\\
&& Loc.A &Loc.B & Loc.C & Loc.D \\\hline
 \multirow{4}{6pt}{\begin{sideways}Truth\end{sideways}}
  &Loc.A&\textbf{.754}&.228&.018&\\
  &Loc.B&.175&\textbf{.772}&.035&.018\\
  &Loc.C&.029&.017&\textbf{.953}&\\
  &Loc.D&.125&&.071&\textbf{.804}\\ 
 \end{tabular}
}
  \label{localisationofwalk}
\end{table}  
We observe that the classification of the location where a person is walking is harder than the classification of the location where a person is standing. 
However, the classification accuracy reached is still near or above $0.8$.

\subsubsection{Changes in the walking speed}\label{sectionEvaluationD}
As detailed in~\cite{AttentionMonitoring_Ferscha_2012}, an important indicator of the attention state of a person are changes in the walking speed.
When a person is interested in a specific content of a poster, she might likely slow down to better perceive the content.

We obtain the walking speed of a passer-by from the fluctuation in ambient FM-radio signals.
We collected for all three subjects and for three different velocities (0.5~m/s, 1~m/s, 2~m/s) samples of a duration of 2~minutes each.
Again, k-NN and DT classifiers are utilised for training and classification. 

Table~\ref{tableWalkingSpeedOne} illustrates our results.
\begin{table}
\setlength{\tabcolsep}{4pt}
\caption{Confusion matrices for the discrimination between walking speeds (0.5~m/s, 1~m/s, 2~m/s) achieved by k-NN and Decision Tree classifiers {\scriptsize (Published by the IEEE CS n 1536-1268/14/\$31.00 \copyright 2014 IEEE}}
  \subfloat[
Classification accuracy achieved by a k-NN classifier]{
     \begin{tabular}{r l | ccc}
 &&\multicolumn{3}{c}{Classification}\\
&&0.5m/s&1m/s&2m/s\\\hline
 \multirow{3}{6pt}{\begin{sideways}Truth\end{sideways}}
  &0.5m/s&\textbf{.641}&.219&.141\\
  &1m/s&.109&\textbf{.453}&.438\\
  &2m/s&.134&.284&\textbf{.582}\\
 \end{tabular}
}\hfill
 \subfloat[Classification accuracy achieved by a DT classifier]{
       \begin{tabular}{r l | ccc}
 &&\multicolumn{3}{c}{Classification}\\
&&0.5m/s&1m/s&2m/s\\\hline
 \multirow{3}{6pt}{\begin{sideways}Truth\end{sideways}}
  &0.5m/s&\textbf{.844}&.094&.063\\
  &1m/s&.094&\textbf{.578}&.328\\
  &2m/s&.164&.343&\textbf{.493}\\
 \end{tabular}
   }
\label{tableWalkingSpeedOne}
\end{table}
We observe that, although there is an indication towards the correct velocity in all cases, the accuracy greatly drops compared to the previous considerations.
The confusion of these velocity levels especially for higher walking speeds is owing to the reduced duration an individual is located in front of a single poster during her walk.

We can, however, achieve a higher recognition accuracy without increasing the distance between posters by abstracting from the 1~m/s walking speed (cf. table~\ref{2speed}), distinguishing only between a slow walk and a running person.
\begin{table}
\setlength{\tabcolsep}{4pt}
\caption{Confusion matrices for the discrimination between walking speeds (0.5~m/s, 2~m/s) achieved by k-NN and Decision Tree classifiers {\scriptsize (Published by the IEEE CS n 1536-1268/14/\$31.00 \copyright 2014 IEEE}}
  \subfloat[
Classification accuracy achieved by a k-NN classifier]{
     \begin{tabular}{r l | l l}
 &&\multicolumn{2}{c}{Classification}\\
&&0.5m/s&2m/s\\\hline
 \multirow{2}{6pt}{\begin{sideways}Truth\end{sideways}}
  &0.5m/s&\textbf{.896}&.104\\
  &2m/s&.219&\textbf{.782}\\
 \end{tabular}
}\hfill
 \subfloat[Classification accuracy achieved by a DT classifier]{
        \begin{tabular}{r l | l l}
 &&\multicolumn{2}{c}{Classification}\\
&&0.5m/s&2m/s\\\hline
 \multirow{2}{6pt}{\begin{sideways}Truth\end{sideways}}
  &0.5m/s&\textbf{.925}&.075\\
  &2m/s&.025&\textbf{.750}\\
 \end{tabular}
   }
  \label{2speed}
\end{table} 
Although we are then not able to distinguish the medium walking speed, note that the attraction of attention of a person in a hurry is not the intention of the considered system.
Rather, we are focusing towards individuals in a relaxed, open state of mind to receive external stimuli and information. 

Since the change in walking speed at a particular location might correspond to the attention level of passer-by, the information on the walking speed, monitored over time, can be utilised to grasp her attention level.

\subsubsection{Altering the count of receive devices}\label{sectionTwoUSRP}
In the above considerations, we have utilised two USRP devices since the experimental setting spans over five meters and the classification accuracy deteriorates with increasing distance to the receive antenna~\cite{Pervasive_Shi_2012b}.
However, for economic reasons, a simple installation might be designed in favour of only one receive device at the cost of a slightly reduced recognition accuracy for greater distances. 
In order to evaluate this impact for the monitoring of attention of passers-by within a corridor, we also consider the classification accuracy when the data from only one of the receive devices is utilised.
The classification system and location of receive devices was not changed.
Figure~\ref{accuracycomparison} depicts our results.
\begin{figure}
      \centering
           \includegraphics[width=9cm]{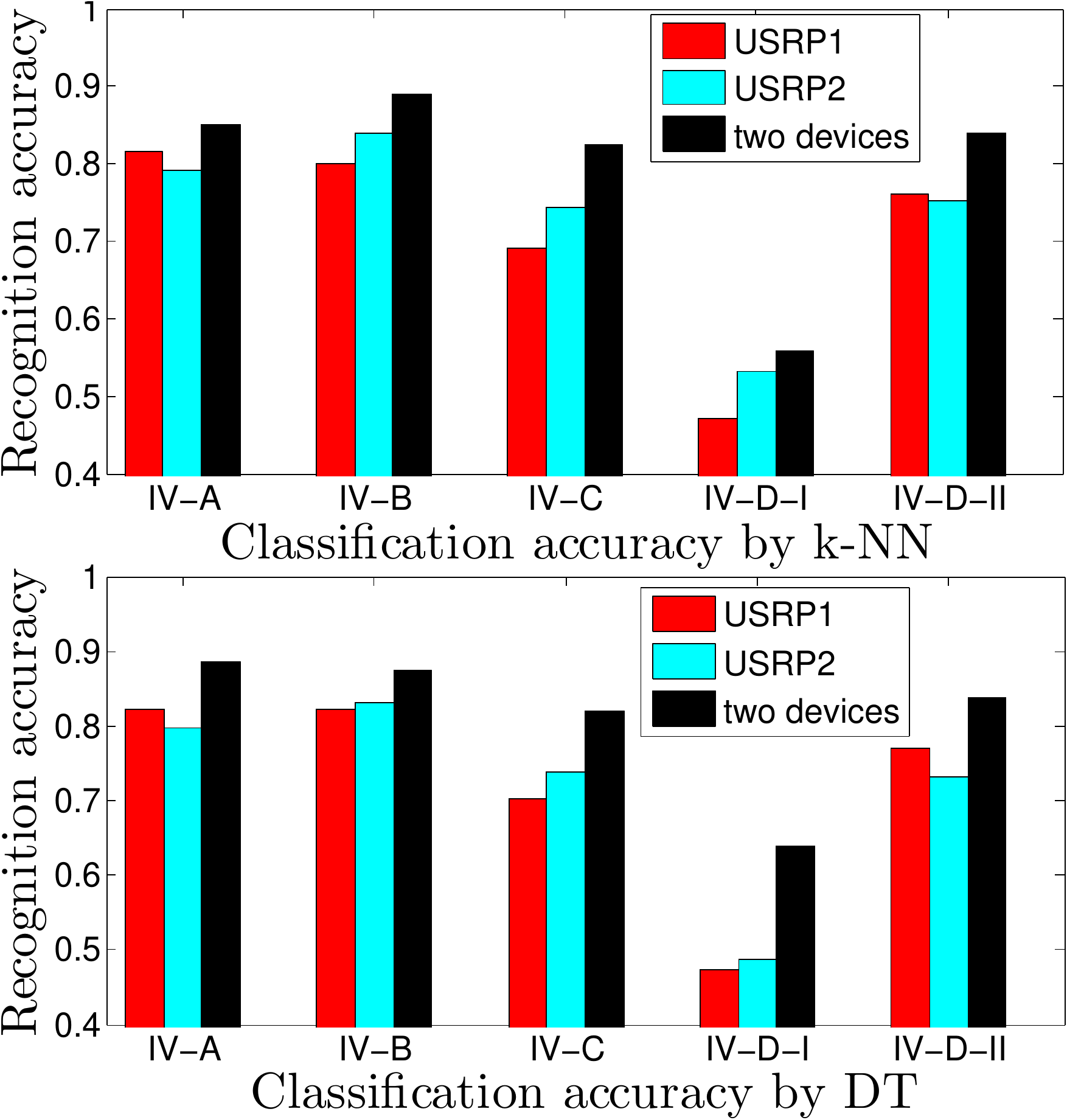}
  \caption{Comparison of the classification accuracy for the 4 cases described 
in section~\ref{sectionEvaluationA} when only one receive device is utlised {\scriptsize (Published by the IEEE CS n 1536-1268/14/\$31.00 \copyright 2014 IEEE}}
       \label{accuracycomparison}
  \end{figure}
We observe that the classification accuracy benefits from the addition of the second device in all cases.
With only one device, the overall classification accuracy drops by about $0.05$ to $0.1$ since the classification accuracy for individuals in greater distance deteriorates.

\subsection{Discussion}\label{sectionConclusionandFutureWork}
The monitoring of the attention state of passers-by towards interactive poster frames can provide additional information 
to the display system when to display which information. 

We demonstrated the distinction of attention classes from features extracted from ambient FM-radio signals.
In particular, we utilised the mean, variance and energy of a signal received at 82.5 MHz in order to distinguish occupancy states in a corridor 
as well as locations at which persons are walking or standing and finally walking speeds.
The attention level of persons is, among other factors, related to walking speeds or changes in velocity or acceleration.
Therefore, we can use the information extracted from the fluctuation in the received FM-signal as an indicator towards various attention states.
This information might control the information provided by a poster frame, conditioned on the attention of passers-by.

Due to the low cost of FM-receiver hardware and the high coverage of FM-radio, the described attention-monitoring system has the potential to be widely deployed in systems that benefit from knowing the attention levels of people in proximity.
Future directions cover the simultaneous detection of attention levels of multiple persons as well as the implementation of the system using off-the-shelf receiver hardware.
\vfill
\pagebreak

\section[The Telepathic Phone: Frictionless Activity Recognition from WiFi-RSSI]{The Telepathic Phone: Frictionless Activity Recognition from WiFi-RSSI \footnote{Originally published as 'Stephan Sigg, Ulf Blanke and Gerhard Troester: The Telepathic Phone: Frictionless Activity Recognition from WiFi-RSSI, IEEE International Conference on Pervasive Computing and Communications (PerCom), Budapest, Hungary, March 24-28, 2014 (DOI: http://dx.doi.org/10.1109/PerCom.2014.6813955)' (978-1-4799-3445-4/14/\$31.00 \copyright 2014 IEEE)}}\label{sectionOriginalRF03}

We investigate the use of WiFi Received Signal Strength Information (RSSI) at a mobile phone for the recognition of situations, activities and gestures.
In particular, we propose a device-free and passive activity recognition system that does not require any device carried by the user and uses ambient signals.
We discuss challenges and lessons learned for the design of such a system on a mobile phone and propose appropriate features to extract activity characteristics from RSSI. 
We demonstrate the feasibility of recognising activities, gestures and environmental situations from RSSI obtained by a mobile phone. 
The case studies were conducted over a period of about two months in which about 12 hours of continuous RSSI data was sampled, in two countries and with 11 participants in total.
Results demonstrate the potential to utilise RSSI for the extension of the environmental perception of a mobile device as well as for the interaction with touch-free gestures.
The system achieves an accuracy of 0.51 while distinguishing as many as 11 gestures and can reach 0.72 on average for four more disparate ones.

\subsection{Introduction}\label{sectionIntroductionRF03}
Mobile phones are a popular sensing platform for the multitude of sensors they incorporate and for their status as personal device kept close to or on the body~\cite{Pervasive_Yang_2012,Pervasive_Lukovicz_2012-2}.
However, these mobile sensing platforms focus on inertial motion to recognize physical activity. 
When a device is no longer worn on the body, its sensing capabilities are greatly reduced.
Indeed, although people are in the same room with their mobile device almost 90\% of the time, their device is within arms reach less than 55\% of a day~\cite{Pervasive_Dey_2011,Pervasive_Patel_2006}.
Therefore, the mobile phone can hardly serve as a continuous sensing platform with sensors such as accelerometers or gyroscopes.

To still obtain information about situations or activities, we need to exploit sensors that react on ambient stimuli.
Possible choices are video~\cite{ActivityRecognition_Aggarwal_2011}, or audio for the classification of device-locations based on audio signatures~\cite{ContextAwareness_Kunze_2007} as well as localisation via audio-based fingerprinting~\cite{Cryptography_Schuerman_2011}.
However, video is restricted by the sensor's field of vision while audio is limited to general locations or situations~\cite{Pervasive_Chaquet_2013}.

\begin{figure}
     \begin{center}
     \includegraphics[width=\columnwidth]{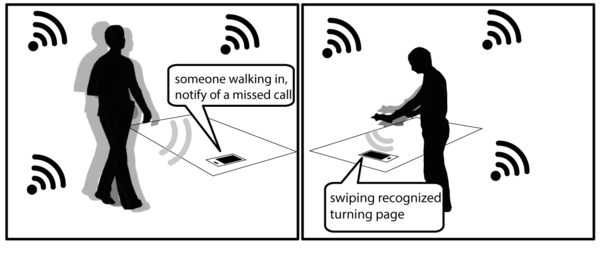}
     \caption{Activity obtained from RSSI-signatures. Two example use-cases: user walking in with the smartphone implicitly reacting (left) and a no-touch explicit interaction (right). {\scriptsize (978-1-4799-3445-4/14/\$31.00 \copyright 2014 IEEE)}}
     \label{teaser}
     \end{center}
\end{figure}
We propose the use of another environmental sensor: the wireless interface to the Radio Frequency (RF) channel.
By monitoring the fluctuation in the received signal strength indicator (RSSI) that is calculated at a receiver for each incoming packet, we attempt to classify the situation (e.g. crowd size), activities or gestures performed in proximity of a mobile phone
(See figure~\ref{teaser}).
This approach allows operation even when the device is not carried by the user but near to her -- a scenario where most activity recognition systems fail.

We can utilise RSSI also in dark or quiet environments when audio or video might not provide sufficient information.
In urban spaces, WiFi connectivity can be presumed (cf. section~\ref{sectionWiFiTraffic}). 
In addition, RF might be perceived as less privacy intrusive when compared to audio or video.

While there is some work on the device-free recognition of activities from RF-channel fluctuation~\cite{Pervasive_Adib_2013,RFsensing_Pu_2013,Pervasive_Sigg_2012}, these systems require sophisticated Software Defined Radio (SDR) devices in order to obtain frequency domain features.
In contrast, we attempt to utilise signal strength fluctuation on off-the shelf mobile phone hardware and from ambient WiFi traffic.
On such devices, already the capturing of RSSI data in sufficient frequency is challenging. 
In addition, the data captured is less accurate and bursty.
We discuss necessary pre-processing as well as the design of features suitable for highly bursty and low-resolution environmental RSSI data together with the final recognition step.
In case studies we demonstrate the potential and limitations of using RSSI for recognition.
The contributions of this work are:
\begin{enumerate}
\item System design and definition of feature space for RSSI-based `frictionless` recognition
\item Analysis of RSSI-influencing factors in a controlled setting (e.g. direction, distance)
\item Feasibility study of situation, activity, and gesture recognition with off-the-shelf mobile phones.
\end{enumerate}
Our results indicate that RF-based sensing of environmental {\em situations}, {\em crowd} and {\em individual activity} provides additional information for activity or context classification tools.

\subsection{Related Work}\label{sectionRelatedWorkRF03}
Device-free RF-based recognition was introduced by Youssef and others~\cite{Pervasive_Youssef_2007} as the localisation of an entity not equipped with any transmitter or receiver.
In recent years, some groups work in this direction using hardware that ranges from SDR devices~\cite{DeviceFreeRecognition_Popleteev_2013}, laptop-class computers~\cite{Pervasive_Seifeldin_2013} over sensor nodes~\cite{Pervasive_Zhang_2012} or RFID tags~\cite{DeviceFreeRecognition_Wagner_2013} and achieve high accuracies of about 1~meter.
This work is also related to a considerable body of practical and theoretical results on passive radar (cf.~\cite{DeviceFreeRecognition_Tan_2005,DeviceFreeRecognition_Colone_2012} and references therein) where vehicles and individuals are detected and tracked from signals such as HF radio, UHF television broadcasts or DAB, DVB and GSM.

Recognition utilising signals on the wireless channel has been generalised in~\cite{Pervasive_Sigg_2012} to activities and we can further imagine also situations~\cite{Pervasive_Scholz_2011}, gestures~\cite{RFsensing_Pu_2013} or attention~\cite{Pervasive_Shi_2014} to be identified by RF-based device-free implementations.
These systems can be grouped into active and passive approaches conditioned on the presence of an active transmitter.
Most previous work in this direction uses SDR devices. 

Kassem et al. sense traffic situations by tracking frequency and speed of passing cars that intercept the direct line of sight between a pair of nodes~\cite{RFSensing_Kassem_2012}. 
The authors of~\cite{Pervasive_Sigg_2012} classify simple activities in an SDR-based active device-free system by extracting and interpreting features from a continuous signal between two nodes. 
Their approach explores also the multipath effects induced by persons that are not intercepting the direct path between nodes.
It was later demonstrated that also simultaneously conducted activities from multiple persons can be distinguished by leveraging purely signal-strength based features~\cite{DeviceFreeRecognition_Sigg_2013}.
Furthermore, it was shown by Pu and others that simultaneous detection of gestures from multiple individuals is possible utilising multi-antenna nodes and micro Doppler fluctuations~\cite{RFsensing_Pu_2013,RFsensing_Kim_2009}.
In a related system, Adib and Katabi employ MIMO interference nulling and combine samples taken over time to achieve the same result while compensating for the missing spatial diversity in a single-antenna system~\cite{Pervasive_Adib_2013}.

While the above are active approaches that require a dedicated transmitter, Ding and others have presented a passive system leveraging RF noise from engines of vehicles~\cite{RFSensing_Ding_2011}. 
In addition, Shi et al. recognised activities and locations from fluctuation in the signal strength of broadcast FM radio~\cite{Pervasive_Shi_2012b}.

Also, active systems utilising non-SDR nodes have been studied.
Most notably, Patwari and others estimated the breathing frequency of an individual surrounded by nodes from the RSSI of exchanged packets~\cite{Pervasive_Patwari_2011b}.
Following other directions, Xu et al. have counted crowd~\cite{RFsensing_Xu_2013} from RSSI within a field of sensor nodes.
Their unsupervised learning approach is able to predict the count of up to 10 stationary or moving individuals.
Recently, the recognition of general activities from RSSI in a sensor network has been considered~\cite{Pervasive_Scholz_2013}.
In particular, the activities standing, sitting, lying, walking and empty have been distinguished with an accuracy of 0.8-0.9.

For these studies, either a sophisticated SDR device or transmit-receive pairs of nodes were required.
Both cases are hard to establish with end-user equipment in spontaneous use.
We propose a usable RF-based device-free recognition approach on phones by leveraging received RSSI from packets of WiFi access points (APs).
We are not aware of previous work on such RSSI-based passive device-free recognition system.

\subsection{Capturing RSSI on Phones}\label{sectionMPRecognition}
In IEEE 802.11, data is exchanged in packets on 11 partly overlapping frequency channels.
In normal communication, a WiFi receiver discards all packets not addressed to itself.
However, we can force the interface into monitor mode to log all traffic. 
For each packet, the receiver calculates the signal strength from the 8 bit preamble.
Due to the lower data rates, control packets differ in their estimated RSSI significantly. 

While the APIs of contemporary mobile phone operating systems (OSs) provide means to access the RSSI, this information is averaged and refreshed at about 1~Hz only. 
Another access to the RSSI is possible via the interface directly with tools such as \texttt{airodump-ng} or \texttt{tcpdump}.
This requires root permissions to access the interface in monitor mode.\footnote{Monitor mode is obligatory in our case since otherwise the tool is executed in Ethernet emulation which does not provide RSSI information}
WiFi-firmware with sufficient access to relevant parameters is sparse.
More severe even for mobile phones, most handsets implement a similar chipset family (e.g. Broadcom bcm4329, bcm4330(B1/B2), bcm4334, bcm4335) for which the default firmware does not provide access to the desired information (even as root).
The only solution to avoid root access and which abstracts from this chipset family is via an external antenna\footnote{github.com/brycethomas/liber80211/blob/master/README.md}.
However, this considerably extends the dimensions and complexity of the hardware, so that we decided against it.
Instead, we used a modified firmware for the above mentioned chipset family~\cite{Wardriving_Ildis_2013-PerCom} on a Nexus One phone running Cyanogen mod 7.2 and executed tcpdump on the interface in monitor mode to capture RSSI of packets.
In monitor mode, no data can be transmitted and consequently no impact can be taken on the frequency in which packets are received.
We can, however, adjust the channel we listen on and might utilise data from multiple APs transmitting on the same channel.
In summary, while it is practically possible to monitor RSSI, the support of manufacturers for the operating systems to perform this out of the box is limited. 
However, we can track RSSI fluctuation with a modified OS, but without hardware modifications.

Figure~\ref{figureRSSISamples} shows an exemplary snippet of sampled RSSI.
\begin{figure}
     \begin{center}
     \includegraphics[width=.9\columnwidth]{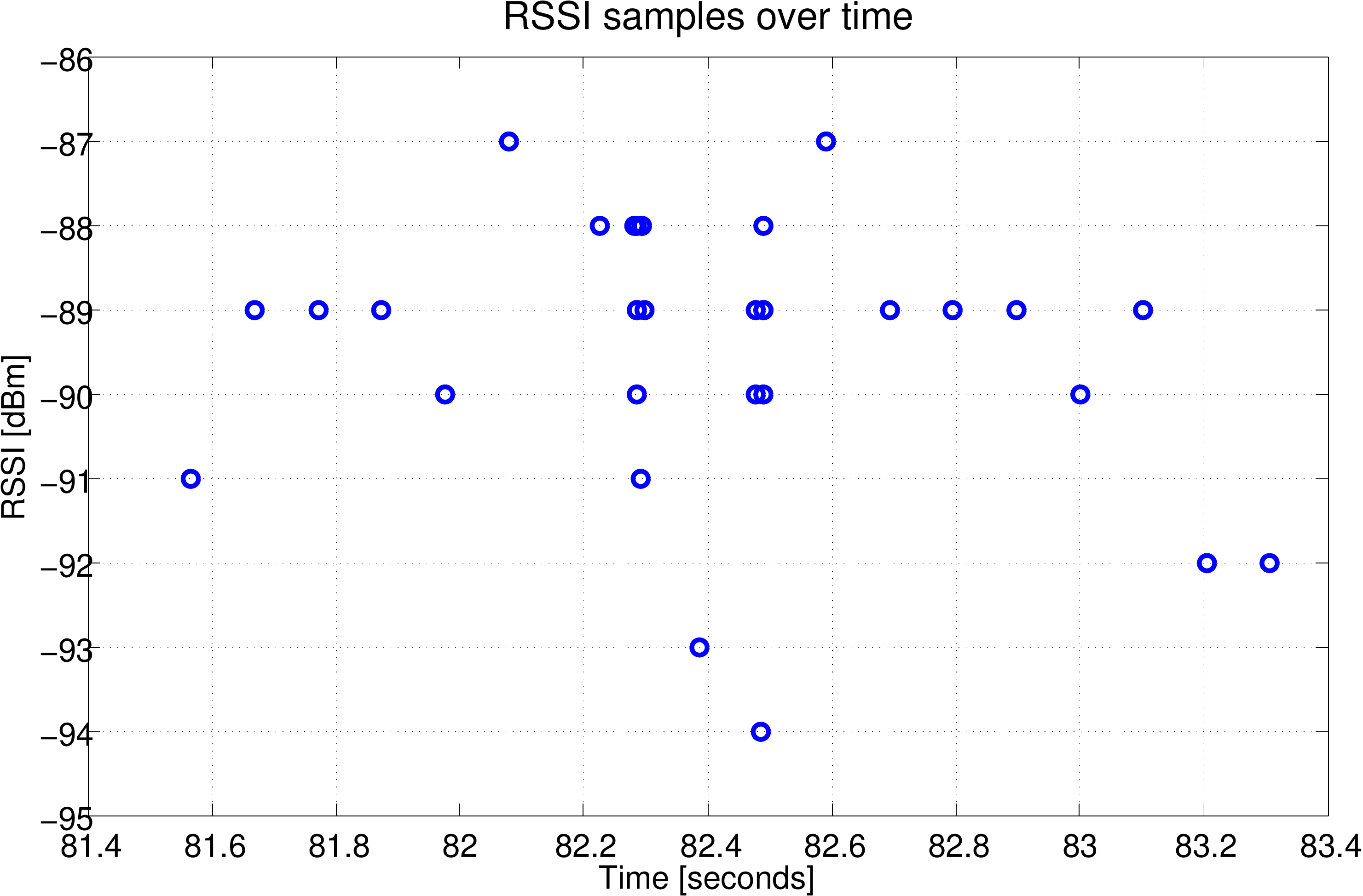}
     \caption{RSSI from packets of a single AP {\scriptsize (978-1-4799-3445-4/14/\$31.00 \copyright 2014 IEEE)}}
     \label{figureRSSISamples}
     \end{center}
\end{figure}
In the experiments conducted, the RSSI usually ranged from -98dBm to -47dBm.
Since the RSSI calculated for control packets differs, we disregarded them for the generation of this data. 
At the time of this recording, the phone was lying on a table within approximately 0.5 meters distance of a person sitting at that table. 
We observe that the data is very bursty.
While there might be only one packet within 0.1 seconds at times, we can also observe five or more packets in the same interval.
Clearly, when compared to SDR-based recognition systems that have direct access to the physical channel, the amount of information available from RSSI is severely reduced. 
Even compared to active RSSI-based systems that contain a transmitter omitting packets at high rate, our passive approach has to deal with more bursty traffic and a lower packet arrival rate.
In addition, the granularity of RSSI is low.
In our case, the 1dB granularity observed in the figure could not be improved for the WiFi interface.

We conclude that it would be hard to apply any curve fitting that could successfully predict the RSSI evolution at a higher sample rate.

\subsection{Features for RSSI-based Recognition}\label{sectionFeatures}
Considering this structure of the data, we used simple features that express general properties such as the overall weight or mass as well as their spread.\footnote{No frequency domain features could be used;  Features as zero crossings or direction changes were not meaningful on the undersampled signal.}
As a tribute to the bursty traffic, the low granularity (cf. figure~\ref{figureRSSISamples}) and a fluctuating packet arrival rate, we simply fixed non-overlapping windows of two seconds and then utilised all RSSI values that would arrive during this period for feature calculation.
The window length was set to $2$s since we aim to design a system that would be practically usable with a good response time.
A higher accuracy can be achieved with increased window size or via majority votes over successively calculated features (cf. section~\ref{sectionWiFiTraffic}).
In total, 18 different features have been considered.
On a data set with the three basic cases 
\begin{enumerate}
 \item A phone lying on a table in an empty room
 \item A phone lying on a table with a person moving
 \item A person holding and handling the phone
\end{enumerate}
we applied a feature selection from the orange data mining toolkit~\footnote{http://orange.biolab.si/}.
From the remaining 9 features, we manually tweaked a combination that achieves good accuracy.
Several combinations of {\em mean, median, variance, maximum} and the {\em difference between minimum and maximum} could achieve best and comparable classification results. 
For the case studies (section~\ref{sectionCaseStudyRF03}), we decided for a combination of {\em mean, variance, maximum} and {\em difference between maximum and minimum}.
For the gesture recognition, also the slope was considered.

\subsection{Case Studies}\label{sectionCaseStudyRF03}
We conducted case studies in indoor environments at ETH Zurich and TU Braunschweig (cf. figure~\ref{figureCaseStudies}).
\begin{figure}
  \subfloat[Office environment at ETH]{\includegraphics[height=5.2cm]{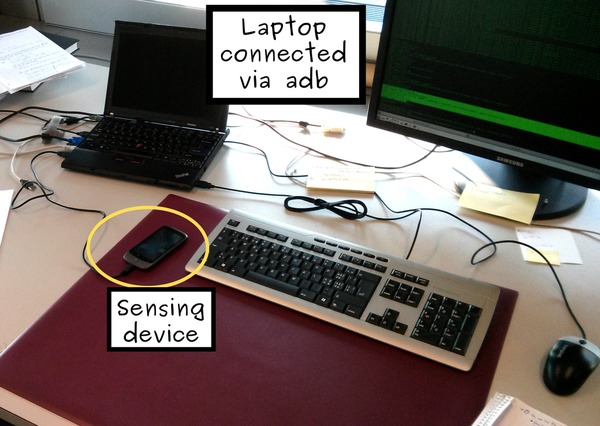}\label{figureCaseStudiesETH}}\hfill
  \subfloat[Lecture room at TU-BS]{\includegraphics[height=5.2cm]{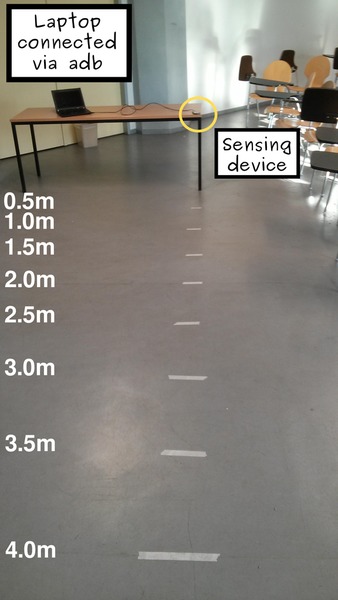}\label{figureCaseStudiesTU}}\hfill
  \subfloat[Scenario for the distinction of walking speed]{\includegraphics[height=5.2cm]{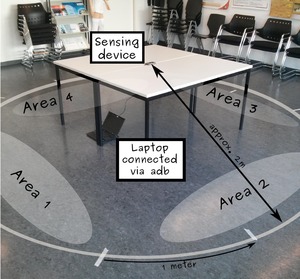}\label{figureCaseStudiesETH2}}
  
  \subfloat[Activities conducted behind a closed door]{\includegraphics[height=4.7cm]{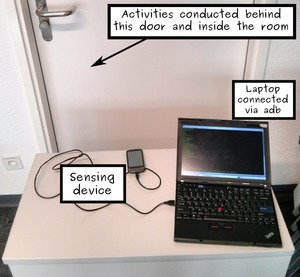}\label{figureCaseStudiesETH3-behind}}\hfill
  \subfloat[Sensing device inside pocket]{\includegraphics[height=4.7cm]{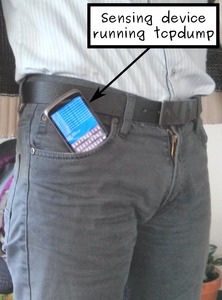}\label{figureCaseStudiesETH4-pocket}}\hfill
  \subfloat[Meeting room at ETH]{\includegraphics[height=4.7cm]{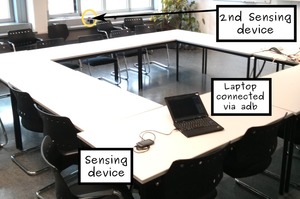}\label{figureCaseStudiesETH5-crowd}}
  \caption{Environments for our case studies. Surrounding furniture and objects were intentionally altered in all cases. {\scriptsize (978-1-4799-3445-4/14/\$31.00 \copyright 2014 IEEE)}}
  \label{figureCaseStudies}
\end{figure}
Occasionally the phone was connected to a computer via \texttt{adb shell} as an alternative to the slow on-screen keyboard which made no difference for the recorded data.
All recordings were conducted multiple times and over several days.
We intentionally altered the environments between recordings (e.g. moving furniture, placing the device slightly different).

Data was processed off-line.
However, we have developed a toolchain for the processing and classification that is sufficiently lightweight to be executed on the phone in realtime.\footnote{The python tools to extract and process RSSI information from pcap files and to classify situations are available at http://www.stephansigg.de/DeviceFree/pcapTools.tar.gz}
The tool groups packets for their source address (since the mean RSSI differs among senders) and disregards control packets (since also their RSSI level differs). 

We now consider general RSSI properties and then investigate limits of RSSI-based recognition.
The studies were conducted over two months in Braunschweig, Germany and Zurich, Switzerland. 
A total of 11 persons (9 male and 2 female; 26 to 37 years) have participated and overall about twelve hours of continuously sampled data has been produced.  

First, we investigate properties of urban WiFi with respect to traffic and sampling rate (section~\ref{sectionWiFiTraffic}). 
Then, we study coarse characteristics with respect to the presence of a user. 
Finally, we provide experiments on fine-grained gesture recognition.

\subsubsection{WiFi Traffic in Urban Spaces}\label{sectionWiFiTraffic}
For the recognition of activities and gestures from RSSI, the rate of incoming packets is essential since this is the rate of fluctuation induced by environmental stimuli.
We sampled packets over some days at various locations in Zurich to estimate a typical rate of packets in urban places.
Figure~\ref{tableIncomingPackets} shows the number of packets per second from the most active AP at various locations on all channels.
Short packets, such as acknowledgements, were removed (cf. section~\ref{sectionFeatures}).
\begin{figure}
\centering
 \includegraphics[width=.6\columnwidth]{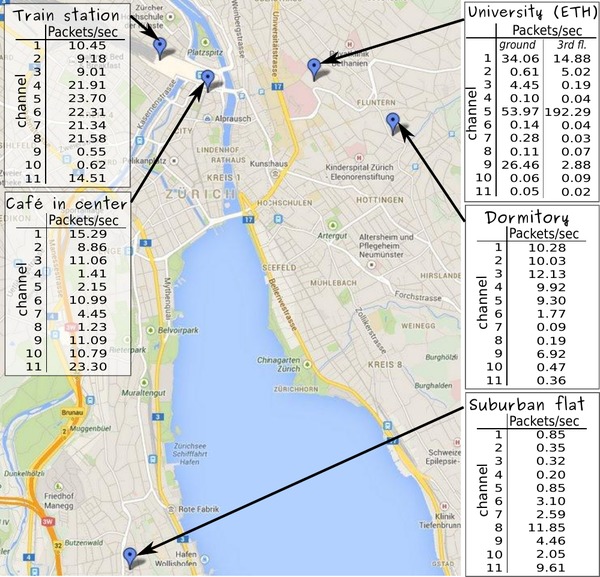}
 \caption{Packets per second from the most active AP at various locations and over all 11 WiFi channels {\scriptsize (978-1-4799-3445-4/14/\$31.00 \copyright 2014 IEEE)}}
 \label{tableIncomingPackets}
\end{figure}

The locations span a University building at two distinct floors, a dormitory, a caf\'e in the city center, the main train station and a flat in a suburb of Zurich.
Only the University locations share APs.
All other locations are well separated over the city.
All locations have characteristic properties.
While the caf\'e has the most equally distributed traffic over all channels, in the dormitory, traffic is clustered in few channels.
University locations feature few, heavily trafficked channels while at the suburban flat only few channels are frequented.
In all cases, we find at least one channel with 10 or more packets per second from a single AP. 
While this most frequented channel might differ spatially, a brief scan easily reveals most suited channels.

Since the receiver has no impact on the packet arrival rate, it relies on traffic from other devices.
We considered the impact of the RSSI samples per second on the classification accuracy. 
In the case study (cf. figure~\ref{figureCaseStudiesETH}), we distinguish an empty office with the mobile phone lying on a table, the same room with a person walking next to the table and a person holding and handling the phone.
Recordings were taken over four days at different times of day. 
Each activity is sampled for five minutes in a row. 
This was repeated on each day twice for all activities. 

Table~\ref{tableResultsSampleRateStatisticsFirst} shows the classification accuracy (CA), information score (IS), Brier score and area under the ROC\footnote{Receiver Operating Characteristic} curve (AUC)~\cite{Statistics_Kononenko_1991,Statistics_Spackman_1989}.
\begin{table}
\centering
\setlength{\tabcolsep}{1pt}
\subfloat[Performance of a k-NN classifier with distinct sample rates]{\begin{tabular}{r|cccc}
&CA&IS&Brier&AUC\\\hline
5 samples/s&.593&.594&.512&.813\\
\cellcolor{gray!20}7 samples/s&\cellcolor{gray!20}.607&\cellcolor{gray!20}.622&\cellcolor{gray!20}.502&\cellcolor{gray!20}.814\\
10 samples/s&.652&.703&.446&.831\\
\cellcolor{gray!20}15 samples/s&\cellcolor{gray!20}.671&\cellcolor{gray!20}.806&\cellcolor{gray!20}.408&\cellcolor{gray!20}.856\\
20 samples/s&.836&1.127&.229&.957
\end{tabular}
\label{tableResultsSampleRateStatisticsFirst}
}\hfill
\subfloat[Confusion matrix for the k-NN classifier with 20 RSSI samples/sec]{\begin{tabular}{cr|cccc}
&&\multicolumn{3}{c}{\begin{scriptsize}Classification\end{scriptsize}}& \\
&&activity&empty&holding&\cellcolor{blue!40}\color{white}{\textbf{recall}}\\\hline
\multirow{3}{5pt}{\begin{sideways}\begin{scriptsize}Gr. truth
\end{scriptsize}\end{sideways}}&activity&\cellcolor{gray!40}{\textbf{.829}}&.014&.157&\cellcolor{blue!40}\color{white}{\textbf{.829}}\\
&empty&.021&\cellcolor{gray!40}{\textbf{.921}}&.057&\cellcolor{blue!40}\color{white}{\textbf{.921}}\\
&holding&.207&.036&\cellcolor{gray!40}{\textbf{.757}}&\cellcolor{blue!40}\color{white}{\textbf{.757}}\\
\multicolumn{2}{r}{ \cellcolor{blue!40}\color{white}{\textbf{precision}}}& \cellcolor{blue!40}\color{white}{\textbf{.784}}&\cellcolor{blue!40}\color{white}{\textbf{.949}} &\cellcolor{blue!40}\color{white}{\textbf{.779}}&\cellcolor{blue!40}
\end{tabular}
\label{tableResultsSampleRateConfusion}
}		
\caption{Impact of the sample rate on the classification {\scriptsize (978-1-4799-3445-4/14/\$31.00 \copyright 2014 IEEE)}}
\label{tableResultsSampleRateStatistics}
\end{table}
The IS measures how well a classifier learned a data set. 
It is higher when the correct class is predicted more often.
Brier score measures the mean squared difference between a predicted probability for an outcome and the actual class.
AUC is the probability that a classifier ranks a random positive instance higher than a random negative one.

For these results, we used a k-NN with $k=20$ (best results reached with $k\in[10..20]$), and a 10-fold cross validation.
While higher sample rates improve accuracy, also 10 to 15 samples per second allow an indication about a class.
The Confusion matrix for 20 samples per second is depicted in table~\ref{tableResultsSampleRateConfusion}. 
Observe that activity and holding suffer from slight confusion. 
In the empty room almost no confusion is seen. 
Then the signal is stable and not influenced by movement.

The classification accuracy is impacted by the sampling window size (cf. figure~\ref{figureMajorityVotes}).
\begin{figure}
\centering
     \includegraphics[width=.7\columnwidth]{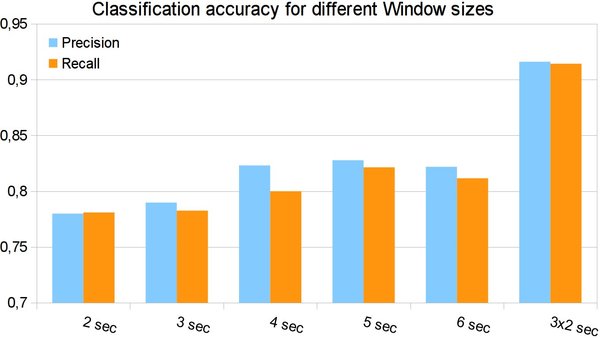}
     \caption{Accuracy for the distinction between three basic cases with varying feature window size. A majority vote over three windows of 2 seconds outperforms greater windows {\scriptsize (978-1-4799-3445-4/14/\$31.00 \copyright 2014 IEEE)}}
     \label{figureMajorityVotes}
\end{figure}
A majority vote over three successive windows of two seconds can reach higher accuracy than a greater window size. 
However, since the system is more responsive with shorter windows, we choose 2s windows.

\subsubsection{Distance to the phone}\label{sectionCaseStudyDistance}
How does the distance to the sensing hardware impact the capability to detect an activity.
The case studies depicted in figure~\ref{figureCaseStudiesTU} were conducted at TU-Braunschweig over two consecutive days with repetitions of experiments on both days.
On the floor, locations were marked in increasing distance of 0.5m up to 4.0m.
At these locations, an individual walked around or move for at least 5min for each distance and day.

We investigated the distinction between an empty environment, a person moving in 4 meters distance and a person moving closer to the mobile phone (cf. table~\ref{tableDistanceStatisticsBest}).
The classification accuracy deteriorates when the locations are closer together (cf. table~\ref{tableDistanceStatistics}). 
\begin{table}
\centering
\setlength{\tabcolsep}{1pt}
\subfloat[Performance using 1 (x) and 2 (O) APs 
]{
\begin{tabular}{c|ccccccccc|cccc}
&\multicolumn{9}{c}{\textit{Distance [meters]}}&&\\
&\begin{sideways}.5\end{sideways}&\begin{sideways}1.0\end{sideways}&\begin{sideways}1.5\end{sideways}&\begin{sideways}2.0\end{sideways}&\begin{sideways}2.5\end{sideways}&\begin{sideways}3.0\end{sideways}&\begin{sideways}3.5\end{sideways}&\begin{sideways}4.0\end{sideways}&\begin{sideways}em\end{sideways}&CA&IS&Brier&AUC\\\hline
\multirow{5}{5pt}{\begin{sideways}\begin{scriptsize}1 AP
\end{scriptsize}\end{sideways}}&x&&&&&&&x&x&.809&1.115&.258&.939\\
&\cellcolor{gray!20}&\cellcolor{gray!20}&\cellcolor{gray!20}&\cellcolor{gray!20}x&\cellcolor{gray!20}&\cellcolor{gray!20}&\cellcolor{gray!20}&\cellcolor{gray!20}x&\cellcolor{gray!20}x&\cellcolor{gray!20}.730&\cellcolor{gray!20}.796&\cellcolor{gray!20}.434&\cellcolor{gray!20}.866\\
&&&&&&&x&x&x&.528&.472&.599&.743\\
&\cellcolor{gray!20}&\cellcolor{gray!20}x&\cellcolor{gray!20}&\cellcolor{gray!20}x&\cellcolor{gray!20}&\cellcolor{gray!20}x&\cellcolor{gray!20}&\cellcolor{gray!20}x&\cellcolor{gray!20}x&\cellcolor{gray!20}.483&\cellcolor{gray!20}.933&\cellcolor{gray!20}.644&\cellcolor{gray!20}.831\\
&x&x&x&x&x&x&x&x&x&.379&1.19&.762&.823\\\hline
\begin{tiny}2AP
\end{tiny}&\cellcolor{gray!20}O&\cellcolor{gray!20}O&\cellcolor{gray!20}O&\cellcolor{gray!20}O&\cellcolor{gray!20}O&\cellcolor{gray!20}O&\cellcolor{gray!20}O&\cellcolor{gray!20}O&\cellcolor{gray!20}O&\cellcolor{gray!20}.427&\cellcolor{gray!20}1.329&\cellcolor{gray!20}.722&\cellcolor{gray!20}.857
\end{tabular}
\label{tableDistanceStatistics}}
\hfill
\subfloat[Classification accuracy with fairly separated locations]{\begin{tabular}{cr|cccc}
&&\multicolumn{3}{c}{\begin{scriptsize}Classification\end{scriptsize}}\\
&&.5m&4.0m&empty&\cellcolor{blue!40}\color{white}{\textbf{recall}}\\\hline
\multirow{3}{5pt}{\begin{sideways}\begin{scriptsize}Gr. truth
\end{scriptsize}\end{sideways}}&.5m&\cellcolor{gray!40}{\textbf{.981}}&.019&&\cellcolor{blue!40}\color{white}{\textbf{.981}}\\
&4.0m&.026&\cellcolor{gray!40}{\textbf{.768}}&.206&\cellcolor{blue!40}\color{white}{\textbf{.768}}\\
&empty&.013&.310&\cellcolor{gray!40}{\textbf{.677}}&\cellcolor{blue!40}\color{white}{\textbf{.677}}\\
\multicolumn{2}{r}{ \cellcolor{blue!40}\color{white}{\textbf{precision}}}& \cellcolor{blue!40}\color{white}{\textbf{.962}}&\cellcolor{blue!40}\color{white}{\textbf{.700}} &\cellcolor{blue!40}\color{white}{\textbf{.766}}&\cellcolor{blue!40}
\end{tabular}\label{tableDistanceStatisticsBest}
}

\subfloat[Confusion matrix over all distances]{
 \includegraphics[height=3cm,width=.6\columnwidth]{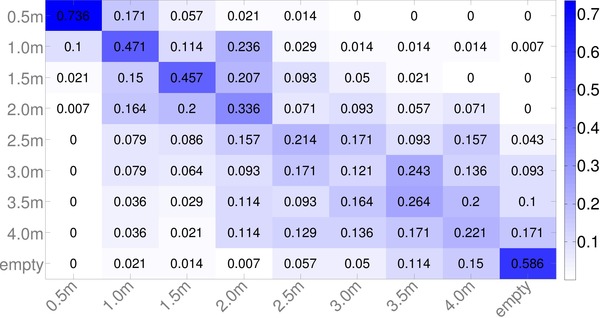}\label{tableDistanceStatisticsConfusion}}
 \subfloat[Accuracy with two APs]{\includegraphics[width=.38\columnwidth]{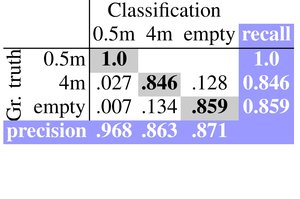}\label{tableDistanceStatisticsMultipleAPs}}
 
 \caption{Classification of activity in various distances. {\scriptsize (978-1-4799-3445-4/14/\$31.00 \copyright 2014 IEEE)}
}
\label{tableDistanceStatisticsAll}
\end{table}
However, when we tolerate an error of about 0.5-1m, reasonable accuracy can be achieved (cf. table~\ref{tableDistanceStatisticsConfusion}).
Furthermore, distance to an activity can be estimated from RSSI.
In conclusion, there is good potential to classify activities also in this distance so that for indoor environments a mobile phone can cover a typical room sufficiently.

In addition, we employed another equally active AP operating at the same frequency.
Although the signal strength between both differed by about 10 dB, classification accuracy was comparable using packets from either AP.
In addition, when features are created from RSSI information of both APs, the accuracy can be further improved (cf. table~\ref{tableDistanceStatisticsMultipleAPs}).
We used the same features for both access points, effectively doubling the number of features for one time window.

\subsubsection{Direction of Movement or Activity} \label{sectionMovement}
To identify locations of performed activities, in addition to distance also relative direction must be distinguished. 
We conducted a study in the environment depicted in figure~\ref{figureCaseStudiesETH2} in which the mobile was placed in the center of a 2m$\times$2m table. 
In parallel to the four borders of the table a subject conducted activity (walking up and down) in approximately 1m distance. 
In figure~\ref{figureCaseStudiesETH2}, the regions are marked as 'Area~1--4'.
The experiment was repeated multiple times for each side and each time for at least five minutes continuously.

We then attempted to distinguish at which side the activity was performed. 
However, it turned out that it is hardly possible to tell this from the RSSI. 
We were not able to find a subset of features that would achieve reasonable accuracy with three distinct classifiers\footnote{For the results depicted in this table, we utilise a Naive Bayes classifier with 100 sample points and a Loess window of $.5$, a classification tree with two or more instances at its leaves and a k-NN classifier with $k=20$} (cf. table~\ref{tableConfusionDistance} for exemplary results).
\begin{table}
\centering
\setlength{\tabcolsep}{1pt}
\begin{footnotesize}
\subfloat[Naive Bayes]{\begin{tabular}{cr|ccccc}
&&\multicolumn{4}{c}{\begin{scriptsize}Classification\end{scriptsize}}\\
&&s~1&s~2&s~3&s~4&\cellcolor{blue!40}\color{white}{\textbf{recall}}\\\hline
\multirow{4}{5pt}{\begin{sideways}\begin{scriptsize}Gr. truth
\end{scriptsize}\end{sideways}}&side~1&\cellcolor{gray!40}{\textbf{.486}}&.193&.121&.2&\cellcolor{blue!40}\color{white}{\textbf{.486}}\\
&side~2&.3&\cellcolor{gray!40}{\textbf{.321}}&.086&.293&\cellcolor{blue!40}\color{white}{\textbf{.321}}\\
&side~3&.286&.136&\cellcolor{gray!40}{\textbf{.179}}&.4&\cellcolor{blue!40}\color{white}{\textbf{.179}}\\
&side~4&.221&.121&.214&\cellcolor{gray!40}{\textbf{.443}}&\cellcolor{blue!40}\color{white}{\textbf{.443}}\\
\multicolumn{2}{r}{ \cellcolor{blue!40}\color{white}{\textbf{precision}}}& \cellcolor{blue!40}\color{white}{\textbf{.376}}&\cellcolor{blue!40}\color{white}{\textbf{.417}} &\cellcolor{blue!40}\color{white}{\textbf{.298}}&\cellcolor{blue!40}\color{white}{\textbf{.332}}&\cellcolor{blue!40}
\end{tabular}
}\hfill
\subfloat[Classification tree]{\begin{tabular}{|ccccc}
\multicolumn{3}{c}{\begin{scriptsize}Classification\end{scriptsize}}\\
s~1&s~2&s~3&s~4&\cellcolor{blue!40}\color{white}{\textbf{recall}}\\\hline
\cellcolor{gray!40}{\textbf{.471}}&.243&.15&.136&\cellcolor{blue!40}\color{white}{\textbf{.471}}\\
.314&\cellcolor{gray!40}{\textbf{.279}}&.236&.171&\cellcolor{blue!40}\color{white}{\textbf{.279}}\\
.3&.314&\cellcolor{gray!40}{\textbf{.214}}&.171&\cellcolor{blue!40}\color{white}{\textbf{.214}}\\
.293&.257&.171&\cellcolor{gray!40}{\textbf{.279}}&\cellcolor{blue!40}\color{white}{\textbf{.279}}\\
\cellcolor{blue!40}\color{white}{\textbf{.342}}&\cellcolor{blue!40}\color{white}{\textbf{.255}} &\cellcolor{blue!40}\color{white}{\textbf{.278}}&\cellcolor{blue!40}\color{white}{\textbf{.368}}&\cellcolor{blue!40}
\end{tabular}
}\hfill
\subfloat[k-NN classifier]{\begin{tabular}{|ccccc}
\multicolumn{3}{c}{\begin{scriptsize}Classification\end{scriptsize}}\\
s~1&s~2&s~3&s~4&\cellcolor{blue!40}\color{white}{\textbf{recall}}\\\hline
\cellcolor{gray!40}{\textbf{.421}}&.207&.221&.15&\cellcolor{blue!40}\color{white}{\textbf{.421}}\\
.271&\cellcolor{gray!40}{\textbf{.336}}&.214&.179&\cellcolor{blue!40}\color{white}{\textbf{.335}}\\
.271&.186&\cellcolor{gray!40}{\textbf{.279}}&.264&\cellcolor{blue!40}\color{white}{\textbf{.279}}\\
.221&.143&.264&\cellcolor{gray!40}{\textbf{.371}}&\cellcolor{blue!40}\color{white}{\textbf{.371}}\\
\cellcolor{blue!40}\color{white}{\textbf{.355}}&\cellcolor{blue!40}\color{white}{\textbf{.385}} &\cellcolor{blue!40}\color{white}{\textbf{.285}}&\cellcolor{blue!40}\color{white}{\textbf{.385}}&\cellcolor{blue!40}
\end{tabular}
}
\end{footnotesize}
     \caption{Confusion matrices for the distinction of the direction in which a person was performing activities {\scriptsize (978-1-4799-3445-4/14/\$31.00 \copyright 2014 IEEE)}}
\label{tableConfusionDistance}
\end{table}

\subsubsection{Detection of activity behind a door/wall}\label{sectionBehind}
WiFi signals can traverse obstacles such as walls or doors but the signal will be damped at this occasion so that the recognition of activity based on this data might be more challenging.
We distinguished activity inside or outside a room.
As depicted in figure~\ref{figureCaseStudiesETH3-behind}, we placed the phone inside a room next to the door.
Then, a person was present and moving either inside or on the other side of the door.
In the third case, nobody was present in the room or outside on the corridor.
For each case, RSSI samples have been recorded for at least five minutes.
Table~\ref{tableResultsDoor} depicts the results.
\begin{table}
\centering
\setlength{\tabcolsep}{3pt}
\subfloat[Performance of various classifiers]{\begin{tabular}{r|cccc}
&CA&IS&Brier&AUC\\\hline
Naive Bayes&.710&.784&.423&.880\\
\cellcolor{gray!20}Classification tree&\cellcolor{gray!20}.669&\cellcolor{gray!20}.855&\cellcolor{gray!20}.663&\cellcolor{gray!20}.795\\
k-NN&.724&.843&.393&.903
\end{tabular}
\label{tableResultsStatisticsDoor}
}\hfill
\setlength{\tabcolsep}{1pt}
\subfloat[Confusion matrix]{\begin{tabular}{cr|ccccc}
&&\multicolumn{4}{c}{\begin{scriptsize}Classification\end{scriptsize}}& \\
&&empty&inside&outside&\cellcolor{blue!40}\color{white}{\textbf{recall}}\\\hline
\multirow{3}{5pt}{\begin{sideways}\begin{scriptsize}Gr. truth
\end{scriptsize}\end{sideways}}&
empty&\cellcolor{gray!40}{\textbf{.814}}&.036&.150&\cellcolor{blue!40}\color{white}{\textbf{.814}}\\
&inside&.064&\cellcolor{gray!40}{\textbf{.743}}&.193&\cellcolor{blue!40}\color{white}{\textbf{.743}}\\
&outside&.2&.186&\cellcolor{gray!40}{\textbf{.614}}&\cellcolor{blue!40}\color{white}{\textbf{.614}}\\
\multicolumn{2}{r}{ \cellcolor{blue!40}\color{white}{\textbf{precision}}}& \cellcolor{blue!40}\color{white}{\textbf{.755}}&\cellcolor{blue!40}\color{white}{\textbf{.770}} &\cellcolor{blue!40}\color{white}{\textbf{.642}}&\cellcolor{blue!40}
\end{tabular}
\label{tableResultsConfusionDoor}
}		
\caption{Classification of activity inside and outside a room {\scriptsize (978-1-4799-3445-4/14/\$31.00 \copyright 2014 IEEE)}}
\label{tableResultsDoor}
\end{table}
While all three cases can be distinguished, the activity conducted outside the room is indeed most confused.
This is, because although there is increased fluctuation, signals are weak so that classes are more likely confused for one of the other classes which represent either stronger activity or weakly fluctuating signals.

\subsubsection{Detection of Walking Speed} \label{sectionSpeed}
Walking speed can be derived from signal strength with an SDR-based active device-free system~\cite{Pervasive_Shi_2014}.
We investigate the performance of a passive RSSI-based system.
In the setting shown in figure~\ref{figureCaseStudiesETH2}, a person walked around the table with the mobile phone in its center in a distance of about 2m.
The phone sampled the RSSI while the person was moving at 0.5$\frac{m}{\mbox{sec}}$, 1$\frac{m}{\mbox{sec}}$ and 2$\frac{m}{\mbox{sec}}$.
This experiment was conducted for at least 5 minutes at each recording and repeated for each velocity twice and also clockwise and counter-clockwise.
The speed was controlled autonomously by the subject.
For this purpose we marked the circle with an interleaving of 1m and equipped the subject with a stopwatch so that she could adjust her speed.
Best accuracy was achieved considering median, mean, minimum and standard deviation.
Results are depicted in table~\ref{tableConfusionSpeed}.
\begin{table}
\centering
\setlength{\tabcolsep}{2pt}
\subfloat[Performance for different sampling rates]{\begin{tabular}{r|cccc}
&CA&IS&Brier&AUC\\\hline
5 samples/s&.681&.777&.409 &.881\\
\cellcolor{gray!20}10 samples/s &\cellcolor{gray!20}.717&\cellcolor{gray!20}.823&\cellcolor{gray!20}.388&\cellcolor{gray!20}.894\\
15 samples/s &.767&.910&.355&.905
\end{tabular}
}\hfill
\setlength{\tabcolsep}{1pt}
\subfloat[Confusion of walking speeds]{\begin{tabular}{cr|ccccc}
&&\multicolumn{3}{c}{\begin{scriptsize}Classification\end{scriptsize}}\\
&&$.5\frac{m}{\mbox{sec}}$&$1.0\frac{m}{\mbox{sec}}$&$2.0\frac{m}{\mbox{sec}}$&\cellcolor{blue!40}\color{white}{\textbf{recall}}\\\hline
\multirow{3}{5pt}{\begin{sideways}\begin{scriptsize}Gr. truth
\end{scriptsize}\end{sideways}}&$.5\frac{m}{\mbox{sec}}$&\cellcolor{gray!40}{\textbf{.864}}&.071&.064&\cellcolor{blue!40}\color{white}{\textbf{.864}}\\
&$1\frac{m}{\mbox{sec}}$&.121&\cellcolor{gray!40}{\textbf{.657}}&.221&\cellcolor{blue!40}\color{white}{\textbf{.657}}\\
&$2\frac{m}{\mbox{sec}}$&.050&.171&\cellcolor{gray!40}{\textbf{.779}}&\cellcolor{blue!40}\color{white}{\textbf{.779}}\\
\multicolumn{2}{r}{ \cellcolor{blue!40}\color{white}{\textbf{precision}}}& \cellcolor{blue!40}\color{white}{\textbf{.834}}&\cellcolor{blue!40}\color{white}{\textbf{.730}} &\cellcolor{blue!40}\color{white}{\textbf{.732}}&\cellcolor{blue!40}
\end{tabular}
}
     \caption{Classification of walking speed ($k=18$; 10$\frac{\mbox{\tiny samp.}}{\mbox{\tiny sec}}$) {\scriptsize (978-1-4799-3445-4/14/\$31.00 \copyright 2014 IEEE)}}
\label{tableConfusionSpeed}
\end{table}
All velocities can be well distinguished.
The confusion is greater for velocity pairs that are closer to each other.

\subsubsection{Sensing Crowd}
An important ingredient for context-recognition is the size of the surrounding crowd. 
Different sizes can indicate different situations. For instance, having a conversation between 
few people or listening to or giving a talk in a meeting with multiple people.
We attempted to distinguish between the empty room depicted in figure~\ref{figureCaseStudiesETH5-crowd} and the same room occupied by 1, 5 or 10 persons.
In the room, two phones where placed to record the RSSI.
Phone~1 is located near the entrance on a table and the second one is placed beside a window across the room.
The latter was farther away from the nearest AP which is located right next to the door outside the room.
For the case study, 10, 5, 1 or no person would be present for at least five minutes.
Participants were instructed not to stand still for longer periods of time but otherwise should move or act freely.
They have then, for instance, moved around, stood in front of a poster and discussed it or leaned over a map to plan a weekend trip.
Table~\ref{tableResultsCrowd} shows that this broad distinction of the number of persons present is possible with reasonable accuracy.
\begin{table}
\centering
\setlength{\tabcolsep}{3pt}
\subfloat[Performance of a k-NN classifier with data from various phones]{\begin{tabular}{r|cccc}
&CA&IS&Brier&AUC\\\hline
Phone~1&.759&1.309&.354&.946\\
\cellcolor{gray!20}Phone~2&\cellcolor{gray!20}.805&\cellcolor{gray!20}1.397&\cellcolor{gray!20}.304&\cellcolor{gray!20}.937
\end{tabular}
\label{tableResultsStatisticsCrowd}
}\hfill
\setlength{\tabcolsep}{1pt}
\subfloat[Confusion matrix (Phone~2)]{\begin{tabular}{cr|ccccc}
&&\multicolumn{4}{c}{\begin{scriptsize}Classification\end{scriptsize}}& \\
&&0P&1P&5P&10P&\cellcolor{blue!40}\color{white}{\textbf{recall}}\\\hline
\multirow{4}{5pt}{\begin{sideways}\begin{scriptsize}Gr. truth
\end{scriptsize}\end{sideways}}&
0~Persons&\cellcolor{gray!40}{\textbf{1.0}}&&&&\cellcolor{blue!40}\color{white}{\textbf{1.0}}\\
&1~Persons&&\cellcolor{gray!40}{\textbf{.857}}&.129&.014&\cellcolor{blue!40}\color{white}{\textbf{.857}}\\
&5~Persons&&.129&\cellcolor{gray!40}{\textbf{.671}}&.2&\cellcolor{blue!40}\color{white}{\textbf{.671}}\\
&10~Persons&&.114&.193&\cellcolor{gray!40}{\textbf{.693}}&\cellcolor{blue!40}\color{white}{\textbf{.693}}\\
\multicolumn{2}{r}{ \cellcolor{blue!40}\color{white}{\textbf{precision}}}& \cellcolor{blue!40}\color{white}{\textbf{1.0}}&\cellcolor{blue!40}\color{white}{\textbf{.779}} &\cellcolor{blue!40}\color{white}{\textbf{.676}}&\cellcolor{blue!40}\color{white}{\textbf{.764}}&\cellcolor{blue!40}
\end{tabular}
\label{tableResultsConfusionCrowd}
}		
\caption{Classification of crowd (k-NN;~20~samples/s) {\scriptsize (978-1-4799-3445-4/14/\$31.00 \copyright 2014 IEEE)}}
\label{tableResultsCrowd}
\end{table}
Empty room is perfectly recognised with 100\% of accuracy. While different crowd sizes are confused in particular for the 5 and 10 persons, the performance is still far above random guess.

Observe that with the data captured by the phone placed near the window (Phone~2) the recognition accuracy is higher. 
We account this to the fact that individuals in the room continuously have resided in the area between the AP outside the room and the window. 
Therefore, the impact on the WiFi packets due to blocking or damping was greater for this phone.

\subsubsection{Detect Activity while the Device is Carried}
When the phone is carried, we expect significant noise for the recognition of situations from packets blocked by the user carrying the phone.
We investigated whether RSSI can still be utilised to classify simple situations.
For instance, it might be possible to derive whether a person is alone or in company. 
For this study, the phone was carried in the pocket of a person (cf. figure~\ref{figureCaseStudiesETH4-pocket}).
Then, the person was standing or walking alone and while another person was walking in proximity. 
For each class, data has been recorded for at least five minutes. 
Table~\ref{tableResultsPocket} depicts the results.
\begin{table}
\centering
\setlength{\tabcolsep}{1pt}
\begin{tabular}{cr|ccccc}
&&\multicolumn{4}{c}{\begin{scriptsize}Classification\end{scriptsize}}& \\
&&Stat.--Empty&Stat.--pres.&Walking--Empty&Walking--pres.&\cellcolor{blue!40}\color{white}{\textbf{recall}}\\\hline
\multirow{4}{5pt}{\begin{sideways}\begin{scriptsize}Gr. truth
\end{scriptsize}\end{sideways}}&
Stationary--Empty&\cellcolor{gray!40}{\textbf{.921}}&.079&&&\cellcolor{blue!40}\color{white}{\textbf{.921}}\\
&Stationary--presence&.15&\cellcolor{gray!40}{\textbf{.693}}&.079&.079&\cellcolor{blue!40}\color{white}{\textbf{.693}}\\
&Walking--Empty&&.171&\cellcolor{gray!40}{\textbf{.457}}&.371&\cellcolor{blue!40}\color{white}{\textbf{.457}}\\
&Walking--presence&&.107&.364&\cellcolor{gray!40}{\textbf{.529}}&\cellcolor{blue!40}\color{white}{\textbf{.529}}\\
\multicolumn{2}{r}{ \cellcolor{blue!40}\color{white}{\textbf{precision}}}& \cellcolor{blue!40}\color{white}{\textbf{.860}}&\cellcolor{blue!40}\color{white}{\textbf{.660}} &\cellcolor{blue!40}\color{white}{\textbf{.508}}&\cellcolor{blue!40}\color{white}{\textbf{.540}}&\cellcolor{blue!40}
\end{tabular}
\caption{Classification of presence when device is carried {\scriptsize (978-1-4799-3445-4/14/\$31.00 \copyright 2014 IEEE)}}
\label{tableResultsPocket}
\end{table}
We are well able to detect whether the person wearing the device is stationary and alone or if there is movement either by the device holder or by someone else. 
However, when the device holder is herself moving, the distinction of other activity is more confused. 

\subsubsection{Recognition of gestures}
Finally, we considered the recognition of gestures.
For this study, we placed the phone on the table as depicted in figure~\ref{figureCaseStudiesETH} and performed 10 single-handed gestures (figure~\ref{figureGestures}).
\begin{figure}
\centering
 \includegraphics[width=0.6\columnwidth]{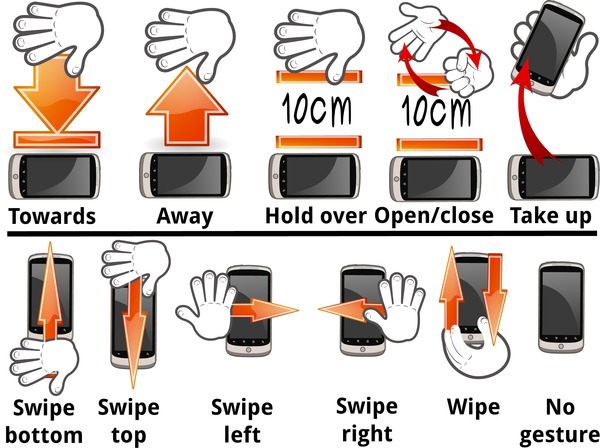}
 \caption{Gestures performed {\scriptsize (978-1-4799-3445-4/14/\$31.00 \copyright 2014 IEEE)}}
 \label{figureGestures}
\end{figure}
Each gesture lasted for approximately 0.4 to 0.7 seconds and was performed in a distance of about 1cm to 20cm.
Only for two gestures, {\em Take up} and {\em Wipe}, the hand was in actual physical contact with it.
Each gesture was repeated 100 times for a total of 1100 recordings of the distinct cases.

Best results have been achieved with the features {\em mean, variance, signal peaks within 10\% of the maximum} and the {\em fraction between the mean of the first and second half of a feature window}. 
Table~\ref{tableConfusionGesture01} shows the classification results.
\begin{table}
\centering
\setlength{\tabcolsep}{2pt}
\begin{tabular}{cr|cccccccccccc}
&&\multicolumn{11}{c}{\begin{scriptsize}Classification\end{scriptsize}}\\
&&\begin{sideways}\begin{scriptsize}Away
\end{scriptsize}\end{sideways}&\begin{sideways}\begin{scriptsize}Hold over
\end{scriptsize}\end{sideways}&\begin{sideways}\begin{scriptsize}Towards
\end{scriptsize}\end{sideways}&\begin{sideways}\begin{scriptsize}No gesture
\end{scriptsize}\end{sideways}&\begin{sideways}\begin{scriptsize}Open/close
\end{scriptsize}\end{sideways}&\begin{sideways}\begin{scriptsize}Take up
\end{scriptsize}\end{sideways}&\begin{sideways}\begin{scriptsize}S. bottom
\end{scriptsize}\end{sideways}&\begin{sideways}\begin{scriptsize}S. left
\end{scriptsize}\end{sideways}&\begin{sideways}\begin{scriptsize}Wipe
\end{scriptsize}\end{sideways}&\begin{sideways}\begin{scriptsize}S. right
\end{scriptsize}\end{sideways}&\begin{sideways}\begin{scriptsize}S. top
\end{scriptsize}\end{sideways}&\cellcolor{blue!40}\color{white}{\textbf{recall}}\\\hline
\multirow{11}{5pt}{\begin{sideways}\begin{scriptsize}Ground truth
\end{scriptsize}\end{sideways}}&Away&\cellcolor{gray!40}{\textbf{.54}}&&&&.06&.03&.15&.12&&.05&.05&\cellcolor{blue!40}\color{white}{\textbf{.540}}\\
&Hold over&&\cellcolor{gray!40}{\textbf{.26}}&.16&.05&.16&.08&.03&.06&.04&.14&.02&\cellcolor{blue!40}\color{white}{\textbf{.260}}\\
&Towards&&.09&\cellcolor{gray!40}{\textbf{.71}}&.07&.04&.01&.01&&.06&&.01&\cellcolor{blue!40}\color{white}{\textbf{.710}}\\
&No gesture&&.04&.06&\cellcolor{gray!40}{\textbf{.67}}&.05&.01&&.01&.15&&.01&\cellcolor{blue!40}\color{white}{\textbf{.670}}\\
&Open/close&&.1&.07&.09&\cellcolor{gray!40}{\textbf{.47}}&.07&.01&&.14&.03&.02&\cellcolor{blue!40}\color{white}{\textbf{.470}}\\
&Take up&.01&.08&.03&.02&.09&\cellcolor{gray!40}{\textbf{.46}}&.06&.03&.09&.06&.07&\cellcolor{blue!40}\color{white}{\textbf{.460}}\\
&S. bottom&.13&.06&&.01&.04&.09&\cellcolor{gray!40}{\textbf{.36}}&.06&&.2&.05&\cellcolor{blue!40}\color{white}{\textbf{.360}}\\
&S. left&.12&.01&.01&.01&&.08&.07&\cellcolor{gray!40}{\textbf{.49}}&&.07&.14&\cellcolor{blue!40}\color{white}{\textbf{.490}}\\
&Wipe&&.04&.1&.08&.16&.09&.01&&\cellcolor{gray!40}{\textbf{.51}}&&.01&\cellcolor{blue!40}\color{white}{\textbf{.510}}\\
&S. right&.03&.03&&.01&.03&.01&.1&.01&.01&\cellcolor{gray!40}{\textbf{.68}}&.09&\cellcolor{blue!40}\color{white}{\textbf{.680}}\\
&S. top&.07&.02&.01&&.03&.08&&.21&&.11&\cellcolor{gray!40}{\textbf{.47}}&\cellcolor{blue!40}\color{white}{\textbf{.470}}\\
\multicolumn{2}{r}{ \cellcolor{blue!40}\color{white}{\textbf{precision}}}& \cellcolor{blue!40}\color{white}{\textbf{.600}}&\cellcolor{blue!40}\color{white}{\textbf{.356}} &\cellcolor{blue!40}\color{white}{\textbf{.617}}&\cellcolor{blue!40}\color{white}{\textbf{.663}}&\cellcolor{blue!40}\color{white}{\textbf{.416}}&\cellcolor{blue!40}\color{white}{\textbf{.455}}&\cellcolor{blue!40}\color{white}{\textbf{.450}}&\cellcolor{blue!40}\color{white}{\textbf{.495}}&\cellcolor{blue!40}\color{white}{\textbf{.510}}&\cellcolor{blue!40}\color{white}{\textbf{.507}}&\cellcolor{blue!40}\color{white}{\textbf{.500}}&\cellcolor{blue!40}
\end{tabular}
     \caption{Confusion matrices for the distinction of gestures {\scriptsize (978-1-4799-3445-4/14/\$31.00 \copyright 2014 IEEE)}}
\label{tableConfusionGesture01}
\end{table}
We observe that, while some gestures have a reasonable accuracy and average results are far above guess, a high confusion for other classes inhibits correct classification.
In particular, the gesture {\em Hold over} can hardly be be distinguished. 
Furthermore, some of the swiping gestures are confused.  
Therefore, we merged cases into a single gesture. 
Table~\ref{tableConfusionGesture02} shows two levels of merging gestures. 
When merging to 7 distinct gestures\footnote{{\em Hold over, Open/close, Take up} and {\em Wipe} were labelled as {\em No gesture}} we achieve a mean accuracy of about 0.56. 
In the table, labels are shortened to the first two letters for space limitations.
While most gestures are well recognized, especially the swipe gestures still achieve mediocre performance.
When further reducing to the four gestures {\em away, towards, no gesture} and {\em swipe} by merging all swipe gestures, an average accuracy of 0.66 is achieved.

\begin{table}
\centering
\setlength{\tabcolsep}{1pt}
\subfloat[Confusion of 7 distinct gestures -- all remaining gestures shifted to 'no gesture']{\begin{tabular}{r|cccccccc}
&\multicolumn{7}{c}{\begin{scriptsize}Classification\end{scriptsize}}\\
&\begin{sideways}\begin{scriptsize}Aw
\end{scriptsize}\end{sideways}&\begin{sideways}\begin{scriptsize}No
\end{scriptsize}\end{sideways}&\begin{sideways}\begin{scriptsize}To
\end{scriptsize}\end{sideways}&\begin{sideways}\begin{scriptsize}Sb
\end{scriptsize}\end{sideways}&\begin{sideways}\begin{scriptsize}Sl
\end{scriptsize}\end{sideways}&\begin{sideways}\begin{scriptsize}Sr
\end{scriptsize}\end{sideways}&\begin{sideways}\begin{scriptsize}St
\end{scriptsize}\end{sideways}&\cellcolor{blue!40}\color{white}{\textbf{recall}}\\\hline
Aw&\cellcolor{gray!40}{\textbf{.58}}&.09&&.13&.11&.05&.04&\cellcolor{blue!40}\color{white}{\textbf{.58}}\\
No&&\cellcolor{gray!40}{\textbf{.872}}&.05&.014&.012&.034&.018&\cellcolor{blue!40}\color{white}{\textbf{.872}}\\
To&&.4&\cellcolor{gray!40}{\textbf{.59}}&&&&.01&\cellcolor{blue!40}\color{white}{\textbf{.59}}\\
Sb&.15&.22&&\cellcolor{gray!40}{\textbf{.32}}&.04&.22&.05&\cellcolor{blue!40}\color{white}{\textbf{.32}}\\
Sl&.12&.11&.01&.06&\cellcolor{gray!40}{\textbf{.48}}&.08&.14&\cellcolor{blue!40}\color{white}{\textbf{.48}}\\
Sr&.04&.15&&.06&.01&\cellcolor{gray!40}{\textbf{.67}}&.07&\cellcolor{blue!40}\color{white}{\textbf{.67}}\\
St&.03&.18&.01&.01&.24&.1&\cellcolor{gray!40}{\textbf{.43}}&\cellcolor{blue!40}\color{white}{\textbf{.43}}\\
\cellcolor{blue!40}\color{white}{\textbf{prec}}& \cellcolor{blue!40}\color{white}{\textbf{.630}}&\cellcolor{blue!40}\color{white}{\textbf{.791}} &\cellcolor{blue!40}\color{white}{\textbf{.686}}&\cellcolor{blue!40}\color{white}{\textbf{.492}}&\cellcolor{blue!40}\color{white}{\textbf{.511}}&\cellcolor{blue!40}\color{white}{\textbf{.519}}&\cellcolor{blue!40}\color{white}{\textbf{.518}}&\cellcolor{blue!40}
\end{tabular} 
}\hfill
\subfloat[Confusion of 4 gestures]{\begin{tabular}{r|ccccc}
&\multicolumn{4}{c}{\begin{scriptsize}Classification\end{scriptsize}}\\
&\begin{sideways}\begin{scriptsize}Away
\end{scriptsize}\end{sideways}&\begin{sideways}\begin{scriptsize}Towards 
\end{scriptsize}\end{sideways}&\begin{sideways}\begin{scriptsize}No gesture
\end{scriptsize}\end{sideways}&\begin{sideways}\begin{scriptsize}Swipe
\end{scriptsize}\end{sideways}&\cellcolor{blue!40}\color{white}{\textbf{recall}}\\\hline
Away&\cellcolor{gray!40}{\textbf{.45}}&.06&&.49&\cellcolor{blue!40}\color{white}{\textbf{.45}}\\
Towards&&\cellcolor{gray!40}{\textbf{.834}}&.052&.114&\cellcolor{blue!40}\color{white}{\textbf{.834}}\\
No gesture&&.41&\cellcolor{gray!40}{\textbf{.56}}&.03&\cellcolor{blue!40}\color{white}{\textbf{.56}}\\
Swipe&.063&.128&.005&\cellcolor{gray!40}{\textbf{.805}}&\cellcolor{blue!40}\color{white}{\textbf{.805}}\\
\cellcolor{blue!40}\color{white}{\textbf{precision}}&\cellcolor{blue!40}\color{white}{\textbf{.643}}&\cellcolor{blue!40}\color{white}{\textbf{.810}}&\cellcolor{blue!40}\color{white}{\textbf{.667}}&\cellcolor{blue!40}\color{white}{\textbf{.747}}&\cellcolor{blue!40}
\end{tabular}
}
     \caption{Performance with fewer gestures {\scriptsize (978-1-4799-3445-4/14/\$31.00 \copyright 2014 IEEE)}}
\label{tableConfusionGesture02}
\end{table}
This suggests that some gestures can indeed be utilised to interact with phones or other WiFi capable devices. 
Possible applications cover a touch-free, frictionless interface to control mobile devices also through clothes, an extended interface for wearable devices or interface-free devices in an Internet of Things.

\subsubsection{Discussion}
We have investigated a passive, device-free RSSI-based activity recognition system considering several situations captured by a mobile phone. 
Figure~\ref{figureSummary} summarises shows the accuracies achieved in our case studies relative to a random classifier\footnote{The random Classifier takes each possible choice with equal probability} as a baseline.
\begin{figure}
\centering
\includegraphics[width=.8\columnwidth]{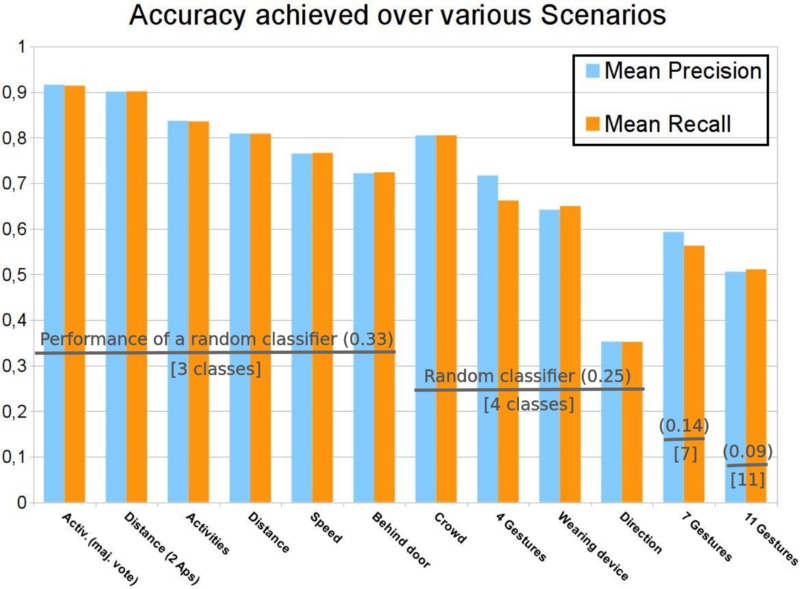}
\caption{Accuracies achieved for various scenarios considered {\scriptsize (978-1-4799-3445-4/14/\$31.00 \copyright 2014 IEEE)}}
\label{figureSummary}
\end{figure}
In general, the overall accuracy falls with increasing number of classes to distinguish.
However, short of the recognition of direction, the results are far above random guess in all cases.
The simple distinction of distance and three well separated situations reached best results and could be further improved by considering multiple APs or majority votes over several windows of features.

The failure in the distinction in which direction activity was performed indicates limitations of passive device-free RSSI-based recognition.
Since the system has to rely on data transmitted from an AP which can be located in an arbitrary direction and the device might be in arbitrary orientation, it is hardly possible to obtain fine grained information on environmental situation. 
The sequence of received RSSI samples is highly bursty and of low granularity and rate. 
Consequently, the classes that can be distinguished are limited too. 
Additional studies we conducted on the recognition of further activities (sitting, standing, walking, reading, typing on a computer) could not yield a useful recognition accuracy.

However, our results show that an RSSI-based passive device-free recognition system can provide basic environmental awareness when classical phone-based recognition systems fail (e.g. when the phone is not carried on the body).
In addition, for special cases such as the distinction of gestures where movement is conducted in close proximity to the device, RSSI-based passive recognition might provide an innovative ad-hoc alternative to more complex solutions.

Unfortunately, our solution requires a modified WiFi firmware, root access and is currently limited to a small set of phones. 
Much work is still required in order to allow operating-system supported non-root access to RSSI information in sufficient frequency. 

\subsection{Conclusion}\label{sectionConclusion}
We have proposed and discussed the utilisation of RSSI information from mobile phones for the characterisation of situations, activities and gestures.
We reported problems to be solved for the acquisition of RSSI from received packets on mobile phones and discussed the structure of the data as well as features suited for the recognition of activities and gestures.

In case studies we investigated the feasibility of RSSI-based recognition on mobile phones for multiple scenarios.
Summarising, these results show that it is possible to distinguish simple activities and to some extent also gestures from RSSI fluctuation captured by a mobile phone.
However, it also shows the limitations of this device-free recognition approach for instance, regarding a localisation of activities.
Furthermore, the accuracies achieved stay below what would be possible with classical sensors such as accelerometers.

However, we could demonstrate, that there is a good potential to extend the perception of a phone beyond its boundaries into the environment.
A recognition in distances of 4 meters is still feasible.
RSSI-based recognition can cover cases where classical sensors can not provide meaningful results.

Regarding the recognition of gestures, we see a good potential to extend the interface of body-worn devices with RSSI-based gesture recognition.
\vfill
\pagebreak

\section[Secure communication based on ambient audio]{Secure communication based on ambient audio \footnote{Originally published as 'Dominik Schuermann and Stephan Sigg: Secure communication based on ambient audio, in IEEE Transactions on Mobile Computing (TMC), Feb. 2013, vol. 12 no. 2 (DOI: http://doi.ieeecomputersociety.org/10.1109/TMC.2011.271)' (1536-1233/13/\$31.00 \copyright 2013 IEEE
 Published by the IEEE CS, CASS, ComSoc, IES, SPS)
}}\label{sectionOriginalSE01}
  We propose to establish a secure communication channel among devices based on similar audio patterns.
     Features from ambient audio are used to generate a shared cryptographic key between devices without exchanging information about the ambient audio itself or the features utilised for the key generation process.
     We explore a common audio-fingerprinting approach and account for the noise in the derived fingerprints by employing error correcting codes.
     This fuzzy-cryptography scheme enables the adaptation of a specific value for the tolerated noise among fingerprints based on environmental conditions by altering the parameters of the error correction and the length of the audio samples utilised.
     In this paper we experimentally verify the feasibility of the protocol in four different realistic settings and a laboratory experiment.
     The case-studies include an office setting, a scenario where an attacker is capable of reproducing parts of the audio context, a setting near a traffic loaded road and a crowded canteen environment.
     We apply statistical tests to show that the entropy of fingerprints based on ambient audio is high.
     The proposed scheme constitutes a totally unobtrusive but cryptographically strong security mechanism based on contextual information.

\subsection{Introduction}
An important factor in the set of security risks is typically the human impact.
People are occasionally careless or incompletely understanding the underlying technology.
This is especially true for wireless communication.
For instance, the communication range or the number of potential communication partners might be underestimated.
This is natural since humans typically base trust on the situation or context they perceive~\cite{Context_Dupuy_2006}.
Nevertheless, the range of a communication network most likely bridges devices in various contexts.

As context, proximity and trust are related~\cite{Context_Dupuy_2006}, a security scheme that utilises common contextual features among communicating devices might provide a sense of security which is perceived as natural by individuals and reduce the number of human errors related to security.

Consider, for instance, a meeting with co-workers of a specific project.
Naturally, workers trust the others based on working agreements.
Every group member needs the permission to access common information like mobile phone numbers or shared files.
Communication between group members, however, should be guarded against access from external devices or individuals.
The meeting room defines the borders which shall not be crossed by any confidential data.
Context information that is unique inside these borders, such as ambient audio, can be exploited as the seed to generate a common secret for the secure information exchange and authentication.

Mobile phones can then synchronise their ID-cards ad-hoc without user interaction and secured by their physical proximity.
Similarly, access to shared files on computers of co-workers and communication links among co-workers can be secured.

Another reason why security cautions might be discarded occasionally is the effort required and inconvenience to establish a secure connection.
This is especially true between devices that communicate seldom or for the first time.

We propose a mechanism to unobtrusively (zero interaction) establish an ad-hoc secure communication channel between unacquainted devices which is conditioned on the surrounding context.
In particular, we consider audio as a source of spatially centred context.
We exploit the similarity of features from ambient audio by devices in proximity to create a secure communication channel exclusively based on these features.
At no point in the protocol the secret itself or information that could be used to derive audio feature values is made public.
In order to do so, we generate synchronised audio-fingerprints from ambient sounds and utilise error correcting codes to account for noise in the feature vector.
On each communicating device the feature vector is then used to create an identical key.
The proposed protocol is non-interactive, unobtrusive and does not require specific or identical hardware at communication partners.

The remainder of this document is structured as follows.
In section~\ref{section2} we introduce related work on context-based security mechanisms and security with noisy input data.
Section~\ref{section3} discusses the algorithmic background required for ambient audio-based key generation and implementation details.
In section~\ref{section4} we discuss the noise and entropy of audio-fingerprints achieved in an offline-experiment with sampled audio sequences.
We show that the similarity in audio-fingerprints is sufficient for authentication but can not be utilised as secure key directly.
In particular, we utilise fuzzy-cryptography schemes to account for noise in the input data.
Section~\ref{section5} presents four case-studies in different environments that exploit the feasibility of the approach in various settings.
The general feasibility of the approach is demonstrated in section~\ref{section5-1} in a controlled office environment. 
Section~\ref{section5-2} then shows that the audio context can be separated between two offices even when a synchronised audio source is located in both places.
Additionally, we studied the feasibility of ambient audio-based key generation at the side of a heavily trafficked road in section~\ref{section5-4} and in a canteen environment in section~\ref{section5-3}.
The entropy of the ambient audio-based characteristic binary sequences generated by our method is discussed in section~\ref{section6}.
In section~\ref{section7} we draw our conclusion.

\subsection{Related work}\label{section2}
In the literature, several authors consider spontaneous authentication or the establishing of a secure communication channel among mobile and ad-hoc devices based on environmental stimuli~\cite{mayrhofer2008spontaneous,Cryptography_Bichler_2007-2,ContextAwareness_Holmquist_2001,Cryptography_Varshavsky_2007}.
So far, shaking processes from accelerometer data and RF-channel measurements have been utilised as unique context source that contains shared characteristic information.

This concept was presented 2001 by Holmquist et al.~\cite{ContextAwareness_Holmquist_2001}. 
The authors propose to utilise the accelerometer of the Smart-It~\cite{5836} device to extract characteristic features from simultaneous shaking processes of two devices.
Later, Mayrhofer et al. presented an authentication mechanism based on this principle~\cite{mayrhofer2007shake}.
The authors demonstrated, that an authentication is possible when devices are shaken simultaneously by a single person, while an authentication was unlikely for a third person trying to mimic the correct movement pattern remotely.
Also, Mayrhofer derived in~\cite{mayrhofer2007candidate} that the sharing of secret keys is possible with a similar protocol.
The proposed protocol that can be utilised with arbitrary context features repeatedly exchanges hashes of key-sub-sequences until a common secret is found.
In this instrumentation, exponentially quantised fast Fourier transformation (FFT) coefficients of a sequence of accelerometer samples are utilised.
In contrast, Bicher et al. describe an approach in which noisy acceleration readings can be utilised to establish a secure communication channel among devices~\cite{Cryptography_Bichler_2007,Cryptography_Bichler_2007-2}.
They utilise a hash function that maps similar acceleration patterns to identical key sequences.
However, their approach suffers from the required exact synchronisation among devices so that the authors computed the correct hash-values offline.
Additionally, the hash function utilised required that the keys computed exactly match and that the neighbourhood around these keys is precisely defined.
When patterns are located at the border of one of the region's neighbourhoods, the tolerance for noise in the input is biased in the direction of the centre of this region.
Additionally, key generation by simultaneous shaking is not unobtrusive.

We utilise an error correction scheme to account for noise in the input data which can be fine-tuned for any Hamming distance desired which is centred around the noisy characteristic sequences generated instead of an artificially defined centre value.
We implement a Network Time Protocol (NTP) based synchronisation mechanism that establishes sufficient synchronisation among nodes.

Another sensor class utilised for context-based device authentication is the RF-channel.
Varshavsky et al. present a technique to authenticate co-located devices based on RF-measurements since channel measurements from devices in near proximity are sufficiently similar to authenticate devices against each other~\cite{Cryptography_Varshavsky_2007}.
Hershey et al. utilise physical layer features to derive secret keys for a pair of devices~\cite{Cryptography_Hershey_1995}.
In the absence of interference and non-linear components, transmitter and receiver experience identical channel response~\cite{Cryptography_Smith_2004}.
This information is utilised to generate a secret key among a node pair.
Since channel characteristics are spatially sharply concentrated and not predictable at a remote location~\cite{Cryptography_Madiseh_2008}, an eavesdropper is not capable of guessing information about the secret.
This scheme was validated in an indoor environment in~\cite{Cryptography_BenHamida_2009}.
Although we consider the keys generated by this scheme as strong, it does not preserve spatial properties.
A device at arbitrary distance could pretend to be a nearby communication partner.

Kunze and Lukowicz recently demonstrated, that audio information indeed suffices to derive spatial information~\cite{ContextAwareness_Kunze_2007}.
They combine audio readings with accelerometer data to classify locations of mobile devices.
In their work, the noise emitted by a vibrating mobile phone was utilised to distinguish among 35 specific locations in three different rooms with over 90\,\% accuracy.

Instead, we utilise purely ambient noise to establish a secure communication channel among devices in spatial proximity.
We record NTP-synchronised audio samples at two locations, generate a characteristic audio-fingerprint and map this fingerprint to a unique secret key with the help of error correcting codes.

The last step is necessary since the similarity between fingerprints is typically not sufficient to establish a secure channel.
With fuzzy-cryptography schemes, the generation of an identical key based on noisy input data~\cite{tuyls2007security} is possible.
Li et al. analyse the usage of biometric or multimedia data as part of an authentication process and propose a protocol~\cite{li2006robust}.
Due to the use of error-tolerant cryptographic techniques, this protocol is robust against noise in the input data.
The authors utilise a secure sketch~\cite{Dodis04fuzzyextractors} to produce public information about an input without revealing it.
The input can then be recovered given another value that is close to it.
A similar study is presented by Miao et al.~\cite{miao2009biometrics}.
The authors establish a key distribution based on a fuzzy vault~\cite{Juels02fuzzyvault} using data measured by devices worn on the human body.
The fuzzy vault scheme, also utilised in~\cite{dodis2006robust}, enables the decryption of a secret with any key that is substantially similar to the key used for encryption.

\subsection{Ad-hoc audio-based encryption}\label{section3}
Originally, audio-fingerprinting was proposed to classify music or speech.
In our work binary fingerprints from ambient audio are used to establish an encrypted connection based on the surrounding audio context.
Due to differences between fingerprints generated by participating devices, a cryptographic protocol is needed that tolerates a specific amount of noise in these keys.

We propose the following scheme. 
A set of devices willing to establish a common key conditioned on ambient audio take synchronised audio samples from their local microphones.
Each device then computes a binary characteristic sequence for the recorded audio: An audio-fingerprint (cf.~section~\ref{section3-1}).
This binary sequence is designed to fall onto a code-space of an error correcting code (cf.~section~\ref{section3-2}).
In general, a fingerprint will not match any of the codewords exactly.
Fingerprints generated from similar ambient audio resemble but due to noise and inaccuracy in the audio-sampling process, it is unlikely that two fingerprints are identical.
Devices therefore exploit the error correction capabilities of the error correcting code utilised to map fingerprints to codewords (cf.~section~\ref{section3-3}).
For fingerprints with a Hamming-distance within the error correction threshold of the error correcting code the resulting codewords are identical and then utilised as secure keys (cf.~section~\ref{section3-4}).
This scheme is in principle not limited in the number of devices that participate.
When devices are synchronised in their local times, they agree on a point in time when audio shall be recorded and proceed with the fingerprint creation and error correction autonomously as described above.
All similar fingerprints will map to an identical codeword.
As detailed in section~\ref{section5-3}, the Hamming distance tolerated in fingerprints rises with increasing distance of devices.

The following sections provide an overview over audio-fingerprinting, our fuzzy commitment implementation, problems we experienced and possible solutions.

\subsubsection{Audio-fingerprinting}\label{section3-1}
Audio-fingerprinting is an approach to derive a characteristic pattern from an audio sequence~\cite{cano2005review}.
Generally, the first step involves the extraction of features from a piece of audio.
These features are usually isolated in a time-frequency analysis after application of Fourier or Cosine transforms.
Some authors also utilise wavelet-transforms~\cite{AudioFingerprinting_Baluja_2008,AudioFingerprinting_Ghouti_2006,AudioFingerprinting_Sukittanon_2002}.
Common applications include the retrieval of a specific music file in an audio database~\cite{Haitsma2003highlyrobust}, duplicate detection in such a database~\cite{AudioFingerprinting_Burges_2005} as well as identification of music based on short samples~\cite{AudioFingerprinting_Bellettini_2010}.
The capabilities of detecting similar audio sequences in the presence of heavy signal distortion are prominently demonstrated by applications such as query by humming~\cite{AudioFingerprinting_Ghias_1995}.
The authors utilise autocorrelation, maximum likelihood and Cepstrum analysis to describe the pitch of an audio sequence as a Parsons encoded music contour~\cite{AudioFingerprinting_Parsons_1975}.
Similar audio sequences are detected by approximate string matching~\cite{AlgorithmsStringMatching_BeezaYates_1992}.
McNab et al. added rhythm information by analysing note duration to match the beginning of a song~\cite{AudioFingerprinting_McNab_1996}.
A similar approach is presented by Prechelt et al.~\cite{AudioFingerprinting_Prechelt_2001}. 
They achieved more accurate results for query by whistling since the frequency range of whistling is much lower than for humming or singing.
In 2002, Chai et al. computed a rough melodic contour by counting the number of equivalent transitions in each beat~\cite{AudioFingerprinting_Chai_2002}.
Notes are detected by amplitude-based note segmentation.
Later, Shiffrin et al. showed that songs can be described by Markov-chains \cite{AudioFingerprinting_Shifrin_2002} where states represent note transitions.
Retrieval of songs is then achieved by the HMM Forward algorithm~\cite{AlgorithmsHMM_Rabiner_1989} so that no database query is required.
In 2003, Zhu et al. addressed practical problems of recently proposed approaches such as the accuracy of the derived description by utilising a dynamic time-warping mechanism~\cite{AudioFingerprinting_Zhu_2003}.

Most of these studies are based on music-specific properties such as rhythm information, pitch or melodic contour. 
Since such features might be missing in ambient audio, these methods are not applicable in our case.
Haitsma et al. presented in~\cite{AudioFingerprinting_Haitsma_2001} an approach applicable for the classification of general audio sequences by extracting a binary representation of audio from changes in the energy of successive frequency bands.
This system was later shown to be highly robust to noise and distortion in audio data~\cite{Haitsma2003highlyrobust}.
Due to its reported robustness, several authors employ slightly modified versions of this approach~\cite{AudioFingerprinting_Bellettini_2010}.
Leboss\'e et al, for instance, add further redundant sub-samples taken from the beginning and the end of an overlapping time window in order to reduce the number of bits in the fingerprint representation~\cite{AudioFingerprinting_Lebosse_2007}.
Alternatively, Burges et al. enhance the former approach by utilising a distortion discriminant analysis~\cite{AudioFingerprinting_Burges_2003}.
Generally, time frames taken from the audio source are mapped successively on smaller time windows in order to generate a condensed characteristic representation of the audio sequence.
An alternative approach based on spectral flatness of a signal is proposed Herre et al.~\cite{AudioFingerprinting_Herre_2001}.

Also, Yang presented a method to utilise characteristic energy peaks in the signal spectrum in order to extract a unique pattern~\cite{AudioFingerprinting_Yang_2001}.
A general framework that supports this scheme was later presented by Yang et al.~\cite{AudioFingerprinting_Yang_2002}.
Building on these ideas, a similar algorithm was then successfully applied commercially by Avery Wang on a huge data base of audio sequences~\cite{AudioFingerprinting_Wang_2006,AudioFingerprinting_Wan_2003}.

To create audio-fingerprints for our studies, we split an audio sequence $S$ with length $|S|=l$ and sample rate $r$ up into $n$ frames $F_1,\dots,F_n$ of identical length $d=|F_i|=r\cdot\frac{l}{n}$.
On each frame a discrete Fourier transformation (DFT) weighted by a Hanning window (HW) is applied: 
\begin{align}
  \forall i &\in\{0, \dots ,n-1 \},\nonumber\\
  S_i&=DFT\left(HW(F_i)\right)\nonumber
\end{align}
  The frames are divided into $m$ non-overlapping frequency bands of width
\begin{align}
  b&=\frac{\mbox{\footnotesize maxfreq}(S_i) - \mbox{\footnotesize minfreq}(S_i)}{m} \mathrm{.}
\end{align}
On each band the sum of the energy values is calculated and stored to an energy matrix $E$ with energy per frame per frequency band.
\begin{align}
  \forall j &\in \{0,\dots ,m-1\},\nonumber\\
  S_{ij}&=\mbox{bandfilter}_{b\cdot j,b\cdot (j+1)}(S_i)\\
  E_{ij}&=\sum_{k} S_{ij}[k]
\end{align}

Using the matrix $E$, a fingerprint $f$ is generated, where $\forall i \in \{1, \dots, n-1\},\forall j \in \{0, \dots, m-2\}$ each bit describes the difference between the energy on frequency bands between two consecutive frames:
\begin{eqnarray}
  f(i,j)=\left\{\begin{array}{rl}
    1, &\begin{array}{ll}
              (E(i,j)-E(i,j+1))-\\
        (E(i-1,j)-E(i-1,j+1))>0
         \end{array}\\[.4cm]
    0, &\mbox{otherwise.}
  \end{array}\right.
\end{eqnarray}
The complete algorithm is detailed in the appendix.

For each synchronisation, we sampled $l=6.375$ seconds of ambient audio at a sample rate of $r=44100$~Hz. 
We split the audio stream into $n=17$ frames of $d=0.375$ seconds each and divide every frame into $m=33$ frequency bands, to obtain a 512 bit fingerprint.
Due to the extensive recording duration, the generated fingerprints show great robustness in real world experiments (cf.~section~\ref{section4} and section~\ref{section5}).
We used a Fast Fourier Transform (FFT) with fixed values on the length of the segments as detailed above.

This audio-fingerprinting scheme utilised in our studies utilises energy differences between frequency bands, as proposed by Haitsma et al.~\cite{Haitsma2003highlyrobust}. 
However, we take a more general approach of classifying ambient audio instead of music.
Commonly, in the literature, the characteristic information is found in a smaller frequency band and a logarithmic scaling is suggested to better represent properties of the human auditory system.
Since our system is not restricted to musical recordings, we expect that all frequency bands are equally important.
Therefore, we divide frames into frequency bands at a linear scale rather than a logarithmic one.
Additionally, we do not use overlapping frames since this has not shown improvements in our case. 
Also, the entropy and therefore the security features of the generated fingerprint is likely to become impaired with overlapping frames~\cite{2108,2109}.

\subsubsection{Audio-fingerprints as cryptographic keys}\label{section3-2}
To use the audio-fingerprints directly as keys for a classic encryption scheme the concurrence of fingerprints generated from related audio sequences has to be $1$ with a considerably high probability~\cite{schneider1996applied}.
Since we experienced a substantial difference in the audio-fingerprints created (cf. section~\ref{section4}) we consider the application of fuzzy-cryptography schemes.
Note that a perfect match in fingerprints is unlikely since devices are spatially separated, not exactly synchronised and utilise possibly different audio hardware.

The proposed cryptographic protocol shall be feasible unattended and ad-hoc with unacquainted devices.
For an eavesdropper in a different audio context it shall be computationally infeasible to use any intercepted data to decrypt a message or parts of it.
Additionally, we want to control the threshold for the tolerated offset between fingerprints based on contextual conditions of different physical locations.

With fuzzy encryption schemes, a secret $\varsigma$ is used to hide the key $\kappa$ in a set of possible keys $\mathcal{K}$ in such a way that only a similar secret $\varsigma'$ can find and decrypt the original key $\kappa$ correctly.
In our case, the secrets which ought to be similar for all communicating devices in the same context are audio-fingerprints.

A Fuzzy Commitment scheme can, for instance, be implemented with Reed-Solomon codes~\cite{Reed60polynomial}.
The following discussion provides a short introduction to these codes.

Given a set of possible words $\mathcal{A}$ of length $m$ and a set of possible codewords $\mathcal{C}$ of length $n$, Reed-Solomon codes $RS(q,m,n)$ are initialised as:
\begin{align}
     \mathcal{A}&=\mathbb{F}^m_{q},\\
     \mathcal{C}&=\mathbb{F}^n_{q},     
\end{align}
with $q=p^k, p \text{ prime}, k \in \mathbb{N}$.
These codes are mapping a word $a\in\mathcal{A}$ of length $m$ uniquely to a specific codeword $c\in \mathcal{C}$ of length $n$:
\begin{equation}
     a \xrightarrow{Encode} c,  
\end{equation}
This step adds redundancy to the original words with $n>m$, based on polynomials over Galois fields~\cite{Reed60polynomial}.

Decoding utilises the error correction properties of the Reed-Solomon-based encoding function to account for differences in the fingerprints created.
The decoding function maps a set of codewords from one group $C=\{c,c',c'',\dots \} \subset \mathcal{C}$ to one single original word. 
It is
\begin{align}
  \tilde{c} \xrightarrow{Decode} a\in \mathcal{A}.
\end{align}
The value
\begin{equation}
     t=\left\lfloor \frac{n-m}{2} \right\rfloor\label{equationT}
\end{equation}
defines the threshold for the maximum number of bits between codewords that can be corrected in this manner to decode correctly to the same word $a$~\cite{Juels99fuzzycommitment}.
In the following algorithms the fingerprints $f$ and $f'$ are used in conjunction with codewords to make use of this error correction procedure.
Dependent on the noise in the created fingerprints, $t$ can then be chosen arbitrarily.

\subsubsection{Commit and Decommit algorithms}\label{section3-3}
We utilise Reed-Solomon error correcting codes in the following scheme to generate a common secret among devices.
A fingerprint $f$ is used to hide a randomly chosen word $a$ as the basis for a key in a set of possible words $a\in\mathcal{A}$.
This is a commit method.
A decommit method is constructed in such a way that only a fingerprint $f'$ with maximum Hamming distance 
\begin{equation}
     \mbox{Ham}(f,f')\leq t
\end{equation}
can find $a$ again.
We use Reed-Solomon $RS(q,m,n)$ codes, with $q=2^k$, $k \in \mathbb{N}$ and $n<2^k$, for our commit and decommit methods.
After initialisation, a private word $a\in\mathcal{A}$ is randomly chosen.
It is then encoded following the Reed-Solomon scheme to a specific codeword $c$.
For a subtract-function $\ominus$ in $\mathcal{C}=\mathbb{F}^n_{2^k}$, the difference to the fingerprint is calculated as
\begin{equation}
     \delta=f\ominus c\mathrm{.}
\end{equation}
Then, a SHA-512 hash~\cite{fips2008180} \texttt{h($a$)} is generated from $a$.
Afterwards, the tuple $(\delta, \mbox{\texttt{h($a$)}})$ containing the difference and the hash is made public.
Note that the transmission of $\mbox{\texttt{h($a$)}}$ is optional and is only required to check whether the decommitted $a'$ on the receiver side equals $a$.
However, provided a sufficiently secure hash function, an eavesdropper does not learn additional information about the key $a$ within reasonable time provided that she is ignorant of a fingerprint sufficiently similar to $f$.

The decommitment algorithm uses the public tuple $(\delta, \mbox{\texttt{h($a$)}})$ together with the secret fingerprint $f'$ to verify the similarity between $f$ and $f'$ and to obtain a shared word $a$.
A codeword $c'$ is calculated by subtracting~$f'$ by~$\delta$ in~$\mathbb{F}^n_{2^k}$.
\begin{equation}
     c'=f'\ominus \delta.
\end{equation}
Afterwards~$c'$ is decoded to~$a'$ as
\begin{equation}
     a' \in \mathcal{A} \xleftarrow{\text{Decode}} c' \in \mathcal{C}.
\end{equation}
From $\mbox{\texttt{h($a$)}}=\mbox{\texttt{h($a'$)}}$ we can conclude $a=a'$ with high probability.
This procedure is capable of correcting up to $t$ (cf. equation~(\ref{equationT})) differing bits between the fingerprints.
The decommitment was then successful and differences between $f$ and $f'$ are $t$ at most. 
The decommitted word~$a'$ is privately saved.

Participants can use their private words to derive keys for encryption.
A simple example for using $a=a'=(a_0, \dots , a_{m-1})$ to generate an encryption key for the Advanced Encryption Standard (AES)~\cite{fips2001197} is to sum over blocks of values of $a$.
For example, when $m=256$ we would sum over blocks with the length $8$ and take these values modulo $2^8-1$ to represent characters for a string with the length $32$, that can be used as a key $\kappa$:
\begin{align*}
  \text{Let } \kappa=(\kappa_0, \dots , \kappa_{31}) \text{, whereas }\\ \kappa_i = \left( \sum_{j=0}^{7} c_{(i*8)+j} \right) \mod 2^8-1
\end{align*}

In our study, for fingerprints of 512 bits we apply Reed-Solomon codes with $RS(q=2^{10},m,n=512)$. 
Given a maximum acceptable Hamming distance $t^*$ (cf. equation~(\ref{equationT})) between fingerprints we can then set~$m$ flexibly to define the minimum required fraction $u$ of identical bits in fingerprints as
\begin{align}
  t^*&=\left\lceil (1-u) \cdot n \right\rceil,\\
  m&=n-2 \cdot t^*\mathrm{.}
\end{align}
Experimentally, we found $u=0.7$ as a good trade-off for common audio environments to allow a sufficient amount of differences among the used fingerprints to pair devices successfully while at the same time providing sufficient cryptographic security against an eavesdropper in a different audio context (cf.~section~\ref{section4}).
\begin{align}
  m&=512-2 \cdot \left\lceil (1-0.7) \cdot 512 \right\rceil\\
  \nonumber &= 204
\end{align}
We therefore use Reed-Solomon codes with 
\begin{equation}
RS(2^{10},204,512)\mathrm{.}
\end{equation}

The commit and decommit algorithms are further detailed in the appendix.

\subsubsection{Synchronising communicating devices}\label{section3-4}
Since audio is time-dependent, a tight synchronisation among devices is required.
In particular, we experienced that fingerprints created by two devices were sufficiently similar only when the synchronisation offset among devices was within tens of milliseconds.
For synchronisation, any sufficiently accurate time protocol such as the Network Time Protocol (NTP)~\cite{rfc5905,AlgorithmsNTP_Mills_1995}, the Precision Time Protocol (PTP)~\cite{Meier_Synchronisation_1998} or a similar time protocol can be utilised. 
Also, synchronisation with GPS time might be a valid option.

When two participants, Alice and Bob, are willing to communicate securely with each other, Alice starts the protocol by requesting a pairing with Bob.
Then, they synchronise their absolute system times using a sufficiently accurate time protocol. 
Afterwards, Alice sends a start time $\tau_{start}$ to Bob.
When their clocks reach $\tau_{start}$, the recording of ambient audio is initiated and audio-fingerprinting is applied.

In our case-studies, synchronisation of devices was a critical issue.
Since the approach bases the binary fingerprints on energy differences of sub-samples of 0.375 seconds width, a misalignment of several hundreds of milliseconds results in completely different fingerprints.
For best results, the start times of the audio recordings should not differ more than about 0.001 seconds.
We successfully tested this with a remote NTP-server and also with one of the devices hosting the server.

Still, since NTP is able to synchronise clocks with an error of several milliseconds~\cite{AlgorithmsNTP_Mills_1994,rfc5905}, some error in the synchronisation of audio samples remains.
For instance, the usage of sound subsystems, like GStreamer \cite{gstreamer2010doc}, to record ambient audio introduces new delays.

Figure~\ref{figure1} illustrates this aspect in the frequency spectrum of two NTP-synchronised recordings.
\begin{figure}
     \centering
    \includegraphics[width=3.5in]{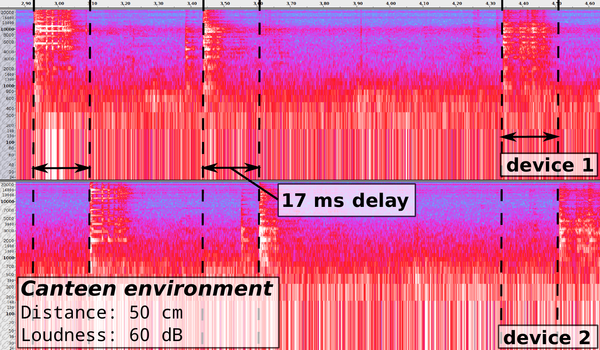}
     \caption{Synchronisation offset of NTP synchronised audio recordings {\scriptsize (1536-1233/13/\$31.00 \copyright 2013 IEEE Published by the IEEE CS, CASS, ComSoc, IES, SPS)
}}
     \label{figure1}
\end{figure}

As a solution, we had the decommiting node create 200 additional fingerprints by shifting the audio sequence in both directions in steps of 0.001 seconds.
The device then tried to create a common key with each of these fingerprints and uses the first successful attempt.
In this way, we could compensate for an error of about 0.2 seconds in the clock synchronisation among nodes.

\subsubsection{Security Considerations and Attack Scenarios}
Privacy leakages translate to leaking partial information about the used audio-fingerprints.
This can simplify the attack when further details of the ambient audio of Alice and Bob is available.

Possible attacks on fuzzy-cryptography are reviewed by Scheirer et al.~\cite{scheirer2007cracking}.
In particular, fuzzy commitment is evaluated regarding information leakage by Ignatenko et al.~\cite{ignatenko2010privacy}.
It was found that the scheme can leak information about the secret key.
However, this is attributable to helper data, a bit sequence at random distance to the secret key, which is made public in traditional fuzzy commitment schemes.
In our case, we do not utilise helper data and only optionally provide the hash of a data sequence with similar purpose.

The publicly available distance $\delta$ between $f$ and $c$ might, however leak information when either the fingerprints $f$, the code-sequences $c\in\mathcal{C}$ or the random word $a\in\mathcal{A}$ are not distributed uniformly at random or have insufficient entropy. 
Generally, it is important that 
\begin{enumerate}
     \item the random function to generate $a$ has a sufficiently high entropy
     \item the codewords $c\in\mathcal{C}$ are independently and uniformly distributed over all possible bit sequences of length $n$
     \item The entropy of the generated fingerprints is high
\end{enumerate}
We address these issues in the following.

1)~The choice of $a \in \mathcal{A}$ has to be done by using a random source with sufficient entropy.
In Linux-based systems \texttt{/dev/urandom} should provide enough entropy for using the output for cryptographic purposes~\cite{fenzi2004linux}.
For generating \texttt{h($a$)} a one-way-function has to be chosen to make sure that no assumptions on $a$ can be made based on \texttt{h($a$)}.
We utilise SHA-512 which is certified by the NIST and was extensively evaluated~\cite{fips2008180}.

2)~We are using 512 bit fingerprints and the Reed-Solomon code $RS(2^{10},204,512)$.
Consequently, sets of words and codewords are defined as $\mathcal{A}=\mathbb{F}^{204}_{2^{10}}$ and $\mathcal{C}=\mathbb{F}^{512}_{2^{10}}$.
A word $a$ out of $2^{10^{204}}=1024^{204}$ possible words is randomly chosen and encoded to $c$.

3)~In order to test the entropy of generated fingerprints we applied the dieHarder \cite{statistical_Brown_0000} set of statistical tests.
Generally, we could not find any bias in the fingerprints created from ambient audio.
Section~\ref{section6} discusses the test results in more detail.

A relevant attack scenario valid in our case is that the attacker is in the same audio context as Alice and Bob.
In this case, no security is provided by the proposed protocol.
Although this is a plausible threat, it can hardly be avoided that the leaking of contextual information poses a thread to a protocol that is designed to base the secure key generation exclusively on exactly this information. 
This principle is essential for the desired unobtrusive and ad-hoc operation.
An overview over possible attack scenarios when the attacker is not inside the same context is listed below.

\paragraph{Brute force}
The set of possible words $\mathcal{A}$ has to be large enough.
It should be computationally infeasible to test every combination to get the used word $a$.
The probability to guess the right $a$ is $1024^{-204}$ in our implementation.
Note that even with $u=0.6$, this probability is still $1024^{-102}$.
  
\paragraph{Denial-of-service (DoS)}
An attacker could stress the communication while Alice and Bob are using the fuzzy pairing.
The pairing would fail if $(\delta, \text{\texttt{h($a$)}})$ is not transmitted correctly.
DoS preventions should be implemented to provide an accurate treatment.
As part of these preventions a maximum number of attempts to pair two devices should be defined.
Generally, this type of attack is only possible when $(\delta, \text{\texttt{h($a$)}})$ or $\delta$ is transmitted.
As mentioned in section~\ref{section3-3}, with a careful choice of the fingerprint mechanism the exchange of data can be avoided.

\paragraph{Man-in-the-middle}
An Eavesdropper Eve could be located in such a way, that she can intercept the wireless connection but is not located in the same physical context as Alice and Bob.
When Eve intercepts the tuple $(\delta, \text{\texttt{h($a$)}})$, she must generate an audio-fingerprint $\overline{f}$ that is sufficiently close to the fingerprints $f$ and $f'$ of Alice and Bob to intercept successfully.
With no knowledge on the audio context, a brute force attack is then required.
This has to be done while Alice and Bob are currently in the phase of pairing.
Therefore Eve is limited by a strict time frame.
Again, this attack can be prevented by avoiding the transmission of $(\delta, \text{\texttt{h($a$)}})$ or $\delta$.

\paragraph{Audio amplification}
An Eavesdropper Eve could be located in physical proximity where the ambient audio used by Alice and Bob to generate their fingerprints is replicated.
Eve can utilise a directional microphone to amplify these audio signals.
In fact, this is a security threat which increases the chance that Eve can reconstruct the fingerprint partly to have a greater probability of guessing the secure secret.
Since our scheme inherently relies on contextual information we can not completely eliminate this threat.
However, we show in section~\ref{section5-2} that the acoustic properties in two rooms are at least sufficiently different to prevent a device with access to the dominant audio source to be successful in more than 50\,\% of all cases.

\subsection{Fingerprint-based authentication}\label{section4}
In a controlled environment we recorded several audio samples with two microphones placed at distinct positions in a laboratory.
The samples were played back by a single audio source.
Microphones were attached to the left and right ports of an audio card on a single computer with audio cables of equal lengths. 
They were placed at 1.5\,m, 3\,m, 4.5\,m and 6\,m distance to the audio source.
For each setting, the two microphones were always located at non-equal distances.
In several experiments, the audio source emitted the samples at quiet, medium and loud volume.
The audio samples utilised consisted of several instances of music, a person clapping her hands, snapping her fingers, speaking and whistling.
Dependent on the specific sample, the mean dB for these loudness levels varied slightly.
The loudness levels for several sample classes experienced in 1.5\,m distance are detailed in table~\ref{tableLoudnessLevels}.
\begin{table}
\renewcommand{\arraystretch}{1.3}
\caption{Approximate mean loudness experienced for several sample classes at 1.5\,m distance {\scriptsize (1536-1233/13/\$31.00 \copyright 2013 IEEE Published by the IEEE CS, CASS, ComSoc, IES, SPS)
}}
\label{tableLoudnessLevels}
\centering
\begin{tabular}{|c||c|c|c|}
\hline
 & loud & median & quiet \\
\hline
 Clap    & 40\,dB & 35\,dB &25\,dB\\
 Music   & 35\,dB & 25\,dB &15\,dB\\
 Snap    & 30\,dB & 25\,dB &10\,dB\\
 Speak   & 25\,dB & 20\,dB &15\,dB\\
 Whistle & 45\,dB & 35\,dB &25\,dB\\
\hline
\end{tabular}
\end{table}     

For these samples recorded by both microphones we created audio-fingerprints and compared their Hamming distances pair-wise.
We distinguish between fingerprints created for audio sampled simultaneously and non-simultaneously.
Overall, 7500 distinct comparisons between fingerprints are conducted in various environmental settings.
From these, 300 comparisons are created for simultaneously recorded samples.

Figure~\ref{figure2} depicts the median percentage of identical bits in the fingerprints for audio samples recorded simultaneously and non-simultaneously for several positions of the microphones and for several loudness levels.
The error bars depict the variance in the Hamming distance.
\begin{figure}
\centerline{
\subfloat[Loud, microphones at 1.5\,m and 3\,m]{\includegraphics[width=2.2in]{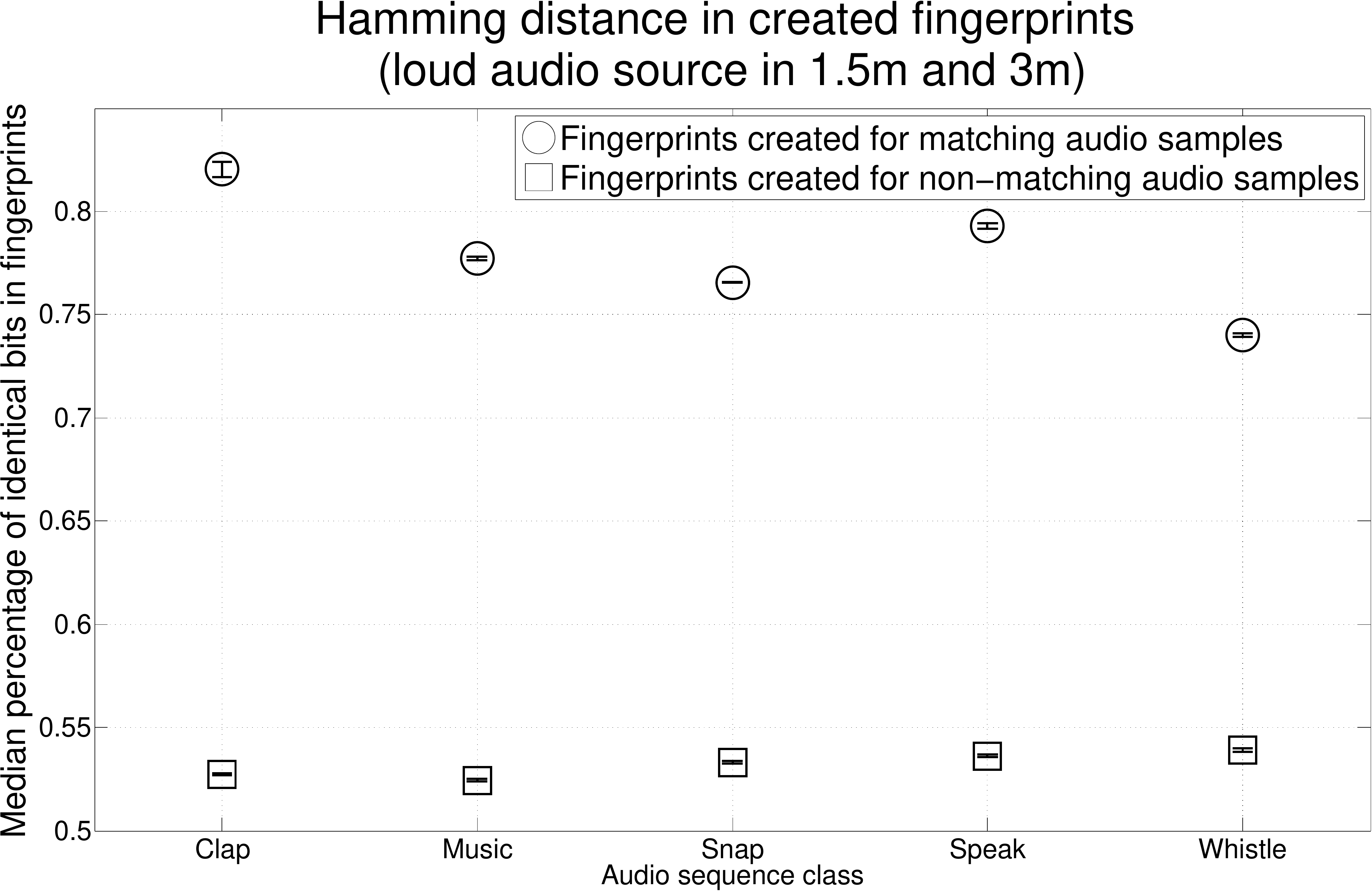}
\label{fig_first_case}}
\hfil
\subfloat[Medium, microphones at 1.5\,m and 3\,m]{\includegraphics[width=2.2in]{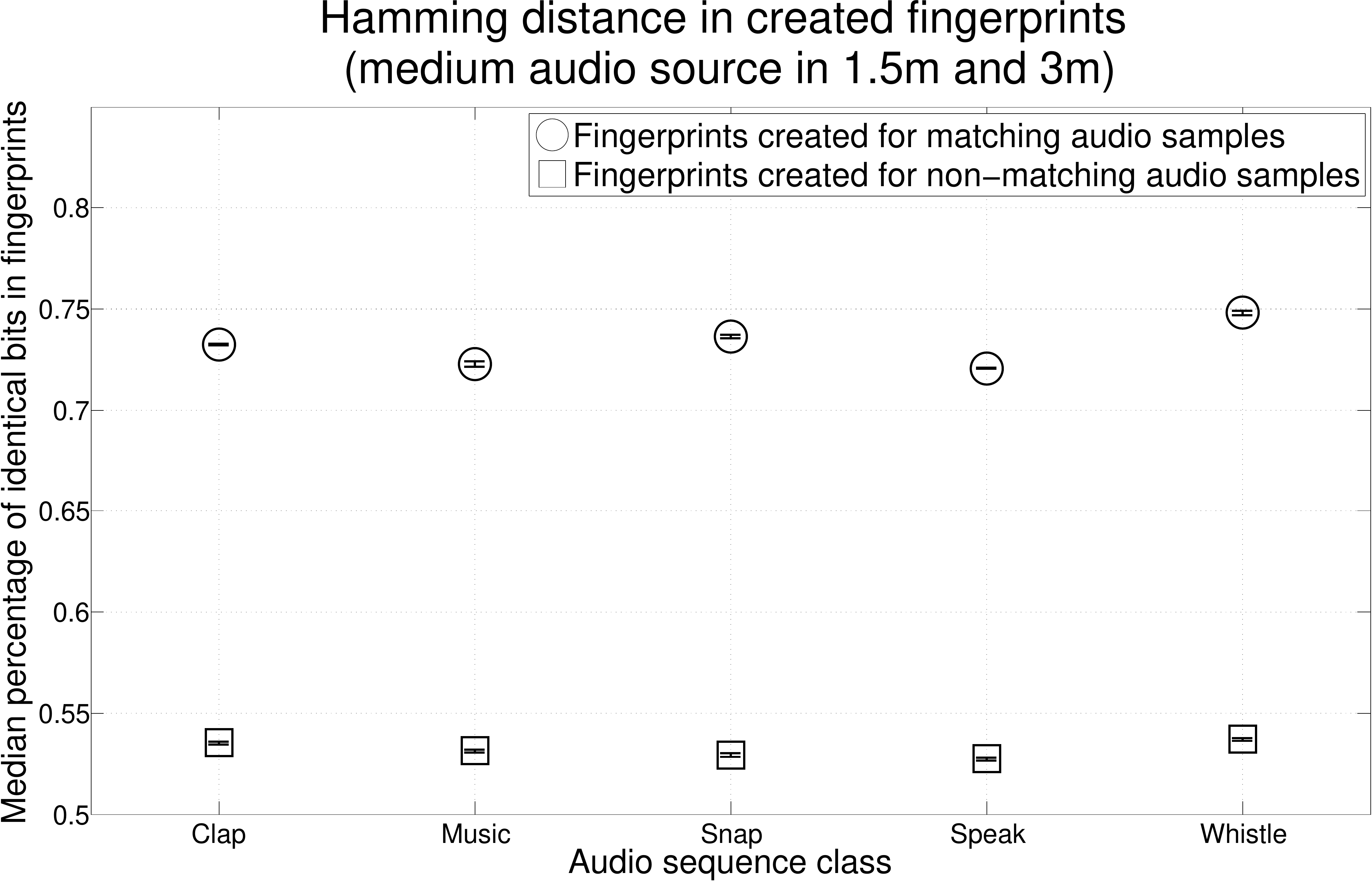}
\label{fig_second_case}}
\hfil
\subfloat[Quiet, microphones at 1.5\,m and 3\,m]{\includegraphics[width=2.2in]{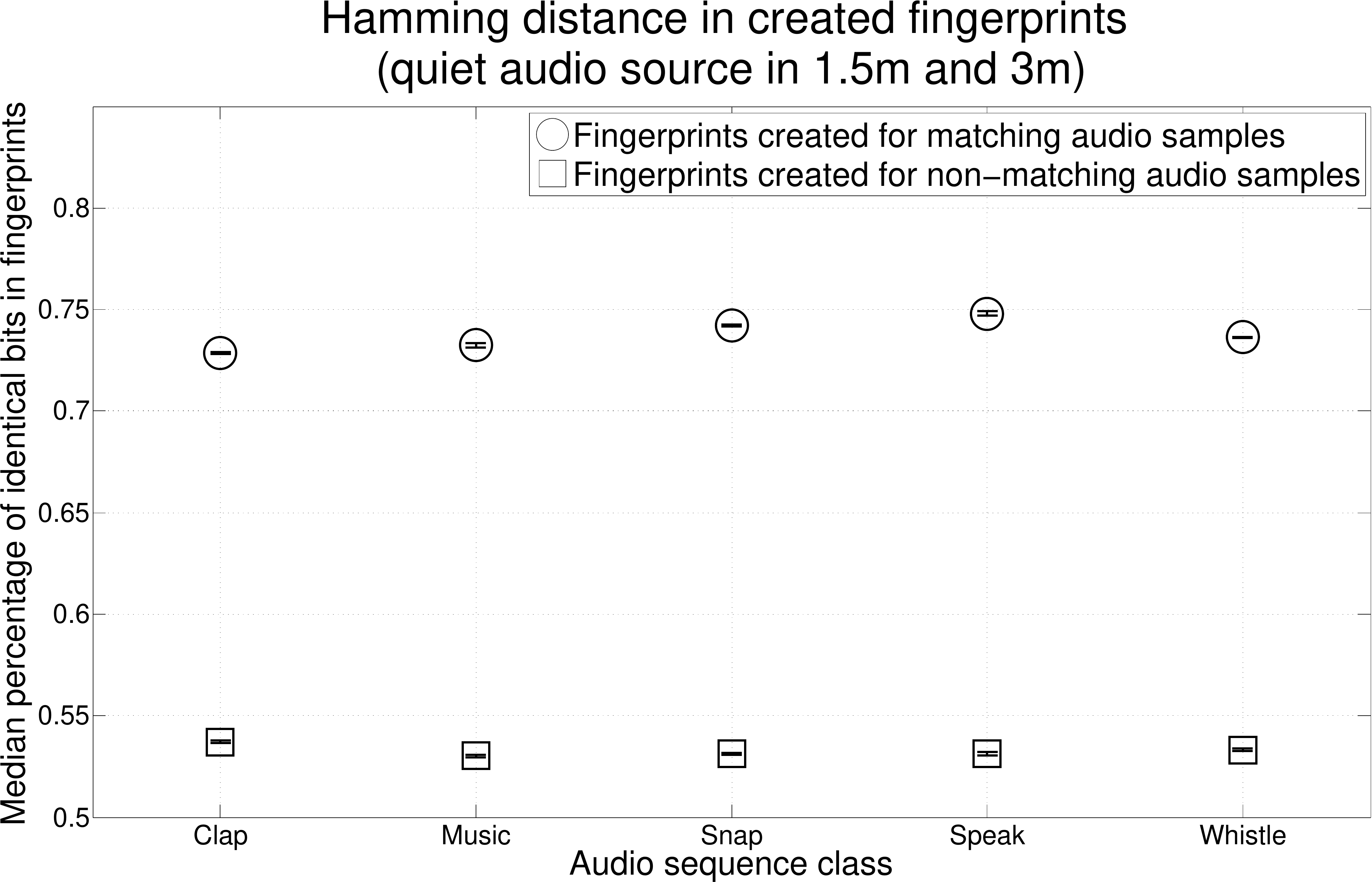}
\label{fig_third_case}}}
\centerline{
\subfloat[Loud, microphones at 3\,m and 4.5\,m]{\includegraphics[width=2.2in]{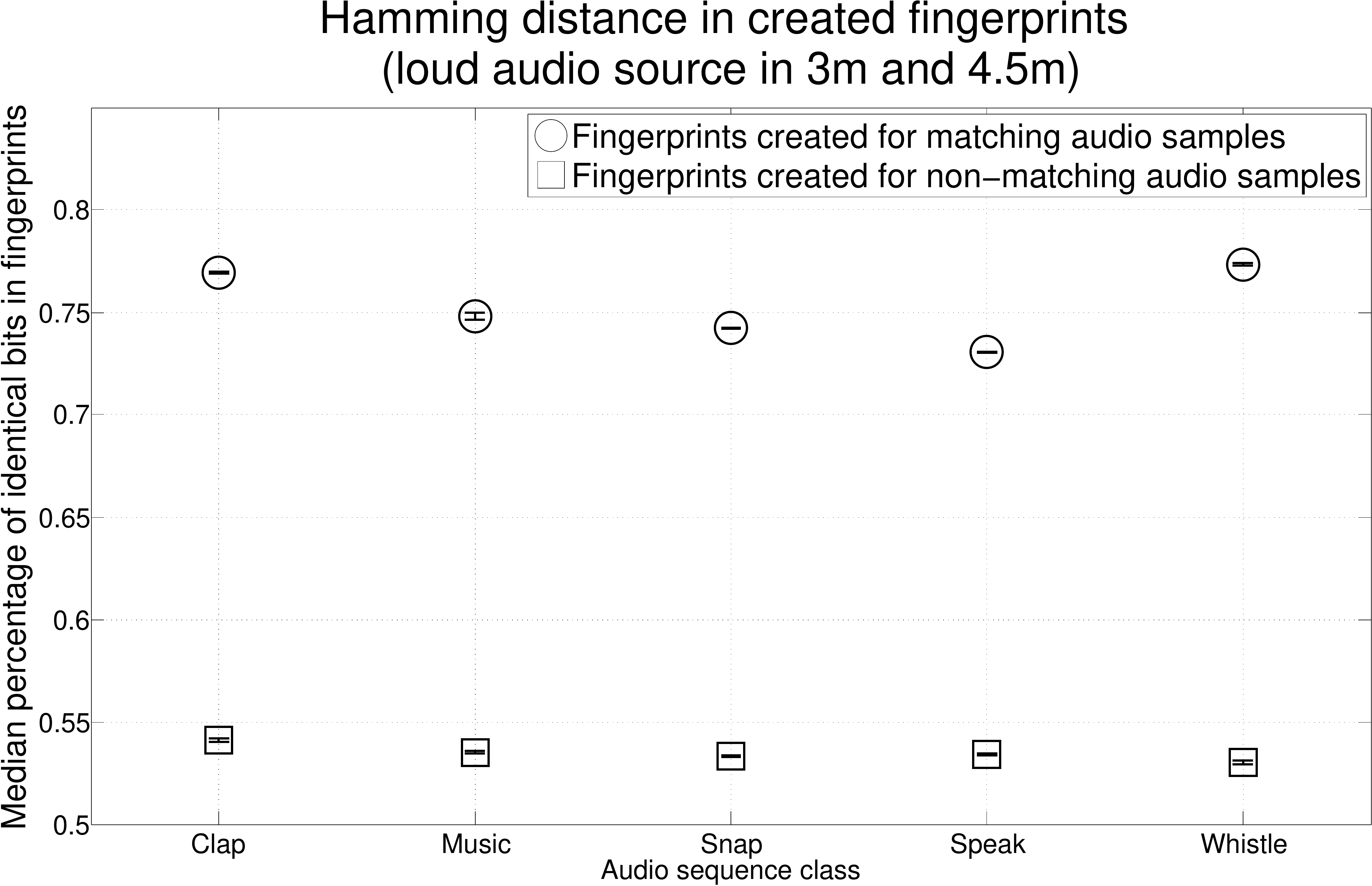}
\label{fig_fourth_case}}
\hfil
\subfloat[Medium, microphones at 3\,m and 4.5\,m]{\includegraphics[width=2.2in]{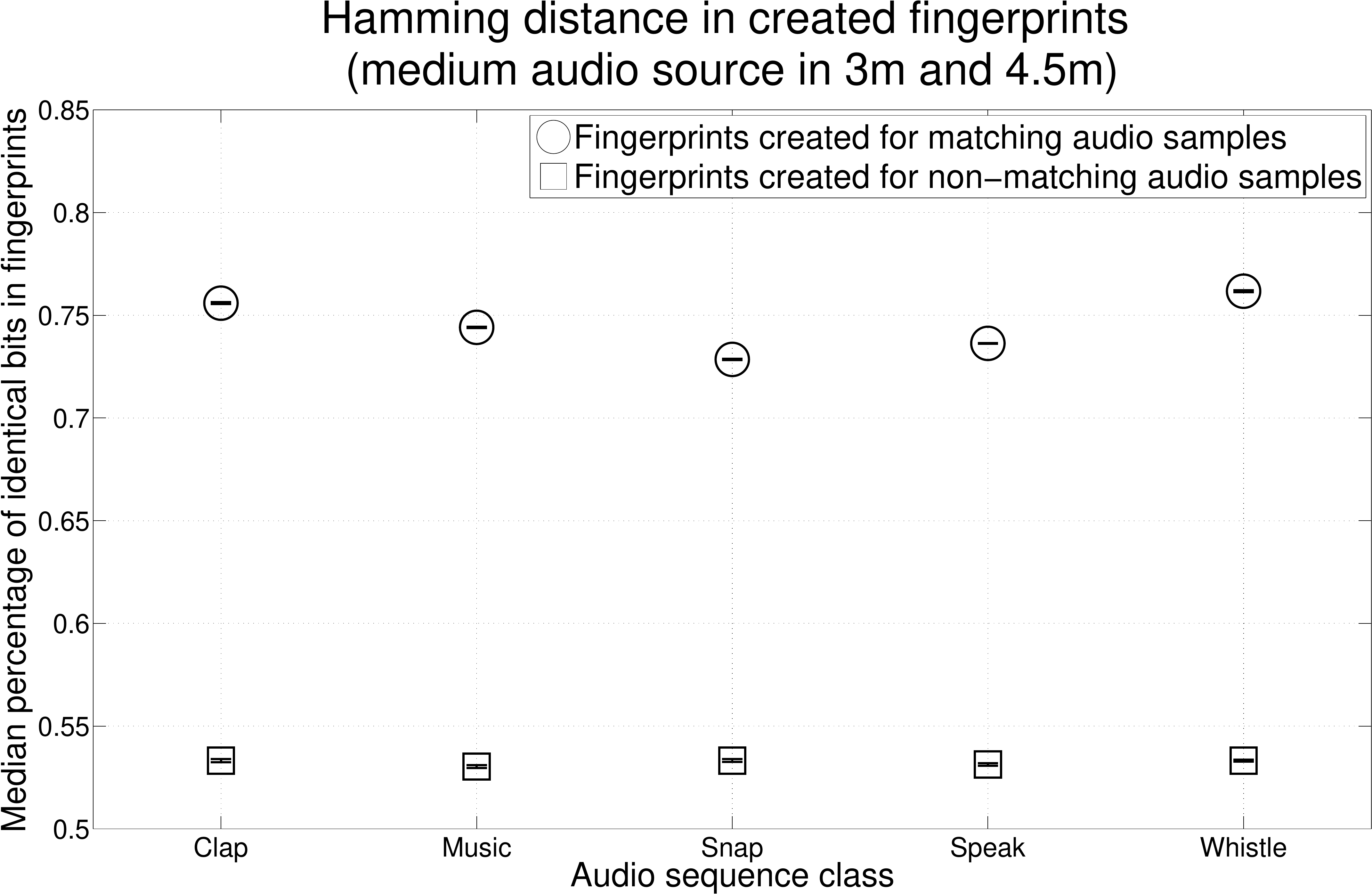}
\label{fig_fifth_case}}
\hfil
\subfloat[Quiet, microphones at 3\,m and 4.5\,m]{\includegraphics[width=2.2in]{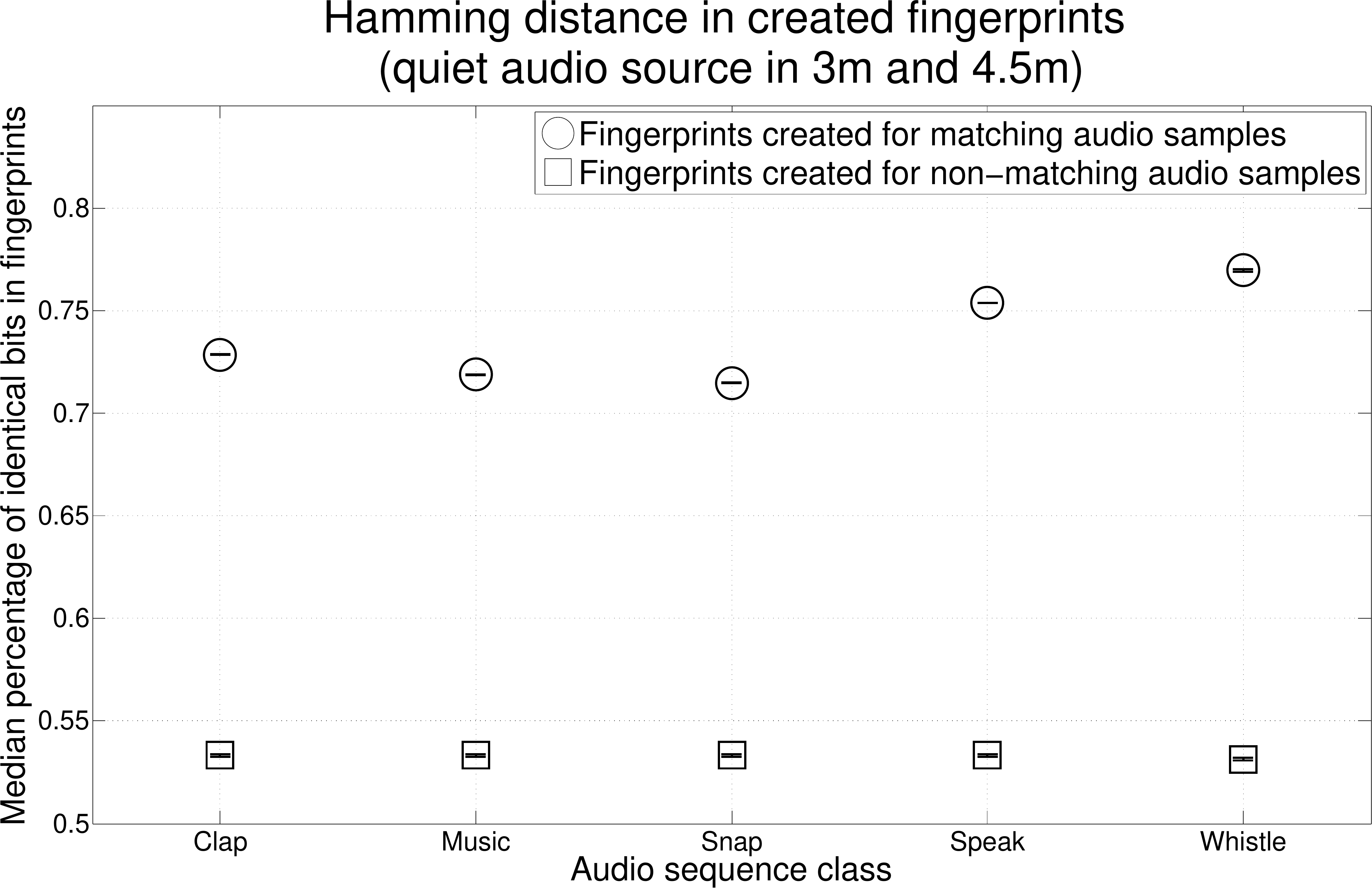}
\label{fig_sixth_case}}}
\centerline{
\subfloat[Loud, microphones at 4.5\,m and 6\,m]{\includegraphics[width=2.2in]{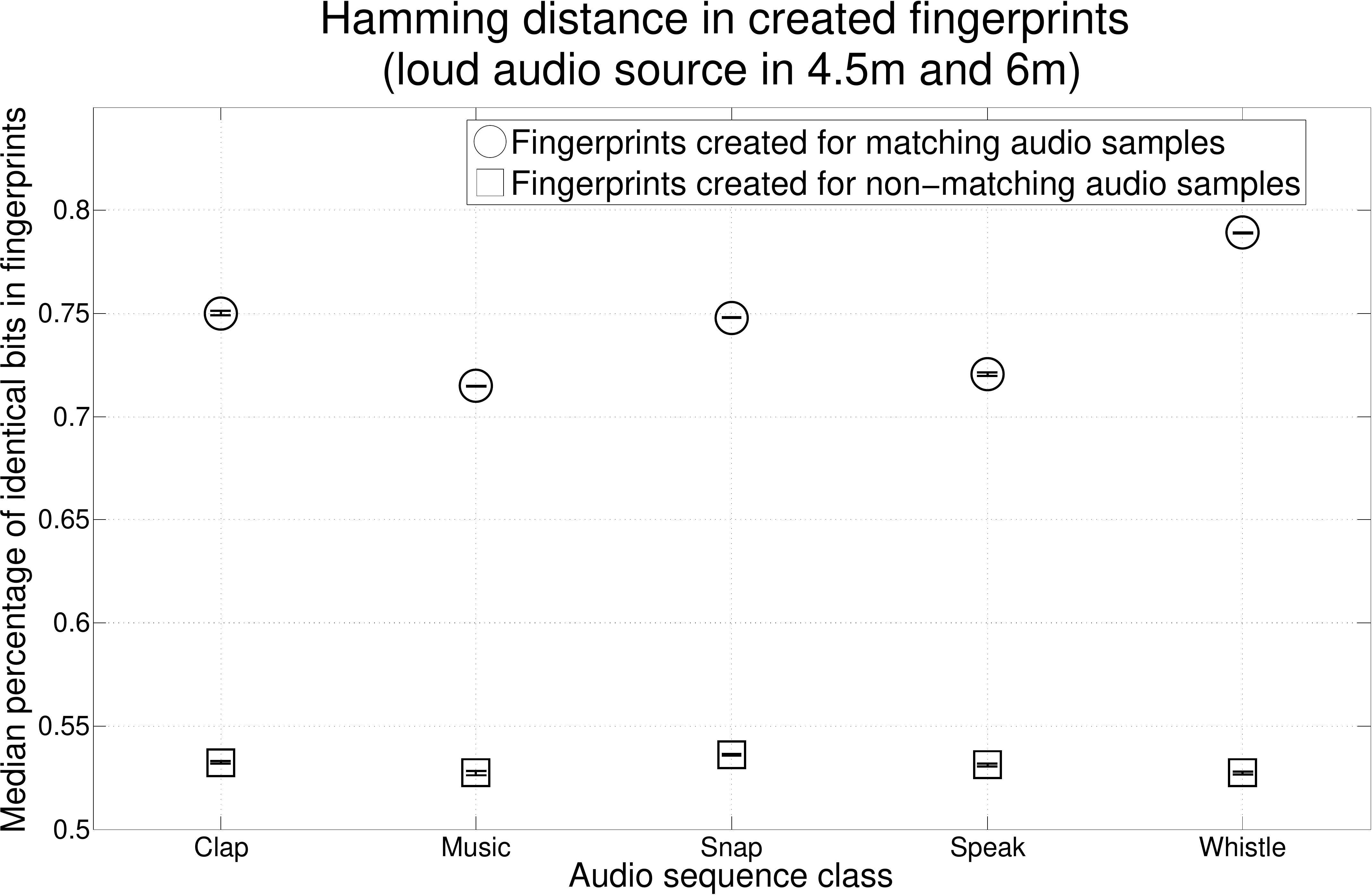}
\label{fig_seventh_case}}
\hfil
\subfloat[Medium, microphones at 4.5\,m and 6\,m]{\includegraphics[width=2.2in]{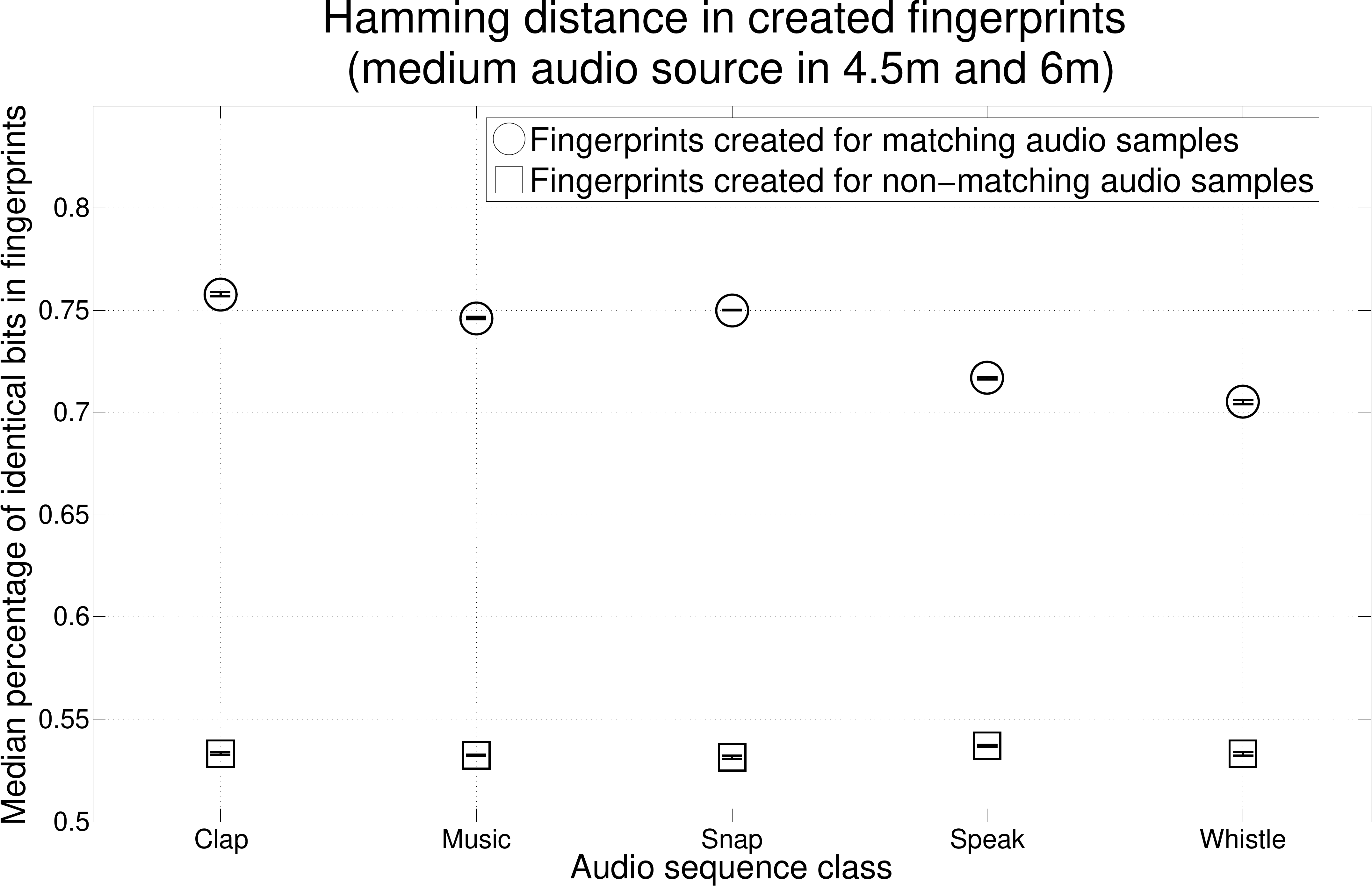}
\label{fig_eigth_case}}
\hfil
\subfloat[Quiet, microphones at 4.5\,m and 6\,m]{\includegraphics[width=2.2in]{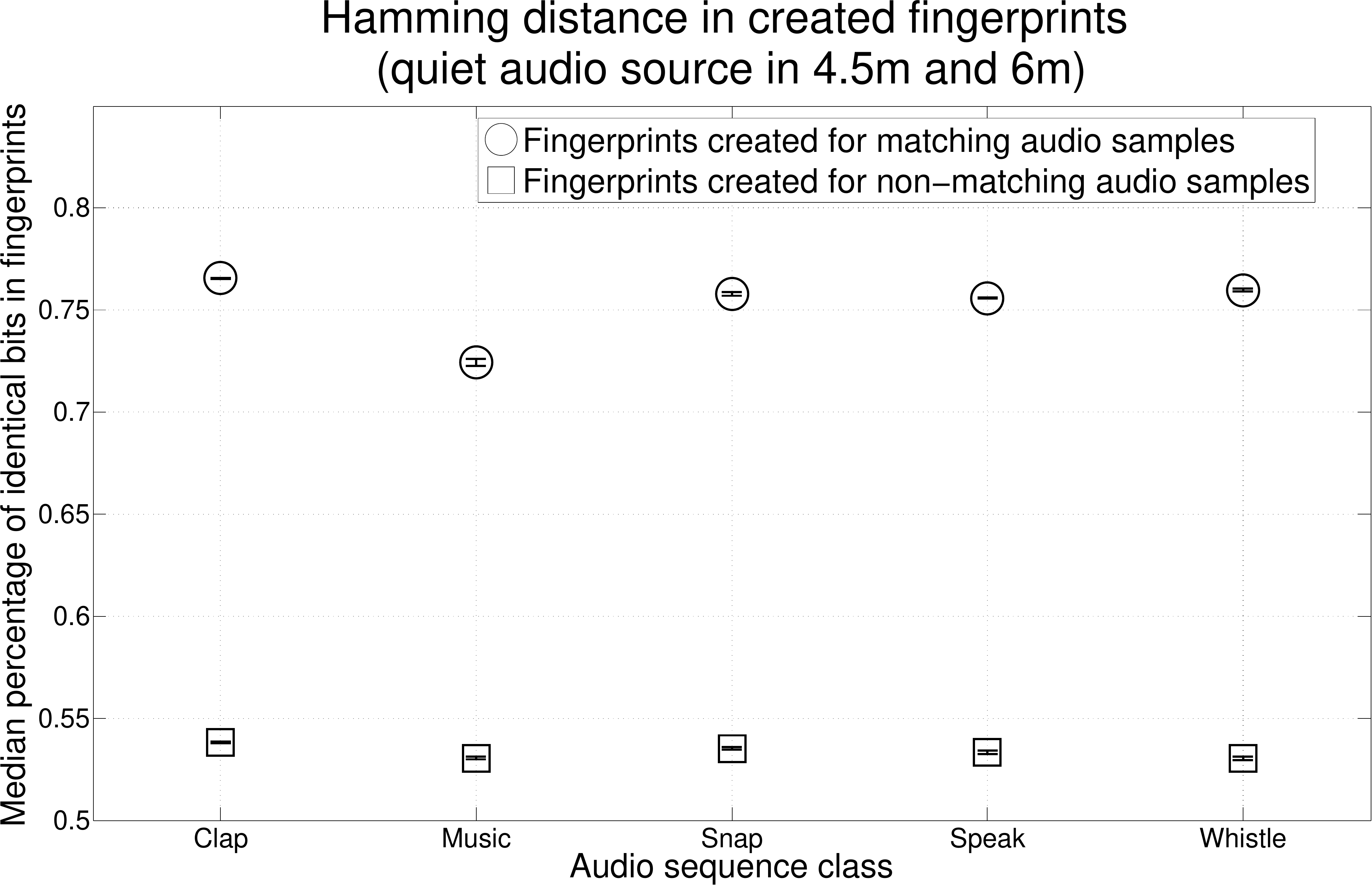}
\label{fig_ninth_case}}}
\centerline{
\subfloat[Loud, microphones at 3\,m and 6\,m]{\includegraphics[width=2.2in]{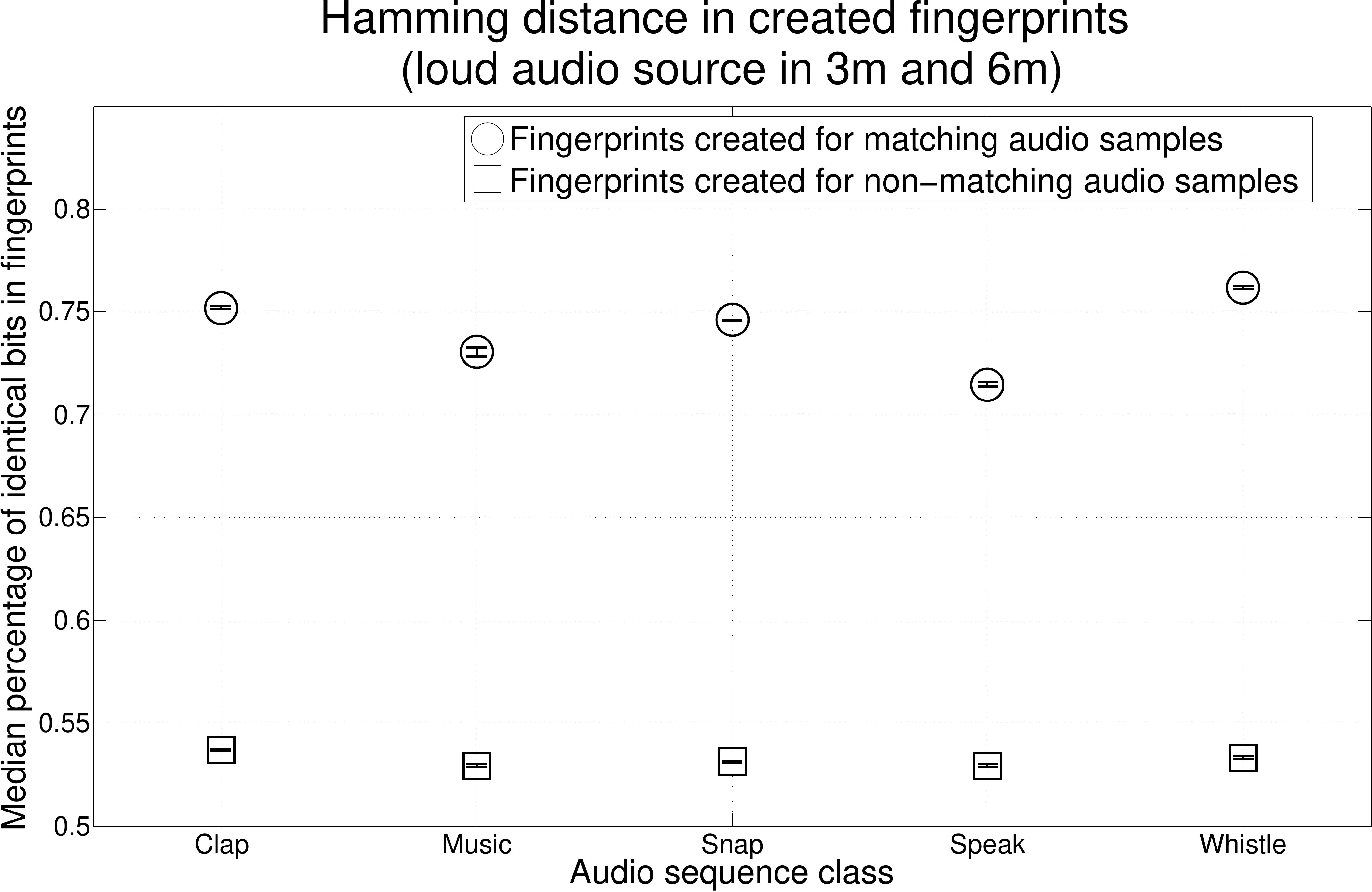}
\label{fig_tenth_case}}
\hfil
\subfloat[Medium, microphones at 3\,m and 6\,m]{\includegraphics[width=2.2in]{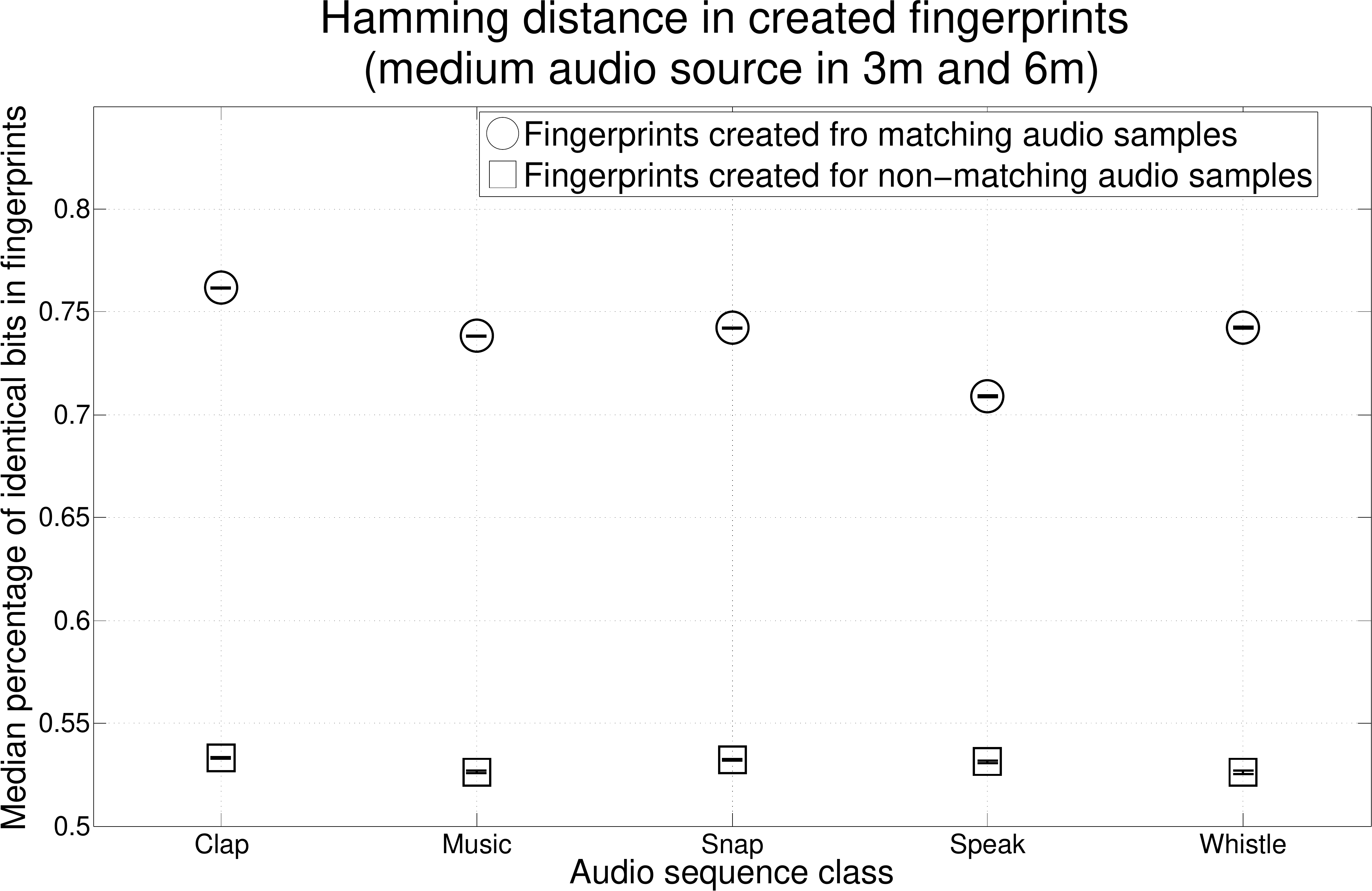}
\label{fig_eleventh_case}}
\hfil
\subfloat[Quiet, microphones at 3\,m and 6\,m]{\includegraphics[width=2.2in]{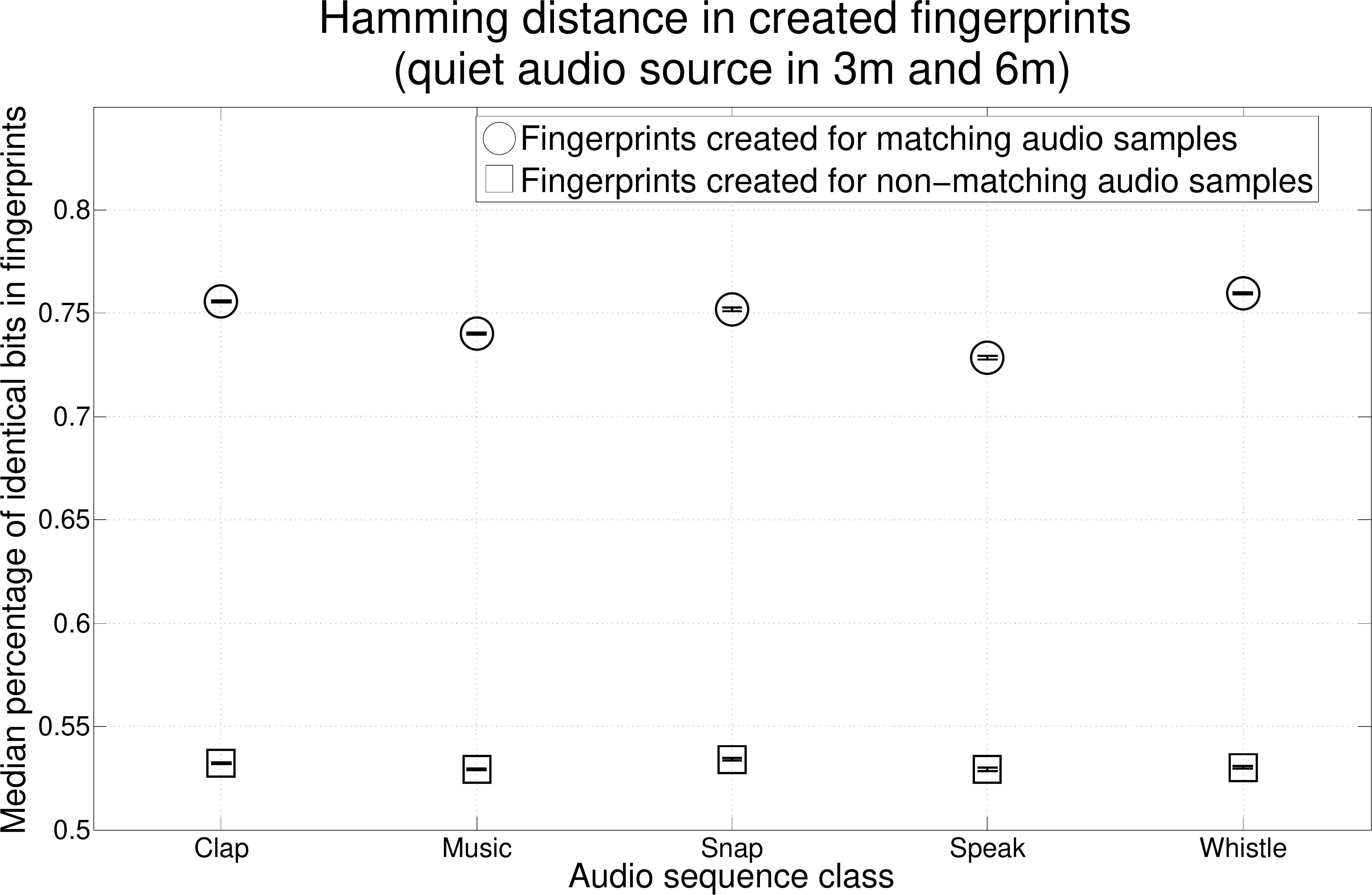}
\label{fig_twelveth_case}}}
\caption{Hamming distance observed for fingerprints created for recorded audio samples at distinct loudness levels and distances between microphones and the audio source {\scriptsize (1536-1233/13/\$31.00 \copyright 2013 IEEE Published by the IEEE CS, CASS, ComSoc, IES, SPS)
}}
\label{figure2}
\end{figure}

First, we observe that the similarity in the fingerprints is significantly higher for simultaneously sampled audio in all cases.
Also, notably, the similarity in the fingerprints of non-simultaneously recorded audio is slightly higher than 50\,\%, which we would expect for a random guess.
The small deviation is a consequence of the monotonous electronic background noise originated by the recording devices consisting of the microphones and the audio chipsets.

Additionally, the distance of the microphones to the audio source has no impact on the similarity of fingerprints.
Similarly, we can not observe a significant effect of the loudness level.
This confirms our expectation since for the fingerprinting approach not the absolute energy on frequency bands but changes in energy over time were considered (cf. section~\ref{section3-1}).
Therefore, changes in the loudness level as, for instance, by altering the distance to the audio source or by changing the volume of the audio, have minor impact on the fingerprints.

Table~\ref{tableAuthentication} depicts the maximum and minimum Hamming distance among all experiments.
\begin{table} 
\renewcommand{\arraystretch}{1.3}
\caption{Percentage of identical bits between fingerprints {\scriptsize (1536-1233/13/\$31.00 \copyright 2013 IEEE Published by the IEEE CS, CASS, ComSoc, IES, SPS)
}}
\label{tableAuthentication}
\centering
  \begin{tabular}{|c||c|c|}
    \hline
    & matching samples & non-matching samples \\
    \hline
    Median   & 0.7617 & 0.5332     \\
    Mean     & 0.7610 & 0.5322     \\
    Variance & 0.0014 & 0.00068342 \\
    Min      & 0.6777 & 0.4414     \\
    Max      & 0.8750 & 0.6484     \\
    \hline
  \end{tabular}
\end{table}     
We observe that one of the comparisons of fingerprints for non-simultaneously recorded audio yielded a maximum similarity of $0.6484$. 
This value is still fairly separated from the minimum bit-similarity observed for fingerprints from simultaneously recorded samples.
Also, this event is very seldom in the $7200$ comparisons since the mean is sharply concentrated around the median with a low variance.
Therefore, by repeating this process for a small number of times, we reduce the probability of such an event to a negligible value.
For instance, only about $3.8$\,\% of the comparisons between fingerprints from non-matching samples have a similarity of more than $0.58$; only $0.4583$\,\% have a similarity of more than $0.6$.
Similarly, only $2.33$\,\% of the comparisons of synchronously sampled audio have a similarity of less than $0.7$.

With these results, we conclude that an authentication based on audio-fingerprints created from synchronised audio samples in identical environmental contexts is feasible.
However, since it is unlikely that the fingerprints match in all bits, it is not possible to utilise the audio-fingerprints directly as a secret key to establish a secure communication channel among devices.
We therefore considered error correcting codes to account for the noise in the fingerprints created.

\subsection{Case-studies}\label{section5}
We implemented the described ambient audio-based secure communication scheme in Python and conducted case-studies in four distinct environments.
The experiments feature differing loudness levels, different background noise figures as well as distinct common situations.
In section~\ref{section5-1}, we observe how the proposed method can establish an ad-hoc secure communication based on audio from ongoing discussions in a general office environment.
Since an adversary able to sneak into the audio context of a given room might be better positioned to guess the secure key, we demonstrate in section~\ref{section5-2} that even for an adversary device that is able to establish a similar dominant audio context in a different room by listening to the same FM-radio-channel, the gap in the created fingerprints is significant.
In these two experiments, we utilised artificial audio sources in a sense that they were specifically placed to create the ambient audio context.
In section~\ref{section5-3} and section~\ref{section5-4} we describe experiments in common environments where ambient audio was utilised exclusively.
In section~\ref{section5-3} we placed devices at distinct locations in a canteen and studied the success probability based on the distance between devices. 
In section~\ref{section5-4} we study the feasibility of establishing a secure communication channel with road-traffic as background noise.
Figure~\ref{figure3} summarises all settings considered.
\begin{figure}
     \centering
      \subfloat[Office setting. Devices and speakers located at distinct positions.]{\includegraphics[width=.4\textwidth,height=6cm]{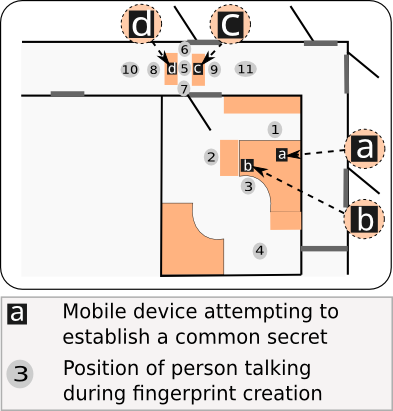}\label{figure3-1}}\hfill
     \subfloat[Office setting. Synchronised dominant audio source established in both rooms via FM radios.]{\includegraphics[width=.4\textwidth,height=6cm]{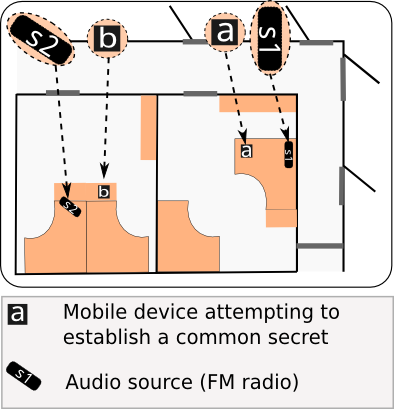}\label{figure3-2}}
     
     \subfloat[Canteen setting. Devices located at various distances. No dominant noise source.]{\includegraphics[width=.4\textwidth,height=6cm]{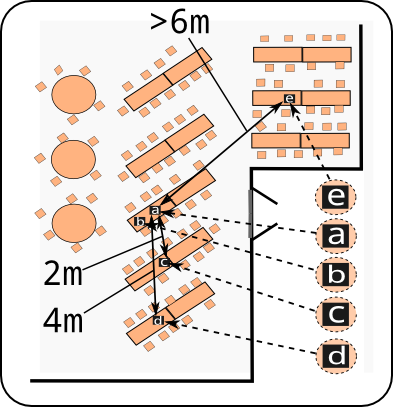}\label{figure3-3}}\hfill
     \subfloat[Roadtraffic setting. Devices arranged alongside a road. No dominant noise source.]{\includegraphics[width=.4\textwidth,height=6cm]{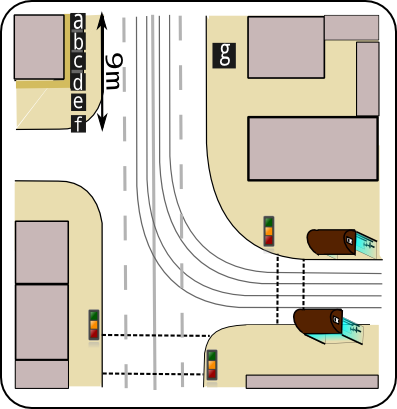}\label{figure3-4}}
     \caption{Environmental settings of the case-studies conducted. {\scriptsize (1536-1233/13/\$31.00 \copyright 2013 IEEE Published by the IEEE CS, CASS, ComSoc, IES, SPS)
}}
     \label{figure3} 
\end{figure}
To capture audio we utilised the build-in microphones of the computers.
The only exception is the reference scenario~\ref{figure3-1} in which simple off-the-shelf external microphones have been utilised.
For both devices, the manufacturer and audio device types differed.
Table~\ref{tableConfigurationCaseStudy} details further configuration of the scenarios conducted and the hardware utilised.
\begin{table}
\renewcommand{\arraystretch}{1.1}
	\centering
		\caption{Configuration of the four scenarios considered {\scriptsize (1536-1233/13/\$31.00 \copyright 2013 IEEE Published by the IEEE CS, CASS, ComSoc, IES, SPS)
}}
	 \begin{small}
	\begin{tabular}{r|ccc}
		\textbf{Microphones (external)}&\\\hline
		Impedance &\multicolumn{3}{r}{$\leq 22\,\mathrm{k}\Omega$}\\
		Current consumption&\multicolumn{3}{r}{$\leq 0.5\,\mathrm{mA}$}\\
		Frequency response &\multicolumn{3}{r}{$100\,\mathrm{Hz} \sim 16\,\mathrm{KHz}$}\\
		Sensitivity &\multicolumn{3}{r}{$ -38\,\mathrm{dB} \pm 2\,\mathrm{dB}$}\\
		&\\
		\textbf{Microphones (internal)}&\\\hline
		Device A &Intel G45 DEVIBX\\
		Device B&Intel 82801I\\
	\end{tabular}
	\end{small}
	\label{tableConfigurationCaseStudy}
\end{table}

\subsubsection{Office environment}\label{section5-1}
In our first case-study, we position two laptops in an office environment.
Ambient audio was originated from individuals speaking inside or outside of the office room.
We conducted several sets of experiments with differing positions of laptop computers and audio sources as depicted in figure~\ref{figure3-1}.
We distinguish four distinct scenarios
\begin{description}
     \item[\ref{figure3-1}$_1$] Both devices inside the office at locations a and b. 
     1-2 Individuals speaking at locations 1 to 4.
     \item[\ref{figure3-1}$_2$] One device inside and one outside the office in front of the open office door at locations a and c. 
     1-2 Individuals speaking at locations 1 and 5.
     \item[\ref{figure3-1}$_3$] Both devices in the corridor in front of the office at locations c and d. 
     1-2 Individuals speaking at locations 5 to 11.
     \item[\ref{figure3-1}$_4$] One device inside and one outside the office in front of the closed office door at locations a and c. 
     1-2 Individuals speaking (damped but audible behind closed door) at locations 1 and\,5.
 \end{description}

In all cases the devices were synchronised over NTP.
For each synchronisation, one device indicated at which point in time it would initiate audio recording.
Both devices then sample ambient audio at that time and create a common key following the protocol described in section~\ref{section3-1}.
For each scenario the key synchronisation process was repeated 10~times with the persons located at different locations.
From these persons, either person~1, person~2 or both were talking during the synchronisation attempts in order to provide the audio context.

The settings \ref{figure3-1}$_1$ and \ref{figure3-1}$_3$ represent the situation of two friendly devices willing to establish a secure communication channel.
The setting \ref{figure3-1}$_2$ could constitute the situation in which a person passing by is accidentally witnessing the communication and part of the audio context.
In setting \ref{figure3-1}$_4$, the communication partners might have closed the office door intentionally in order to keep information secure from persons outside the office.

In scenario \ref{figure3-1}$_1$, where both devices share the same audio context a fraction of $0.9$ of all synchronisation attempts have been successful.
Also, for scenario \ref{figure3-1}$_3$, the fraction of successful synchronisation attempts was as high as $0.8$.
Consequently, when both devices are located in the same audio context, a successful synchronisation is possible with high probability.

For scenario \ref{figure3-1}$_2$, where the device outside the ajar door could partly witness the audio context, we had a success probability of $0.4$.
Although this means that less than every second approach was successful, this is clearly not acceptable in most cases.
Still, this low success probability it is remarkable since the person speaking in the office or on the corridor was clearly audible at the respective other location.

In scenario \ref{figure3-1}$_4$, however, when the audio context was separated by the closed door, no synchronisation attempt was successful.
Remarkably in this case, the person speaking was, although hardly comprehensible, still audible at the other side of the door. 

Finally, we attempted to establish a synchronisation in the scenarios \ref{figure3-1}$_1$, \ref{figure3-1}$_2$ and \ref{figure3-1}$_3$ when only background noise was present.
This means that no sound was emitted from a source located in the same location as one of the devices.
Some distant voices and indistinguishable sounds could occasionally be observed.
After a total of twelve tries in these three scenarios, not a single one resulted in a successful synchronisation between devices.
We conclude that a dominant noise source or at least more dominant background noise needs to be present in the same physical context as the devices that want to establish a common key.

\subsubsection{Context replication with FM-radio}\label{section5-2}
A straightforward security attack for audio-based encryption could be for the attacker to extract information about the audio context and use this in order to guess the secret key created.
We studied this threat by trying to generate a secret key between two devices in different rooms but with similar audio contexts.
In particular, we placed two FM-radios, tuned to the same frequency in both rooms (cf. figure~\ref{figure3-2}).

The audio context was therefore dominated by the synchronised music and speech from the FM-radio channel.
No other audio sources have been present in the rooms so that additional background noise was negligible.
We conducted two experiments in which the devices were first located in the same room and then in different rooms with the same distance to the audio source.
The loudness level of the audio source was tuned to about 50\,dB in both rooms.
Figure~\ref{figure4} depicts the median bit-similarity achieved when the devices were placed in the same room and in different rooms respectively.
\begin{figure}
     \centering
     \includegraphics[width=8.9cm]{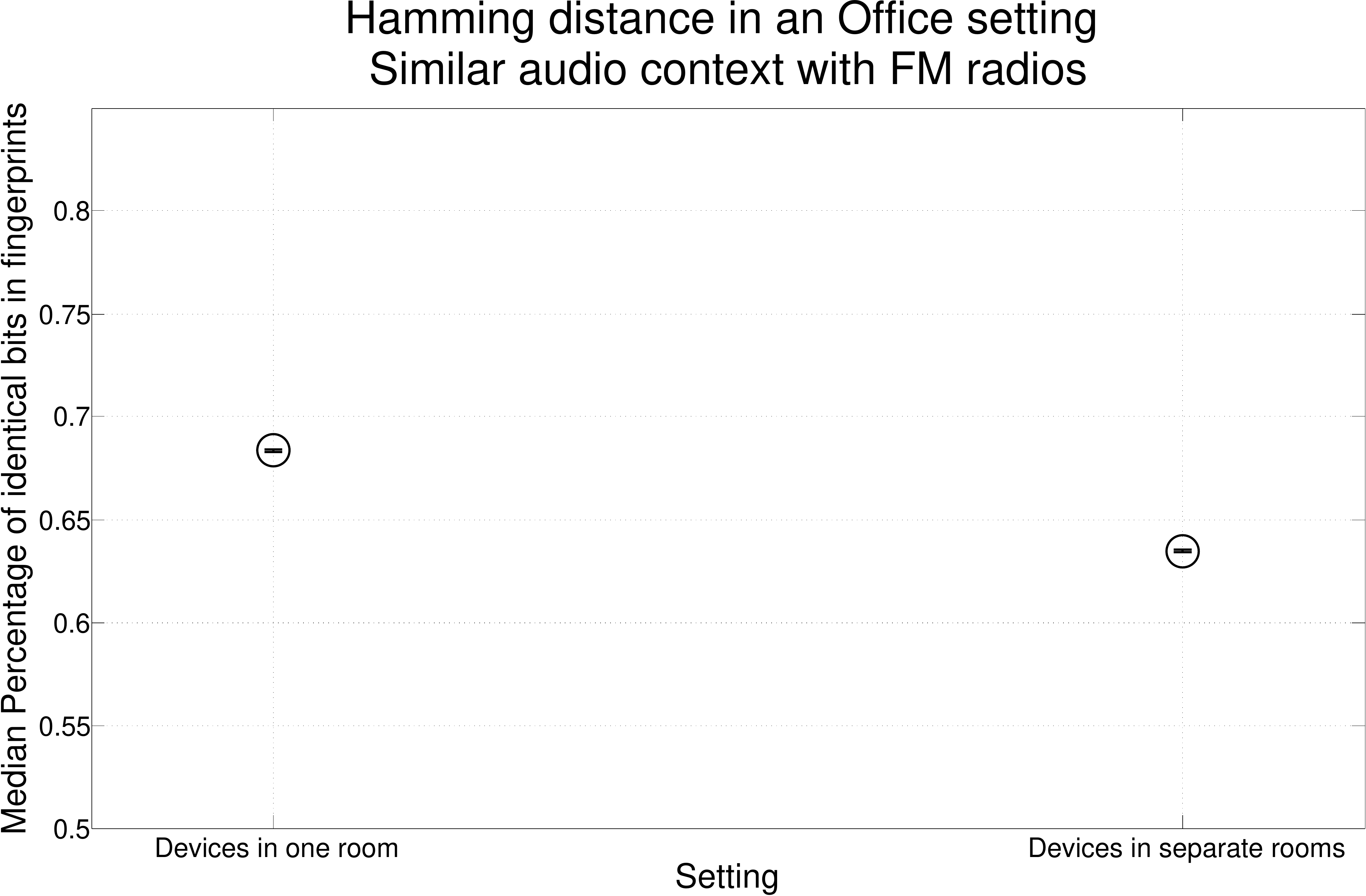}
     \caption{Median percentage of bit errors in fingerprints generated by two mobile devices in an office setting. The audio context was dominated by an FM radio tuned to the same channel. {\scriptsize (1536-1233/13/\$31.00 \copyright 2013 IEEE Published by the IEEE CS, CASS, ComSoc, IES, SPS)}}
     \label{figure4}
\end{figure}

We observe that in both cases the variance in the bit errors achieved is below $0.01$\,\%.
When both devices are placed in the same room, the median Hamming distance between fingerprints is only $31.64$\,\%.
We account this high similarity and the low variance to the fact that background noise was negligible in this setting since the FM-radio was the dominant audio source.

When the devices are placed in different rooms, the variance in bit error rates is still low with $0.008$\,\%.
The median Hamming distance rose in this case to $36.52$\,\%.	

Consequently, although the dominant audio source in both settings generated identical and synchronised content, the Hamming distance drops significantly when both devices are in an identical room.
With sufficient tuning of the error correction method conditioned on the Hamming distance, an eavesdropper can be prevented from stealing the secret key even though information on the audio context might be leaking.

\subsubsection{Canteen environment}\label{section5-3}
We studied the accuracy of the approach in the canteen of the TU Braunschweig (cf.~figure~\ref{figure3-3}).
At different tables, laptop computers have been placed.
For each configuration we conducted 10 attempts to establish a unique key based on the fingerprints.
We conducted all experiments during 11:30 and 14:00 on a business day in a well populated canteen. 
The ambient noise in this experiment was approximately at 60\,dB. 
Apart from the audible discussion on each table, background noise was characterised by occasional high pitches of clashing cutlery.

Figure~\ref{figure5} depicts the results achieved.
The figure shows the median percentage of bit errors between the fingerprints generated by both devices.
\begin{figure}
     \centering
     \includegraphics[width=8.9cm]{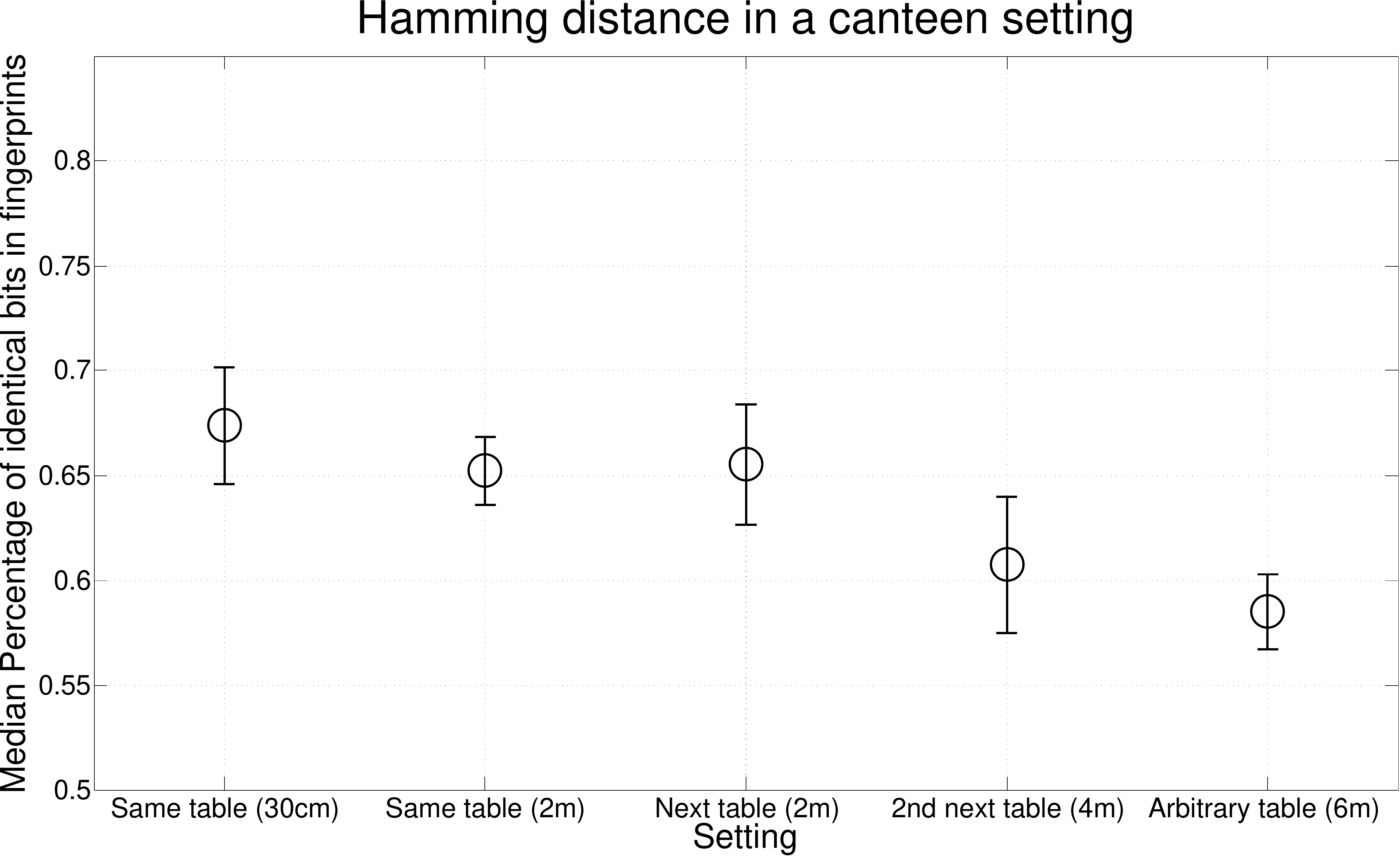}
     \caption{Median percentage of bit errors in fingerprints generated by two mobile devices in a canteen environment. {\scriptsize (1536-1233/13/\$31.00 \copyright 2013 IEEE Published by the IEEE CS, CASS, ComSoc, IES, SPS)}}
     \label{figure5}
\end{figure}

We observe that generally the percentage of identical bits in the fingerprint decreases with increasing distance. 
With about 2\,m distance the percentage of identical bits is still quite similar to the similarity achieved when devices are only 30\,cm apart.
This is also true when one of the devices is placed at the next table.
However, with a distance of about 4 meters and above, the percentage of bit errors are well separated so that also the error correction could be tuned such that a generation of a unique key is not feasible at this a distance.

\subsubsection{Outdoor environment}\label{section5-4}
In this instrumentation the two computers were located at the side of a well trafficked road.
The study has been conducted during the rush hour between 17:00 and 19:00 at a regular working day.
The road was frequented by pedestrians, bicycles, cars, lorries and trams.
The data was measured not far off a headlight so that traffic occasionally stopped with running motors in front of the measurements.
The loudness level was about $60$\,dB for both devices.
The setting is depicted in figure~\ref{figure3-4}.
We gradually increased the distance among devices.
Devices have been placed with a distance between their microphones of 0.5\,m, 3\,m, 5\,m, 7\,m and 9\,m at one side of the road.
Additionally, for one experiment devices are placed at opposite sides of the road.
For each configuration 10 to 13 experiments have been conducted.
The results are depicted in figure~\ref{figure6}.
\begin{figure}
     \centering
     \includegraphics[width=8.9cm]{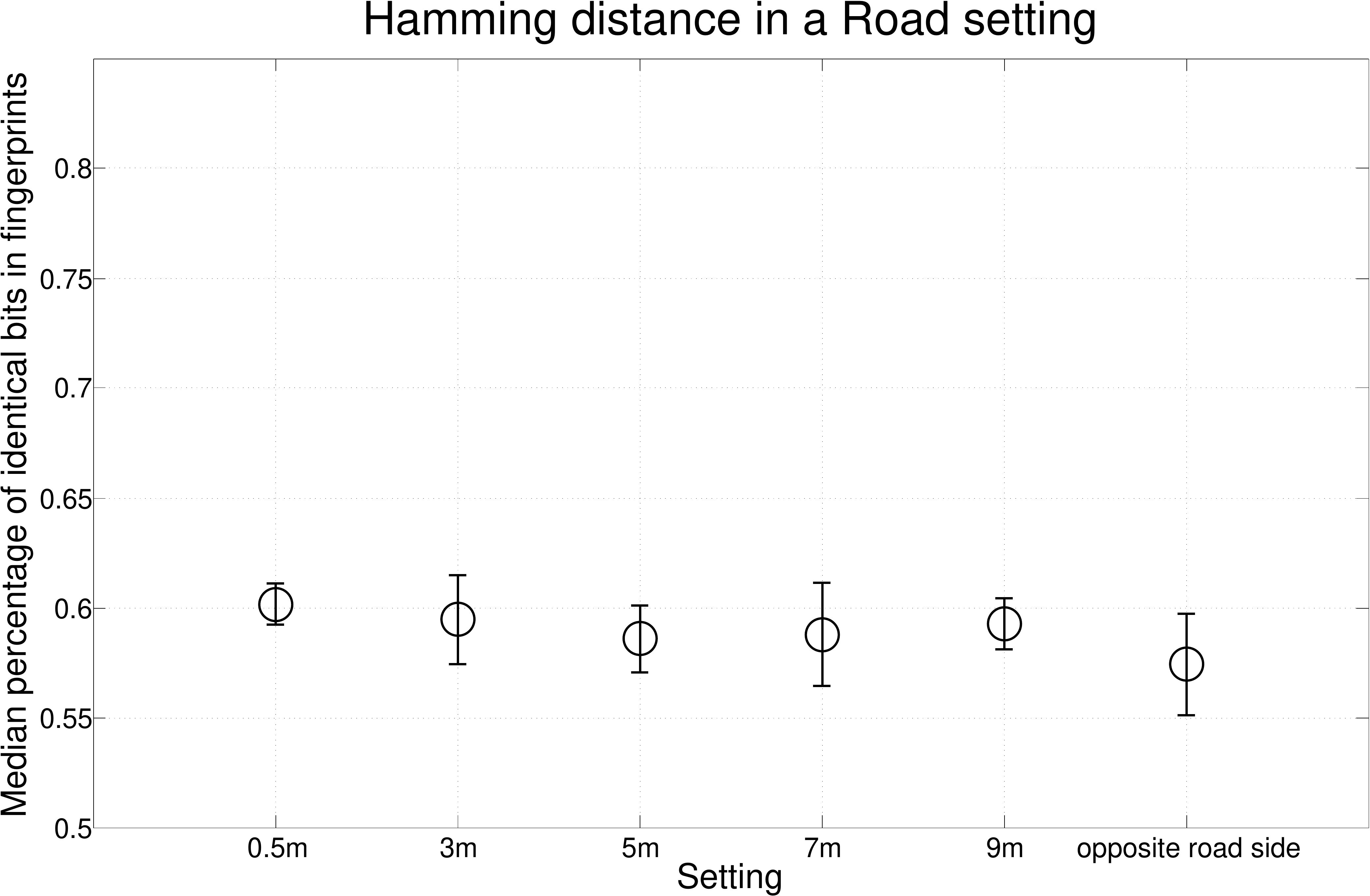}
     \caption{Median percentage of bit errors in fingerprints from two mobile devices beside a heavily trafficked road. {\scriptsize (1536-1233/13/\$31.00 \copyright 2013 IEEE Published by the IEEE CS, CASS, ComSoc, IES, SPS)}}
     \label{figure6}
\end{figure}
The figure depicts the median Hamming distance and variance for the respective configurations applied.

Not surprisingly, we observe that the Hamming distance between fingerprints generated by both devices is lowest when devices are placed next to each other.
With increasing distance, the Hamming distance increases slightly but then stays similar also for greater distances.
 
At the opposite side of the road, however, the Hamming distance drops more significantly.
When both devices are at the same side of the road, the probability to guess the secret key is high even for greater distances between the devices.
We believe that this property is attributable to the very monotonic background noise generated by the vehicles on the road.
The audio-context is therefore similar also in greater distances.

Only when one of the devices is located at the opposite side of the road, a more significant distinction between the generated fingerprints is possible.
This may account to the different reflection of audio off surrounding buildings and to the fact that vehicles on the other lane generate a different dominant audio footprint.

Generally, these results suggest that audio-based key generation is hardly feasible in this scenario.
Audio-based generation of secret keys is not well suited in an environment with very monotonic and unvaried background noise.
Although a light protection from intruders on the different side of the road is possible, the radius in which similar fingerprints are generated on one side of the road is unacceptably high.

\subsection{Entropy of fingerprints}\label{section6}
Although these results suggest that it is unlikely for a device in another audio context to generate a fingerprint which is sufficiently similar, an active adversary might analyse the structure of fingerprints created to identify and explore a possible weakness in the encryption key.
Such a weakness might be constituted by repetitions of subsequences or by an unequal distribution of symbols.
A message encrypted with a key biased in such a way may leak more information about the encrypted message than intended.

We estimated the entropy of audio-fingerprints generated for audio-sub-sequences by applying statistical tests on the distribution of bits.
In particular, we utilised the dieHarder~\cite{statistical_Brown_0000} set of statistical tests.
This battery of tests calculates the p-value of a given random sequence with respect to several statistical tests.
The p-value denotes the probability to obtain an input sequence by a truly random bit generator~\cite{Statistical_Kuiper_1962}. 
All tests are applied to a set of fingerprints of 480 bits length.
We utilised all samples obtained in section~\ref{section4} and section~\ref{section5}.

From 7490 statistical-test-batches consisting of 100 repeated applications of one specific test each, only 173, or about 2.31\%\ resulted in a p-value of less than 0.05\footnote{All results are available at http://www.ibr.cs.tu-bs.de/users/sigg/StatisticalTests/TestsFingerprints\_110601.tar.gz}.  
Each specific test was repeated at least 70 times.
The p-values are calculated according to the statistical test of Kuiper~\cite{Statistical_Kuiper_1962,Statistical_stephens_1965}.

\begin{figure}
	\centering

	\subfloat[Proportion of sequences from an indoor laboratory environment passing a test]{\includegraphics[width=0.48\textwidth, height=6cm]{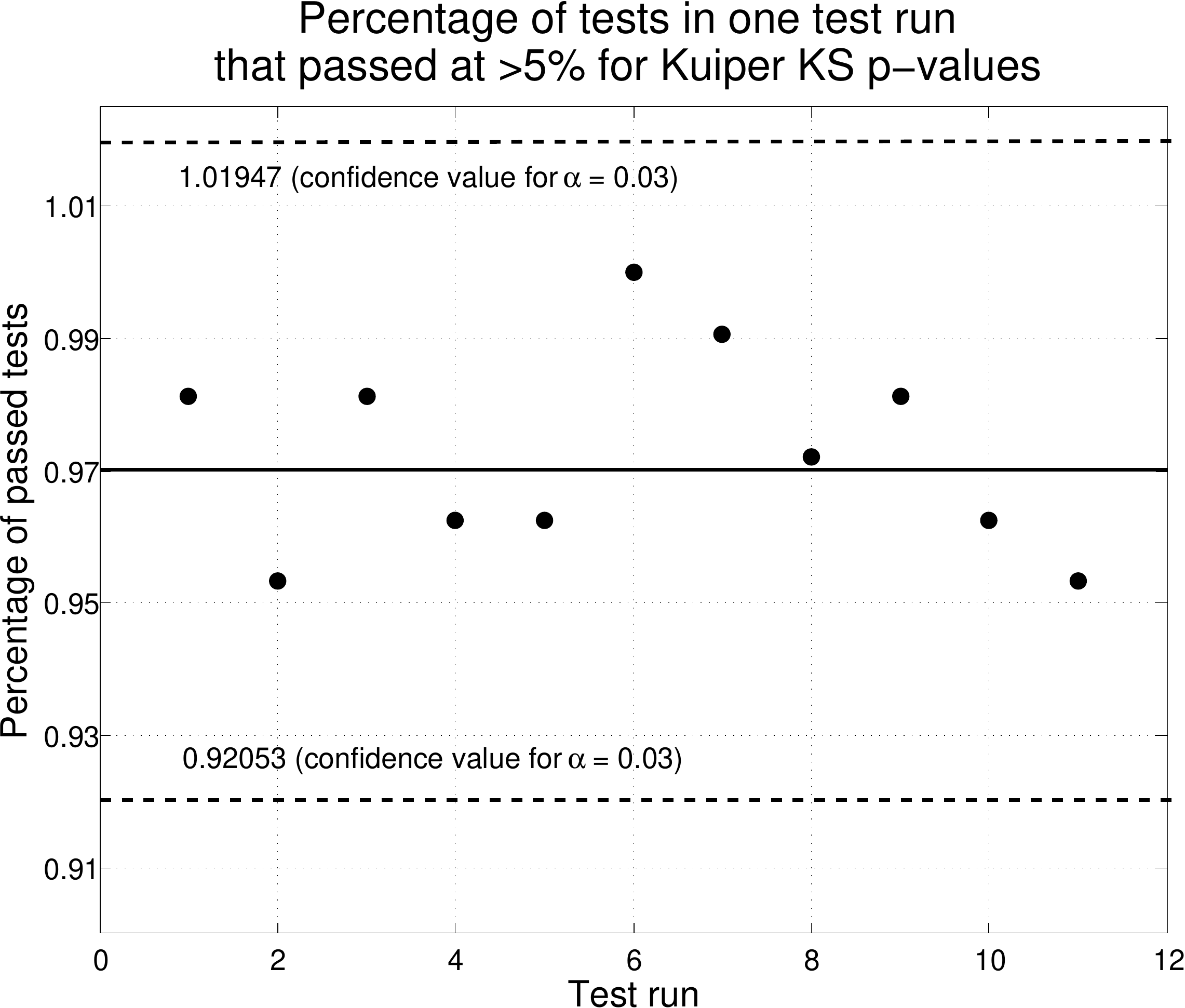}\label{figure7-1}}\hfill 
	\subfloat[Proportion of sequences from various outdoor environments passing a test]{\includegraphics[width=0.48\textwidth, height=6cm]{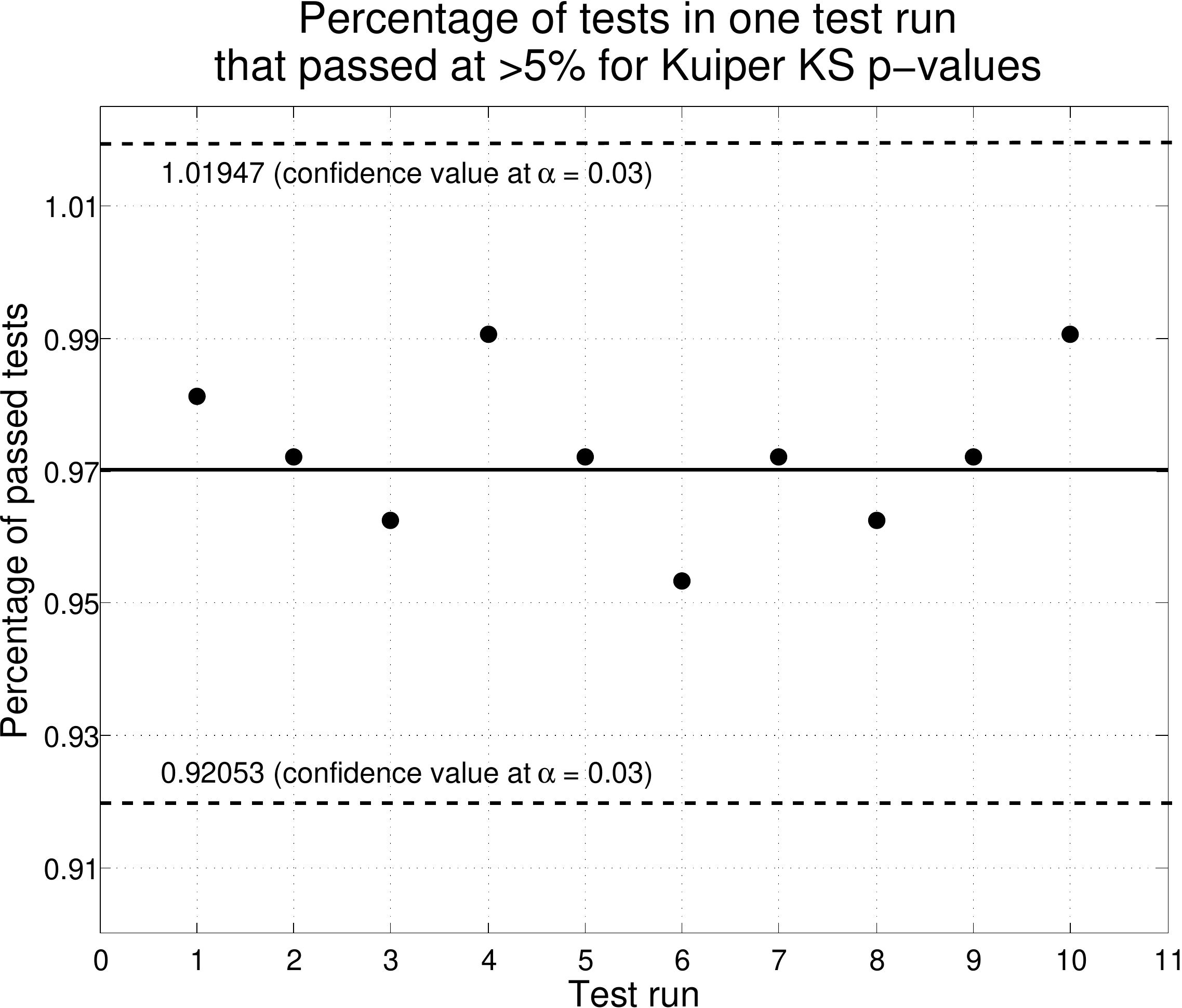}\label{figure7-2}}
	
	\subfloat[Proportion of sequences from all but music samples passing a test]{\includegraphics[width=0.48\textwidth, height=6cm]{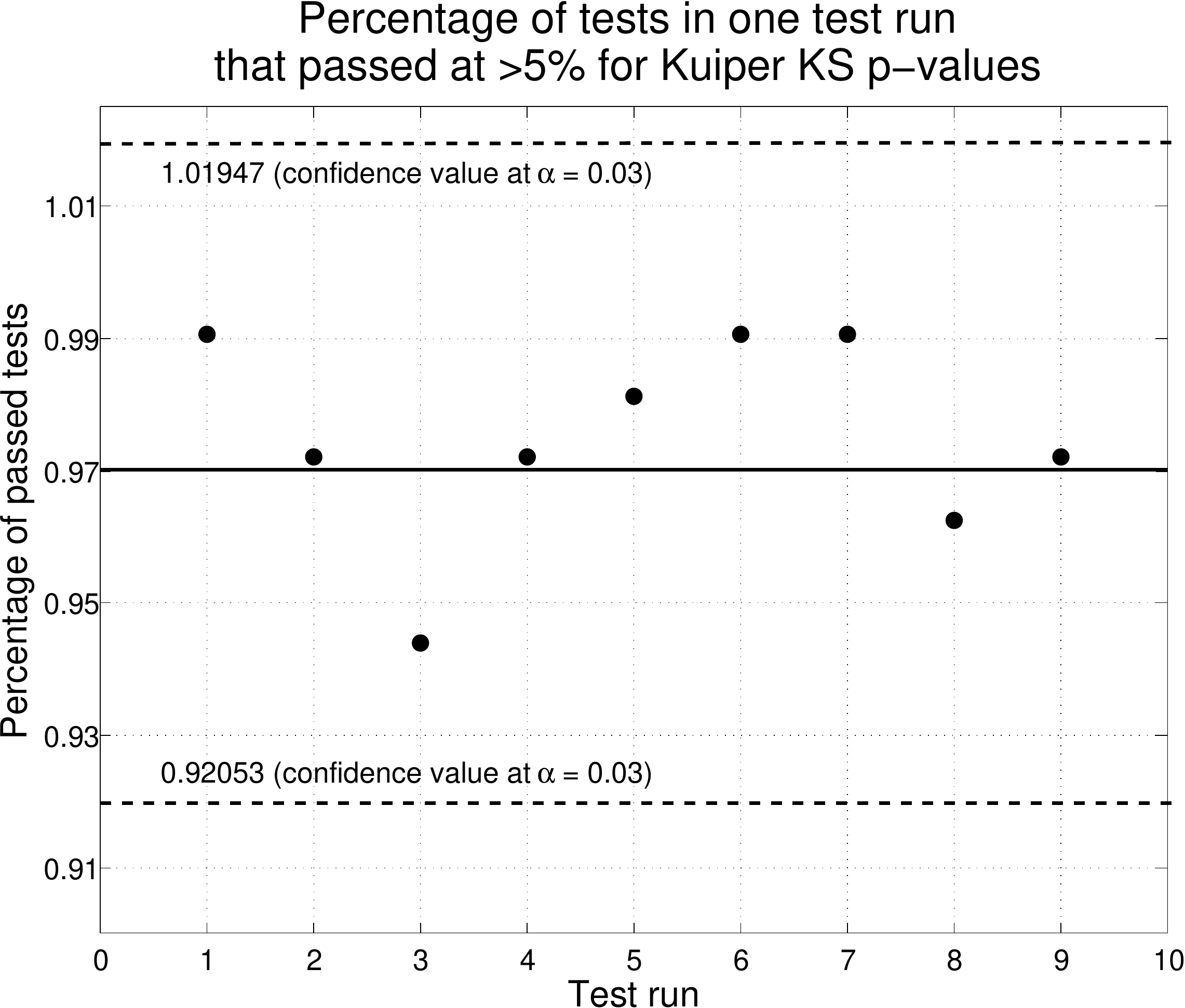}\label{figure7-3}}\hfill
	\subfloat[Proportion of sequences belonging to a specific audio class passing a test]{\includegraphics[width=0.48\textwidth, height=6cm]{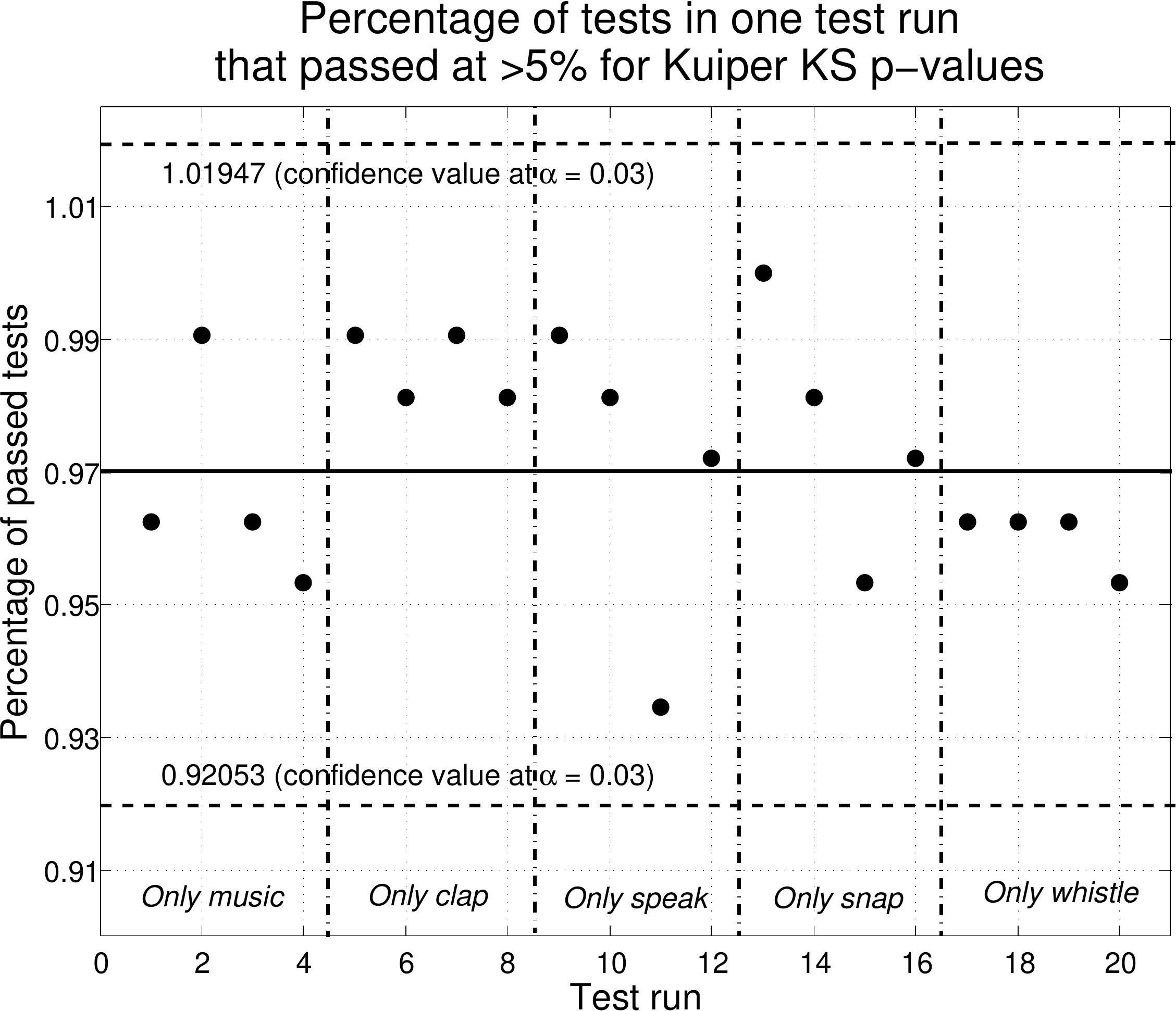}\label{figure7-4}} 
	\caption{Illustration of P-Values obtained for audio-fingerprints by applying the DieHarder battery of statistical tests. {\scriptsize (1536-1233/13/\$31.00 \copyright 2013 IEEE Published by the IEEE CS, CASS, ComSoc, IES, SPS)}}
	\label{figure7}
\end{figure}
Figure~\ref{figure7} depicts for all test-series conducted the fraction of tests that did not pass a sequence of 100 consecutive runs at $>5\%$ for Kuiper KS p-values~\cite{Statistical_Kuiper_1962} for all 107 distinct tests in the DieHarder battery of statistical tests. 
Generally, we observe that for all test-runs conducted, the number of tests that fail is within the confidence interval with a confidence value of $\alpha=0.03$.
The confidence interval was calculated for $m=107$ tests as 
\begin{equation}
     1-\alpha\pm3\cdot\sqrt{\frac{(1-\alpha)\cdot\alpha}{m}}.     
\end{equation}
Alternatively, we could not observe any distinction between indoor and outdoor settings (cf. figure~\ref{figure7-1} and figure~\ref{figure7-2}) and conclude that also the increasing noise figure and different hardware utilised \footnote{Overall, the microphones utilised (2~internal, 2~external) were produced by three distinct manufacturers} does not impact the test results.
Since music might represent a special case due to its structured properties and possible repetitions in an audio sequence, we considered it separately from the other samples.
We could not identify a significant impact of music on the outcome of the test results (cf.~figure~\ref{figure7-3}).

Additionally, we separated audio samples of one audio class and used them exclusively as input to the statistical tests.
Again, there is no significant change for any of the classes (cf.~figure~\ref{figure7-4}).

We conclude that we could not observe any bias in fingerprints based on ambient audio.
Consequently, the entropy of fingerprints based on ambient audio can be considered as high.
An adversary should gain no significant information from an encrypted message eavesdropped.

\subsection{Conclusion}\label{section7}
We have studied the feasibility to utilise contextual information to establish a secure communication channel among devices.
The approach was exemplified for ambient audio and can be similarly applied to alternative features or context sources.
The proposed fuzzy-cryptography scheme is adaptable in its noise tolerance through the parameters of the error correcting code utilised and the audio sample length.

In a laboratory environment, we utilised sets of recordings for five situations at three loudness levels and four relative positions of microphones and audio source.
We derived in 7500 experiments the expected Hamming distance among audio-fingerprints.
The fraction of identical bits is above $0.75$ for fingerprints from the same audio context and below $0.55$ otherwise.
This gap in the Hamming distance can be exploited to generate a common secret among devices in the same audio context.
We detailed a protocol utilising fuzzy-cryptography schemes that does not require the transmission of any information on the secure key.
The common secret is instead conditioned on fingerprints from synchronised audio-recordings.
The scheme enables ad-hoc and unobtrusive generation of a secure channel among devices in the same context.
We conducted a set of common statistical tests and showed that the entropy of audio-fingerprints based on energy differences in adjacent frequency bands is high and sufficient to implement a cryptographic scheme.

In four case-studies, we verified the feasibility of the protocol under realistic conditions.
The greatest separation between fingerprints from identical and non-identical audio-contexts was observed indoor with low background noise and a single dominant audio source.
In such an environment we could distinguish devices in the same and in different audio contexts.
It was even possible to clearly identify a device that replicated dominant audio from another room with an equally tuned FM-radio at similar loudness level.

In a case-study conducted in a crowded canteen environment, we observed that the synchronisation quality was generally impaired due to the absence of a dominant audio source.
However, it was still possible to establish a privacy area of about 2\,m inside which the Hamming distance of fingerprints was distinguishably smaller than for greater distances.
The worst results have been obtained in a setting conducted beside a heavily trafficked road.
In this case, when the noise component becomes dominant and considerably louder, the synchronisation quality was further reduced.
Additionally, due to the increased loudness level, a similar synchronisation quality was possible also at distances of about 9\,m.
We conclude that in this scenario, a secure communication channel based purely on ambient audio is hard to establish.

We claim that the synchronisation quality in scenarios with more dominant noise components can be further improved with improved features and fingerprint algorithms.
Currently, most ideas are lent from fingerprinting algorithms and features designed to distinguish between music sequences.
Although algorithms have been adapted to better capture characteristics of ambient audio, we believe that features and fingerprint generation to classify ambient audio might be further improved.
Additionally, the consideration of additional contextual features such as light or RF-channel-based should improve the robustness of the presented approach.

In our implementation we faced difficulties to achieve sufficiently accurate (in the order of few milliseconds) time-synchronisation among wireless devices.
In our current studies we tested several sample windows of NTP-synchronised recordings in order to achieve a feasible implementation on standard hardware.
However, a more exact time synchronisation would further reduce the accuracy and computational complexity of the approach.

\subsection*{APPENDIX A: Fingerprint creation}\label{sectionFingerprintingImplementation}
The fingerprinting method implemented is detailed in algorithm~\ref{algorithm:audio_fingerprinting}.
\begin{algorithm}
\SetKwFunction{FFT}{FFT}
\SetKwFunction{Abs}{Abs}
\SetKwFunction{HW}{HW}
\SetKwFunction{Flatten}{Flatten}
\KwIn{Audio sequence $S$ with sample rate $r$}
\KwData{$l$: length of $S$ in seconds, $n$: number of frames, $m$: number of frequency bands}
\KwOut{fingerprint $f$ as bit sequence}
\BlankLine
\Begin{
  $d \leftarrow r \cdot \frac{l}{n}$ \tcp*[l]{length of one frame}
  $\mathcal{F} \longleftarrow \{F_0, \dots , F_{n-1}\}$\;
  \BlankLine
  \For{$i \leftarrow 0$ \KwTo $n-1$}{
    $s \leftarrow i\cdot d$\;
    $F_i \leftarrow S[s:s+d]$ \tcp*[l]{split into frames}
  }
  \BlankLine
  \ForEach{$F_i$ in $\mathcal{F}$}{
    $F_i \leftarrow$ \HW{$F_i$}\;
    $F_i \leftarrow$ \Abs{\FFT{$F_i$}}\;
  }
  \tcp{calculate energy per frequency band on every frame}
  \For{$i \leftarrow 0$ \KwTo $n-1$}{
    \emph{divide into frequency bands} $B_0,\dots,B_{m-1}$\;
    \For{$j \leftarrow 0$ \KwTo $m-1$}{
      $E[i,j] \leftarrow \sum_{k} B_j[k]$\;
    }
  }
  \BlankLine
  \For{$i \leftarrow 1$ \KwTo $n-1$}{
    \For{$j \leftarrow 0$ \KwTo $m-2$}{
      \If{$E[i,j]-E[i,j+1]-(E[i-1,j]-E[i-1,j+1])>0$}{
        $f \leftarrow f||1$\;
      }
      \Else{
        $f \leftarrow f||0$\;
      }
    }
  }
  \BlankLine
  \Return $f$
}
\caption{Fingerprint($S$)\label{algorithm:audio_fingerprinting} {\scriptsize (1536-1233/13/\$31.00 \copyright 2013 IEEE Published by the IEEE CS, CASS, ComSoc, IES, SPS)}}
\end{algorithm}

After initialisation, the audio sequence is split into frames $F_i$. 
For each frame, a Hanning window weighted absolute Fourier transform is then applied.
Afterwards, the energy difference between successive frequency bands is calculated and concatenated to a fingerprint.
When the energy was increased, the corresponding position in the binary fingerprint is associated with the value 1, and else with 0.

\subsection*{APPENDIX B: Commit}\label{sectionCommitImplementation}
Algorithm~\ref{algorithm:commit} details the commit function we utilised to create a public pair $(\delta,\mbox{\texttt{h}}(a))$ describing the difference between a fingerprint and a randomly chosen code value.
\begin{algorithm}
\SetKwFunction{h}{h}\SetKwFunction{SavePrivate}{SavePrivate}
\KwIn{fingerprint $f$}
\KwData{$k,m,n$ for initialising $RS(2^k,m,n)$}
\KwOut{$(\delta, \text{\h{a}})$}
\BlankLine
\Begin{
  $\mathcal{A}=\mathbb{F}^m_{2^k}$\;
  $\mathcal{C}=\mathbb{F}^n_{2^k}$\;
  \BlankLine
    \emph{randomly choose }$a=(a_0,\dots ,a_{m-1})\in\mathcal{A}$\;
    \emph{generate} \h{a}\;
    \SavePrivate{$a$}\;
  \BlankLine
  $c \in \mathcal{C} \xleftarrow{\text{Encode}} a \in \mathcal{A}$\;
  $\delta \leftarrow f \ominus c$ \tcp*[l]{$\ominus :$ subtraction in $\mathcal{C}=\mathbb{F}^n_{2^k}$}
  \BlankLine
  \Return $(\delta, \text{\h{a}})$
}
\caption{Commit($f$)\label{algorithm:commit} {\scriptsize (1536-1233/13/\$31.00 \copyright 2013 IEEE Published by the IEEE CS, CASS, ComSoc, IES, SPS)}}
\end{algorithm}
Generally, after randomly choosing $a$ from $\mathcal{A}$, a hash $\mbox{\texttt{h}}(a)$ is generated and $a$ is stored privately. 
Following a suitable error correcting code (in our case Reed-Solomon codes with $RS(2^{10},204,512)$), a corresponding codeword $c$ is derived from $a$.
After calculating the distance $\delta$ between $c$ and $f$, pair $(\delta,\mbox{\texttt{h}}(a))$ is then published as a verification that sufficiently similar fingerprints have been created by both devices.

\subsection*{APPENDIX C: Decommit}\label{sectionDecommitImplementation}
The decommitment algorithm, utilised to verify the similarity of fingerprints created by remote devices, is detailed in algorithm~\ref{algorithm:decommit}.
\begin{algorithm}
\SetKwFunction{h}{h}\SetKwFunction{SavePrivate}{SavePrivate}
\KwIn{fingerprint $f'$, received $(\delta, \text{\h{$a$}})$}
\KwData{$k,m,n$ for initialising $RS(2^k,m,n)$}
\BlankLine
\Begin{
  $\mathcal{A}=\mathbb{F}^m_{2^k}$\;
  $\mathcal{C}=\mathbb{F}^n_{2^k}$\;
  \BlankLine
  $c' \leftarrow f' \ominus \delta$ \tcp*[l]{$\ominus :$ subtraction in $\mathcal{C}=\mathbb{F}^n_{2^k}$}
  $a' \in \mathcal{A} \xleftarrow{\text{Decode}} c' \in \mathcal{C}$\;
  \emph{generate} \h{$a'$}\;
  \BlankLine
  \If{\h{$a'$}$==$\h{$a$}}{
    \Return \emph{"decommitment successful"}\;
    \SavePrivate{$a'$}\;
  }
  \Else{
    \Return \emph{"decommitment failed"}\;
  }
}
\caption{Decommit($f',(\delta, \text{\texttt{h($a$)}})$)\label{algorithm:decommit} {\scriptsize (1536-1233/13/\$31.00 \copyright 2013 IEEE  Published by the IEEE CS, CASS, ComSoc, IES, SPS)}}
\end{algorithm}
The public pair $(\delta,\mbox{\texttt{h}}(a))$ provided by algorithm~\ref{algorithm:commit} is used in order to verify the similarity between $f$ and $f'$.
Given the fingerprint $f'$, the codeword $c'$ is calculated as the codeword with distance $\delta$ to $f'$.
This value is then decoded to a word $a'\in\mathcal{A}$.
Due to the properties of the error correction methods, up to $t=\left\lfloor \frac{n-m}{2} \right\rfloor$ bits difference between the fingerprints can be corrected.
Since $c'$ and $c$ have the distance $\delta$ to the fingerprints $f'$ and $f$ in common, they also share the same Hamming distance.
Consequently, from the Hamming distance $\mbox{Ham}(f,f')$ between the two fingerprints we obtain
\begin{eqnarray}
     & &\mbox{Ham}(f,f')\leq t\\
&\Rightarrow& \mbox{\texttt{h}}(a')=\mbox{\texttt{h}}(a).\\
&\Rightarrow& a'=a
\end{eqnarray}
Therefore, when the hash values are observed to be identical, $a$ can be used as the common secret among devices.
Otherwise, the pairing failed.
\vfill
\pagebreak

\section[Pattern-based Alignment of Audio Data for Ad-hoc Secure Device Pairing]{Pattern-based Alignment of Audio Data for Ad-hoc Secure Device Pairing \footnote{Originally published as ' Ngu Nguyen, Stephan Sigg, An Huynh and Yusheng Ji: Pattern-based Alignment of Audio Data for Ad-hoc Secure Device Pairing, in 2012 16th International Symposium on Wearable Computers (ISWC), pp.88-91, 18-22 June 2012 (DOI: http://dx.doi.org/10.1109/ISWC.2012.14)' (1550-4816/12 \$26.00 \copyright 2012 IEEE)}}\label{sectionOriginalSE02}
When studying the use of ambient audio to generate a secure cryptographic shared key among mobile phones, we encounter a misalignment problem for recorded audio data.
The diversity in software and hardware causes mobile phones to produce badly-aligned audio chunks.
It decreases the identical fraction in audio samples recorded in nearby mobile phones and consequently the common information available to create a secure key.
Unless the mobile devices are real-time capable, this problem can not be solved with standard distributed time synchronisation approaches.
We propose a pattern-based approximative matching process to achieve synchronisation without communication independently on each device.
Our experimental results show that this method can help to improve the similarity of the audio fingerprints, which are the source to create the communication key.

\subsection{Introduction}
\label{sectionIntroductionSE02}
With recent advances in smart-phone dissemination and their computational capabilities, smart-phones can be seen as a kind of wearable device for the masses.
These general-purpose devices are capable of solving several wearable computing tasks.
Due to their high penetration, security in communication among devices becomes a relevant issue.
Common security schemes for mobile devices require explicit user input to provide a shared piece of information.
A wearable device, however, should not distract its holder from other tasks.
How can we provide security among possibly unacquainted devices without any user interaction?

We consider an interaction-free common key generation scheme for proximate devices conditioned on ambient audio.
Since the seed to the key is implicit with the context, no information that could be used to reconstruct the key is transmitted on a wireless channel during key generation.
Each device computes a binary characteristic sequence for a synchronised recording: An audio-fingerprint~\cite{Haitsma2003highlyrobust, AudioFingerprinting_Wan_2003, wavelets, Audio_fingerprint_review_Cano_2005, Audio_fingerprint_review_Chandrasekhar_2011}.
This binary sequence is designed to fall onto a code-space of an error correcting code~\cite{Reed60polynomial}.
Devices then exploit the error correction capabilities of the error correcting code to map fingerprints to codewords as described in~\cite{Juels99fuzzycommitment,schneider1996applied}.
For fingerprints with a Hamming-distance within the error correction threshold of the error correcting code the resulting codewords are identical and then utilized as secure keys.
The Hamming distance in fingerprints rises with increasing distance of devices so that distant devices are unlikely to guess the correct key.
Our fingerprint extraction scheme is adapted from \cite{Haitsma2003highlyrobust} to extract fingerprints from synchronized ambient audio recordings in a noisy environment without exchanging any information about the resources among devices.
However, when audio sequences utilized are not well aligned, similarity in fingerprints decreases.
This is due to the fingerprint generation which exploits the relative fluctuation of energy over time.
Small shifts in the audio data will likely produce completely differing fingerprints.
This is a relevant problem since simple time synchronisation approaches, such as for instance the network time protocol (NTP), are not suitable to sufficiently synchronise audio recordings due to the delays in the recording hardware.

This paper addresses the accurate alignment of recorded audio sequences from remote devices.
The challenging point here is to achieve an alignment between audio samples taken from distinct devices interaction free and without any inter-device communication other than an initial plain pairing request.

We will in section~\ref{sectionRelatedWorkAF01} discuss related work on secure ad-hoc pairing of mobile devices.
Problems that prevent accurate audio sequence alignment and our pattern-based approximative matching method to reduce the mismatching are detailed in in section~\ref{sectionProblemsAddressed}.
Section~\ref{sectionResultsAF01} describes a case study conducted with smart-phone devices to investigate the accuracy of our approach in a realistic setting and results of our alignment scheme.
Section~\ref{sectionConclusionAF01} draws our conclusion.

\subsection{Related Work}
\label{sectionRelatedWorkAF01}
Contextual or sensor information of mobile devices can be incorporated as a solution for authentication~\cite{ContextAwareness_Holmquist_2001}.
When the seed to the key is implicit with the context, no information that could be used to reconstruct the key is transmitted during key generation.
For instance, McCune et al.~\cite{McCune_Cryptography_2005} introduced \textit{Seeing-Is-Believing}, utilizing the camera of a mobile device to capture a 2D barcode which is displayed on the screen of another device.
\textit{Loud and Clear} of Goodrich et al.~\cite{Goodrich_Cryptography_2006} implements a similar scheme but exploits spoken audio.
A user reads aloud a text message displayed on one device and a second device recognizes the speech for authentication.
A further example mechanism by Mayrhofer et al.~\cite{mayrhofer2007shake} uses accelerometer readings when devices are shaken simultaneously by a single person.
Also, Mayrhofer derived in~\cite{mayrhofer2007candidate} that the sharing of secret keys is possible with a similar protocol by repeatedly exchanging hashes of key-sub-sequences until a common secret is found.
Bichler et al. generalise this approach to noisy acceleration readings~\cite{Cryptography_Bichler_2007,Cryptography_Bichler_2007-2}.
They utilize a hash function that maps similar acceleration patterns to identical key sequences.
These approaches require explicit user interaction.

By utilising a context source that provides a sufficient amount of unique, context-related information, such as audio or radio frequency (RF), it is possible to get the user out of the loop.
Mathur et al. introduced ProxiMate that enables wireless devices in proximity to pair automatically and securely using their shared ambient RF-signals~\cite{Cryptography_Mathur_2011}. 
They generate fingerprints from RF-channel fluctuations and map these onto a codespace of an error-correcting code.
By correcting potential errors in the fingerprints, they are mapped onto the closest regular codeword in the codespace.
When the similarity between fingerprints is high, codewords are identical.
Sigg et. al proposed to use  audio instead of RF in a similar implementation~\cite{Cryptography_Sigg_2011-5}.
They study the entropy of audio fingerprints and identify time synchronisation as a main hindrance to practically apply the method for mobile devices.
Their instrumentation requires idealized conditions regarding the synchronisation of devices and to account for this a high number of fingerprints must be created (201 in their experiments) in order to find one matching fingerprint.
For extensive computational load, this is feasible only in an offline approach.
The high number of fingerprints created, however, was necessary since the utilized NTP synchronisation is not sufficiently accurate.

In this paper, we present an alignment mechanism that enables a synchronisation accuracy of recorded audio in the order of less than 10 milliseconds among unsynchronised mobile devices in the same context without transmitting information about the audio sequence over the wireless channel.
The synchronisation is achieved by processing a weakly NTP-synchronised recording without additional communication among the devices. 

\subsection{Extracting fingerprints from ambient audio}
\label{sectionExtractingAudioFingerprints}
Most of the previous studies applied audio fingerprints to a large audio database of musical information.
In this type of data, there exists a dominant sound, which is the song itself.
In our research, we need to extract a fingerprint that can represent all of the characteristics of ambient sounds around the devices.
For instance, when the user is in a lobby room, there are various audio sources such as human voice, music, noise of opening and closing doors, etc.
All of the recorded sounds are potentially equally important to describe a context and should contribute to form an audio fingerprint.

Our approach is adapted from \cite{Haitsma2003highlyrobust} to extract the most similar fingerprint from synchronized ambient audio recordings in a noisy environment without exchanging any information about the resources among devices.
In our research, each bit in the binary audio-fingerprint expresses the difference between energy on frequency bands.

An ambient audio chunk is split into non-overlapping equal-length frames. 
Then we perform a Discrete Fast Fourier Transform on these frames. 
After that, each frame is divided into a set of frequency bands with identical width. 

An energy matrix $E$ is created as follows. 
Its size is \textit{the number of frames} $\times$ \textit{the number of frequency bands per frame}. 
Each value of the matrix represents the total energy of a frequency band in the corresponding frame.

From the matrix $E$, we generate the binary fingerprint $f$ whose bits contain the information about the energy change of frequency bands on two successive frames.
\begin{displaymath}
   f(i,j) = 
   	 \begin{cases}
       1 & \text{, if $\begin{array}{l}
	       (E[i,j] - E[i,j+1]) \\
		 - (E[i-1,j] - E[i-1,j+1]) > 0
		 \end{array}$} \\
       0 & \text{, otherwise.}
     \end{cases}
\end{displaymath}

In the above formula, if the energy gain between two consecutive frames is positive, the corresponding element in the binary sequence is assigned with the value of 1; otherwise, it has the value of 0.
Each audio fingerprint contains 512 bits.
Because the fingerprint extraction is performed based on the relative fluctuation of energy, there is no effect caused by the difference of raw intensity values in time domain for distinct devices.
However, since the audio signal is of consecutive time windows, small shifts in the audio data will likely produce completely differing fingerprints.
When this method is utilized for the generation of identical binary strings on remote mobile devices, the timing difference likely prevents the generation of sufficiently similar information on both devices.

Also, if there are more audio sources around the device, the audio fingerprints are more sensitive with inter-device distance.
The closer the mobile devices to each other, the more similar their fingerprints are.

To use the audio-fingerprints directly as keys for a classic encryption scheme the concurrence of fingerprints generated from related audio sequences has to be $1$ with a considerably high probability~\cite{schneider1996applied}.
Since we experienced a substantial difference in the audio-fingerprints created we consider the application of fuzzy cryptography schemes.
Note that a perfect match in fingerprints is unlikely since devices are spatially separated, not exactly synchronised and utilize possibly different audio hardware.

The proposed cryptographic protocol shall work completely ad-hoc with devices previously not known to each other.
For an eavesdropper in a different audio context it shall be computationally infeasible to use any intercepted data to decrypt a message or parts of it.
Apart from these requirements we want to choose the threshold for a working encrypted communication based on contextual conditions of different physical locations.

With fuzzy encryption schemes, we are able to overcome these challenges.
Generally, a secret $\varsigma$ is used to hide the key $\kappa$ in a set of possible keys $\mathcal{K}$ in such a way that only a similar secret $\varsigma'$ can find and decrypt the original key $\kappa$ correctly.
In our case, the secrets which are similar for all communicating devices in the same context are constituted by fingerprints generated from ambient audio.

We implemented a Fuzzy Commitment scheme with Reed-Solomon codes~\cite{Reed60polynomial}.
The following discussion provides a short introduction into these codes.

Given a set of possible words $\mathcal{A}$ of length $m$ and a set of possible codewords $\mathcal{C}$ of length $n$, Reed-Solomon codes $RS(q,m,n)$ are initialized as:
\begin{eqnarray}
     \mathcal{A}&=\mathcal{F}^m_{q},\\
     \mathcal{C}&=\mathcal{F}^n_{q},     
\end{eqnarray}
with $q=p^k, p \text{ prime}, k \in \mathcal{N}$.
These codes are mapping a word $a\in\mathcal{A}$ of length $m$ uniquely to a specific codeword $c\in \mathcal{C}$ of length $n$:
\begin{equation}
     a \xrightarrow{Encode} c,  
\end{equation}
This step adds redundancy to the original words with $n>m$, based on polynomials over Galois fields~\cite{Reed60polynomial}.

Decoding utilizes the error correction properties of the Reed-Solomon-based encoding function to account for differences in the fingerprints created.
The decoding function maps a set of codewords from one group $C=\{c,c',c'',\dots \} \subset \mathcal{C}$ to one single original word. 
It is 
\begin{eqnarray}
  \tilde{c} \xrightarrow{Decode} a\in \mathcal{A}.
\end{eqnarray}
The value
\begin{eqnarray}
     t=\left\lfloor \frac{n-m}{2} \right\rfloor
\end{eqnarray}
defines the threshold for the maximum number of bits between codewords that can be corrected in this manner to decode correctly to the same word $a$~\cite{Juels99fuzzycommitment}.
In the algorithm implemented, the fingerprints $f$ and $f'$ are used in conjunction with codewords to make use of this error correction procedure.
Dependent on the noise in the created fingerprints, $t$ can then be chosen arbitrarily.

Our current fingerprint extraction scheme, which is based on energy in frequency domain, can deal with the difference in amplitude values.
However, when the audio chunks are shifted in time, for example, since the ambient sounds are not recorded at the same time, the audio fingerprints misalign so that we can not create the common key.
In our experiments, we observed this effect.
In particular, although the NTP-synchronised devices intend to start their recordings at the same time, we observed significant differences that exceed the time offset due to NTP in accuracy by several milliseconds.
This effect can be observed in figure~\ref{figurebeforeSync_DiffDevices}.
\begin{figure}\centering
     \includegraphics[width=9.5cm,height=6cm]{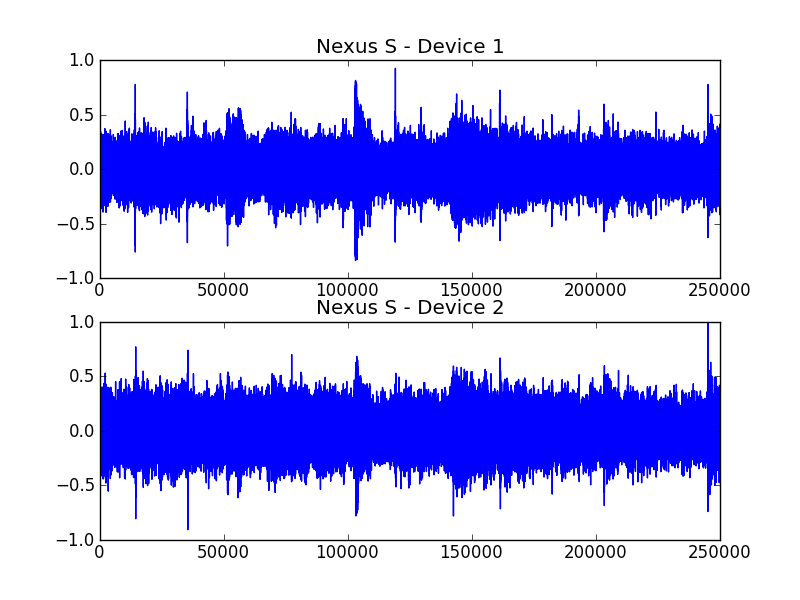}
     \caption{Audio files recorded on two Samsung Nexus S mobile phones. {\scriptsize (1550-4816/12 \$26.00 \copyright 2012 IEEE)}}
     \label{figurebeforeSync_SameDevices}
\end{figure}
\begin{figure}
\centering
     \includegraphics[width=9.5cm,height=6cm]{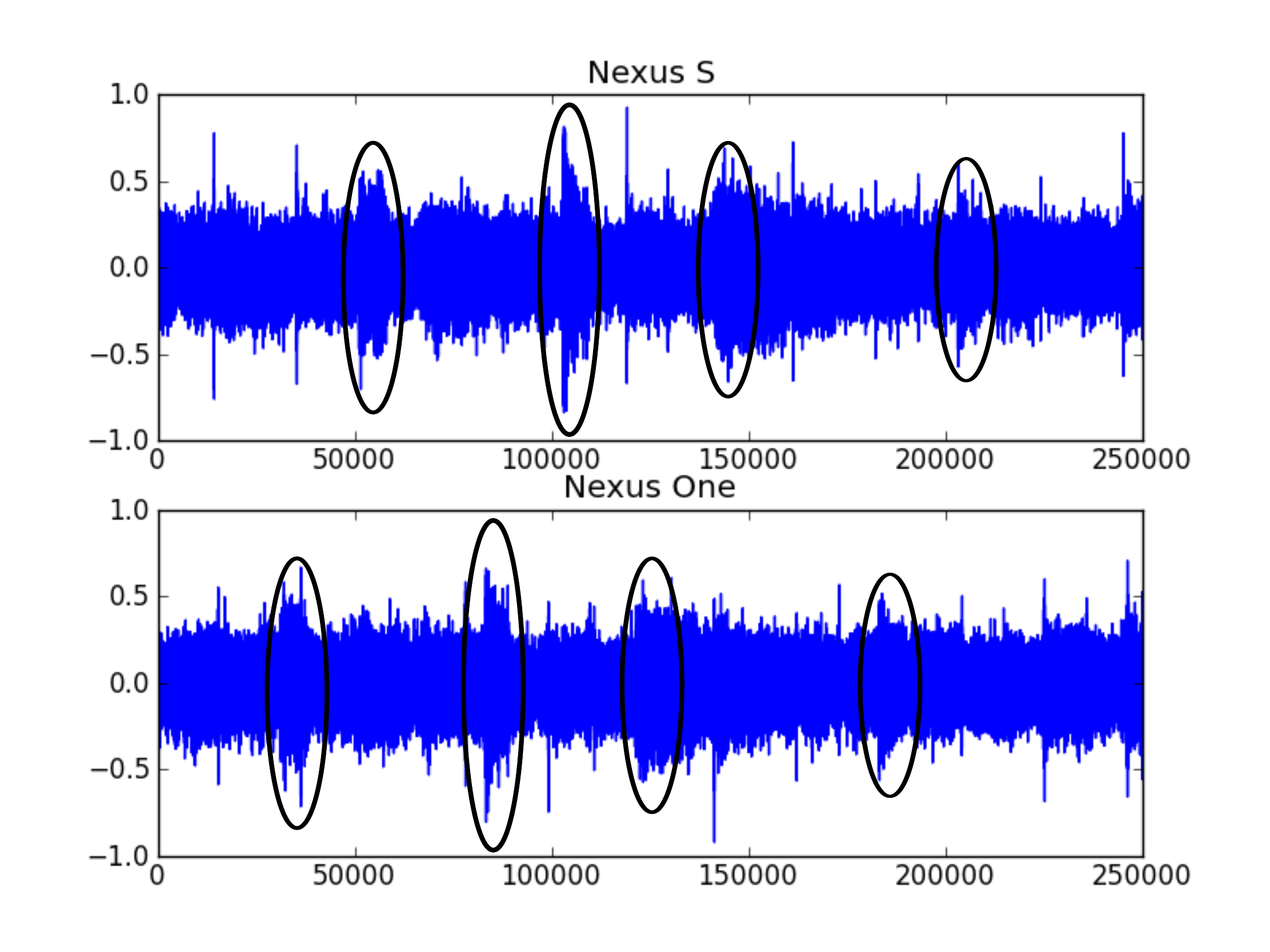}
     \caption{Misalignment of audio files from difference mobile phones. The above is a Samsung Nexus S and the below is a HTC Nexus One. {\scriptsize (1550-4816/12 \$26.00 \copyright 2012 IEEE)}}
     \label{figurebeforeSync_DiffDevices}
\end{figure}

\subsection{Pattern-based alignment of audio data}\label{sectionProblemsAddressed}
When developing the scheme of unobtrusive secure device pairing with audio fingerprints for Android-based mobile devices, we encountered practical issues not evident when considering the problem theoretically.
One issue in practical implementation is differing audio hardware.
For instance, the Samsung Google Nexus~S\footnote{Nexus S: \url{http://www.google.com/nexus/tech-specs.html}} and HTC Google Nexus~One\footnote{Nexus One: \url{http://www.google.com/phone/detail/nexus-one}} devices we utilized apply different audio-pre-processing routines that render the unprocessed audio outputs on these devices unusable for the generation of identical fingerprints.
Furthermore, time synchronisation is a serious problem for the approach.
In particular, not only the clocks on remote devices have to be synchronised as usual for distributed devices, but also the generally unknown and possibly non-constant hardware specific delays on both devices need to be taken into account.

\subsubsection{Misalignment of audio data}
Due to the differences in software and hardware of the devices, the recorded audio sequences are not exactly the same.
An example of this phenomenon is shown in figure~\ref{figureProblems}.
All audio files are recorded with the same settings and at the same time (clocks synchronised by a NTP service\footnote{Navy Clock II application: \url{https://market.android.com/details?id=com.cognition.navyclock}}).
Basing on some similar visual appearances in the waveform format of the signals in figure~\ref{figureProblems}, the recording start time of the Nexus One was heavily delayed when compared to the Nexus S.
In our recorded audio files, time difference values can be from 0.3 second to more than 1 second.
The time difference in some pairs of audio files is shown in Table~\ref{tableOriginalTimeDiff}.

\begin{table}
	\caption{Time Difference before Pattern-based Alignment. {\scriptsize (1550-4816/12 \$26.00 \copyright 2012 IEEE)}}
	\label{tableOriginalTimeDiff}
	\begin{center}
		\begin{tabular}{| c | c |}
		\hline
		Audio file pair & Time difference [seconds]\\
		\hline
		1 & 0.320 \\
		\hline
		2 & 0.849\\
		\hline
		3 & 0.740\\
		\hline
		4 & 0.459\\
		\hline
		5 & 1.341\\
		\hline
		6 & 0.450\\
		\hline
		7 & 0.765\\
		\hline
		8 & 0.562\\
		\hline
		9 & 0.744\\
		\hline
		10 & 0.375\\
		\hline
		\end{tabular}
	\end{center}
\end{table}

We observe that the offset when both recordings start is fluctuating and in all cases much higher than the accuracy expected from NTP synchronisation.
\begin{figure}
     \includegraphics[width=17cm]{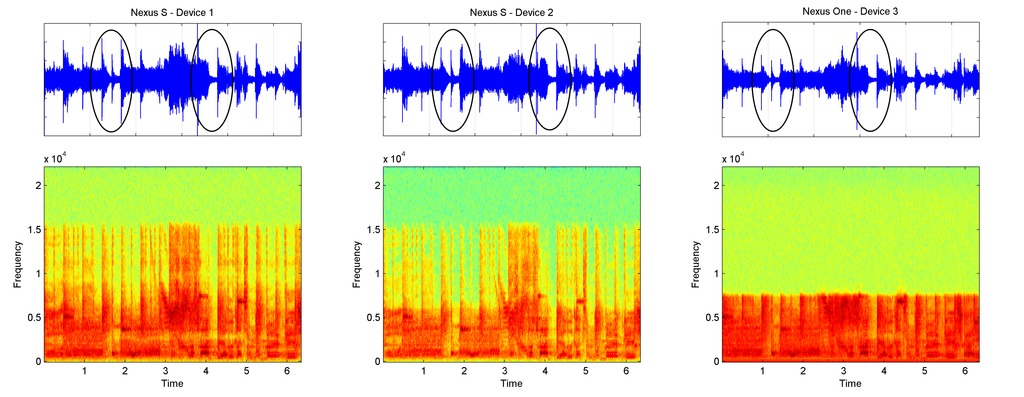}
     \caption{Waveform format (upper) and spectrogram (lower) representation of audio recordings from three devices. For the Nexus One, the hardware noise cancellation cuts higher parts of the frequency spectrum. The recording of the Nexus One is delayed comparing to the two Nexus S devices which are tightly synchronised. {\scriptsize (1550-4816/12 \$26.00 \copyright 2012 IEEE)}}
     \label{figureProblems}
\end{figure}
Additionally, we observed from the figure that clearly, the higher frequency bands of the signal available at the Nexus One device completely differ from that of Nexus S readings because the Nexus~One employs a hardware noise cancellation. 
There is no way to bridge the noise cancellation on that device to obtain the unmodified signal.
Both these effects are unfortunate for our fingerprinting method. 
\subsubsection{Aligning recorded audio data}\label{sectionAlignment}
One important condition for our secure device pairing scheme is that no information regarding the recorded audio shall be exchanged between the devices. 
Otherwise, the security of the key might be impaired by information leaking from these transmissions.
We propose a synchronization scheme to reduce the mismatch between audio data from neighbouring devices.
Our solution is based on the Smith-Waterman algorithm~\cite{SmithWatermanAlg}, an approximative pattern matching technique.
In the algorithm, two strings are compared for the best approximative pattern among them, as depicted in figure~\ref{figurePatternMatching}.
\begin{figure}
\centering
      \includegraphics[width=10cm]{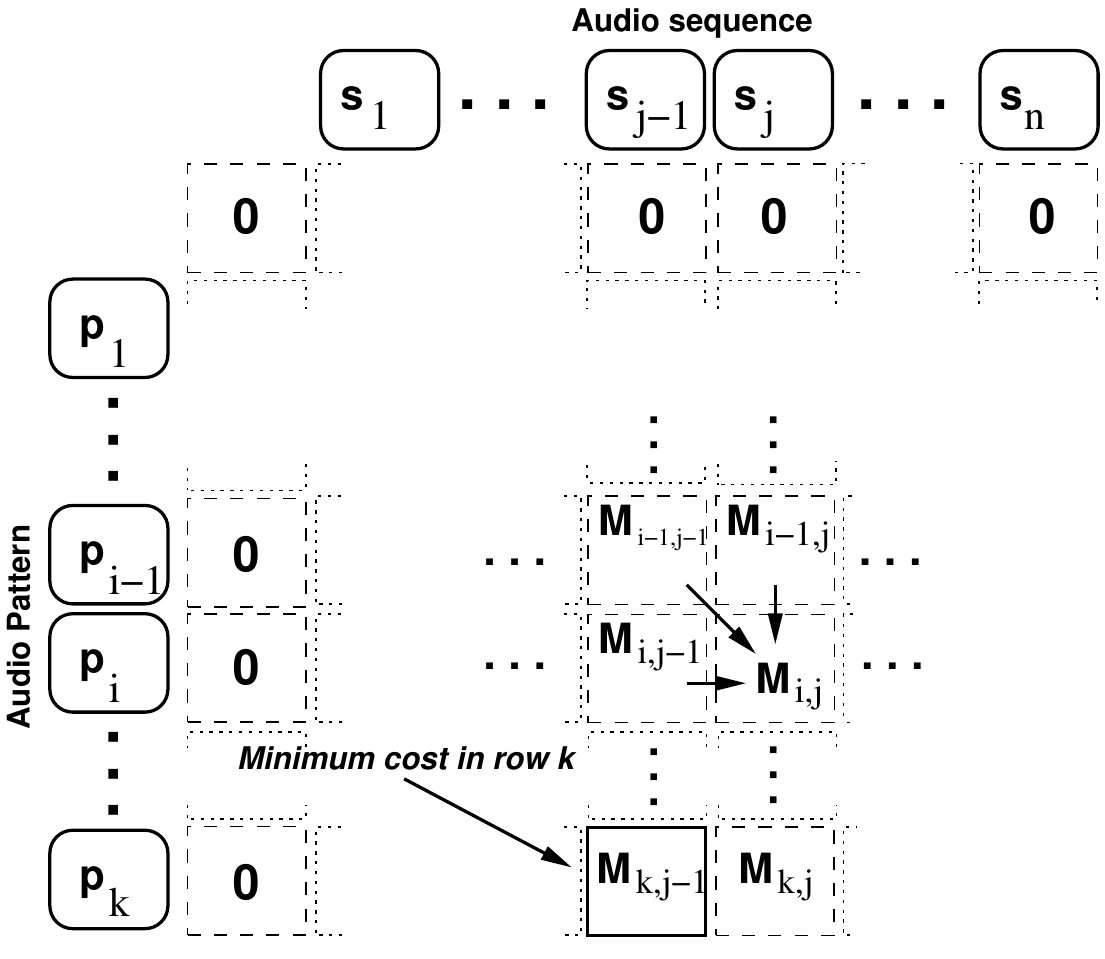}
     \caption{Schematic illustration of the Smith-Waterman pattern matching {\scriptsize (1550-4816/12 \$26.00 \copyright 2012 IEEE)}}
     \label{figurePatternMatching}
\end{figure}
The figure illustrates the operation to find a sub-sequence similar to a pattern $p=p_1\dots p_k$ in a longer sequence $s=s_1\dots s_n$.
The matrix is filled from left to right and from top to bottom.
Each matrix entry $M_{ij}$ contains the minimum cost to align $p_1\dots p_i$ with $s_1\dots s_j$.
It is obtained iteratively from the cost of the cells $M_{i-1,j-1},M_{i-1,j},M_{i,j-1}$ by choosing the minimum cost for an extension of any of these three alignments (extension of both strings, extension of only the first string or extension of only the second string).
The first row is initialised with $0$ to allow the pattern to start at any position within $s$.
The careful choice of the gap-penalty $g$ is crucial for the approach to produce good matching results.
The best matching string is obtained by backtracking the smallest entry in the last row.

Since the matching is approximative, we will always find a best matching position even though the absolute similarity of this very position might not be high.
We exploit this property to be able to resign from any information transmitted among devices on the actual recorded audio.
In particular, we utilize a predefined, characteristic pattern on both devices.
since the pattern is known in advance, no information need to be transmitted.
The downside to this implementation is, of course, that the pattern utilised might be very different from the actual audio recorded.
However, since the Smith-Waterman algorithm always computes a best matching, we can speculate that this best matching is found at similar positions in the audio recordings, provided that the data has significant similarity.
The best matching, although it might not be a good matching in absolute terms, is likely to be found at a similar position in the recorded audio of the remote devices.
The specific pattern used for matching is extracted randomly from consecutive samples of an arbitrary audio sequence. 
In our experiments, its length is 100 samples (longer patterns increase running time of the algorithm).
The pattern $p$ is matched in the first $100000$-sample part of each audio file. 
The matching score of $p$ and a local part $l$ of $s$ is the difference between amplitude values in $p$ and $l$. 
The less the score, the more similar $p$ and $l$ are. 
According to our experiments, the gap penalty that can yield an acceptable matching is $150$. 
Then, we eliminate all samples preceding the matching positions and generate the audio fingerprints from the remaining ones.

It turned out that we achieve only seldom a perfect synchronizing result among two best matching points on both devices. 
We therefore calculated a set of $k$ best matching parts.
We assume that one device chooses the best matching locations, encodes a data sequence with the resulting keys and transmits this to the second device. 
This receiver device then attempts to decode the data with the keys generated from its best matching alignment points.
This process might require the transmit device to transmit a data sequence several times, encrypted with different keys each time.
Although this increases the transmit load, this process could be implemented in an iterative fashion and is required only for the first encounter of the devices to derive the secret key.
Since all data blocks are encrypted, only marginal additional information is provided to a potential adversary.

\subsection{Experimental results}
\label{sectionResultsAF01}
\label{sectionInstrumentation}
We utilize the Android-based mobile phones Samsung Google Nexus~S and HTC Google Nexus~One.
The Nexus One smart-phone has a secondary microphone dedicated for dynamic noise suppression, while the Nexus S devices utilise software noise cancellation.
The ambient audio data are recorded for $6375$ milliseconds at a sampling rate of $44100~Hz$.
In our experiments, we attempt to generate a secure key among a Nexus~S and a Nexus~One device after applying a bandpass filter to the recorded data and aligning audio sequences with the pattern matching approach.
The smart-phones are put at the same distance $d$ from an audio source and we increase the distance between two devices from $10 cm$ to $100 cm$.
With each distance, we record 10 audio files on each phone.
Before matching with the arbitrary pattern, we perform a bandpass filtering which only retains signals whose frequency is between $4000 Hz$ and $4500 Hz$.

\subsection{Key generation without sequence alignment}
\label{sectionCaseStudySE02}
We conducted a case study with two Nexus~S and one Nexus~One device to study the alignment of  audio data at various inter-device distances.
The bit-errors in fingerprints were not related to the distance to an audio source or the loudness of the sound.
Therefore, we consider the deterioration in the similarity of fingerprints while the distance among recording devices is gradually increased.
We are interested in the Hamming distance between the generated fingerprints. 

As detailed in section~\ref{sectionExtractingAudioFingerprints}, we can use error correcting codes (ECC) to account for the Hamming distance in fingerprints and generate an identical cryptographic key at devices independently without communication, provided that the Hamming distance is significantly smaller than 50\%\ of the sequence length (which would be the Hamming distance to a random binary sequence).
We implemented the fingerprinting method based on energy differences as described above. 

Devices were placed in an angle of $22.5^\circ$ and $-22.5^\circ$ to an audio source and the distance between devices was altered gradually.
In several sets of experiments, we increased the distance among devices and to the audio source while keeping the angle to the audio source fixed.
For each distance the experiment was repeated at least 10 times.

\begin{figure}
\centering
     \includegraphics[width=12cm]{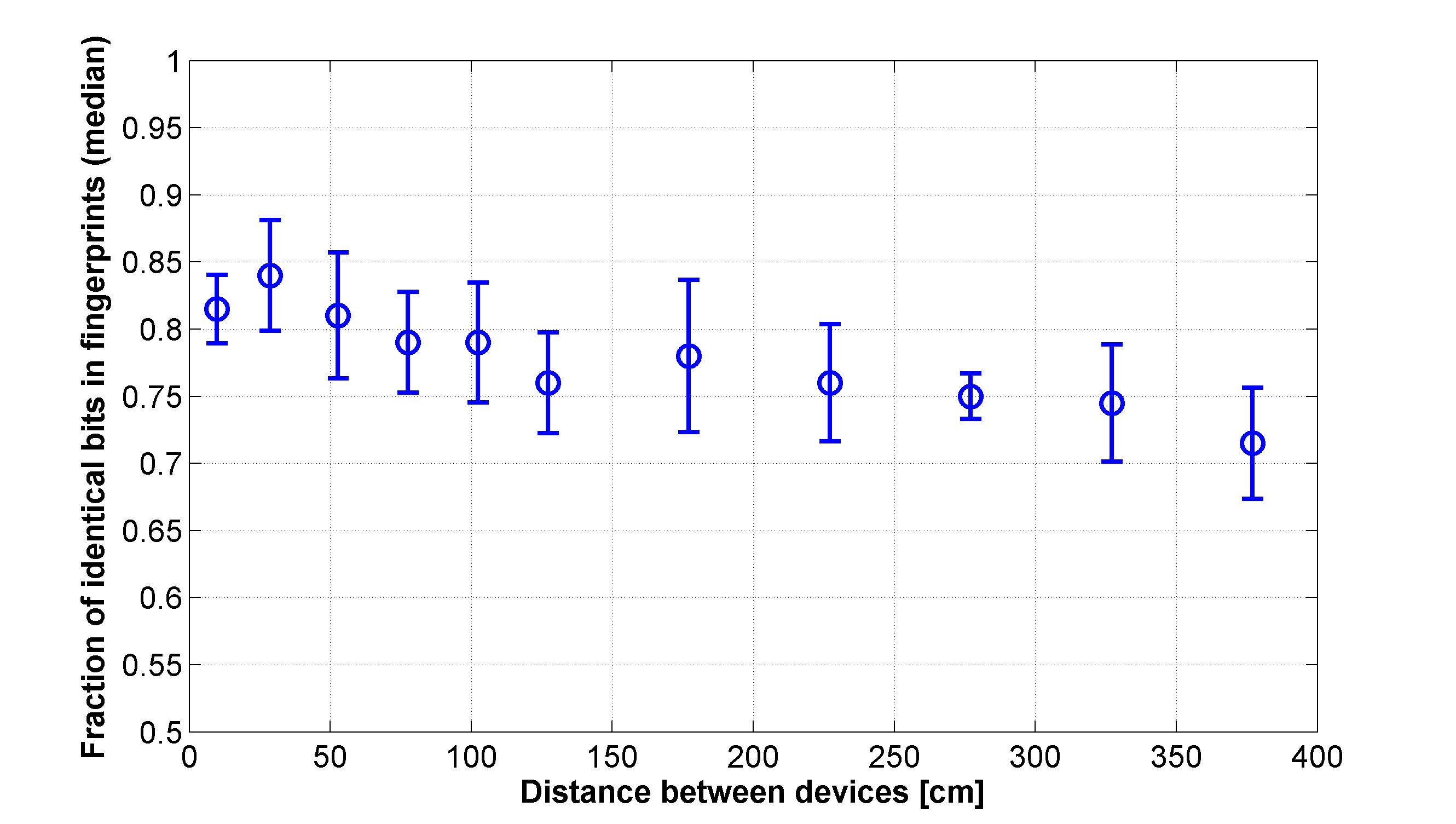}
     \caption{Median fraction of identical bits in fingerprints created from synchronised audio recordings at two Nexus S devices {\scriptsize (1550-4816/12 \$26.00 \copyright 2012 IEEE)}}
     \label{figureResults}
\end{figure}

The figure shows the median fraction of identical bits in fingerprints created by the devices.
Each point was created from 10-12 separate experiments with identical environmental conditions.
The method achieves an accuracy of about $0.75$ to $0.85$ in all tests.
In particular, the Hamming distance in fingerprints only slightly decreases with increasing distance.
This demonstrates that a good synchronisation is feasible on mobile devices.
Also, devices in about 2~m distance can be considered to be in the same security aura for this environment.
Devices farther away can be excluded by properly configuring the error correction threshold of the ECC utilized.

In particular, when the ECC is configured to correct up to about 17\%\ of the bits in a fingerprint, devices in more than 60~cm distance would only occasionally be able to successfully generate the same key.
For devices in more than 2 meters distance, it is highly unlikely that they generate a matching key.
Observe that, for a 512 bit key, although a Hamming distance of 1\%\ of all bits translates to only about 32 different possibilities, the positions of the respective 5 bits in the 512 bits are not known.
This leaves a total of 
\begin{equation}
     \left(
	  \begin{array}{c}
	       5\\
	       512
     \end{array}\right)\cdot 32 > 9\cdot10^{12}
\end{equation}
different possibilities for an attacker to reduce the Hamming distance further by only 1\%.

For the Nexus~One, only about 50\%\ of the bits were identical in all cases, regardless of the distance among devices and the specific Nexus~S device used for pairing.
Since this is the similarity, we would expect for a random guess of the fingerprint, we conclude that a secure key can not be derived with this approach from these devices which differ in hardware and software.

Figure~\ref{Sync_Receiver3-10} depicts the median time difference after alignment when increasing the number of trials in the transmitter. The receiver  attempts to decode the data with its sequences from the top 3 and 10 matching positions.

We can achieve a synchronisation in the order of $10$ milliseconds already with 10 trials at the transmitter and only 3 trials at the receiver (cf. figure~\ref{Sync_Receiver3-10}).

Moreover, when the receiver tries 10 times, we observe that, already with the best matching trial at the transmitter, the matching is greatly improved compared to time difference between misaligned audio sequences in table~\ref{tableOriginalTimeDiff}.
With 3 trials of the transmitter, the synchronisation time is accurate enough for our secret key generation approach.
Since the alignment approach calculates all possible alignments at once, the cost for additional trials at the receiver is low.
Additional trials at the transmitter, however, directly impact the communication load.

\begin{figure}
\centering
     \includegraphics[width=12cm]{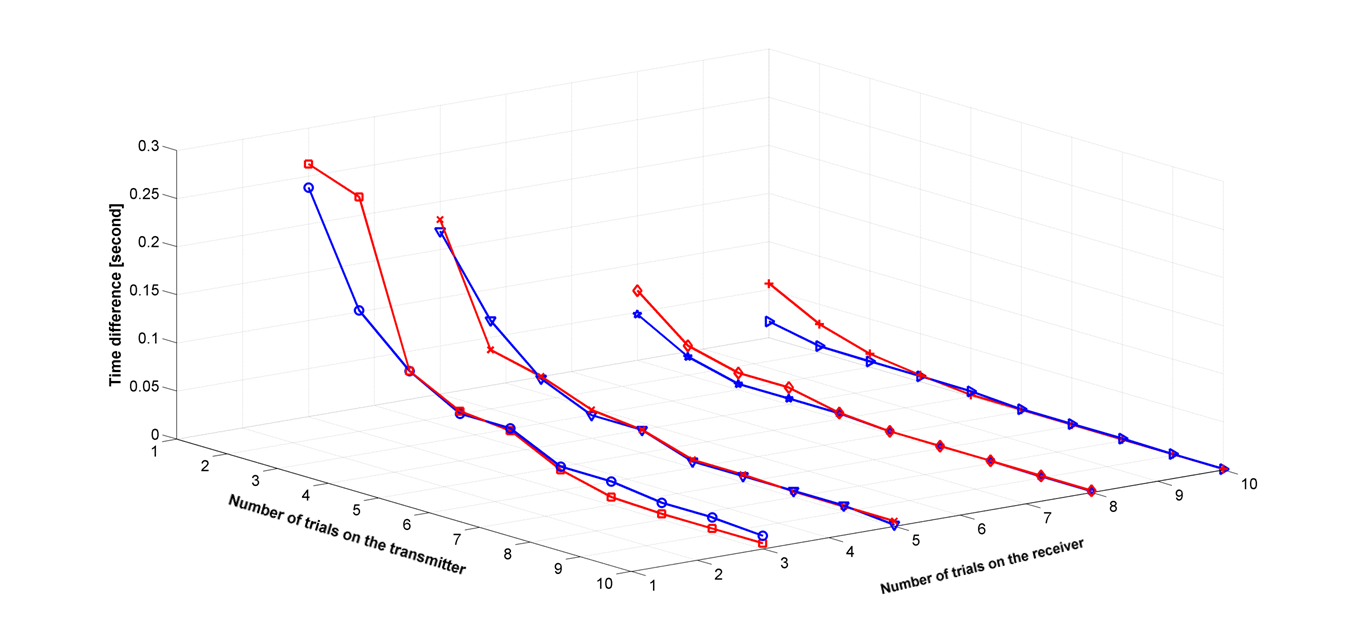}
     \caption{Time difference of the audio after alignment with increasing number of trials on the transmitter. The number of trials on the receiver is respectively 3 and 10. {\scriptsize (1550-4816/12 \$26.00 \copyright 2012 IEEE)}}
     \label{Sync_Receiver3-10}
\end{figure}

\begin{figure}
\centering
     \includegraphics[width=12cm]{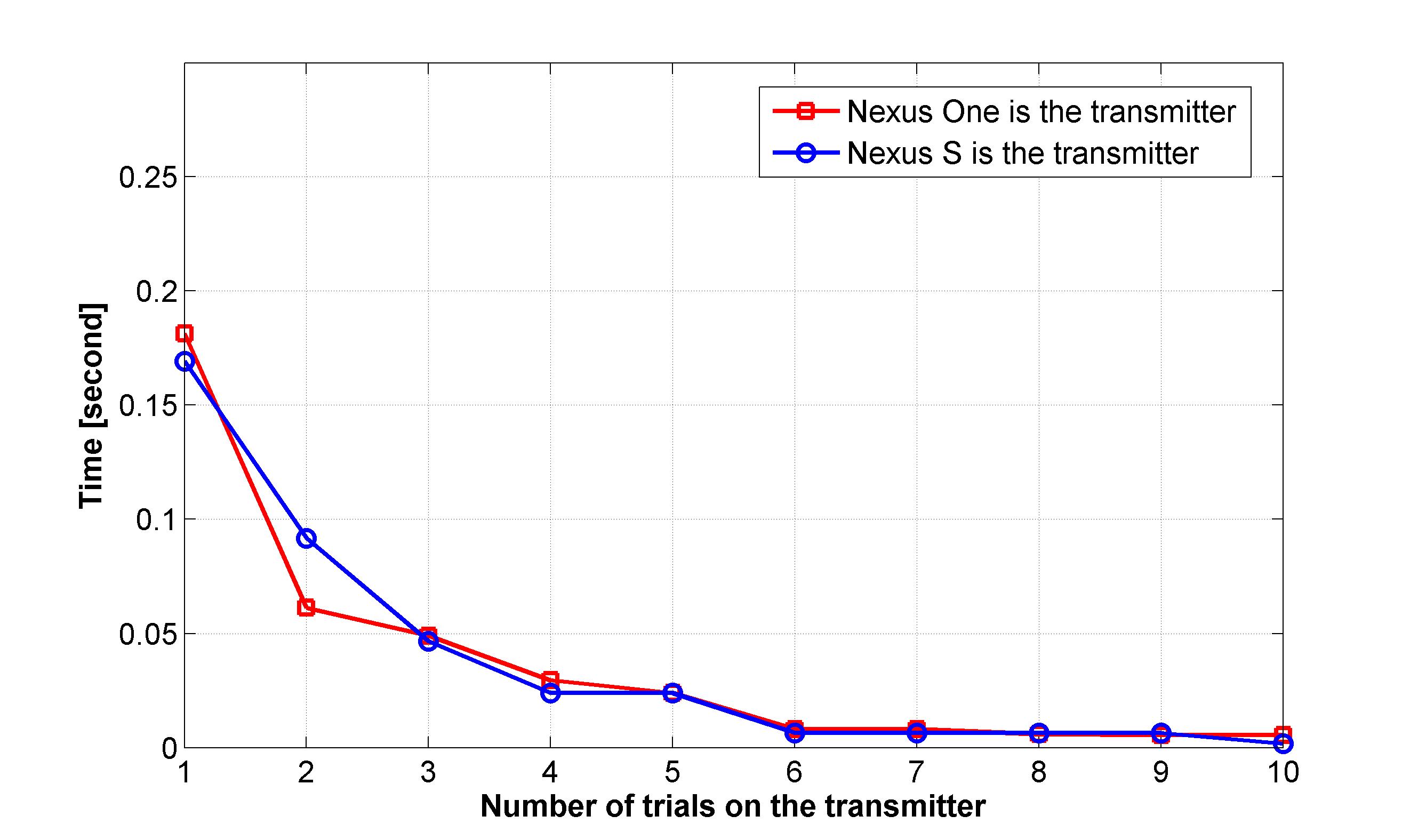}
     \caption{Time offset of the audio after synchronisation with increasing number of trials on the transmitter.
	       The number of trials on the receiver is always 5. {\scriptsize (1550-4816/12 \$26.00 \copyright 2012 IEEE)}}
     \label{Sync_Receiver5}
\end{figure}
\begin{figure}
\centering
     \includegraphics[width=12cm]{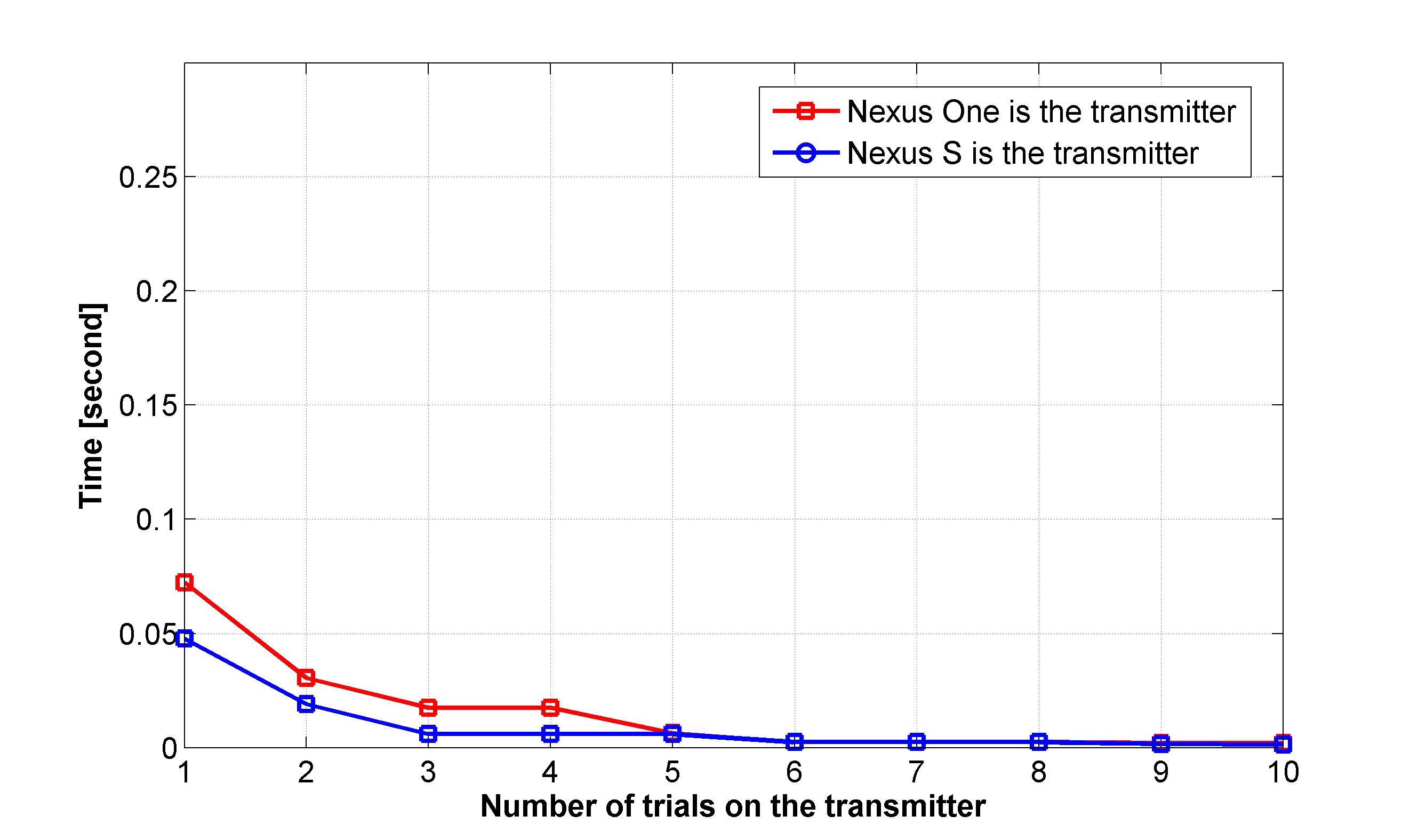}
     \caption{Time offset of the audio after synchronisation with increasing number of trials on the transmitter.
	       The number of trials on the receiver is always 8. {\scriptsize (1550-4816/12 \$26.00 \copyright 2012 IEEE)}}
     \label{Sync_Receiver8}
\end{figure}

\begin{figure}
\centering
     \includegraphics[width=12cm]{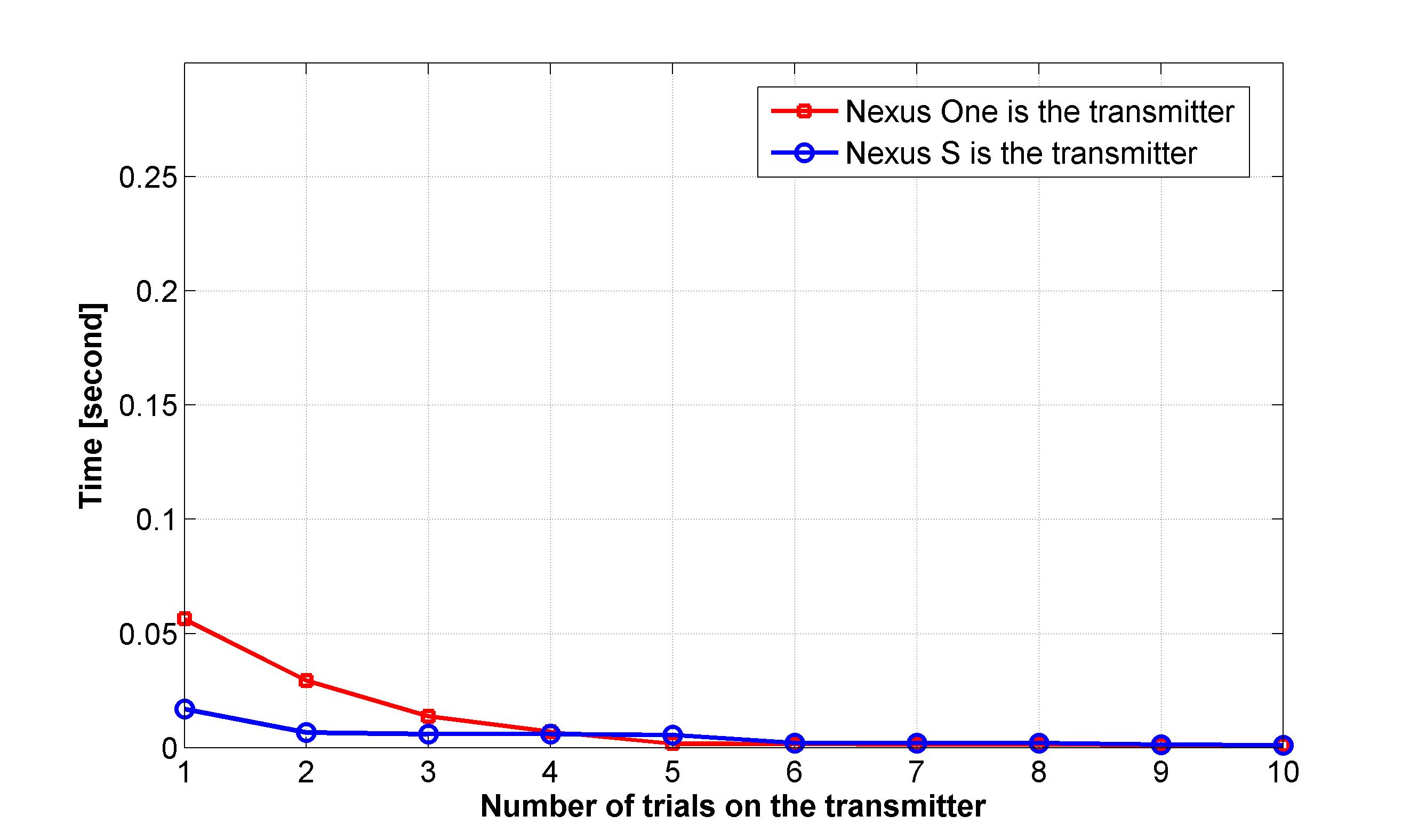}
     \caption{Time offset of the audio after synchronisation with increasing number of trials on the transmitter. 
	       The number of trials on the receiver is always 10. {\scriptsize (1550-4816/12 \$26.00 \copyright 2012 IEEE)}}
     \label{Sync_Receiver10}
\end{figure}

We conclude, that the alignment approach greatly improves the synchronisation of audio recordings on remote devices. 
This can be seen from figure~\ref{AudioFingerprint_BeforeAfterSync}.
The figure depicts the fraction of identical bits among fingerprints before and after the alignment matching was applied.
While the fingerprint similarity only marginally deviates from the similarity to a random sequence, it is greatly improved after the alignment matching was applied.
The fraction of identical bits in the audio fingerprints decreases when inter-device distance increases.
In particular, also the variance in the data is reduced.
These characteristics allow for a sharper threshold of the error correcting code.

We also extracted the pattern from one file in each misaligned audio file pair and found it in the other.
The quality of the audio fingerprints is not much higher than with the arbitrary patterns, even with the best results of matching. The comparison is shown in figure~\ref{AudioFingerprint_BeforeAfterSync}.
Moreover, in this approach, there is data transmission between the devices.
\begin{figure}
\centering
     \includegraphics[width=12cm]{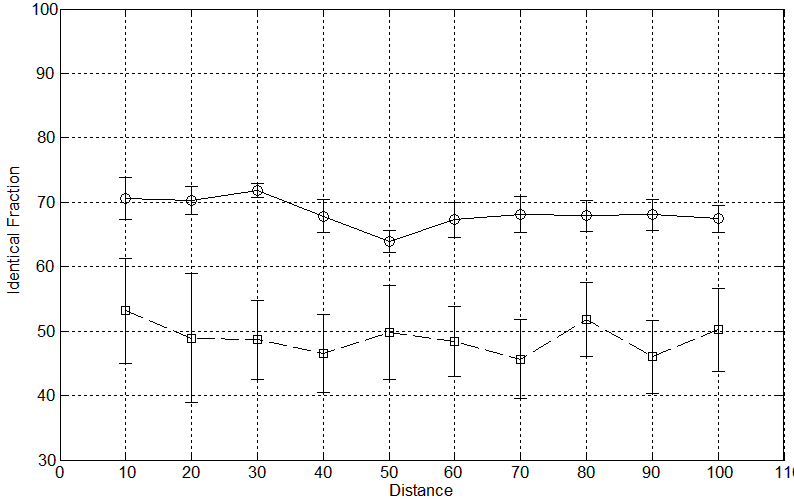}
     \caption{Identical fraction in bits of audio fingerprints created on Nexus~S and Nexus~One devices when the inter-device distance increases. {\scriptsize (1550-4816/12 \$26.00 \copyright 2012 IEEE)}}
     \label{AudioFingerprint_BeforeAfterSync}
\end{figure}
\subsection{Conclusion}
\label{sectionConclusionAF01}
When implementing an audio-based ad-hoc secure device pairing mechanism for previously unacquainted mobile devices, the diversity of hardware and software can affect the offset in audio recordings of even clock-synchronised mobile devices. 
We propose an approximative pattern matching algorithm to align the corresponding audio without communication between the devices.
The devices synchronise their audio sequences without any knowledge about the recorded audio on the remote device other than their own recorded contextual information.
Hence, no information about the audio utilized as a seed for the secure key generation, can leak. 
To improve the alignment quality, we can choose more than one matching position on each device at the cost of increasing the communication load for one device. 
We can obtained a synchronization among devices of less than $2$ milliseconds when both devices utilize up to 10 trials.
With moderate additional communication load, a synchronisation in the order of $10$ milliseconds is reasonable. 

\chapter{Discussion}
\section{Feedback based closed-loop carrier synchronisation: A sharp asymptotic bound, an asymptotically optimal approach, simulations and experiments}
Solutions proposed for carrier synchronisation comprise open-loop synchronisation methods such as round-trip synchronisation \cite{5931,5932,5933}.
In this scheme, the destination transmits beacons in opposed directions along a multi-hop circle in which each of the nodes appends its part of the overall message to the beacons.
Beamforming is achieved when the processing time along the multi-hop chain is identical in both directions.
This approach, however, does not scale with the size of a network.

Closed loop feedback based approaches include full-feedback techniques, in which carrier synchronisation is achieved in a master-slave manner.
The phase-offset among the carrier signals of destination nodes is corrected by a receiver node.
Diversity between RF-transmit signal components is achieved over CDMA channels \cite{5934}.
This approach is applicable only to small network sizes and requires sophisticated processing capabilities at the source nodes.

A more simple and less resource demanding implementation is the one-bit feedback based closed-loop synchronisation considered in \cite{5934,5920}.
The authors describe an iterative process in which $n$ source nodes $i\in[1,\dots,n]$ randomly adapt the phases $\gamma_i$ of their carrier signal $\Re\left(m(t)e^{j(2\pi (f_c+f_i)t+\gamma_i)}\right)$.
Here, $m(t)$ is the transmit message and $f_i$ denotes the frequency offset of node $i$ to a common carrier frequency $f_c$. 
Initially, i.i.d. phase offsets $\gamma_i$ of carrier signals are assumed.
When a receiver requests a transmission from the network, carrier phases are synchronised in an iterative process.

Our study of binary-feedback-based distributed adaptive transmit beamforming derived for the first time a concise theoretical understanding of this carrier synchronisation approach.
The potential to achieve carrier synchronisation by this iterative method had been reported in~\cite{5919} and brief estimates on its performance had been done based on empirical studies, for instance in~\cite{5918}.

The main breakthrough towards the estimation of the optimisation time has been done in preliminary work~\cite{4023,4024,4025} where the model was already described as evolutionary random optimisation. 
In previous empirical studies, the implemented algorithm frequently employed a normal random distribution for the mutation of transmit phases and a mutation probability of~$1$.
The upper and lower bounds have then been derived with standard tools for the analysis of simple evolutionary algorithms (assuming standard bit mutation and no crossover).
Namely, we utilised the method of the expected progress in order to derive the lower bound and the method of fitness-based partitions for the upper bound.

Based on the understanding of the problem domain, it was then also possible to derive further algorithmic improvements on the optimisation time.
Two examples for this are included in the publication, namely the hierarchical clustering of transmit nodes and the asymptotically optimal optimisation method. 
The former exploits the fact that the synchronisation time grows more than linearly with the number of nodes that participate in the synchronisation. 
Therefore, hierarchically clustering node sets and finding synchronisation for these smaller sets can reduce the overall synchronisation time. 

The asymptotically optimal approach exploits that, after random initialisation, with high probability, already more than half of all nodes are well synchronised to each other.
It is then sufficient to alter the phase of each node in the corresponding direction. 
This is done with a minimum number of transmissions by solving an equation with unknowns of the underlying fitness function.

Summarising, this publication solves the main open issues regarding binary-feedback-based distributed carrier synchronisation and, with the concise mathematical description of the underlying problem space, provides tools to design optimal beamformers based on randomly distributed phase perturbations. 

\section{A fast binary feedback-based distributed adaptive carrier synchronisation for transmission among clusters of disconnected IoT nodes in smart spaces}
The above studied closed-loop approach is computationally cheap at the cost of increasing the time required for carrier synchronisation. 
It utilises a binary feedback on the achieved synchronisation quality that is transmitted in each iteration from a remote receiver~\cite{Mudumbai_2009,5920}.
In particular, such binary feedback can be implemented by a simple on/off burst scheme also for sharply resource restricted devices.

The central optimisation procedure consists of $n$ devices $i\in[1,\dots,n]$ randomly altering the phases $\gamma_i$ of their carrier signal $\zeta_i(t)$
in each iteration.
Implicitly, with this process a global random search is implemented.
The search space $\mathcal{S}$ is spanned by all possible combinations of carrier frequencies and carrier phase offsets for all transmit nodes.

Each specific phase-frequency combination $s\in\mathcal{S}$ is associated with a score $\mathcal{F}_{sc}:\mathcal{S}\rightarrow\mathds{R}^+_0$ that denotes its synchronisation quality.
Without loss of generality we assume that the optimisation aim is to maximise $\mathcal{F}_{sc}$.
A natural choice to compute such a score value is, for instance, the Signal-to-Noise-Ratio (SNR) of the received sum signal.

Intuitively, each node may in one iteration alter its transmit carrier phase offset, superimpose a synchronisation signal simultaneously with all other nodes and receive a binary feedback on the quality of the synchronisation.
These iterations are repeated until a random distribution of carrier phases is achieved that scores a sufficient synchronisation quality~\cite{Seo_2008,Mudumbai_2010b,4024}.
Initially, independent and identically distributed (i.i.d.) phase offsets $\gamma_i$ of carrier signals are assumed.
Since a decreasing signal quality is not accepted, and since a global random search is implemented by this approach (every possible combination of carrier phase offsets of nodes has a positive probability in each iteration) the method eventually converges to the optimum with probability~1~\cite{Mudumbai_2010b}. 
For this result an idealised environment without noise and interference was considered.
In a realistic environment, the impact of the noise figure determines the accuracy that can be achieved.

In~\cite{CarrierSynchronisation_2012_Rahman}, an implementation of this carrier synchronisation approach was presented for software defined radio (SDR) devices which does not rely on any wired connections between devices (for instance, for clock synchronisation of the SDR nodes). 

The authors of~\cite{Seo_2008} then demonstrated in a case study that the method is feasible to synchronise frequency as well as phase of carrier signal components. 

In all previous studies, a global random search is considered, in which nodes choose their next carrier phase and frequency offset uniformly at random from all possible values.
However, this global random optimisation approach comparably slow in the optimisation process.

While an asymptotically optimal algorithm has been derived in~\cite{4022}, this requires more than binary feedback. 
In particular, function values of the feedback function have to be transmitted over the wireless channel. 
Assuming that the nodes seek for a common phase for transmission in order to establish a minimum Signal-To-Noise Ratio (SNR) at the receiver, this might be critical.
Binary feedback can be encoded as simple burst-protocol (transmission vs. no transmission) which can be read out also at very low SNR. 
More advanced protocols capable of transmitting additional information, however, also require an increased SNR.

Since the search space does not contain local optima~\cite{4023} we restrict the search neighbourhood to reduce the number of possible next  configurations in one iteration that would worsen the synchronisation quality.
We propose to modify the algorithm to follow a local random search instead of the previously applied global random search mechanism.
In particular, a node~$i$ will, when it changes its phase and frequency offset, draw the new values from a restricted neighbourhood of size~$\mathcal{N}$ that is centred around the current values of $\gamma_i$ (and $f_i$).     
This addresses a recent critique expressed in~\cite{CarrierSynchronisation_2012_Mudumbai} regarding the convergence speed for this binary feedback-based iterative adaptive carrier synchronisation.

The derived sharp asymptotic bounds on the expected optimisation time are again improved compared to the earlier studied global random search method.

Another question which was motivated by these studies on distributed adaptive transmit beamforming was on the impact of environmental effects.
In our case studies, we had observed that the performance of the binary feedback-based carrier synchronisation approach was perceptive to presence and movement of individuals.
In a first step, this motivated a series of studies on an adaptive learning approach for binary-feedback-based distributed adaptive beamforming in which the parameters of the optimisation algorithm -- namely the mutation probability and the variance of the normal distribution on the neighbourhood function -- were adapted according to the environment~\cite{OrganicComputing_Sigg_2011}.

These considerations let to our work on RF-based device-free recognition systems (cf. section~\ref{sectionOriginalRF01}, section~\ref{sectionOriginalRF02} and section~\ref{sectionOriginalRF03}) in which we considered the inmpact of different environmental conditions on the evolution of received RF-signals.

\section{RF-sensing of activities from non-cooperative subjects in device-free recognition systems using ambient and local signals}
In the approaching Internet of Things (IoT), virtually all entities in our environment will be enhanced by sensing, communication and computational capabilities~\cite{IoT_Li_2012,IoT_Haller_2010}.
These entities will provide information on environmental situations, interact in the computation and processing of data~\cite{FunctionComputation_Sigg_2012} and store information.
In order to sense environmental situations, common sensors in current applications are light, movement, pressure, audio or temperature~\cite{5840}.
Clearly, for reasons of cost and sensor size it is desired to minimise the count of distinct sensors in IoT entities.
The one sensor class that defines the minimum set naturally available in virtually all IoT devices is the Radio Frequency (RF)-transceiver to communicate with other wireless entities. 
It is also shipped with nearly every contemporary electronic device like mobile phones, notebooks, media players, printers as well as keyboards, mouses, watches, shoes and rumour has spread about even media cups. 
Therefore, the RF transceiver is a ubiquitously available sensor class.
It is capable of sensing changes or fluctuation in a received RF-signal.
Radio waves are blocked, reflected or scattered at objects.
At a receiver, the signal components from distinct signal paths add up to form a superimposition.
When objects that block or reflect the signal path of some of these signal components are moved, this is reflected in the superimposition of signal waves at the receiver.
We assert that specific activities in the proximity of a receiver generate characteristic patterns in the received superimposed RF-signal.
By identifying and interpreting these patterns, it is possible to detect activities of non-cooperating subjects in an RF-receiver's proximity.

In the context of indoor localisation, Youssef defines this approach as Device-Free Localisation (DFL) in~\cite{Pervasive_Youssef_2007} to localise or track a person using RF-Signals while the entity monitored is not required to carry an active transmitter or receiver.
They localised individuals by exchanging packets between 802.11b nodes in corners of a room and analysed the moving average and its variance of the RSSI~\cite{Pervasive_Youssef_2007}.
A passive radio map was constructed offline before a Bayesian-based inference algorithm estimated the most probable location.
These experiments have been conducted under Line-of-Sight (LoS) conditions.
Also, Wilson and Patwari showed in conformance with the findings of Kosba et al.~\cite{Pervasive_Kosba_2012b} that the variance of the RSSI can be used as an indicator of motion of non-actively transmitting individuals regardless of the average path loss that occurs due to dense walls and stationary objects~\cite{RFSensing_Wilson_2009}.
The area in which environmental changes impact signal characteristics was then considered by Zhang et al.
They used 870~MHz nodes arranged in a grid to show that for each link an elliptical area of about $0.5$ to 1 meters diameter exists for which RSSI fluctuation caused by an object traversing this area exceeds measurements in a static environment~\cite{RFSensing_Zhang_2009}.
They identified a valid region for detecting the impact (i.e. the RSSI fluctuations exceeding the measured threshold in a static environment) for transceiver distances from 2~m to 5~m for the considered 870~MHz frequency range~\cite{RFSensing_Zhang_2011}.
By dividing a room into hexagonal cell-clusters with measurements following a TDMA scheduling, an object position could be derived with an accuracy of around 1~meter.
This accuracy was further improved by Wilson and Patwari in 2011~\cite{RFSensing_Wilson_2009}.
They utilised a dense node array to locate individuals within a room with an average error of about $0.5$~meters.
This was possible by instrumenting a tomographic image over the 2-way RSSI fluctuations of nodes~\cite{RFSensing_Wilson_2010}.
All these studies consider a single experimental setting.

The simultaneous localisation of multiple individuals at the same time was first mentioned and studied by Patwari and Wilson in~\cite{Pervasive_Patwari_2011b}.
The authors derive a statistical model to approximate the position of a person based on RSSI variance which can be extended to multiple persons.
This aspect together with the previously untackled problem that environmental changes over time might necessitate frequent calibration of the location system was approached by Zhang and others in~\cite{Pervasive_Zhang_2012}.
The authors isolate the LoS path by extracting phase information from the differences in the RSS on various frequency spectra at distributed nodes.
Their experimental system is with this approach able to simultaneously and continuously localise up to 5 persons in a changing environment with an accuracy of 1~meter.

While the localisation of individuals based on features from the radio channel was well elaborated, activity recognition from RF was still in its infants.

Patwari et al. monitored breathing based on RSS analysis~\cite{RFSensing_Patwari_2011}.
The monitored area was surrounded by twenty 2.4~GHz nodes and the two-way RSSI was measured.
Using a maximum likelihood estimator they approximated the breathing rate within $0.1$ to $0.4$~beats accuracy.

Our presented work was preceded by several smaller studies in which we were considering a simple set of classes to distinguish~\cite{Pervasive_Scholz_2011} and also open issues in device-free recognition of situations and activities~\cite{Pervasive_Scholz_2011b}. 

Our study mainly focused on the recognition of the activities 'standing, lying, crawling, walking and empty' which we believe to be most expressive in emergency or elderly care situations.

The recognition of these activities has been considered by three device-free systems which we distinguish by the transmit and receive systems utilised. 
In particular, we utilised an active Software Defined Radio (SDR)-based Device-Free Activity Recognition (DFAR) system utilising Universal Software Radio Peripheral (USRP) transmit and receive device, an active RSSI-based DFAR system utilising sensor nodes as transmitters and recievers and a passive DFAR system exploiting ambient FM-radio. 
In all cases, unlike in most related work, we have restricted the system to a single receive device.
While the accuracy can be improved by the utilisation of multiple receive nodes, we envision that in practical applications the number of available nodes will also be limited. 

In comparison with a body-worn accelerometer, we could show that the recognition accuracy of the device-free systems on the recognition of the above mentioned classes in comparable.
Consequently, this shows that, for the recognition of simple classes in, for instance, emergency or elderly care situations, device-free systems can replace body-worn sensing systems.

This opens also a range of new applications in which the monitored entity is not required to wear any device.
Examples are cases in which the monitored entity is not cooperating or where the requirement to wear a device induces additional cost (e.g. monetary, comfort, freedom of movement).

The optimum features identified for all systems differ slightly but mainly cover the mean, variance or maximum of the received signal strength. 

\section{Monitoring of Attention Using Ambient FM-radio Signals}
After our initial studies which demonstrated the general feasibility of RF-based device-free recognition of activities via various recognition systems, we have considered various aspects of these systems regarding the recognition performance. 
Examples are the recognition of mulitple individuals~\cite{DeviceFreeRecognition_Sigg_2013} and also studies considering the distance of the monitored entity to a receive antenna~\cite{Pervasive_Sigg_2013,DeviceFreeRecognition_Shi_2013}.
One result from these studies was that static classes such as lying or standing, although they can in principle be detected from RF-based DFAR systems, require extensive training in a specific environment and it is therefore unrealistic to sense such classes in practical applications. 
Dynamic classes, however, such as walking, crawling or also walking speed can be well recognised also across different environments and slight modifications of parameters of the sensing system such as topology or transmission power. 

Likewise, other groups have been approaching RF-based passive recognition utilising different sensing technology like, for instance, RFID nodes~\cite{DeviceFreeRecognition_Wagner_2013} or also more sophisticated recognition algorithms~\cite{DeviceFreeRecognition_Hong_2013}. 
In this sense, RF-based device-free recognition received a broad recognition from various directions and we were also trying to push the limits of the recognition algorithms further towards the detection of non-directly observable properties such as attention.

In the featured study, we were considering attention of individuals towards poster frames in a corridor. 
Attention determines for a system the potential to impact the actions and decisions taken by an individual~\cite{AttentionMonitoring_Xu_2012}.
The management of attention covers the activation of attention as well as its detection and timely exploitation.
In the literature, we find various definitions that classify attention as well as its determining characteristics~\cite{AttentionMonitoring_Wu_2007,AttentionMonitoring_Wickens_1984}.
A straightforward measure of attention might be the tracking of gaze~\cite{AttentionMonitoring_Yonezawa_2007}. 
In general, aspects such as Saliency, Effort, Expectancy and Value are important indicators of attention~\cite{AttentionMonitoring_Wickens_2008,AttentionMonitoring_Wickens_1984}.
Others extended this model and put a greater stress on the effort a person takes towards an object~\cite{AttentionMonitoring_Ferscha_2012}.
Such detection and management of attention may require elaborate installations and very specific sensors in order to accurately sense quantities such as Saliency, Effort, Expectancy and Value~\cite{AttentionMonitoring_Xu_2012,AttentionMonitoring_Gollan_2011}.
However, we believe that for many commercial installations, cost and ease of installation and not primarily the highest achievable accuracy are most important.
Also more general, environmental sensors can provide sufficient information to estimate the attention state of individuals.

We propose to utilise ambient FM-radio signals for the detection of attention since it has a nearly perfect coverage in populated areas and features cheap receiver hardware~\cite{Percom_Popleteev_2012}.
In particular, authors in~\cite{AttentionMonitoring_Ferscha_2012} discuss various aspects of attention and identify as most distinguishing factors changes in walking speed, direction or orientation.
We therefore argue that attention levels can be inferred upon interpretation of the changes in walking speed or direction as derived from our system.

Our passive, FM-radio based DFAR system then detects the location and walking speed of persons in that corridor. 
From this information, for instance the location of the person and change in her walking speed, it is then possible to estimate attention levels towards specific poster frames in the corridor. 

This work also raises questions whether further complex classes, such as, for instance, sentiment can be reliably detected from device free RF-based sensing system.

\section{The Telepathic Phone: Frictionless Activity Recognition from WiFi-RSSI}
Throughout 2013, research campaigns in particular at the University of Washington and at the MIT have further pushed the limits of which classes can be accurately sensed from RF-based device-free recognition systems. 
In particular, the Doppler shift was exploited by these groups as a powerful feature to accurately sense fine-grained movement.

When an object reflecting a signal wave is in motion, this causes Doppler Shift. 
The direction and speed of the movement conditions the strength and nature of this frequency shift.
Pu and others showed that simultaneous detection of gestures from multiple individuals is possible by utilising multi-antenna nodes and micro Doppler fluctuations~\cite{RFsensing_Pu_2013,RFsensing_Kim_2009}.
They utilise a USRP SDR multi antenna receiver and one or more single antenna transmitters distributed in the environment to distinguish between a set of 9 gestures with an average accuracy of 0.94. 
Their active device-free system exploits a MIMO receiver in order to recognise gestures from different persons present at the same time. 
By leveraging a preamble gesture pattern, the receiver estimates the MIMO channel that maximises the reflections of the desired user.

A main challenge was for them that the Doppler shift from human movement was several magnitudes smaller than the bandwidth of the signal employed.
The authors therefore proposed to transform the received signal into several narrowband pulses which are then analysed for possible Doppler fluctuation.
The group discussed application possibilities of their system in~\cite{RFSensing_Kellog_2014}.

In a related system, Adib and Katabi employ MIMO interference nulling and combine samples taken over time to achieve a similar result while compensating for the missing spatial diversity in a single-antenna receiver system~\cite{Pervasive_Adib_2013}.
In their system, they leverage standard WiFi hardware at 2.4~GHz.

Later, this work was extended to 3D motion tracking by utlising three or more directional receive antennas in exactly defined relative orientation~\cite{RFSensing_Adib_2014}. 
In particular, the system is able to track the center of a human body with an error below 21cm in any direction and can also detect movement of body parts and directions of a pointing body part, such as a hand. 
This localisation is possible through time-of-flight estimation and triangulisation.
Higher accuracy of this estimation is granted by utilising frequency modulated carrier waves (sending a signal that changes linearly in frequency with time) over a bandwidth of 1.69~GHz.
Impact of static objects could be mitigated by subtracting successive sample means whereas noise was filtered by its speed of changes in energy over frequency bands. 

For these systems, however, sophisticated SDR devices are are necessary as well as extensive training of the system.

Towards another dimension, we were therefore considering the RF-based device-free recognition of activities, situations and gestures on consumer devices. 
On such device classes, the high accuracy and also the direct access to signals on the wireless channel are not provided. 
Instead, utilising IEEE 802.11g signals, we are restricted to the RSSI information calculated for each received packet.
It turned out that this information is particularly challenging since the packet reception is sparse, bursty and RSSI calculation is inaccurate and discrete.
We therefore limited our study to very simple time-domain features, mainly the mean RSSI, variance as well as minimum and maximum RSSI.
In this first study on passive RSSI-based DFAR, we investigated the distinction of distance, movement speed, crowd size as well as simple presence and interaction with the phone.
With about 20 packets per second, these classes can be recognised with an accuracy between 0.9 and 0.7. 
It was, however, not possible to distinguish the direction in which an activity was performed. 
Furthermore, gestures showed only a low recognition accuracy. 

With this study, the first to utilise passive RSSI on a smartphone device, we presented an alternative to the frequent use of accelerometers for the recognition of activities. 
When the smartphone is not worn on the body, which occurs more than 50\%\ of the time according to~\cite{Pervasive_Dey_2011,Pervasive_Patel_2006}, accelerometer-based sensing of activities is infeasible.
In such cases, i.e. when the smartphone is lying on the table, RSSI-based device-free sensing paradigms still enable the detection of simple activities or situations. 
Furthermore, this sensing paradigm enables a range of novel application cases in which environmental conditions might trigger actions on a smartphone. 

Rencent developments further advanced these ideas to also cover the recognition of emotional states from gestures captured from RSSI or channel state information (CSI) fluctuation~\cite{Muneeba_2017_Geospatial,Raja_2016_CoSDEO}.

\section{Secure communication based on ambient audio}
Secure authentication and communication in wireless settings is still a topic with a number of unsolved issues.
In mobile computing, when two devices meet for the first time, how is it possible to authenticate devices and to establish a secure key without the need of a trusted third party.

In the absence of a shared key, the first communication over the channel is necessarily unencrypted. 
Hence, this opens the door for an eavesdropper within the communication range. 
Furthermore, how can the communication partner be authenticated.
It is easy in such wireless scenario for an attacker to forge the identity of the legal communication partner.

Context as an implicit input might help to limit the amount of potential communication partners.
This concept was presented first 2001 by Holmquist et al.~\cite{ContextAwareness_Holmquist_2001}. 
The authors propose to utilise the accelerometer of the Smart-It~\cite{5836} device to extract characteristic features from simultaneous shaking processes of two devices.
Later, Mayrhofer et al. presented an authentication mechanism based on this principle~\cite{mayrhofer2007shake}.
The authors demonstrated, that an authentication is possible when devices are shaken simultaneously by a single person, while an authentication was unlikely for a third person trying to mimic the correct movement pattern remotely.
Also, Mayrhofer derived in~\cite{mayrhofer2007candidate} that the sharing of secret keys is possible with a similar protocol.
The proposed protocol that can be utilised with arbitrary contextual features repeatedly exchanges hashes of key-sub-sequences until a common secret is found.
In this instrumentation, exponentially quantised fast Fourier transformation (FFT) coefficients of a sequence of accelerometer samples are utilised.
In contrast, Bicher et al. describe an approach in which noisy acceleration readings can be utilised to establish a secure communication channel among devices~\cite{Cryptography_Bichler_2007,Cryptography_Bichler_2007-2}.
They utilise a hash function that maps similar acceleration patterns to identical key sequences.
However, their approach suffers from the required exact synchronisation among devices so that the authors computed the correct hash-values offline.
Additionally, the hash function utilised required that the keys computed exactly match and that the neighbourhood around these keys is precisely defined.
When patterns are located at the border of one of the region's neighbourhoods, the tolerance for noise in the input is biased in the direction of the centre of this region.
Additionally, key generation by simultaneous shaking is not unobtrusive.

In our publication we utilise an error correction scheme to account for noise in the input data which can be fine-tuned for any Hamming distance desired which is centred around the noisy characteristic sequences generated instead of an artificially defined centre value.
We implement a Network Time Protocol (NTP) based synchronisation mechanism that establishes sufficient synchronisation among nodes.

Another sensor class utilised for context-based device authentication is the RF-channel.
Varshavsky et al. present a technique to authenticate co-located devices based on RF-measurements since channel measurements from devices in near proximity are sufficiently similar to authenticate devices against each other~\cite{Cryptography_Varshavsky_2007}.
Hershey et al. utilise physical layer features to derive secret keys for a pair of devices~\cite{Cryptography_Hershey_1995}.
In the absence of interference and non-linear components, transmitter and receiver experience identical channel response~\cite{Cryptography_Smith_2004}.
This information is utilised to generate a secret key among a node pair.
Since channel characteristics are spatially sharply concentrated and not predictable at a remote location~\cite{Cryptography_Madiseh_2008}, an eavesdropper is not capable of guessing information about the secret.
This scheme was validated in an indoor environment in~\cite{Cryptography_BenHamida_2009}.
Although we consider the keys generated by this scheme as strong, it does not preserve spatial properties.
A device at arbitrary distance could pretend to be a nearby communication partner.

In our work we present an approach to establish a common cryptographic key for co-located devices from ambient audio, without leaking information to eavesdroppers in a different audio context.
We utilise purely ambient noise to establish a secure communication channel among devices in spatial proximity.
We record NTP-synchronised audio samples at two locations, generate a characteristic audio-fingerprint and map this fingerprint to a unique secret key with the help of error correcting codes.
The last step is necessary since the similarity between fingerprints is typically not sufficient to establish a secure channel.
With fuzzy-cryptography schemes, the generation of an identical key based on noisy input data~\cite{tuyls2007security} is possible.
The strength of the key established is conditioned on the length of the audio sequences considered but sequence of less than 6 seconds are already sufficient to obtain 512 bit keys. 

The main novelty of our approach compared to other context-based device pairing mechanisms is the use of fuzzy cryptography to mitigate differences in the audio fingerprints of co-located audio recordings.
Through the strong error correcting codes employed (we utilise Reed-Solomon error correcting codes), it is possible to exactly define how many differences are tolerated in fingerprints of co-located devices so as to establish identical cryptographic keys. 

We implemented the approach using python on a pair of desktop computers.
A critical problem we faced was the time synchronisation among devices. Standard NTP-based synchronisation showed to be too inaccurate. 
In particular, since the fingerprints are created from chunks of audio recordings in the time domain of length 0.375 seconds, a time difference of more than 20 milliseconds already results in significant difference of the generated fingerprint. 

Our study demonstrated that ambient audio is a rich source of randomness with good Entropy.
Indeed, employing sets of statistical tests we could not identify a bias in the generated fingerprints.

Our case studies also showed that even with partial information on the audio context, an eavesdropper can be prevented from obtaining a valid fingerprint. 
In this case, however, the parameters of the Reed-Solomon algorithm have to be stricter so that also more of the legitimate pairing attempts fail.
An automatic and adaptive choice of parameters in a given environment is yet an open issue to solve.

The presented approach raises the bar for attacks in wireless environments since attackers would have to establish physical proximity in order to be able to authenticate. 
Furthermore, it enables simpler, easier to use security schemes since part of the key can be derived automatically via ambient audio.

Li et al. analyse the usage of biometric or multimedia data as part of an authentication process and propose a protocol~\cite{li2006robust}.
Due to the use of error-tolerant cryptographic techniques, this protocol is robust against noise in the input data.
The authors utilise a secure sketch~\cite{Dodis04fuzzyextractors} to produce public information about an input without revealing it.
The input can then be recovered given another value that is close to it.
A similar study is presented by Miao et al.~\cite{miao2009biometrics}.
The authors establish a key distribution based on a fuzzy vault~\cite{Juels02fuzzyvault} using data measured by devices worn on the human body.
The fuzzy vault scheme, also utilised in~\cite{dodis2006robust}, enables the decryption of a secret with any key that is substantially similar to the key used for encryption.

Recently, the scheme developed by us in this work has been further advanced towards other domains, such as e.g. pairing conditioned on gait or co-presence on the same body~\cite{schurmann2017bandana}.

\section{Pattern-based Alignment of Audio Data for Ad-hoc Secure Device Pairing}
Our instrumentation described above requires idealised conditions regarding the synchronisation of devices and to account for this a high number of fingerprints must be created (201 in our experiments) in order to find one matching fingerprint.
For extensive computational load, this is feasible only in an offline approach.
The high number of fingerprints created, however, was necessary since the utilised NTP synchronisation is not sufficiently accurate.

In a later paper, we present an alignment mechanism which enables a synchronisation accuracy of recorded audio in the order of less than 10 milliseconds among unsynchronised mobile devices in the same context without transmitting information about the audio sequence over the wireless channel.
The synchronisation is achieved by processing a weakly NTP-synchronised recording without additional communication among the devices. 

This implementation solved several practical issues and constitutes a prototype implementation for smart and pervasive environments. 
The android-implementation is available online\footnote{https://github.com/stephansigg/AdhocPairing.git} and the method was also included in the Open-UAT API for secure applications on android~\footnote{http://www.openuat.org/}.

While the main pairing approach is identical, challenges arised through hardware peculiarities of the phones utilised. 

In particular, the hardware-based audio-pre-processing on one of the phones and the non-real-time capability of the operating system posed greatest difficulties.
While the hardware-pre-processing issue could be easily solved by restricting the frequency range from which audio-fingerprints are created, the non-real-time capability of the phone posed greater difficulties.
In particular, for the start of the audio recording, we experienced unpredictable delays of up to two seconds.
We solved the challenge to find an identical starting point in these sequences without inter-device-communication by utilising approximate pattern matching between a common short audio sequence and the recording on each of the devices.
While the common sequence was uncorrelated to the recording, a best matching would be found with high probability approximately at the same location within the audio recording provided that the recordings are similar. 

With this approach we have been able to identify an identical starting time for the generation of the fingerprints with an offset in the order of ten milliseconds.

\chapter{Acknowledgements}
A number of people have contributed to this thesis in various ways.
My work and studies towards this thesis have profited much from these positive contributions which have shaped and impacted the direction of my work.\\

First of all, I would like to thank my three girls, Nelja, Freyja and Aniko.
Thank you for your patience, understanding and support!\\

I would like to thank Professor Dr. Michael Beigl for giving me the opportunity to work with much freedom on these topics.
Without the financial support he mobilised to acquire the hardware for the experimental studies, many findings and experiences could not have been made.
Likewise, Professor Dr. Yusheng Ji has greatly supported this work with suggestions and advises and also financial support for equipment. 
Professor Dr. Lars Wolf helped me significantly with many issues after my return from japan and shared an office in his group which greatly helped to improve my working conditions at that time. 
Professor Dr. Xiaoming Fu greatly supported and re-integrated me into the German academic system and thereby essentially enabled the writing on this document. \\

Over the last seven years, I have had the opportunity to cooperate with many great scientists which have had a strong impact on my work and the results presented in this document. 
In particular, I would like to thank the team at TecO for stirring discussions and helpful advice on very diverse topics.
I have experienced a very competitive research environment in this group.
Dawud Gordon, in particular, has always been available for discussions and practical advice.
In Behnam Banitalebi I have found an electrical engineer with whom I have had the good fortune of working with.
Many of the studies on device-free recognition have been conducted in close cooperation with Markus Scholz. 
I am thankful for constructive discussions on this topic and for his patience. 
Furthermore, Matthias Budde, Till Riedel, Hedda Schmidtke, Predrag Jakimovski and Martin Berchtold have contributed to parts of the work presented here and I am very thankful for having had the chance to cooperate with them.

I would further like to thank the k-lab team at the National Institute for Informatics, in particular Shuyu Shi. 
I am honoured to have worked closely together with Shuyu during 2011 and 2012. 
Your passion and clear focus on research has inspired me greatly. 
Furthermore, I have had the chance to cooperate with passionate advanced scientists at NII, in particular, Olga Streibel, Sven Wohlgemuth and Christoph Lofi.

From the IBR at TU-Braunschweig, in particular Dominik Schuermann, Felix Buesching and Sebastian Schildt had a great impact on parts of the work presented here. 
Dominik has helped me a lot with security and cryptography-related issues and has significantly contributed to the implementation of the audio-based device pairing approach. 
Felix has introduced the RF-based recognition to the inga sensor nodes and has significantly contributed to our studies on the active RSSI-based device-free activity recognition system.

I am very thankful towards the team of the wearable computing laboratory of Professor Dr. Gerhard Troester at ETH Zurich. 
I was impressed by the high scientific standard each team member maintained. 
In particular, Ulf Blanke greatly contributed to our study on passive device-free activity recognition from received 802.11g packets.
I greatly enjoyed the constructive and positive discussions with Sinziana Mazilu, Alberto Calatroni, Sebastian Feese, Franz Gravenhorst, Giovanni Antonio Salvatore, Zhu Zack and with my roommate Simon Christen.

I would also like to mention the very warm welcome in the Nodes group of the University of Helsinki and good and fruitful discussions, in particular with Sasu Tarkoma, N. Asokan, Eemil Lagerspetz, Ella Peltronen, Petteri Nurmi, Hien Truong and Markus Miettinen. 
I have learned a lot during this stay, in particular, about big data management and algorithmic tools. 

Also, I am grateful for discussions and support with all kinds of academic issues by the ComTec team at Kassel University, in particular, Professor, Dr. Klaus David, Niklas Klein, Rico Kusber and Andreas Jahn.

Finally, the members of the ComNet group at University of Goettingen provide me with a new academic home and have inspired me with new ideas and research aspects towards future networking.
In particular, David Koll, Mayutan Arumaithurai, Konglin Zhu and Jiachen Chen. 

During these years I am grateful to have had the opportunity to closely work together with very talented and passionate students.
The interesting and deep discussions have often led to refinements or clarifications of my work.
In particular, I would like to mention Jialin Wang, Weiwei Liang, Rayan Merched El Masri, Julian Ristau, Aaron Israel, Georg von Zengen, Gerrit Bagschik, Toni Günther, Johannes Starosta, Sebastian Schwarzl, Markus Reschke, Timo Schulz, Sascha Lity, Stephen Roettger, Sergei Dechand, Ngu Nguyen, David Rieger, Mario Hock, An Huynh, Stephan Mueller, Philipp Specht, Christoph Rauterberg, Marko Becker, Matthias Velten and Chuong Thach Nguyen. \\

Finally, I would like to acknowledge partial funding by the 'Deutsche Forschungsgemeinschaft' (DFG) for the project ''Emergent radio'' as part of the priority program 1183 ''Organic Computing'' as well as funding by the German Academic Exchange Service (DAAD) for funding in the frame of the FIT-Weltweit program.

\printindex
\listoftables
\listoffigures

\small
\bibliographystyle{splncs}

\begin{thebibliography}{100}

\bibitem{4022}
Sigg, S., Masri, R.M.E., Beigl, M.:
\newblock Feedback based closed-loop carrier synchronisation: A sharp
  asymptotic bound, an asymptotically optimal approach, simulations and
  experiments.
\newblock Transactions on mobile computing \textbf{10}(11) (2011)  1605--1617

\bibitem{Beamforming_Sigg_2014}
Sigg, S.:
\newblock A fast binary feedback-based distributed adaptive carrier
  synchronisation for transmission among clusters of disconnected iot nodes in
  smart spaces.
\newblock Ad Hoc Netw. (2014) http://dx.doi.org/10.1016/j.adhoc.2013.12.006.

\bibitem{Pervasive_Sigg_2012}
Sigg, S., Scholz, M., Shi, S., Ji, Y., Beigl, M.:
\newblock Rf-sensing of activities from non-cooperative subjects in device-free
  recognition systems using ambient and local signals.
\newblock IEEE Transactions on Mobile Computing \textbf{13}(4) (2013)

\bibitem{Pervasive_Shi_2014}
Shi, S., Sigg, S., Zhao, W., Ji, Y.:
\newblock Monitoring of attention from ambient fm-radio signals.
\newblock IEEE Pervasive Computing, Special Issue - Managing Attention in
  Pervasive Environments (2014)

\bibitem{RFSensing_Sigg_2014}
Sigg, S., Blanke, U., Troester, G.:
\newblock The telepathic phone: Frictionless activity recognition from
  wifi-rssi.
\newblock In: IEEE International Conference on Pervasive Computing and
  Communications (PerCom). PerCom '14 (2014)

\bibitem{Cryptography_Mathur_2011}
Mathur, S., Miller, R., Varshavsky, A., Trappe, W., Mandayam, N.:
\newblock Proximate: Proximity-based secure pairing using ambient wireless
  signals.
\newblock In: Proceedings of the ninth International Conference on Mobile
  Systems, Applications and Services (MobiSys 2011). (2011)

\bibitem{Cryptography_Schuerman_2011}
Schuermann, D., Sigg, S.:
\newblock Secure communication based on ambient audio.
\newblock IEEE Transactions on mobile computing \textbf{12}(2) (2013)

\bibitem{Cryptography_Nguyen_2012-2}
Nguyen, N., Sigg, S., Huynh, A., Ji, Y.:
\newblock Pattern-based alignment of audio data for ad-hoc secure device
  pairing.
\newblock In: Proceedings of the 16th annual International Symposium on
  Wearable Computers (ISWC). (2012)

\bibitem{FunctionComputation_Sigg_2012}
Sigg, S., Jakimovski, P., Beigl, M.:
\newblock Calculation of functions on the rf-channel for iot.
\newblock In: 3rd International Conference on the Internet of Things (IOT).
  (2012)  107--113

\bibitem{InNetworkProcessing_Jakimovski_2012}
Jakimovski, P., Schmidtke, H.R., Sigg, S., Weiss, L., Chaves, F., Beigl, M.:
\newblock Collective communication for dense sensing environments.
\newblock Journal of Ambient Intelligence and Smart Environments (JAISE)
  \textbf{4}(2) (2012)

\bibitem{FunctionComputation_Sigg_2013}
Sigg, S., Jakimovski, P., Ji, Y., Beigl, M.:
\newblock Utilising an algebra of random functions to realise function
  calculation via a physical channel.
\newblock In: 14th IEEE workshop on Signal Processing Advances in Wireless
  Communications (SPAWC). (2013)

\bibitem{5911}
Culler, D., Estrin, D., Srivastava, M.:
\newblock Overview of sensor networks.
\newblock IEEE Computer \textbf{37}(8) (2004)  41--49

\bibitem{5912}
Zhao, F., Guibas, L.:
\newblock Wireless Sensor Networks: An Information Processing Approach.
\newblock Morgan Kaufmann, Los Altos, CA (2004)

\bibitem{5137}
Norman, D.:
\newblock The invisible computer.
\newblock MIT press (1999)

\bibitem{5916}
Butera, W.J.:
\newblock Programming a paintable computer.
\newblock PhD thesis, Massachusetts Institute of Technology (2002)

\bibitem{5908}
Pillutla, L., Krishnamurthy, V.:
\newblock Joint rate and cluster optimisation in cooperative mimo sensor
  networks.
\newblock In: Proceedings of the 6th IEEE Workshop on signal Processing
  Advances in Wireless Communications. (2005)  265--269

\bibitem{5884}
Scaglione, A., Hong, Y.W.:
\newblock Opportunistic large arrays: Cooperative transmission in wireless
  multihop ad hoc networks to reach far distances.
\newblock IEEE Transactions on Signal Processing \textbf{51}(8) (2003)
  2082--2092

\bibitem{5893}
Sendonaris, A., Erkop, E., Aazhang, B.:
\newblock Increasing uplink capacity via user cooperation diversity.
\newblock In: IEEE Proceedins of the International Symposium on Information
  Theory (ISIT). (2001)  156

\bibitem{5894}
Laneman, J., Wornell, G., Tse, D.:
\newblock An efficient protocol for realising cooperative diversity in wireless
  networks.
\newblock In: Proceedings of the IEEE International Symposium on Information
  Theory. (2001)  294

\bibitem{5885}
Hong, Y.W., Scaglione, A.:
\newblock Critical power for connectivity with cooperative transmission in
  wireless ad hoc sensor networks.
\newblock In: IEEE Workshop on Statistical Signal Processing. (2003)

\bibitem{5888}
Hong, Y.W., Scaglione, A.:
\newblock Energy-efficient broadcasting with cooperative transmission in
  wireless sensor networks.
\newblock IEEE Transactions on Wireless communications (2005)

\bibitem{5939}
Jayaweera, S.K.:
\newblock Energy analysis of mimo techniques in wireless sensor networks.
\newblock In: 38th conference on information sciences and systems. (2004)

\bibitem{5907}
{del Coso}, A., Sagnolini, U., Ibars, C.:
\newblock Cooperative distributed mimo channels in wireless sensor networks.
\newblock IEEE Journal on Selected Areas in Communications \textbf{25}(2)
  (2007)  402--414

\bibitem{4019}
Sigg, S., Beigl, M.:
\newblock Collaborative transmission in wsns by a (1+1)-ea.
\newblock In: Proceedings of the 8th International Workshop on Applications and
  Services in Wireless Networks (ASWN'08). (2008)

\bibitem{4020}
Sigg, S., Beigl, M.:
\newblock Randomised collaborative transmission of smart objects.
\newblock In: 2nd International Workshop on Design and Integration principles
  for smart objects (DIPSO2008) in conjunction with Ubicomp 2008. (2008)

\bibitem{Mudumbai_2009}
Mudumbai, R., Brown, D.R., Madhow, U., Poor, H.V.:
\newblock Distributed transmit beamforming: Challenges and recent progress.
\newblock IEEE Communications Magazine (2009)  102--110

\bibitem{5919}
Mudumbai, R., Wild, B., Madhow, U., Ramchandran, K.:
\newblock Distributed beamforming using 1 bit feedback: from concept to
  realization.
\newblock In: Proceedings of the 44th Allerton conference on communication,
  control and computation. (2006)  1020--1027

\bibitem{Barriac_2004}
Barriac, G., Mudumbai, R., Madhow, U.:
\newblock Distributed beamforming for information transfer in sensor networks.
\newblock In: Proceedings of the third International Workshop on Information
  Processing in Sensor Networks. (2004)

\bibitem{5923}
Mudumbai, R., Barriac, G., Madhow, U.:
\newblock On the feasibility of distributed beamforming in wireless networks.
\newblock IEEE Transactions on Wireless communications \textbf{6} (2007)
  1754--1763

\bibitem{5930}
Ochiai, H., Mitran, P., Poor, H.V., Tarokh, V.:
\newblock Collaborative beamforming for distributed wireless ad hoc sensor
  networks.
\newblock IEEE Transactions on Signal Processing \textbf{53}(11) (2005)  4110
  -- 4124

\bibitem{5940}
Chen, W., Yuan, Y., Xu, C., Liu, K., Yang, Z.:
\newblock Virtual mimo protocol based on clustering for wireless sensor
  networks.
\newblock In: Proceedings of the 10th IEEE Symposium on Computers and
  Commmunications. (2005)

\bibitem{5937}
Youssef, M., Yousif, A., El-Sheimy, N., Noureldin, A.:
\newblock A novel earthquake warning system based on virtual mimo wireless
  sensor netwroks.
\newblock In: Canadian conference on electrical and computer engineering.
  (2007)  932--935

\bibitem{5941}
{del Coso}, A., Savazzi, S., Spagnolini, U., Ibars, C.:
\newblock Virtual mimo channels in cooperative multi-hop wireless sensor
  networks.
\newblock In: 40th annual conference on information sciences and systems.
  (2006)  75--80

\bibitem{5938}
Jayaweera, S.K.:
\newblock Energy efficient virtual mimo based cooperative communications for
  wireless sensor networks.
\newblock IEEE Transactions on Wireless communications \textbf{5}(5) (2006)
  984--989

\bibitem{5898}
Laneman, J., Wornell, G.:
\newblock Distributed space-time coded protocols for exploiting cooperative
  diversity in wireless networks.
\newblock IEEE Transactions on Information theory \textbf{49}(10) (2003)
  2415--2425

\bibitem{5899}
Sendonaris, A., Erkip, E., Aazhang, B.:
\newblock User cooperation diversity -- part i: System description.
\newblock IEEE Transactions on Communications \textbf{51}(11) (2003)
  1927--1938

\bibitem{5900}
Zimmermann, E., Herhold, P., Fettweis, G.:
\newblock On the performance of cooperative relaying protocols in wireless
  networks.
\newblock European Transactions on Telecommunications \textbf{16}(1) (2005)
  5--16

\bibitem{5901}
Cover, T.M., Gamal, A.A.E.:
\newblock Capacity theorems for the relay channel.
\newblock IEEE Transactions on Information Theory \textbf{525}(5) (1979)
  572--584

\bibitem{5903}
Kramer, G., Gastpar, M., Gupta, P.:
\newblock Cooperative strategies and capacity theorems for relay networks.
\newblock IEEE Transactions on Information Theory \textbf{51}(9) (2005)
  3037--3063

\bibitem{5807}
Scaglione, A., Hong, Y.W.:
\newblock Cooperative models for synchronization, scheduling and transmission
  in large scale sensor networks: An overview.
\newblock In: 1st IEEE International Workshop on Computational Advances in
  Multi-Sensor Adaptive Processing. (2005)  60--63

\bibitem{5909}
Gupta, P., Kumar, R.P.:
\newblock The capacity of wireless networks.
\newblock IEEE Transactions on Information Theory \textbf{46}(2) (2000)
  388--404

\bibitem{Ochiai_2005}
Ochiai, H., Mitran, P., Poor, H.V., Tarokh, V.:
\newblock Collaborative beamforming for distributed wireless ad hoc sensor
  networks.
\newblock IEEE Transactions on Signal Processing \textbf{53}(11) (2005)  4110
  -- 4124

\bibitem{5905}
Simeone, O., Spagnolini, U.:
\newblock Capacity region of wireless ad hoc networks using opportunistic
  collaborative communications.
\newblock In: Proceedings of the International Conference on Communications
  (ICC). (2006)

\bibitem{5844}
Krohn, A., Beigl, M., Decker, C., Varona, D.G.:
\newblock Increasing connectivity in wireless sensor network using cooperative
  transmission.
\newblock In: 3rd International Conference on Networked Sensing Systems (INSS).
  (2006)

\bibitem{5843}
Krohn, A.:
\newblock Optimal non-coherent m-ary energy shift keying for cooperative
  transmission in sensor networks.
\newblock In: 31st IEEE International Conference on Acoustics, Speech, and
  Signal Processing (ICASSP). (2006)

\bibitem{5886}
Hong, Y.W., Scaglione, A.:
\newblock Cooperative transmission in wireless multi-hop ad hoc networks using
  opportunistic large arrays.
\newblock In: 4th IEEE workshop on Signal Processing Advances in Wireless
  Communications (SPAWC). (2003)

\bibitem{5931}
Brown, D.R., Prince, G., McNeill, J.:
\newblock A method for carrier frequency and phase synchronization of two
  autonomous cooperative transmitters.
\newblock In: sixth IEEE workshop on signal processing advances in wireless
  communications. (2005)

\bibitem{5932}
Brown, D.R., Poor, H.V.:
\newblock Time-slotted round-trip carrier synchronisation for distributed
  beamforming.
\newblock IEEE Transactions on Signal Processing \textbf{56} (2008)  5630--5643

\bibitem{5933}
Ozil, I., Brown, D.R.:
\newblock Time-slotted round-trip carrier synchronisation.
\newblock In: Proceedings of the 41st Asilomar conference on signals, signals
  and computers. (2007)  1781--1785

\bibitem{5934}
Tu, Y., Pottie, G.:
\newblock Coherent cooperative transmission from multiple adjacent antennas to
  a distant stationary antenna through awgn channels.
\newblock In: Proceedings of the IEEE Vehicular Technology Conference. (2002)
  130--134

\bibitem{5920}
Mudumbai, R., Hespanha, J., Madhow, U., Barriac, G.:
\newblock Scalable feedback control for distributed beamforming in sensor
  networks.
\newblock In: Proceedings of the IEEE International Symposium on Information
  Theory. (2005)  137--141

\bibitem{Mudumbai_2010b}
Mudumbai, R., Hespanha, J., Madhow, U., Barriac, G.:
\newblock Distributed transmit beamforming using feedback control.
\newblock IEEE Transactions on Information Theory \textbf{56}(1) (2010)

\bibitem{Seo_2008}
Seo, M., Rodwell, M., Madhow, U.:
\newblock A feedback-based distributed phased array technique and its
  application to 60-ghz wireless sensor network.
\newblock In: IEEE MTT-S International Microwave Symposium Digest. (2008)
  683--686

\bibitem{Bucklew_2008}
Bucklew, J.A., Sethares, W.A.:
\newblock Convergence of a class of decentralised beamforming algorithms.
\newblock IEEE Transactions on Signal Processing \textbf{56}(6) (2008)
  2280--2288

\bibitem{2105}
Bennett, W.:
\newblock Introduction to signal transmission.
\newblock McGraw-Hill (1971)

\bibitem{5811}
Krohn, A.:
\newblock Superimposed Radio Signals for Wireless Sensor Networks.
\newblock PhD thesis, Technical University of Braunschweig (2007)

\bibitem{062}
3GPP:
\newblock 3rd generation partnership project; technical specification group
  radio access networks; 3g home nodeb study item technical report (release 8).
\newblock Technical Report 3GPP TR 25.820 V8.0.0 (2008-03) (March)

\bibitem{InNetworkProcessing_Jakimovski_2011}
Jakimovski, P., Becker, F., Sigg, S., Schmidtke, H.R., Beigl, M.:
\newblock Collective communication for dense sensing environments.
\newblock In: 7th IEEE International Conference on Intelligent Environments
  (IE). (2011) (**Best paper**).

\bibitem{4023}
Sigg, S., Beigl, M.:
\newblock Algorithms for closed-loop feedback based distributed adaptive
  beamforming in wireless sensor networks.
\newblock In: Proceedings of the fifth International Conference on Intelligent
  Sensors, Sensor Networks and Information Processing - Symposium on Adaptive
  Sensing, Control, and Optimization in Sensor Networks. (2009)

\bibitem{Hagmann_1981}
Hagmann, W.:
\newblock Network synchronisation techniques for satellite communication
  systems.
\newblock PhD thesis, USC, Los Angeles (1981)

\bibitem{DistributedBeamforming_Mudumbai_2011}
Mudumbai, R., Madhow, U., Brown, R., Bidigare, P.:
\newblock Dsp-centric algorithms for distributed transmit beamforming.
\newblock In: 2011 Conference Record of the 45th Asilomar Conference on
  Signals, Systems and Computers (ASILOMAR). (2011)  93--98

\bibitem{DistributedBeamforming_Quitin_2013}
Quitin, F., Rahman, M.M.U., Mudumbai, R., Madhow, U.:
\newblock A scalable architecture for distributed transmit beamforming with
  commodity radios: Design and proof of concept.
\newblock IEEE Transactions on Wireless Communications (\textbf{12}(3))
  1418--1428

\bibitem{CarrierSynchronisation_2012_Mudumbai}
Mudumbai, R., Bidigare, P., Pruessing, S., Dasgupta, S., Oyarzun, M., Raeman,
  D.:
\newblock Scalable feedback algorithms for distributed transmit beamforming in
  wireless networks.
\newblock In: IEEE Internatinal Conference on Acoustics, Speech and Signal
  Processing (ICASSP). (2012)  5213 --5216

\bibitem{4032}
Masri, R.M.E., Sigg, S., Beigl, M.:
\newblock An asymptotically optimal approach to the distributed adaptive
  transmit beamforming in wireless sensor networks.
\newblock In: Proceedings of the 16th European Wireless Conference. (2010)

\bibitem{4024}
Sigg, S., Beigl, M.:
\newblock Algorithmic approaches to distributed adaptive transmit beamforming.
\newblock In: Fifth International Conference on Intelligent Sensors, Sensor
  Networks and Information Processing - Symposium on Theoretical and Practical
  Aspects of Large-scale Wireless Sensor Networks. (2009)

\bibitem{CarrierSynchronisation_2012_Rahman}
Rahman, M.M., Baidoo-Williams, H.E., Mudumbai, R., Dasgupta, S.:
\newblock Fully wireless implementation of distributed beamforming on a
  software-defined radio platform.
\newblock In: Proceedings of the 11th international conference on Information
  Processing in Sensor Networks. IPSN '12 (2012)  305--316

\bibitem{DistributedBeamforming_Quitin_2012}
Quitin, F., Rahman, M.M.U., Mudumbai, R., Madhow, U.:
\newblock Distributed beamforming with software-defined radios: frequency
  synchronisation and digital feedback.
\newblock In: Proceedings of the 55th International Global Communications
  Conference (Globecom). (2012)

\bibitem{Algorithm_Savage_1995}
Savage, C., Winkler, P.:
\newblock Monotone gray codes and the middle levels problem.
\newblock Journal of Combinatorial Theory \textbf{70}(2) (1995)  230--248

\bibitem{Algorithms_Knuth_2011}
Knuth, D.E.:
\newblock The Art of Computer Programming -- Combinatarial Algorithms, Part 1.
\newblock Addison-Wesley (2011)

\bibitem{4036}
Reschke, M., Starosta, J., Schwarzl, S., Sigg, S.:
\newblock Situation awareness based on channel measurements.
\newblock In: Vehicular Technology Conference (VTC Spring), 2011 IEEE 73rd.
  (2011)

\bibitem{ContextAwareness_Sigg_2011}
Reschke, M., Schwarzl, S., Starosta, J., Sigg, S., Beigl, M.:
\newblock Context awareness through the rf-channel.
\newblock In: Proceedings of the 2nd workshop on Context-Systems Design,
  Evaluation and Optimisation. (2011)

\bibitem{IoT_Li_2012}
Li, T., Chen, L.:
\newblock Internet of things: Priinciples, frameworks and applications.
\newblock In: Proceedings of the Future Wireless Networks and Information
  Systems. Number 144 in LNEE (2008)  477--842

\bibitem{IoT_Haller_2010}
Haller, S.:
\newblock The things in the internet of things.
\newblock In: Proceedings of the Internet of Things Conference 2010. (2010)

\bibitem{5840}
Beigl, M., Krohn, A., Zimmer, T., Decker, C.:
\newblock Typical sensors needed in ubiquitous and pervasive computing.
\newblock In: First International Workshop on Networked Sensing Systems.
  Volume~4 of Society of Instrument and Control Engineers. (2004)  153--158

\bibitem{Pervasive_Scholz_2011}
Scholz, M., Sigg, S., Shihskova, D., von Zengen, G., Bagshik, G., Guenther, T.,
  Beigl, M., Ji, Y.:
\newblock Sensewaves: Radiowaves for context recognition.
\newblock In: Video Proceedings of the 9th International Conference on
  Pervasive Computing (Pervasive 2011). (2011)

\bibitem{RFSensing_Woyach_2006}
Woyach, K., Puccinelli, D., Haenggi, M.:
\newblock Sensorless sensing in wireless networks: implementation and
  measurements.
\newblock In: Proceedings of the Second International Workshop on Wireless
  Network Measurement (WiNMee). (2006)

\bibitem{RFSensing_Muthukrishnan_2007}
Muthukrishnan, K., Lijding, M., Meratnia, N., Havinga, P.:
\newblock Sensing motion using spectral and spatial analysis of wlan rssi.
\newblock In: Proceedings of Smart Sensing and Context. (2007)

\bibitem{Pervasive_Youssef_2007}
Youssef, M., Mah, M., Agrawala, A.:
\newblock Challenges: Device-free passive localisation for wireless
  environments.
\newblock In: Proceedings of the 13th annual ACM international Conference on
  Mobile Computing and Networking (MobiCom 2007). (2007)  222--229

\bibitem{Pervasive_Seifeldin_2013}
Seifeldin, M., Saeed, A., Kosba, A., El-Keyi, A., Youssef, M.:
\newblock Nuzzer: A large-scale device-free passive localization system for
  wireless environments.
\newblock IEEE Transactions on Mobile Computing (TMC) \textbf{12}(7) (2013)

\bibitem{Pervasive_Scholz_2011b}
Scholz, M., Sigg, S., Schmidtke, H.R., Beigl, M.:
\newblock Challenges for device-free radio-based activity recognition.
\newblock In: Proceedings of the 3rd workshop on Context Systems, Design,
  Evaluation and Optimisation, in conjunction with MobiQuitous 2011. (2011)

\bibitem{Pervasive_Shi_2012}
Shi, S., Sigg, S., Ji, Y.:
\newblock Passive detection of situations from ambient fm-radio signals.
\newblock In: Proceedings of the 2012 ACM Conference on Ubiquitous Computing.
  UbiComp '12 (2012)

\bibitem{Pervasive_Shi_2012b}
Shi, S., Sigg, S., Ji, Y.:
\newblock Activity recognition from radio frequency data: Multi-stage
  recognition and features.
\newblock In: IEEE Vehicular Technology Conference (VTC Fall). (2012)

\bibitem{Pervasive_Ravi_2005}
Ravi, N., Dandekar, N., Mysore, P., Littman, M.L.:
\newblock Activity recognition from accelerometer data.
\newblock In: Proceedings of the 17th conference on Innovative applications of
  artificial intelligence - Volume 3. IAAI'05 (2005)  1541--1546

\bibitem{Pervasive_Cao_2012}
Cao, H., Nguyen, M.N., Phua, C.C.W., Krishnaswamy, S., Li, X.:
\newblock An integrated framework for human activity classification.
\newblock In: Proceedings of the 14th ACM International Conference on
  Ubiquitous Computing (UbiComp 2012). (2012)

\bibitem{Pervasive_Bao_2004}
Bao, L., Intille, S.S.:
\newblock Activity recognition from user-annotated acceleration data.
\newblock In: Proceedings of PERVASIVE 2004. Volume LNCS 3001. (2004)

\bibitem{Pervasive_Ploetz_2012}
Ploetz, T., Hammerla, N.Y., Rozga, A., Reavis, A., Call, N., Abowd, G.D.:
\newblock Automatic assessment of problem behavior in individuals with
  developmental disabilities.
\newblock In: Proceedings of the 14th ACM International Conference on
  Ubiquitous Computing (Ubicomp 2012). (2012)

\bibitem{Pervasive_Abdullah_2012}
Abdullah, S., Lane, N.D., Choudhury, T.:
\newblock Towards population scale activity recognition: A scalable framework
  for handling data diversity.
\newblock In: Proceedings of the 26th conference on artificial intelligence
  (AAAI 2012). (2012)

\bibitem{Pevasive_Chavarriaga_2011}
Chavarriaga, R., Bayati, H., del R.~Millan, J.:
\newblock Unsupervised adaptation for acceleration-based activity recognition:
  robustness to sensor displacement and rotation.
\newblock Personal and Ubiquitous Computing (2011)

\bibitem{Pervasive_Cohn_2012}
Cohn, G., Morris, D., Patel, S.N., Tan, D.S.:
\newblock Humantenna: Using the body as an antenna for real-time whole-body
  interaction.
\newblock In: Proceedings of ACM CHI 2012. (2012)

\bibitem{Pervasive_Patridge_2008}
Patridge, K., Golle, P.:
\newblock On using existing time-use study data for ubiquitous computing
  applications.
\newblock In: Proceedings of the 10th International Conference on Ubiquitous
  Computing. (2008)  144--153

\bibitem{Pervasive_Varshavsky_2008}
Varshavsky, A., Patel, S.N.
\newblock Ubiquitous Computing Fundamentals. In: Location in Ubiquitous
  Computing. Taylor and Francis Group (2008)  285--319

\bibitem{Pervasive_Lane_2010}
Lane, N.D., Miluzzo, E., Lu, H., Peebles, D., Choudhury, T., Campbell, A.T.:
\newblock A survey of mobile phone sensing.
\newblock IEEE Communications magazine \textbf{48}(9) (2010)  140--150

\bibitem{Pervasive_Dey_2011}
Dey, A.K., Wac, K., Ferreira, D., Tassini, K., Hong, J.H., Ramos, J.:
\newblock Getting closer: An empirical investigation of the proximity of user
  to their smart phones.
\newblock In: Proceedings of the 13th international conference on Ubiquitous
  computing. (2011)

\bibitem{Pervasive_Patel_2006}
Patel, S.N., Kientz, J.A., Hayes, G.R., Bhat, S., Abowd, G.D.:
\newblock Farther than you may think: An empirical investigation of the
  proximity of users to their mobile phones.
\newblock In: Proceedings of the 8th international conference on Ubiquitous
  computing. (2006)  123--140

\bibitem{Pervasive_Hongeng_2004}
Hongeng, S., Nevatia, R., Bremond, F.:
\newblock Video based event recognition: Activity representation and
  probabilistic recognition methods.
\newblock Computer Vision and Image Understanding \textbf{96} (2004)  129--162

\bibitem{Pervasive_Bannach_2008}
Bannach, D., Lukowicz, P., Amft, O.:
\newblock Rapid prototyping of activity recognition applications.
\newblock Pervasive Computing, IEEE \textbf{7}(2) (2008)  22 --31

\bibitem{Pervasive_Patel_2007}
Patel, S.N., Robertson, T., Kientz, J.A., Reynolds, M.S., Abowd, G.D.:
\newblock At the flick of a switch: Detecting and classifying unique electrical
  events on the residential power line.
\newblock In: Proceedings of the 9th International Conference on Ubiquitous
  Computing (UbiComp 2007). (2007)  271--288

\bibitem{RFSensing_Gupta_2010}
Gupta, S., Reynolds, M.S., Patel, S.N.:
\newblock Electrisense: Single-point sensing using emi for electrical event
  detection and classificaiton in the home.
\newblock In: Proceedings of the 13th international conference on Ubiquitous
  computing. (2010)

\bibitem{Pervasive_Gupta_2011}
Gupta, S., Chen, K.Y., Reynolds, M.S., Patel, S.N.:
\newblock Lightwave: Using compact fluorescent lights as sensors.
\newblock In: Proceedings of the 13th international conference on Ubiquitous
  computing. (2011)

\bibitem{Pervasive_Campbell_2010}
Campbell, A.T., Larson, E., Cohn, G., Froehlich, J., Alcaide, R., Patel, S.N.:
\newblock Wattr: A method for self-prowered wireless sensing of water activity
  in the home.
\newblock In: Proceedings of the 12th international conference on Ubiquitous
  computing. (2010)

\bibitem{Pervasive_Thomaz_2012}
Thomaz, E., Bettadapura, V., Reyes, G., Sandesh, M., Schindler, G., Ploetz, T.,
  Abowd, G.D., Essa, I.:
\newblock Recognizing water-based activities in the home through
  infrastructure-mediated sensing.
\newblock In: Proceedings of the 14th ACM International Conference on
  Ubiquitous Computing (Ubicomp 2012). (2012)

\bibitem{Pervasive_Cohn_2010}
Cohn, G., Stuntebeck, E., Pandey, J., Otis, B., Abowd, G.D., Patel, S.N.:
\newblock Snupi: Sensor nodes utilizing powerline infrastructure.
\newblock In: Proceedings of the 12th international conference on Ubiquitous
  computing (UbiComp 2010). (2010)

\bibitem{Pervasive_Bahl_2000}
Bahl, P., Padmanabhan, V.:
\newblock Radar: an in-building rf-based user location and tracking system.
\newblock In: Proceedings of the 19th IEEE International Conference on Computer
  Communications (Infocom). (2000)

\bibitem{Pervasive_Otsason_2005}
Otsason, V., Varshavsky, A., LaMarca, A., de~Lara, E.:
\newblock Accurate gsm indoor localisation.
\newblock In: Proceedings of the 7th ACM International Conference on Ubiquitous
  Computing (Ubicomp 2005). (2005)

\bibitem{Pervasive_Varshavsky_2007}
Varshavsky, A., de~Lara, E., Hightower, J., LaMarca, A., Otsason, V.:
\newblock Gsm indoor localization.
\newblock Pervasive and Mobile Computing \textbf{3} (2007)

\bibitem{Pervasive_Krumm_2003}
Krumm, J., Cermak, G.:
\newblock Rightspot: A novel sense of location for a smart personal object.
\newblock In: Proceedings of the 5th ACM International Conference on Ubiquitous
  Computing (Ubicomp 2003). (2003)

\bibitem{Pervasive_Youssef_2005}
Youssef, A., Krumm, J., Miller, E., Cermak, G., Horvitz, E.:
\newblock Computing location from ambient fm radio signals.
\newblock In: Proceedings of the IEEE Wireless Communications and Networking
  Conference. (2005)

\bibitem{Pervasive_Stuntebeck_2008}
E.P.~Stuntebeck, S.P., Robertson, T., Reynolds, M., Abowd, G.:
\newblock Wideband powerline positioning for indoor localization.
\newblock In: Proceedings of the 10th ACM International Conference on
  Ubiquitous Computing (Ubicomp 2008). (2008)

\bibitem{Pervasive_Jiang_2012}
Jiang, Y., Pan, X., Li, K., Lv, Y., Dick, R.P., Hannigan, M., Shang, L.:
\newblock Ariel: Automatic wi-fi based room fingerprinting for indoor
  localization.
\newblock In: Proceedings of the 14th ACM International Conference on
  Ubiquitous Computing (Ubicomp 2012). (2012)

\bibitem{Pervasive_Pulkkinen_2012}
Pulkkinen, T., Nurmi, P.:
\newblock Awesom: Automatic discrete partitioning of indoor spaces for wifi
  fingerprinting.
\newblock In: Proceedings of the 10th International Conference on Pervasive
  Computing. (2012)

\bibitem{Pervasive_Schougaard_2012}
Schougaard, K.R., Gronbaek, K., Scharling, T.:
\newblock Indoor pedestrian navigation based on hybrid route planning and
  location modelling.
\newblock In: Proceedings of the 10th International Conference on Pervasive
  Computing (Pervasive2012). (2012)

\bibitem{Pervasive_Wang_2012}
Wang, H., Sen, S., Elgohary, A., Farid, M., Youssef, M., Choudhury, R.R.:
\newblock No need to war-drive -- unsupervised indoor localization.
\newblock In: Proceedings of the 10th International Conference on Mobile
  Systems, Applications and Services (Mobisys2012). (2012)

\bibitem{Pervasive_Chen_2012}
Chen, L., Nugent, C., Wang, H.:
\newblock A knowledge-driven approach to activity recognition in smart homes.
\newblock Knowledge and Data Engineering, IEEE Transactions on \textbf{24}(6)
  (2012)  961--974

\bibitem{RFSensing_Anderson_2006}
Anderson, I., Muller, H.:
\newblock Context awareness via gsm signal strength fluctuation.
\newblock In: 4th international conference on pervasive computing, late
  breaking results. (2006)

\bibitem{RFSensing_Sohn_2006}
Sohn, T., Varshavsky, A., LaMarca, A., Chen, M.Y., Choudhury, T., Smith, I.,
  Consolvo, S., Hightower, J., Grisworld, W.G., de~Lara, E.:
\newblock Mobility detection using everyday gsm traces.
\newblock In: Proceedings of the 8th international conference on Ubiquitous
  computing. (2006)

\bibitem{RFSensing_Sen_2012}
Sen, S., Radunovic, B., Choudhury, R.R., Minka, T.:
\newblock You are facing the mona lisa -- spot localization using phy layer
  information.
\newblock In: Proceedings of the 10th International Conference on Mobile
  Systems, Applications and Services (Mobisys2012). (2012)

\bibitem{Pervasive_Kosba_2012b}
Kosba, A., Youssef, M.:
\newblock Rasid demo: A robust wlan device-free passive motion detection
  system.
\newblock In: 2012 IEEE International Conference on Pervasive Computing and
  Communications Workshops (PERCOM Workshops). (2012)  531--533

\bibitem{RFSensing_Wilson_2009}
Wilson, J., Patwari, N.:
\newblock See-through walls: Motion tracking using variance-based radio
  tomography.
\newblock IEEE Transactions on Mobile Computing \textbf{10}(5) (2011)  612--621

\bibitem{RFSensing_Zhang_2009}
Zhang, D., Ni, L.:
\newblock Dynamic clustering for tracking multiple transceiver-free objects.
\newblock In: Proceedings of the 7th IEEE International Conference on Pervasive
  Computing and Communications. (2009)

\bibitem{RFSensing_Zhang_2011}
Zhang, D., Liu, Y., Ni, L.:
\newblock Rass: A real-time, accurate and scalable system for tracking
  transceiver-free objects.
\newblock In: Proceedings of the 9th IEEE International Conference on Pervasive
  Computing and Communications (PerCom2011). (2011)

\bibitem{RFSensing_Wilson_2010}
Wilson, J., Patwari, N.:
\newblock Radio tomographic imaging with wireless networks.
\newblock IEEE Transactions on Mobile Computing \textbf{9} (2010)  621--632

\bibitem{RFSensing_Lee_2010}
Lee, P.W.Q., Seah, W.K.G., Tan, H.P., Yao, Z.:
\newblock Wireless sensing without sensors - an experimental study of
  motion/intrusion detection using rf irregularity.
\newblock Measurement science and technology \textbf{21} (2010)

\bibitem{Pervasive_Kosba_2011}
Kosba, A.E., Saeed, A., Youssef, M.:
\newblock Rasid: A robust wlan device-free passive motion detection system.
\newblock In: IEEE International Conference on Pervasive Computing and
  Communications (PerCom). (2012)

\bibitem{Pervasive_Patwari_2011b}
Patwari, N., Wilson, J., Ananthanarayanan, S., Kasera, S.K., Westenskow, D.:
\newblock Monitoring breathing via signal strength in wireless networks (2011)
  submitted to IEEE Transactions on Mobile Computing, 18 Sept., 2011,
  available: arXiv:1109.3898v1.

\bibitem{Pervasive_Zhang_2012}
Zhang, D., Liu, Y., Guo, X., Gao, M., Ni, L.M.:
\newblock On distinguishing the multiple radio paths in rss-based ranging.
\newblock In: Proceedings of the 31st IEEE International Conference on Computer
  Communications. (2012)

\bibitem{RFSensing_Patwari_2011}
Patwari, N., Wilson, J.:
\newblock Spatial models for human motion-induced signal strength variance on
  static links.
\newblock IEEE Transactions on Information Forensics and Security \textbf{6}(3)
  (2011)  791--802

\bibitem{OrganicComputing_Sigg_2011}
Sigg, S., Beigl, M., Banitalebi, B.:
\newblock 5.4.
\newblock Organic Computing - A Paradigm Shift for Complex Systems, Autonomic
  Systems Series. In: Efficient adaptive communication from multiple resource
  restricted transmitters. Springer (2011)

\bibitem{Pervasive_Song_2010}
Song, W.Z., Huang, R., Xu, M., Shirazi, B., LaHusen, R.:
\newblock Design and deployment of sensor network for real-time high-fidelity
  volcano monitoring.
\newblock IEEE Transactions on Parallel and Distributed Systems \textbf{21}(11)
  (2010)  1658--1674

\bibitem{Pervasive_Palmer_2012}
Palmer, N., Kemp, R., Kielmann, T., Bal, H.:
\newblock The case for smartphones as an urgent computing client platform.
\newblock In: Proceedings of the International Conference on Computational
  Science, ICCS 2012. Volume~9. (2012)  1667--1676

\bibitem{Pervasive_Ramirez_2011}
Ramirez, L., Dyrks, T., Gerwinski, J., Betz, M., Scholz, M., Wulf, V.:
\newblock Landmarke: an ad hoc deployable ubicomp infrastructure to support
  indoor navigation of firefighters.
\newblock Personal and Ubiquitous Computing \textbf{16}(8) (2012)  1025--1038

\bibitem{HassLass2002}
for the Prevention~of Accidents~(RoSPA), T.R.S.:
\newblock Home and leisure accident surveillance system.
\newblock http://www.hassandlass.org.uk/query/index.htm (2002)

\bibitem{Pervasive_Nouri_2007}
Noury, N., Fleury, A., Rumeau, P., Bourke, A., Laighin, G., Rialle, V., Lundy,
  J.:
\newblock Fall detection - principles and methods.
\newblock In: 29th Annual International Conference of the IEEE Engineering in
  Medicine and Biology Society. (2007)  1663--1666

\bibitem{Pervasive_Cucchiara_2007}
Cucchiara, R., Prati, A., Vezzani, R.:
\newblock A multi-camera vision system for fall detection and alarm generation.
\newblock Expert Systems \textbf{24}(5) (2007)  334--345

\bibitem{Pervasive_Lin_2007}
Lin, C.W., Ling, Z.H.:
\newblock Automatic fall incident detection in compressed video for intelligent
  homecare.
\newblock In: Proceedings of 16th International Conference on Computer
  Communications and Networks (ICCCN). (2007)  1172--1177

\bibitem{Noury_Book_2011}
Noury, N., Poujaud, J., Fleury, A., Nocua, R., Haddidi, T., Rumeau, P.:
\newblock Smart sweet home... a pervasive environment for sensing our daily
  activity?
\newblock In Chen, L., Nugent, C.D., Biswas, J., Hoey, J., Khalil, I., eds.:
  Activity Recognition in Pervasive Intelligent Environments. Volume~4 of
  Atlantis Ambient and Pervasive Intelligence.
\newblock Atlantis Press (2011)  187--208

\bibitem{Pervasive_Vacher_2011}
Vacher, M., Istrate, D., Portet, F., Joubert, T., Chevalier, T., Smidtas, S.,
  Meillon, B., Lecouteux, B., Sehili, M., Chahuara, P., Meniard, S.:
\newblock The sweet-home project: Audio technology in smart homes to improve
  well-being and reliance.
\newblock In: Proceedings of the 33rd Annual International Conference of the
  IEEE Engineering in Medicine and Biology Society (EMBC 2011). (2011)

\bibitem{Pervasive_Lukovicz_2012}
Lukowicz, P., Pentland, S., Ferscha, A.:
\newblock From context awareness to socially aware computing.
\newblock Pervasive Computing, IEEE \textbf{11}(1) (2012)  32 --41

\bibitem{MachineLearning_Kononenko_1991}
Kononenko, I., Bratko, I.:
\newblock Information-based evaluation criterion for classifier's performance.
\newblock Machine Learning \textbf{6} (1991)  67--80

\bibitem{AttentionMonitoring_Xu_2012}
Xu, Y., Stojanovic, N., Stojanovic, L., Schuchert, T.:
\newblock Efficient human attention detection based on intelligent complex
  event processing.
\newblock In: Proceedings of the 6th ACM International Conference on
  Distributed Event-Based Systems. DEBS '12 (2012)  379--380

\bibitem{AttentionMonitoring_Wu_2007}
Wu, F., Hubermann, B.:
\newblock Novelty and collective attention.
\newblock In: Proceedings of the National Academics of Sciences. Volume 104.
  (2007)  17599--17601

\bibitem{AttentionMonitoring_Yonezawa_2007}
Yonezawa, T., Yamazoe, H., Utsumi, A., Abe, S.:
\newblock Gaze-communicative behavior of stuffed-toy robot with joint attention
  and eye contact based on ambient gaze-tracking.
\newblock In: ICMI. (2007)  140--145

\bibitem{AttentionMonitoring_Wickens_2008}
Wickens, C., McCarley, J.:
\newblock Applied attention theory.
\newblock CRC Press (2008)

\bibitem{AttentionMonitoring_Ferscha_2012}
Ferscha, A., Zia, K., Gollan, B.:
\newblock Collective attention through public displays.
\newblock In: 2012 IEEE Sixth International Conference on Self-Adaptive and
  Self-Organizing Systems (SASO). (2012)  211 --216

\bibitem{Percom_Popleteev_2012}
Popleteev, A., Osmani, V., Mayora, O.:
\newblock Investigation of indoor localization with ambient {FM} radio
  stations.
\newblock In: Pervasive Computing and Communications (PerCom), 2012 IEEE
  International Conference on. (2012)  171 --179

\bibitem{Pervasive_Sigg_2013}
Sigg, S., Shi, S., Buesching, F., Ji, Y., Wolf, L.:
\newblock Leveraging rf-channel fluctuation for activity recognition.
\newblock In: Proceedings of the 11th International Conference on Advances in
  Mobile Computing and Multimedia (MoMM2013). (2013)

\bibitem{DeviceFreeRecognition_Tan_2005}
Tan, D., Sun, H., Lu, Y., Lesturgie, M., Chan, H.:
\newblock Passive radar using global system for mobile communication signal:
  theory, implementation and measurements.
\newblock IEE Proceedings - Radar, Sonar and Navigation \textbf{152}(3) (2005)
  116--123

\bibitem{PIMRC_Nishi_2006}
Nishi, M., Takahashi, S., Yoshida, T.:
\newblock Indoor human detection systems using vhf-fm and uhf-tv broadcasting
  waves.
\newblock In: Personal, Indoor and Mobile Radio Communications, 2006 IEEE 17th
  International Symposium on. (2006)  1--5

\bibitem{Pervasive_Yang_2012}
Yang, D., Xue, G., Fang, X., Tang, J.:
\newblock Crowdsourcing to smartphones: incentive mechanism design for mobile
  phone sensing.
\newblock In: Proceedings of the 18th annual international conference on Mobile
  computing and networking. Mobicom '12 (2012)  173--184

\bibitem{Pervasive_Lukovicz_2012-2}
Lukowicz, P., Pentland, A., Ferscha, A.:
\newblock From context awareness to socially aware computing.
\newblock IEEE Pervasive Computing \textbf{11}(1) (2012)  32--41

\bibitem{ActivityRecognition_Aggarwal_2011}
Aggarwal, J., Ryoo, M.:
\newblock Human activity analysis: A review.
\newblock ACM Computing Surveys \textbf{43}(3) (2011)  16:1--16:43

\bibitem{ContextAwareness_Kunze_2007}
Kunze, K., Lukowicz, P.:
\newblock Symbolic object localization through active sampling of acceleration
  and sound signatures.
\newblock In: Proceedings of the 9th International Conference on Ubiquitous
  Computing. (2007)

\bibitem{Pervasive_Chaquet_2013}
Chaquet, J.M., Carmona, E.J., Fern\'{a}Ndez-Caballero, A.:
\newblock A survey of video datasets for human action and activity recognition.
\newblock Comput. Vis. Image Underst. \textbf{117}(6) (2013)  633--659

\bibitem{Pervasive_Adib_2013}
Adib, F., Katabi, D.:
\newblock See through walls with wi-fi.
\newblock In: ACM SIGCOMM'13. (2013)

\bibitem{RFsensing_Pu_2013}
Pu, Q., Gupta, S., Gollakota, S., Patel, S.:
\newblock Whole-home gesture recognition using wireless signals.
\newblock In: The 19th Annual International Conference on Mobile Computing and
  Networking (Mobicom'13). (2013)

\bibitem{DeviceFreeRecognition_Popleteev_2013}
Popleteev, A.:
\newblock Device-free indoor localization using ambient radio systems.
\newblock In: Adjunct Proceedings of the 2013 ACM International Joint
  Conference on Pervasive and Ubiquitous Computing (UbiComp 2013). UbiComp '13
  (2013)

\bibitem{DeviceFreeRecognition_Wagner_2013}
Wagner, B., Timmermann, D.:
\newblock Adaptive clustering for device-free user positioning utilizing
  passive rfid.
\newblock In: Adjunct Proceedings of the 2013 ACM International Joint
  Conference on Pervasive and Ubiquitous Computing (UbiComp 2013). UbiComp '13
  (2013)

\bibitem{DeviceFreeRecognition_Colone_2012}
Colone, F., Falcone, P., Bongioanni, C., Lombardo, P.:
\newblock Wifi-based passive bistatic radar: Data processing schemes and
  experimental results.
\newblock IEEE Transactions on Aerospace and Electronic Systems \textbf{48}(2)
  (2012)  1061--1079

\bibitem{RFSensing_Kassem_2012}
Kassem, N., Kosba, A., Youssef, M.:
\newblock Rf-based vehicle detection and speed estimation.
\newblock In: 75th IEEE Vehicular Technology Conference (VTC Spring). (2012)
  1--5

\bibitem{DeviceFreeRecognition_Sigg_2013}
Sigg, S., Shi, S., Ji, Y.:
\newblock Rf-based device-free recognition of simultaneously conducted
  activities.
\newblock In: Adjunct Proceedings of the 2013 ACM International Joint
  Conference on Pervasive and Ubiquitous Computing (UbiComp 2013). UbiComp '13
  (2013)

\bibitem{RFsensing_Kim_2009}
Kim, Y., Ling, H.:
\newblock Human activity classification based on micro-doppler signatures using
  a support vector machine.
\newblock IEEE Transactions on Geoscience and Remote Sensing \textbf{47}(5)
  (2009)  1328--1337

\bibitem{RFSensing_Ding_2011}
Ding, Y., Banitalebi, B., Miyaki, T., Beigl, M.:
\newblock Rftraffic: Passive traffic awareness based on emitted rf noise from
  the vehicles.
\newblock In: ITS Telecommunications (ITST), 2011 11th International Conference
  on. (2011)  393 --398

\bibitem{RFsensing_Xu_2013}
Xu, C., Firner, B., Moore, R.S., Zhang, Y., Trappe, W., Howard, R., Zhang, F.,
  An, N.:
\newblock Scpl: Indoor device-free multi-subject counting and localization
  using radio signal strength.
\newblock In: The 12th ACM/IEEE Conference on Information Processing in Sensor
  Networks (ACM/IEEE IPSN). (2013)

\bibitem{Pervasive_Scholz_2013}
Scholz, M., Riedel, T., Hock, M., Beigl, M.:
\newblock Device-free and device-bound activity recognition using radio signal
  strength.
\newblock In: Proceedings of the 4th Augmented Human International Conference
  in cooperation with ACM SIGCHI. (2013)

\bibitem{Wardriving_Ildis_2013-PerCom}
Ildis, O., Ofir, Y., Feinstein, R.:
\newblock Wardriving from your pocket (2013)

\bibitem{Statistics_Kononenko_1991}
Kononenko, I., Bratko, I.:
\newblock {Information-based evaluation criterion for classifier's
  performance}.
\newblock Machine Learning \textbf{6}(1) (1991)  67--80

\bibitem{Statistics_Spackman_1989}
Spackman, K.A.:
\newblock Signal detection theory: valuable tools for evaluating inductive
  learning.
\newblock In: 6th international workshop on Machine learning. (1989)  160--163

\bibitem{Context_Dupuy_2006}
Dupuy, C., Torre, A.
\newblock In: Local Clusters, trust, confidence and proximity, Clusters and
  Globalisation: The development of urban and regional economies. (2006)  pp.
  175--195

\bibitem{mayrhofer2008spontaneous}
Mayrhofer, R., Gellersen, H.:
\newblock {Spontaneous mobile device authentication based on sensor data}.
\newblock information security technical report \textbf{13}(3) (2008)  136--150

\bibitem{Cryptography_Bichler_2007-2}
Bichler, D., Stromberg, G., Huemer, M., Loew, M.:
\newblock Key generation based on acceleration data of shaking processes.
\newblock In Krumm, J., ed.: Proceedings of the 9th International Conference on
  Ubiquitous Computing. (2007)

\bibitem{ContextAwareness_Holmquist_2001}
Holmquist, L.E., Mattern, F., Schiele, B., , Alahuhta, P., Beigl, M.,
  Gellersen, H.W.:
\newblock Smart-its friends: A technique for users to easily establish
  connections between smart artefacts.
\newblock In: Proceedings of the 3rd International Conference on Ubiquitous
  Computing. (2001)

\bibitem{Cryptography_Varshavsky_2007}
Varshavsky, A., Scannell, A., LaMarca, A., de~Lara, E.:
\newblock Amigo: Proximity-based authentication of mobile devices.
\newblock International Journal of Security and Networks (2009)

\bibitem{5836}
Gellersen, H.W., Kortuem, G., Schmidt, A., Beigl, M.:
\newblock Physical prototyping with smart-its.
\newblock IEEE Pervasive computing \textbf{4}(1536-1268) (2004)  10--18

\bibitem{mayrhofer2007shake}
Mayrhofer, R., Gellersen, H.:
\newblock {Shake well before use: Authentication based on accelerometer data}.
\newblock Pervasive Computing (2007)  144--161

\bibitem{mayrhofer2007candidate}
Mayrhofer, R.:
\newblock {The Candidate Key Protocol for Generating Secret Shared Keys from
  Similar Sensor Data Streams}.
\newblock Security and Privacy in Ad-hoc and Sensor Networks (2007)  1--15

\bibitem{Cryptography_Bichler_2007}
Bichler, D., Stromberg, G., Huemer, M.:
\newblock Innovative key generation approach to encrypt wireless communication
  in personal area networks.
\newblock In: Proceedings of the 50th International Global Communications
  Conference. (2007)

\bibitem{Cryptography_Hershey_1995}
Hershey, J., Hassan, A., Yarlagadda, R.:
\newblock Unconventional cryptographic keying variable management.
\newblock IEEE Transactions on Communications \textbf{43} (1995)  3--6

\bibitem{Cryptography_Smith_2004}
Smith, G.:
\newblock A direct derivation of a single-antenna reciprocity relation for the
  time domain.
\newblock IEEE Transactions on Antennas and Propagation \textbf{52} (2004)
  1568--1577

\bibitem{Cryptography_Madiseh_2008}
Madiseh, M.G., McGuire, M.L., Neville, S.S., Cai, L., Horie, M.:
\newblock Secret key generation and agreement in uwb communication channels.
\newblock In: Proceedings of the 51st International Global Communications
  Conference (Globecom). (2008)

\bibitem{Cryptography_BenHamida_2009}
Hamida, S.T.B., Pierrot, J.B., Castelluccia, C.:
\newblock An adaptive quantization algorithm for secret key generation using
  radio channel measurements.
\newblock In: Proceedings of the 3rd International Conference on New
  Technologies, Mobility and Security. (2009)

\bibitem{tuyls2007security}
Tuyls, P., Skoric, B., Kevenaar, T.:
\newblock {Security with Noisy Data}.
\newblock Springer-Verlag (2007)

\bibitem{li2006robust}
Li, Q., Chang, E.C.:
\newblock {Robust, short and sensitive authentication tags using secure
  sketch}.
\newblock In: Proceedings of the 8th workshop on Multimedia and security, ACM
  (2006)  56--61

\bibitem{Dodis04fuzzyextractors}
Dodis, Y., Ostrovsky, R., Reyzin, L., Smith, A.:
\newblock Fuzzy extractors: How to generate strong keys from biometrics and
  other noisy data.
\newblock EUROCRYPT 2004 (2004)  79--100

\bibitem{miao2009biometrics}
Miao, F., Jiang, L., Li, Y., Zhang, Y.T.:
\newblock {Biometrics based novel key distribution solution for body sensor
  networks}.
\newblock In: Engineering in Medicine and Biology Society, 2009. EMBC 2009.
  Annual International Conference of the IEEE, IEEE (2009)  2458--2461

\bibitem{Juels02fuzzyvault}
Juels, A., Sudan, M.:
\newblock {A Fuzzy Vault Scheme}.
\newblock Proceedings of IEEE Internation Symposium on Information Theory
  (2002)  408

\bibitem{dodis2006robust}
Dodis, Y., Katz, J., Reyzin, L., Smith, A.:
\newblock {Robust fuzzy extractors and authenticated key agreement from close
  secrets}.
\newblock Advances in Cryptology-CRYPTO 2006 (2006)  232--250

\bibitem{cano2005review}
Cano, P., Batlle, E., Kalker, T., Haitsma, J.:
\newblock {A Review of Algorithms for Audio Fingerprinting}.
\newblock The Journal of VLSI Signal Processing \textbf{41}(3) (2005)  271--284

\bibitem{AudioFingerprinting_Baluja_2008}
Baluja, S., Covell, M.:
\newblock Waveprint: Efficient wavelet-based audio fingerprinting.
\newblock Pattern Recognition \textbf{41}(11) (2008)

\bibitem{AudioFingerprinting_Ghouti_2006}
Ghouti, L., Bouridane, A.:
\newblock A robust perceptual audio hashing using balanced multiwavelets.
\newblock In: Proceedings of the 5th IEEE International Conference on
  Acoustics, Speech, and Signal Processing. (2006)

\bibitem{AudioFingerprinting_Sukittanon_2002}
Sukittanon, S., Atlas, L.:
\newblock Modulation frequency features for audio fingerprinting.
\newblock In: Proceedings of the 2nd IEEE International Conference on
  Acoustics, Speech, and Signal Processing. (2002)

\bibitem{Haitsma2003highlyrobust}
Haitsma, J., Kalker, T.:
\newblock {A Highly Robust Audio Fingerprinting System}.
\newblock Journal of New Music Research \textbf{32}(2) (2003)  211--221

\bibitem{AudioFingerprinting_Burges_2005}
Burges, C., Plastina, D., Platt, J., Renshaw, E., Malvar, H.:
\newblock Using audio fingerprinting for duplicate detection and thumbnail
  generation.
\newblock In: Acoustics, Speech, and Signal Processing, 2005. Proceedings.
  (ICASSP '05). IEEE International Conference on. Volume~3. (2005)
  iii/9--iii12 Vol. 3

\bibitem{AudioFingerprinting_Bellettini_2010}
Bellettini, C., Mazzini, G.:
\newblock A framework for robust audio fingerprinting.
\newblock Journal of Communications \textbf{5}(5) (2010)

\bibitem{AudioFingerprinting_Ghias_1995}
Ghias, A., Logan, J., Chamberlin, D., Smith, B.C.:
\newblock Query by humming.
\newblock In: Proceedings of the ACM Multimedia. (1995)

\bibitem{AudioFingerprinting_Parsons_1975}
Parsons, D.:
\newblock The directory of tunes and musical themes.
\newblock Cambridge University press (1975)

\bibitem{AlgorithmsStringMatching_BeezaYates_1992}
Baeza-Yates, R.A., perleberg~perleberg {perleberg perleberg}, C.H.:
\newblock Fast and practical approximate string matching.
\newblock Third annual symposium on combinatorial pattern matching (1992)

\bibitem{AudioFingerprinting_McNab_1996}
McNab, R.J., Smith, L.A., Witten, I.H., Henderson, C.L., Cunningham, S.J.:
\newblock Towards the digital music library: tune retrieval from acoustic
  iinput.
\newblock Proceedings of the ACM (1996)

\bibitem{AudioFingerprinting_Prechelt_2001}
Prechelt, L., Typke, R.:
\newblock An interface for melody input.
\newblock ACM Transactions on Computer Human Interactions \textbf{8} (2001)

\bibitem{AudioFingerprinting_Chai_2002}
Chai, W., Vercoe, B.:
\newblock Melody retrieval on the web.
\newblock In: Proceedings of the ACM/SPIE conference on Multimedia Computing
  and Networking. (2002)

\bibitem{AudioFingerprinting_Shifrin_2002}
Shifrin, L., Pardo, B., Birmingham, W.:
\newblock Hmm-based musical query retrieval.
\newblock In: Proceedings of the joint conference on digital libraries. (2002)

\bibitem{AlgorithmsHMM_Rabiner_1989}
Rabiner, L.:
\newblock A tutorial on hidden markov models and selected accplications in
  speech recognition.
\newblock Proceedings of the IEEE \textbf{77}(2) (1989)

\bibitem{AudioFingerprinting_Zhu_2003}
Zhu, Y., Shasha, D.:
\newblock Warping indexes with envelope transforms for query by humming.
\newblock In: Proceedings of the ACM SIGMOD International Conference on
  Management of Data. (2003)

\bibitem{AudioFingerprinting_Haitsma_2001}
Haitsma, J., Kalker, T.:
\newblock Robust audio hashing for content identification.
\newblock In: In Content-Based Multimedia Indexing (CBMI. (2001)

\bibitem{AudioFingerprinting_Lebosse_2007}
Leboss{\'e}, J., Brun, L., Pailles, J.C.:
\newblock A robust audio fingerprint's based identification method.
\newblock In: Proceedings of the 3rd Iberian conference on Pattern Recognition
  and Image Analysis, Part I. IbPRIA '07, Berlin, Heidelberg, Springer-Verlag
  (2007)  185--192

\bibitem{AudioFingerprinting_Burges_2003}
Burges, C., Platt, J., Jana, S.:
\newblock Distortion discriminant analysis for audio fingerprinting.
\newblock Speech and Audio Processing, IEEE Transactions on \textbf{11}(3)
  (2003)  165--174

\bibitem{AudioFingerprinting_Herre_2001}
Herre, J., Allamanche, E., Hellmuth, O.:
\newblock Robust matching of audio signals using spectral flatness features.
\newblock In: Applications of Signal Processing to Audio and Acoustics, 2001
  IEEE Workshop on the. (2001)  127--130

\bibitem{AudioFingerprinting_Yang_2001}
Yang, C.:
\newblock Macs: music audio characteristic sequence indexing for similarity
  retrieval.
\newblock In: Applications of Signal Processing to Audio and Acoustics, 2001
  IEEE Workshop on the. (2001)  123--126

\bibitem{AudioFingerprinting_Yang_2002}
Yang, C.:
\newblock Efficient acoustic index for music retrieval with various degrees of
  similarity.
\newblock In: Proceedings of the tenth ACM international conference on
  Multimedia. MULTIMEDIA '02, New York, NY, USA, ACM (2002)  584--591

\bibitem{AudioFingerprinting_Wang_2006}
Wang, A.:
\newblock The shazam music recognition service.
\newblock Communications of the ACM \textbf{49}(8) (2006)

\bibitem{AudioFingerprinting_Wan_2003}
Wang, A.:
\newblock An industrial strength audio search algorithm.
\newblock In: Proceedings of the International Conference on Music Information
  Retrieval (ISMIR). (2003)

\bibitem{2108}
Hao, F.:
\newblock On using fuzzy data in security mechanisms.
\newblock PhD thesis, Queens College, Cambridge (2007)

\bibitem{2109}
Ibarrola, A.C., Chavez, E.:
\newblock A robust entropy-based audio-fingerprint.
\newblock In: Proceedings of the 2006 International Conference on Multimedia
  and Expo (ICME 2006). (2006)

\bibitem{schneider1996applied}
Schneider, B.:
\newblock {Applied Cryptography: Protocols, Algorithms, and Source Code in C}.
  2 edn.
\newblock John Wiley and Sons, Inc. (1996)

\bibitem{Reed60polynomial}
Reed, I., Solomon, G.:
\newblock {Polynomial codes over certain finite fields}.
\newblock Journal of the Society for Industrial and Applied Mathematics (1960)
  300--304

\bibitem{Juels99fuzzycommitment}
Juels, A., Wattenberg, M.:
\newblock {A Fuzzy Commitment Scheme}.
\newblock Sixth ACM Conference on Computer and Communications Security (1999)
  28--36

\bibitem{fips2008180}
{National Institute of Standards and Technology}:
\newblock {180-3, Secure Hash Standard (SHS)}.
\newblock Federal Information Processing Standards Publications (FIPS PUBS)
  (2008)

\bibitem{fips2001197}
{National Institute of Standards and Technology}:
\newblock {197, Advanced Encryption Standard (AES)}.
\newblock Federal Information Processing Standards Publications (FIPS PUBS)
  (2001)

\bibitem{rfc5905}
Mills, D., Martin, J., Burbank, J., Kasch, W.:
\newblock {Network Time Protocol Version 4: Protocol and Algorithms
  Specification}.
\newblock RFC 5905 (Proposed Standard) (2010)

\bibitem{AlgorithmsNTP_Mills_1995}
Mills, D.L.:
\newblock Improved algorithms for synchronising computer network clocks.
\newblock IEEE/ACM Transactions on Networking \textbf{3}(3) (1995)

\bibitem{Meier_Synchronisation_1998}
Meier, S., Weibel, H., Weber, K.:
\newblock Ieee 1588 syntonization and synchronization functions completely
  realized in hardware.
\newblock In: International IEEE Symposium on Precision Clock Synchronization
  for Measurement, Control and Communication (ISPCS 2008). (2008)

\bibitem{AlgorithmsNTP_Mills_1994}
Mills, D.L.:
\newblock Precision synchronisation of computer network clocks.
\newblock ACM Computer Communication Review \textbf{24}(2) (1994)

\bibitem{gstreamer2010doc}
Freedesktop.org:
\newblock {GStreamer Documentation}. (2010)

\bibitem{scheirer2007cracking}
Scheirer, W.J., Boult, T.E.:
\newblock {Cracking fuzzy vaults and biometric encryption}.
\newblock In: Proceedings of Biometrics Symposium, Baltimore, USA. (2007)

\bibitem{ignatenko2010privacy}
Ignatenko, T., Willems, F.M.J.:
\newblock {Information Leakage in Fuzzy Commitment Schemes}.
\newblock In: IEEE Transactions on Information Forensics and Security.
  Volume~5. (2010)  337

\bibitem{fenzi2004linux}
Fenzi, K., Wreski, D.:
\newblock Linux Security HOWTO. (2004)

\bibitem{statistical_Brown_0000}
Brown, R.G.:
\newblock Dieharder: A random number test suite.
\newblock http://www.phy.duke.edu/$\sim$rgb/General/dieharder.php (2011)

\bibitem{Statistical_Kuiper_1962}
Kuiper, N.:
\newblock Tests concerning random points on a circle.
\newblock In: Proceedings of the Koinklijke Nederlandse Akademie van
  Wetenschappen. Volume Series a 63. (1962)  38--47

\bibitem{Statistical_stephens_1965}
Stephens, M.:
\newblock The goodness-of-fit statistic $v\_n$: Distribution and significance
  points.
\newblock Biometrika \textbf{52} (1965)

\bibitem{wavelets}
Baluja, S., Covell, M.:
\newblock Content fingerprinting using wavelets.
\newblock In: Proceedings of the Conference of Visual Media Production, London,
  UK (2006)

\bibitem{Audio_fingerprint_review_Cano_2005}
Cano, P., Batlle, E., Kalker, T., Haitsma, J.:
\newblock A review of audio fingerprinting.
\newblock Journal of VLSI Signal Processing Systems \textbf{41 Issue 3} (2005)

\bibitem{Audio_fingerprint_review_Chandrasekhar_2011}
Chandrasekhar, V., Sharifi, M., Ross, D.:
\newblock Survey and evaluation of audio fingerprinting schemes for mobile
  audio search.
\newblock In: International Symposium on Music and Information Retrieval
  (ISMIR), Miami, Florida (2011)

\bibitem{McCune_Cryptography_2005}
McCune, J.M., Perrig, A., Reiter, M.K.:
\newblock Seeing-is-believing: Using camera phones for human-verifiable
  authentication.
\newblock In: Proceedings of the 2005 IEEE Symposium on Security and Privacy.
  (2005)

\bibitem{Goodrich_Cryptography_2006}
Goodrich, M.T., Sirivianos, M., Solis, J., Tsudik, G., Uzun, E.:
\newblock Loud and clear: Human-verifiable authentication based on audio.
\newblock In: Proceedings of the 26th IEEE International Conference on
  Distributed Computing Systems. (2006)

\bibitem{Cryptography_Sigg_2011-5}
Sigg, S., Schuermann, D., Ji, Y.:
\newblock Pintext: A framework for secure communication based on context.
\newblock In: Proceedings of the Eighth Annual International ICST Conference on
  Mobile and Ubiquitous Systems:Computing, Networking and Services (MobiQuitous
  2011). (2011)

\bibitem{SmithWatermanAlg}
Smith, T.F., Waterman, M.S.:
\newblock Identification of common molecular subsequences.
\newblock Journal of molecular biology \textbf{147}(1) (1981)  195--197

\bibitem{5918}
Mudumbai, R., Brown, D.R., Madhow, U., Poor, H.V.:
\newblock Distributed transmit beamforming: Challenges and recent progress.
\newblock IEEE Communications Magazine (2009)  102--110

\bibitem{4025}
Sigg, S., Masri, R., Ristau, J., Beigl, M.:
\newblock Limitations, performance and instrumentation of closed-loop feedback
  based distributed adaptive transmit beamforming in wsns.
\newblock In: Fifth International Conference on Intelligent Sensors, Sensor
  Networks and Information Processing - Symposium on Theoretical and Practical
  Aspects of Large-scale Wireless Sensor Networks. (2009)

\bibitem{DeviceFreeRecognition_Shi_2013}
Shi, S., Sigg, S., Ji, Y.:
\newblock Joint localisation and activity recognition from ambient fm broadcast
  signals.
\newblock In: Adjunct Proceedings of the 2013 ACM International Joint
  Conference on Pervasive and Ubiquitous Computing (UbiComp 2013). UbiComp '13
  (2013)

\bibitem{DeviceFreeRecognition_Hong_2013}
Hong, J., Ohtsuki, T.:
\newblock Ambient intelligence sensing using array sensor: Device-free radio
  based approach.
\newblock In: Adjunct Proceedings of the 2013 ACM International Joint
  Conference on Pervasive and Ubiquitous Computing (UbiComp 2013). UbiComp '13
  (2013)

\bibitem{AttentionMonitoring_Wickens_1984}
Wickens, C.:
\newblock Processing resources in attention.
\newblock Academic Press (1984)

\bibitem{AttentionMonitoring_Gollan_2011}
Gollan, B., Wally, B., Ferscha, A.:
\newblock Automatic attention estimation in an interactive system based on
  behaviour analysis.
\newblock In: Proceedings of the 15th Portuguese Conference on Artificial
  Intelligence (EPIA2011). (2011)

\bibitem{RFSensing_Kellog_2014}
:
\newblock Bringing gesture recognition to all devices.
\newblock In: Presented as part of the 11th USENIX Symposium on Networked
  Systems Design and Implementation, Berkeley, CA, USENIX (2014)

\bibitem{RFSensing_Adib_2014}
Adib, F., Kabelac, Z., Katabi, D., Miller, R.C.:
\newblock 3d tracking via body radio reflections.
\newblock In: Usenix NSDI. Volume~14. (2014)

\bibitem{Muneeba_2017_Geospatial}
Raja, M., Exler, A., Hemminki, S., Konomi, S., Sigg, S., Inoue, S.:
\newblock Towards pervasive geospatial affect perception.
\newblock Springer GeoInformatica (2017)

\bibitem{Raja_2016_CoSDEO}
Raja, M., Sigg, S.:
\newblock Applicability of rf-based methods for emotion recognition: A survey.
\newblock In: 2016 IEEE International Conference on Pervasive Computing and
  Communication Workshops (PerCom Workshops). (2016)  1--6

\bibitem{schurmann2017bandana}
Sch{\"u}rmann, D., Br{\"u}sch, A., Sigg, S., Wolf, L.:
\newblock Bandana—body area network device-to-device authentication using
  natural gait.
\newblock In: Pervasive Computing and Communications (PerCom), 2017 IEEE
  International Conference on, IEEE (2017)  190--196

\end{thebibliography}

\end{document}